\documentclass[a4paper,openany]{amsbook}
\pdfoutput=1
\usepackage{libertine}
\usepackage[libertine]{newtxmath}
\usepackage{akenthmslight}
\usepackage{aken}

\usepackage{dirhott}
\usepackage{mathpartir}

\declaretheorem[style=mythm,sibling=definition,]{thesis}
\crefname{thesis}{thesis}{theses}
\Crefname{thesis}{Thesis}{Theses}

\renewcommand{\todoi}[1]{}
\newcommand{\todo}[1]{}

\usepackage[english]{babel}
\usepackage{a4wide}
\usepackage{graphicx}
\usepackage{wrapfig}

\usepackage{chngcntr}
\counterwithin{section}{chapter}

\newcounter{equationbis}

\newcommand{\commentout}[1]{}

\newcommand{\pullbackcorner}[1][dr]{\save*!/#1-1.2pc/#1:(-1,1)@^{|-}\restore}

\newcommand{\inference}[3]{\inferenceright{#1}{#2}{\text{#3}}}
\newcommand{\binference}[3]{\binferenceright{#1}{#2}{\text{#3}}}
\newcommand{\vmess}[1]{\begin{array}{c} #1 \end{array}}
\newcommand{\interpinference}[3]{\interp{\vmess{\inference{#1}{#2}{#3}}}}

\renewcommand{\loch}{{\llcorner\hspace{-.4ex}\lrcorner}}

\newcommand{\pparen}[1]{{\llparenthesis #1 \rrparenthesis}}

\newcommand{\Ty}{\mathrm{Ty}}
\newcommand{\Tm}{\mathrm{Tm}}

\newcommand{\Arr}{\name{Arr}}
\newcommand{\Empty}{\name{Empty}}

\newcommand{\Disc}{\name{Disc}}
\newcommand{\uniPsh}{\cat U^{\name{Psh}}}
\newcommand{\uniNDD}{\cat U^{\name{NDD}}}
\newcommand{\uniDD}{\cat U^{\name{DD}}}
\newcommand{\PropD}{\Prop^{\name{D}}}

\newcommand{\unidepth}[2]{\cat U_{#1}^{#2}}

\newcommand{\IB}{\mathbb{B}}
\newcommand{\IP}{\mathbb{P}}
\newcommand{\IE}{\mathbb{E}}
\newcommand{\XX}{\mathbb{X}}
	\newcommand{\ctxbrid}[1]{{#1} : \IB}
	\newcommand{\ctxpath}[1]{{#1} : \IP}
	\newcommand{\ctxedge}[1]{{#1} : \IE}
	\newcommand{\ctxline}[1]{{#1} : \XX}

\newcommand{\IX}{\mathbb{I}}
\newcommand{\IF}{\mathbb{F}}
\newcommand{\IJud}{\mathbb{J}}
\newcommand{\Idim}[1]{{\llparenthesis #1 \rrparenthesis}}

\newcommand{\cwfpair}[2]{{#1, #2}}

\newcommand{\ap}{\name{ap}}
\newcommand{\J}{\name{J}}
\newcommand{\tycode}[1]{\ulcorner #1 \urcorner}
\newcommand{\dtycode}[1]{\underline \ulcorner #1 \underline \urcorner}
\newcommand{\dEl}{\underline \El}

\newcommand{\tycodeDD}[1]{\ulcorner #1 \urcorner^{\name{DD}}}
\newcommand{\ElDD}{\El^{\name{DD}}}
\newcommand{\tycodeD}[1]{\ulcorner #1 \urcorner^{\name{D}}}
\newcommand{\ElD}{\El^{\name{D}}}

\newcommand{\dtype}{\,\name{dtype}}

\newcommand{\ftype}{\,\name{ftype}}
\newcommand{\prop}{\,\name{prop}}
\newcommand{\eqrel}{\,\name{eqrel}\,}

\newcommand{\yoneda}{\mathbf{y}}

\newcommand{\homclass}[1]{\mathscr{#1}}

\newcommand{\sheq}{\name{SE}}
\newcommand{\fpsh}[1]{{#1}^\dagger}%
\newcommand{\fpshadj}[1]{\underline{\alpha}_{#1}}

\newcommand{\lpsh}[1]{\widehat{#1}}
\newcommand{\rpsh}[1]{{#1}_\ddagger}
\newcommand{\rpshadj}[1]{\underline{\beta}_{#1}}

\newcommand{\ftrtm}[2]{{}^{#1} #2}
\newcommand{\ftrvar}[2]{\ftrtm{#1}{\var{#2}}}

\newcommand{\shp}{{\rbag}}%
\newcommand{\quotshp}{\mathrlap{\shp}{\circ}}
	\newcommand{\inquotshp}{\varsigma_\circ}
	\newcommand{\hatinquotshp}{\inquotshp}
\newcommand{\tmshp}[1]{\hatinquotshp(#1)}

\newcommand{\coshp}{\P}

\newcommand{\fibrepl}{\cat R}
\newcommand{\infibrepl}{\name{in}\cat R}
\newcommand{\fibcorepl}{\cat Q}
\newcommand{\outfibcorepl}{\name{out}\cat Q}

\newcommand{\catTop}{\name{Top}}
\newcommand{\cohpi}{{\sqcap}}
\newcommand{\cohdisc}{{\vartriangle}}
\newcommand{\cohfget}{{\sqcup}}
\newcommand{\cohcodisc}{{\triangledown}}
\newcommand{\cohpaths}{{\boxminus}}
\newcommand{\shirr}{\%}

\newcommand{\cohshirr}{{\Diamond}}

\newcommand{\rellist}[2]{{\left({#1}\middle|{#2}\right)}}
\newcommand{\reshlist}[2]{{\left\langle{#1}\middle|{#2}\right\rangle}}
\newcommand{\rlookup}[2]{{{#1} \cdot {#2}}}

\newcommand{\eqty}{{=}}

\renewcommand{\isvar}{\sharp}
\renewcommand{\invar}{\circ}

\newcommand{\sys}[1]{\paren{#1}}
\newcommand{\sysclauseb}[2]{#1\,?\,#2}
\newcommand{\sysclause}[2]{\parclr{#1}\,?\,#2}

\newcommand{\Glue}{\name{Glue}}
\newcommand{\Gluesys}[2]{\Glue \accol{#1 \leftarrow \sys{#2}}}
\newcommand{\Gluesysclause}[3]{#1\,?\,{#2, \ptwclr{#3}}}
\newcommand{\Gluesysclauseb}[3]{#1\,?\,{#2, #3}}
\newcommand{\glue}{\name{glue}}
\newcommand{\gluesys}[2]{\glue \accol{#1 \mapsfrom \sys{#2}}}
\newcommand{\unglue}{\name{unglue}}
\newcommand{\ungluesys}[1]{\unglue\,\sys{#1}}
\newcommand{\Gluetp}[4]{\Gluesys{#2}{\Gluesysclause{#1}{#3}{#4}}}
\newcommand{\gluetm}[3]{\gluesys{#2}{\sysclause{#1}{#3}}}
\newcommand{\ungluetm}[3]{\ungluesys{\sysclause{#1}{\ptwclr{#2}}}\,#3}

\newcommand{\Weld}{\name{Weld}}
\newcommand{\Weldsys}[2]{\Weld \accol{#1 \to \sys{#2}}}
\newcommand{\Weldsysclause}[3]{#1\,?\,{#2, \ptwclr{#3}}}
\newcommand{\Weldsysclauseb}[3]{#1\,?\,{#2, #3}}
\newcommand{\weld}{\name{weld}}
\newcommand{\weldsys}[1]{\weld\,\sys{#1}}

\newcommand{\dweld}{\underline{\weld}}
\newcommand{\dunweld}{\underline{\name{unweld}}}
\newcommand{\Weldtp}[4]{\Weldsys{#2}{\Weldsysclause{#1}{#3}{#4}}}
\newcommand{\weldtm}[3]{\weldsys{\sysclause{#1}{\ptwclr{#2}}}\,#3}

\newcommand{\parclr}[1]{\textcolor{purple}{#1}}
\newcommand{\ptwclr}[1]{\textcolor{cyan}{#1}}
\newcommand{\varclr}[1]{\textcolor{orange}{#1}}
\newcommand{\leftflat}[1]{\sharp \setminus {#1}}
\newcommand{\leftsharp}[1]{\coshp \setminus {#1}}
\newcommand{\judty}[1]{\sez \parclr{ #1 } \type}

\newcommand{\judtm}[2]{\sez #1 : \parclr{ #2 }}
\newcommand{\judtmeq}[3]{\sez #1 \jeq #2 : \parclr{ #3 }}

\newcommand{\idmod}{\mathbf{id}}

\newcommand{\parmod}{\mathbf{par}}

\newcommand{\shirrmod}{\mathbf{shi}}
\newcommand{\strmod}{\mathbf{str}}
\newcommand{\hocmod}{\mathbf{hoc}}

\newcommand{\ctxctu}[2]{#1 : \parclr{ #2 }}
\newcommand{\ctxpar}[2]{\parclr{ #1 }^\sharp : \parclr{ #2 }}
\newcommand{\ctxptw}[2]{\ptwclr{ #1 }^\coshp : \parclr{ #2 }}
\renewcommand{\ctxvar}[3]{\varclr{ #2 }^{#1} : \parclr{ #3 }}
\newcommand{\reshsep}{\shortmid}
\newcommand{\ctxresh}[3]{#1 \reshsep #2 : #3}

\newcommand{\tyresh}[2]{#1 \reshsep #2}

\newcommand{\ctxface}[1]{\parclr{#1}}
\newcommand{\Pipar}{\forall}
\newcommand{\Piptw}{\Pi^\coshp}
\newcommand{\Pivar}[1]{\Pi^{#1}}
\newcommand{\Sigmapar}{\exists}
\newcommand{\Sigmaptw}{\Sigma^\coshp}
\newcommand{\Sigmavar}[1]{\Sigma^{#1}}
\newcommand{\prodpar}[2]{\Pipar\paren{#1 : #2}}
\newcommand{\prodctu}[2]{\Pi\paren{#1 : #2}}
\newcommand{\prodptw}[2]{\Piptw\paren{\ptwclr{#1} : #2}}
\renewcommand{\prodvar}[3]{\Pivar{#1}\paren{\varclr{#2} : #3}}
\newcommand{\sumpar}[2]{\Sigmapar\paren{#1 : #2}}
\newcommand{\sumctu}[2]{\Sigma\paren{#1 : #2}}
\newcommand{\sumptw}[2]{\Sigmaptw\paren{\ptwclr{#1} : #2}}
\renewcommand{\sumvar}[3]{\Sigmavar{#1}\paren{\varclr{#2} : #3}}

\newcommand{\lamannotpar}[2]{\lambda (\parclr{#1}^\sharp : \parclr{#2})}

\newcommand{\lamannotvar}[3]{\lambda (\varclr{#2}^{#1} : \parclr{#3})}

\newcommand{\appar}[2]{#1\,\parclr{#2}^\sharp}
\newcommand{\apvar}[3]{#2\,\varclr{#3}^{#1}}
\newcommand{\apresh}[3]{{#2 \angles{#1} #3}}

\renewcommand{\pairvar}[3]{(\varclr{#2}^{#1}, #3)}
\newcommand{\pairpar}[2]{(\parclr{#1}^\sharp, #2)}

\newcommand{\pairresh}[3]{(#1 \reshsep #2, #3)}
\newcommand{\fstptw}[1]{\name{fst}^\coshp\,\parclr{#1}}
\newcommand{\sndptw}[1]{\name{snd}^\coshp\,#1}

\newcommand{\idpr}[2]{#1 \doteq #2}

\newcommand{\degax}{\name{degax}}
\newcommand{\degaxof}[1]{\degax\,\parclr{#1}}

\newcommand{\Nat}{\name{Nat}}
\renewcommand{\Size}{\name{Size}}
	\newcommand{\szero}{0_{\name{S}}}
	\newcommand{\ssuc}{{\uparrow}}
	
	\newcommand{\smax}[2]{#1 \sqcup #2}
	\newcommand{\sfix}{\name{fix}}
	\newcommand{\sfill}{\name{fill}}
	\newcommand{\sfillsys}[1]{\name{fill}\sys{#1}}
	\newcommand{\sfillsyssharp}[1]{\name{fill}_\sharp\sys{#1}}
	\newcommand{\sfillsysclause}[2]{#1\,?\,#2}
\newcommand{\leqfill}{\name{fill}_\leq}
\newcommand{\leqfillsys}[1]{\leqfill\sys{#1}}
\newcommand{\leqfillsysclause}[2]{\sysclause{#1}{#2}}

\renewcommand{\suc}{\name{s}}

\newcommand{\idtpresh}[4]{{#3 =_{\tyresh{#1}{#2}} #4}}

\newcommand{\forsub}[1]{#1 !}

\newcommand{\fix}{\name{fix}}

\newcommand{\cubecat}{\name{Cube}}

\newcommand{\bpcubecat}{\name{BPCube}}
	\newcommand{\facewkn}[1]{#1 / \novar}
	
	\newcommand{\bpdisc}{\widehat{\bpcubecat}_\Disc}
		\newcommand{\DTy}{\Ty^\Disc}
	
	\newcommand{\ddisc}[1]{\widehat{\dcubecat{#1}}_\Disc}
	
	\newcommand{\pshfib}[1]{\widehat{#1}_{\name{Fib}}}
\newcommand{\dcubecat}[1]{\cubecat_{#1}}
\newcommand{\pointcat}{{\name{Point}}}
\newcommand{\reshufflecat}{\name{Reshuffle}}
\newcommand{\reshufflecatll}{\name{Reshuffle}^{\circ \circ \bullet}}
\newcommand{\reshufflecatlr}{\name{Reshuffle}^{\circ \bullet \circ}}
\newcommand{\reshufflecatrr}{\name{Reshuffle}^{\bullet \circ \circ}}
\newcommand{\reshufflecatllr}{\name{Reshuffle}^{\circ \circ \bullet \circ}}
\newcommand{\catW}{{\cat W}}
\newcommand{\catV}{{\cat V}}

\newcommand{\RGcat}{{\name{RG}}}
	\newcommand{\RGdisc}{\widehat{\RGcat}_\Disc}

\newcommand{\novar}{\oslash}

\newcommand{\PSub}[2]{#1 \Rrightarrow #2}
\newcommand{\DSub}[2]{#1 \Rightarrow #2}
\newcommand{\Dsez}{\mathrel{\rhd}}
\newcommand{\sub}[1]{\brac{#1}}
\newcommand{\dsub}[1]{{\left[ #1 \right\rangle}}
\newcommand{\psub}[1]{\angles{#1}}
\newcommand{\ssub}[1]{\left\{ #1 \right \}}
\newcommand{\dlambda}{\underline \lambda}
\newcommand{\dap}{\underline \ap}

\newcommand{\emptysub}{\bullet}

\newcommand{\textdef}[1]{\textbf{#1}}

\newcommand{\var}[1]{{\mathbf{#1}}}
\newcommand{\wknvar}[1]{\wkn{\var{#1}}}
\newcommand{\wkn}[1]{\pi^{#1}}
\newcommand{\subext}{{+}}
\newcommand{\fresh}{\mathbf{fr}}

\newcommand{\fst}{\name{fst}}
\newcommand{\snd}{\name{snd}}

\newcommand{\textand}{\,\text{and}\,}

\newcommand{\thetitle}{Presheaf Models of Relational Modalities in Dependent Type Theory}
\newcommand{\theauthor}{Andreas Nuyts}
\newcommand{\theinstitution}{KU Leuven}
\newcommand{\thedepartment}{imec-DistriNet \\ Dept. of Computer Science}

\begin{document}
	\addtolength{\voffset}{-.5in}

\title{\thetitle}
\date{\today}
\author{\theauthor}

\begin{titlepage}
	\centering
	{\scshape\LARGE \theinstitution \par}
	{\scshape\Large \thedepartment\par}
	\vspace{1cm}
	{\scshape\Large Technical Report\par}
	\vspace{1.5cm}
	{\huge\bfseries \thetitle \par}
	\vspace{2cm}
	{\Large\itshape \theauthor \par}
	\vfill

	\vfill

	{\large \today\par}
\end{titlepage}

\setcounter{tocdepth}{2}
\tableofcontents

\chapter*{Introduction}
This report is an extension of \cite{bpcubicalsets}. The purpose of this text is to prove all technical aspects of our model for dependent type theory with parametric quantifiers \cite{paramdtt} and with degrees of relatedness \cite{reldtt}.

\section*{Overview of the report}
In \textbf{\cref{part:prerequisites}}, we review and develop some important prerequisites for modelling modal dependent type theory in presheaf categories.
\begin{itemize}
	\item In \textbf{\cref{ch:psh}}, we review the main concepts of categories with families \cite{dybjer-cwf}, and the standard presheaf model of dependent type theory \cite{Hofmann97-presheaf-chapter,psh-universes}, and we establish the notations we will use.
	
	\item In \textbf{\cref{ch:cwf-morphisms}}, we capture morphisms of CwFs, and natural transformations and adjunctions between them, in typing rules. We especially study morphisms of CwFs between presheaf categories, that arise from functors between the base categories. This chapter was modified since \cite{bpcubicalsets}: we now relax Dybjer's definition of CwF morphisms \cite{dybjer-cwf} by requiring that the empty context and context extension are preserved up to isomorphism, instead of on the nose.
\end{itemize}

In \textbf{\cref{part:paramdtt}}, we build a presheaf model for the type system ParamDTT (parametric dependent type theory) described in \cite{paramdtt}. We construct our model by defining the base category $\bpcubecat$ of \emph{bridge/path cubes} and adapting the general presheaf model over $\bpcubecat$ to suit our needs. Our model is heavily based on the models by Atkey, Ghani and Johann \cite{dtt-parametricity}, Huber \cite{huber}, Bezem, Coquand and Huber \cite{model-cubical}, Cohen, Coquand, Huber and M\"ortberg \cite{cubical}, Moulin \cite{moulin} and Bernardy, Coquand and Moulin \cite{moulin-param3}.
\begin{itemize}
	\item In \textbf{\cref{ch:bpcubecat}}, we introduce the category $\bpcubecat$ of bridge/path cubes --- which are cubes whose dimensions are all annotated as either bridge or path dimensions --- and its presheaf category $\widehat{\bpcubecat}$ of bridge/path cubical sets. There is a rich interaction with the category of cubical sets $\widehat{\cubecat}$ which we investigate more closely using ideas from axiomatic cohesion \cite{adjoint-logic}.
	\item In \textbf{\cref{ch:discreteness}}, we define discrete types and show that they form a model of dependent type theory. We prove some infrastructural results.
	\item In \textbf{\cref{ch:paramdtt}}, we give an interpretation of the typing rules of ParamDTT \cite{paramdtt} in $\widehat{\bpcubecat}$.
\end{itemize}

In \textbf{\cref{part:reldtt}}, which was not present in \cite{bpcubicalsets}, we build a presheaf model for the type system described in \cite{reldtt}. The main difference with the model from \cref{part:paramdtt} is that we now annotate cube dimensions with a degree of relatedness, rather than just `bridge' or `path'.
\begin{itemize}
	\item In \textbf{\cref{ch:dcubecat}}, we introduce the category $\dcubecat n$ of depth $n$ cubes, whose dimensions are annotated with a degree of relatedness $0 \leq i \leq n$, and its presheaf category $\widehat{\dcubecat n}$ of depth $n$ cubical sets. We generalize the interaction between $\widehat{\bpcubecat}$ and $\widehat{\cubecat}$ to a collection of CwF morphisms, which we call \emph{reshuffling functors}, between the categories of cubical sets of various depths.
	
	\item In \textbf{\cref{ch:fibrancy}}, we introduce the concept of a \emph{robust} notion of fibrancy and prove some results about this concept in arbitrary CwFs. After that, we define and study discreteness of depth $n$ cubical sets, which is a robust notion of fibrancy. %
	
	\item In \textbf{\cref{ch:reldtt}}, we will give an interpretation of the typing rules from \cite{reldtt} in the categories of depth $n$ cubical sets.
\end{itemize}

\section*{Acknowledgements}
\subsection*{Regarding \cref{part:paramdtt} (originally published in June 2017)}
Special thanks goes to Andrea Vezzosi. A cornerstone of this model was Andrea's insight that a shape modality on reflexive graphs is relevant to modelling parametricity. The other foundational ideas -- in particular the use of (cohesive-like) endofunctors of a category with families and the internalization of them as modalities -- were formed in discussion with him. He also injected some vital input during the formal elaboration process and pointed out the relevance of the $\Glue$-operator from cubical type theory \cite{cubical}.

Also thanks to Andreas Abel, Paolo Capriotti, Jesper Cockx, Dominique Devriese, Dan Licata and Sandro Stucki for many fruitful discussions.

\subsection*{Regarding \cref{part:reldtt}} Thanks to Paolo Capriotti for some relevant discussions, to Dominique Devriese, Andrea Vezzosi and Andreas Abel for being excellent sounding boards, and additionally to Andrea Vezzosi for his aid in making sense of definitional equality with type-checking time erasure of irrelevant subterms.

\subsection*{General acknowledgements}
The author holds a Ph.D. Fellowship from the Research Foundation - Flanders (FWO).

\part{Prerequisites}\label{part:prerequisites}
\chapter{The standard presheaf model of Martin-L\"of Type Theory}\label{ch:psh}
In this chapter, we introduce the notion of a category with families (CwF, \cite{dybjer-cwf}) and show that every presheaf category constitutes a CwF that supports various interesting type formers. Most of this has been shown by \cite{Hofmann97-presheaf-chapter,psh-universes}; the construction of $\Glue$-types has been shown by \cite{cubical}. The construction of the $\Weld$-type is new. 

\section{Categories with families}
We state the definition of a CwF without referring to the category $\name{Fam}$ of families. Instead, we will make use of the category of elements:
\begin{definition}
	Let $\cat C$ be a category and $A : \cat C \to \Set$ a functor. Then the category of elements $\int_{\cat C} A$ is the category whose
	\begin{itemize}
		\item objects are pairs $(c, a)$ where $c$ is an object of $\cat C$ and $a \in A(c)$,
		\item morphisms are pairs $(\vfi | a) : (c, a) \to (c', a')$ where $\vfi : c \to c'$ and $a' = A(\vfi)(a)$.
	\end{itemize}
	If the functor $A : \cat C\op \to \Set$ is contravariant, we define $\int_{\cat C} A := \paren{\int_{\cat C\op} A}\op$. Thus, its morphisms are pairs $(\vfi|a') : (c,a) \to (c', a')$ where $\vfi : c \to c'$ and $a = A(\vfi)(a')$.
\end{definition}
\begin{definition}\label{def:cwf}
	A \textbf{category with families} (CwF) \cite{dybjer-cwf} consists of:
	\begin{enumerate}
		\item A category $\Ctx$ whose objects we call \textbf{contexts}, and whose morphisms we call \textbf{substitutions}. We also write $\Gamma \ctx$ to say that $\Gamma$ is a context.
		\item A contravariant functor $\Ty : \Ctx\op \to \Set$. The elements $T \in \Ty(\Gamma)$ are called \textbf{types} over $\Gamma$ (also denoted $\Gamma \sez T \type$). The action $\Ty(\sigma) : \Ty(\Gamma) \to \Ty(\Delta)$ of a substitution $\sigma : \Delta \to \Gamma$ is denoted $\loch[\sigma]$, i.e. if $\Gamma \sez T \type$ then $\Delta \sez T[\sigma] \type$.
		\item A contravariant functor $\Tm : \paren{\int_{\Ctx} \Ty}\op \to \Set$ from the category of elements of $\Ty$ to $\Set$. The elements $t \in \Tm(\Gamma, T)$ are called \textbf{terms} of $T$ (also denoted $\Gamma \sez t : T$). The action $\Tm(\sigma | T) : \Tm(\Gamma, T) \to \Tm(\Delta, T[\sigma])$ of $(\sigma | T) : (\Delta, T[\sigma]) \to (\Gamma, T)$ is denoted $\loch[\sigma]$, i.e. if $\Gamma \sez t : T$, then $\Delta \sez t[\sigma] : T[\sigma]$.
		\item A terminal object $()$ of $\Ctx$ called the \textbf{empty context}.
		\item A \textbf{context extension} operation: if $\Gamma \ctx$ and $\Gamma \sez T \type$, then there is a context $\Gamma.T$, a substitution $\pi : \Gamma.T \to \Gamma$ and a term $\Gamma.T \sez \xi : T[\pi]$, such that for all $\Delta$, the map
		\begin{equation*}
			\Hom(\Delta,\Gamma.T) \to \Sigma(\sigma : \Hom(\Delta, \Gamma)). \Tm(\Delta, T[\sigma]) : \tau \mapsto (\pi  \tau , \xi[\tau])
		\end{equation*}
		is invertible. We call the inverse $\loch, \loch$.
		Note that for more precision and less readability, we could write $\pi_{\Gamma, T}$, $\xi_{\Gamma, T}$ and $(\loch, \loch)_{\Gamma, T}$.
		
		If $\sigma : \Delta \to \Gamma$, then we will write $\sigma \subext = (\sigma \pi, \xi) : \Delta.T[\sigma] \to \Gamma.T$.
	
		Sometimes, for clarity, we will use variable names: we write $\Gamma, \var x : T$ instead of $\Gamma.T$, and $\wknvar x : (\Gamma, \var x : T) \to \Gamma$ and $\Gamma, \var x : T \sez \var x : T[\wknvar x]$ for $\pi$ and $\xi$. Their joint inverse will be called $(\loch, \loch/\var x)$.
	\end{enumerate}
\end{definition}

\section{Presheaf categories are CwFs}\label{sec:psh-cwf}
This is proven elaborately in \cite{Hofmann97-presheaf-chapter}, though we give an unconventional treatment that views the Yoneda-embedding truly as an embedding, i.e. treating the base category $\catW$ as a fully faithful subcategory of $\widehat{\catW}$.

\subsection{Contexts} Pick a base category $\catW$. We call its objects \textdef{primitive contexts} and its morphisms \textdef{primitive substitutions}, denoted $\vfi : \PSub V W$. The presheaf category $\widehat \catW$ over $\catW$ is defined as the functor space $\Set^{\catW\op}$. We will use $\widehat \catW$ as $\Ctx$. A context $\Gamma$ is thus a \textdef{presheaf} over $\catW$, i.e. a functor $\Gamma : \catW\op \to \Set$. We denote its action on a primitive context $W$ as $\DSub W \Gamma$, and the elements of that set are called \textdef{defining substitutions} from $W$ to $\Gamma$. The action of $\Gamma$ on $\vfi : \PSub V W$ is denoted $\loch \vfi : (\DSub W \Gamma) \to (\DSub V \Gamma)$ and is called \textdef{restriction} by $\vfi$. A substitution $\sigma : \Delta \to \Gamma$ is then a natural transformation $\sigma \loch : (\DSub \loch \Delta) \to (\DSub \loch \Gamma)$.

\subsection{Types} A type $\Gamma \sez T \type$ is a \textdef{dependent presheaf} over $\Gamma$. In categorical language, this is a functor $T : \paren{\int_{\catW} \Gamma}\op \to \Set$. We denote its action on an object $(W, \gamma)$, where $\gamma : \DSub W \Gamma$, as $T \dsub \gamma$; the elements $t \in T \dsub \gamma$ will be called \textdef{defining terms} and denoted $W \Dsez t : T \dsub \gamma$. The action of $T$ on a morphism $(\vfi | \gamma) : (V, \gamma \vfi) \to (W, \gamma)$ is denoted $\loch \psub \vfi : T \dsub \gamma \to T \dsub{\gamma \vfi}$ and is again called \textdef{restriction}. We have now defined $\Ty(\Gamma)$.

Given a substitution $\sigma : \Delta \to \Gamma$, we need an action $\loch \sub \sigma : \Ty(\Gamma) \to \Ty(\Delta)$. This is defined by setting $T \sub \sigma \dsub \delta := T \dsub{\sigma \delta}$ and defining $\loch \psub \vfi^{T[\sigma]} : T \sub \sigma \dsub \delta \to T \sub \sigma \dsub{\delta \vfi}$ as $\loch \psub \vfi^T : T \dsub{\sigma \delta} \to T \dsub{\sigma \delta \vfi}$.

\subsection{Terms} A term $\Gamma \sez t : T$ consists of, for every $\gamma : \DSub W \Gamma$, a defining term $W \Dsez t \dsub \gamma : T \dsub \gamma$. Moreover, this must be natural in $W$, i.e. for every $\vfi : V \to W$, we require $t \dsub \gamma \psub \vfi = t \dsub{\gamma \vfi}$.

\subsection{The empty context} We set $(\DSub W ()) = \accol{\emptysub}$. The unique substitutions $\Gamma \to ()$ will also be denoted $\emptysub$.

\subsection{Context extension} We set $(\DSub W \Gamma.T) = \set{(\gamma, t)}{\gamma : \DSub W \Gamma \textand W \Dsez t : T \dsub \gamma}$, and $(\gamma, t)\vfi = (\gamma \vfi, t \psub \vfi)$. Of course $\pi(\gamma, t) = \gamma$ and $\xi \dsub{(\gamma, t)} = t$. In variable notation, we will write $(\gamma, t/\var x)$ for $(\gamma, t)$.

\subsection{Yoneda-embedding} There is a fully faithful embedding $\yoneda : \catW \to \widehat \catW$, called the Yoneda embedding, given by $(\DSub V {\yoneda W}) := (\PSub V W)$. Fully faithful means that $(\PSub V W) \cong (\yoneda V \to \yoneda W)$. We have moreover that $(\DSub V \Gamma) \cong (\yoneda V \to \Gamma)$ and $\Tm(\yoneda V, T) \cong T \dsub \id_V$ meaning that terms $\yoneda V \sez t : T$ correspond to defining terms $V \Dsez t \dsub \id : T \dsub \id$. We will omit notations for each of these isomorphisms, effectively treating them as equality.

\section{$\Sigma$-types}
\begin{definition}
	We say that a CwF \textdef{supports $\Sigma$-types} if it is closed under the following rules:
	\begin{equation}
		\inference{\Gamma \sez A \type \qquad \Gamma.A \sez B \type}{\Gamma \sez \Sigma A B \type}{} \qquad
		\inference{\Gamma \sez a : A \qquad \Gamma \sez b : B[\id, a]}{\Gamma \sez (a , b) : \Sigma A  B}{}
	\end{equation}
	\begin{equation}
		\inference{\Gamma \sez p : \Sigma A B}{\Gamma \sez \fst\,p : A}{} \qquad
		\inference{\Gamma \sez p : \Sigma A B}{\Gamma \sez \snd\,p : B[\id, \fst\,p]}{}
	\end{equation}
	where $(\fst, \snd)$ and $(\loch, \loch)$ are inverses and all four operations are natural in $\Gamma$:
	\begin{align*}
		(\Sigma A B)[\sigma] &= \Sigma (A[\sigma])(B[\sigma \subext]), \\
		(a , b)[\sigma] &= (a[\sigma] , b[\sigma]), \qquad
		(\fst\,p)[\sigma] = \fst(p[\sigma]), \qquad
		(\snd\,p)[\sigma] = \snd(p[\sigma]).
	\end{align*}
\end{definition}
\begin{proposition}
	Every presheaf category supports $\Sigma$-types.
\end{proposition}
\begin{proof}
	Given $\gamma : \DSub W \Gamma$, we set $(\Sigma A B) \dsub \gamma = \set{(a, b)}{W \Dsez a : A \dsub \gamma \textand W \Dsez b : B \dsub{\gamma, a}}$, and $(a, b) \psub \vfi = (a \psub \vfi, b \psub \vfi)$, which is natural in $\Gamma$.
	
	We define the pair term $(a, b)$ by $(a, b) \dsub \gamma = (a \dsub \gamma, b \dsub \gamma)$; $\fst\,p$ by $(\fst\,p) \dsub \gamma = p \dsub \gamma_1$ and $(\snd\,p) \dsub \gamma = p \dsub \gamma_2$. All of this is easily seen to be natural in $\Gamma$ and $W$.
\end{proof}

\section{$\Pi$-types}\label{sec:psh-pi-types}
\begin{definition}\label{def:pi-types}
	We say that a CwF \textdef{supports $\Pi$-types} if it is closed under the following rules:
	\begin{equation}
		\inference{\Gamma \sez A \type \qquad \Gamma.A \sez B \type}{\Gamma \sez \Pi A B \type}{} \qquad
		\inference{\Gamma.A \sez b : B}{\Gamma \sez \lambda b : \Pi A B}{} \qquad
		\inference{\Gamma \sez f : \Pi A B}{\Gamma.A \sez \ap f : B}{}
	\end{equation}
	such that $\ap$ and $\lambda$ are inverses and such that all three operations commute with substitution:
	\begin{align*}
		(\Pi A B)[\sigma] = \Pi (A[\sigma])(B[\sigma \subext]), \qquad
		(\lambda b)[\sigma] = \lambda(b[\sigma \subext]), \qquad
		(\ap f)[\sigma \subext] = \ap (f[\sigma])
	\end{align*}
	We will write $f\,a$ for $(\ap\,f)\sub{\id, a}$.
\end{definition}
\begin{proposition}
	Every presheaf category supports $\Pi$-types.
\end{proposition}
\begin{proof}
	Given $\gamma : \DSub W \Gamma$, we set $(\Pi A B) \dsub \gamma = \set{\dlambda b}{\yoneda W.A[\gamma] \sez b : B[\gamma \subext]}$, where the label $\dlambda$ is included for clarity but can be implemented as the identity function; and $(\dlambda b) \psub \vfi = \dlambda (b \sub{\vfi \subext})$.
	To see that this is natural in $\Gamma$, take $\sigma : \Delta \to \Gamma$ and $\delta : \DSub W \Delta$ and unfold the definitions of $(\Pi A B)\sub \sigma \dsub \delta$ and $(\Pi (A \sub \sigma) (B \sub{\sigma \subext})) \dsub \delta$.
	
	We define $\lambda b$ by $(\lambda b) \dsub \gamma = \dlambda(b[\gamma \subext])$. To see that $\lambda b$ is a term:
	\begin{equation}
		(\lambda b) \dsub \gamma \psub \vfi
		= (\dlambda (b \sub{\gamma\subext})) \psub \vfi
		= \dlambda (b \sub{\gamma\subext} \sub{\vfi\subext})
		= \dlambda (b \sub{(\gamma \vfi)\subext}) = (\dlambda b) \dsub{\gamma \vfi}.
	\end{equation}
	One easily checks that $\lambda$ is natural in $\Gamma$.
	
	Let $\dap$ be the inverse of $\dlambda$. Then $\dap$ satisfies $(\dap\,f)[\vfi \subext] = \dap\,(f \psub \vfi)$. Write $f \cdot a$ for $(\dap\,f) \dsub{\cwfpair{\id}{a}}$. We have
	\begin{equation}
		(\dap (f)) \dsub{\cwfpair{\vfi}{a}}
		= (\dap (f)) \sub{\vfi \subext} \dsub{\cwfpair{\id}{a}}
		= (\dap (f \psub{\vfi})) \dsub{\cwfpair{\id}{a}}
		= f \psub \vfi \cdot a,
	\end{equation}
	so that a defining term $W \Dsez f : (\Pi A B) \dsub \gamma$ is fully determined if we know $f \psub \vfi \cdot a$ for all $\vfi : \PSub V W$ and $V \Dsez a : A \dsub{\gamma \vfi}$.
	Similarly, a term $\Gamma \sez f : \Pi A B$ is fully determined if we know $f \dsub \gamma \cdot a$ for all $\gamma : \DSub V \Gamma$ and $V \Dsez a : A \dsub{\gamma}$.
	
	We define $\ap\,f$ by $(\ap\,f) \dsub{\gamma, a} = f \dsub \gamma \cdot a$. To see that this is a term:
	\begin{align}
		(f \dsub \gamma \cdot a) \psub \vfi
		&= (\dap\,(f \dsub \gamma)) \dsub{\id, a} \psub \vfi
		= (\dap\,(f \dsub \gamma)) \sub{\vfi \subext} \dsub{\id, a \psub \vfi}
		= (\dap\,(f \dsub{\gamma \vfi})) \dsub{\id, a \psub \vfi}
		\nn \\
		&= (f \dsub{\gamma \vfi}) \cdot (a \psub \vfi)
		= (\ap\,f) \dsub{(\gamma, a) \vfi}.
	\end{align}
	One easily checks that $\ap$ is natural in $\Gamma$.
	
	To see that $\ap\,\lambda b = b$, we can unfold
	\begin{equation}
		(\ap\,\lambda b) \dsub{\gamma, a}
		= (\lambda b) \dsub \gamma \cdot a
		= \dlambda (b \sub{\gamma \subext}) \cdot a
		= b \sub{\gamma \subext} \dsub{\id, a}
		= b \dsub{\gamma, a}.
	\end{equation}
	To see that $\lambda\,\ap\,f = f$:
	\begin{equation}
		(\lambda\,\ap\,f) \dsub \gamma \cdot a
		= \dlambda((\ap\,f) \sub{\gamma \subext}) \cdot a
		= (\ap\,f) \sub{\gamma \subext} \dsub{\id, a}
		= (\ap\,f) \dsub{\gamma, a}
		= f \dsub \gamma \cdot a. \qedhere
	\end{equation}
\end{proof}

\section{Identity type}\label{sec:cwf-idtp}
\begin{definition}
	A CwF \textdef{supports the identity type} if it is closed under the following rules:
	\begin{equation}
		\inference{
			\Gamma \sez A \type \\
			\Gamma \sez a, b : A
		}{\Gamma \sez a \idtp A b \type}{}
		\qquad
		\inference{
			\Gamma \sez a : A
		}{\Gamma \sez \refl\,a : a \idtp A a}{}
		\qquad
		\inference{
			\Gamma \sez a, b : A \\
			\Gamma, \var y : A, \var w : (a[\wknvar y] \idtp{A[\wknvar y]} \var y) \sez C \type \\
			\Gamma \sez e : a \idtp A b \\
			\Gamma \sez c : C[\id, a / \var y, \refl\,a / \var w]
		}{\Gamma \sez \J(a, b, \var y.\var w.C, e, c) : C[\id, b / \var y, e / \var w]}{}
	\end{equation}
	such that all three operations commute with substitution:
	\begin{align*}
		(a \idtp A b)[\sigma] &= \paren{a[\sigma] \idtp{A[\sigma]} b[\sigma]} \\
		(\refl\,a)[\sigma] &= \refl\,(a[\sigma]) \\
		\J(a, b, \var y.\var w.C, p, c)[\sigma] &= \J(a[\sigma], b[\sigma], \var y.\var w.C[\sigma \wknvar y \wknvar w, \var y/\var y, \var w/\var w], p[\sigma], c[\sigma])
	\end{align*}
	and such that $\J(a, a, \var y.\var w.C, \refl\,a, c) = c$.
\end{definition}
\begin{proposition}
	Every presheaf category supports the identity type.\footnote{The identity type in the standard presheaf model, expresses equality of mathematical objects. It supports the reflection rule and axiom K, and not the univalence axiom.}
\end{proposition}
\begin{proof}
	We set $(a \idtp A b)\dsub \gamma$ equal to $\accol \star$ if $a \dsub \gamma = b \dsub \gamma$, and make it empty otherwise. We define $(\refl\,a)\dsub \gamma = \star$ and $\J(a, b, \var y.\var w.C, e, c) \dsub \gamma = c \dsub \gamma$, which is well-typed since $e$ witnesses that $a = b$ and all terms of the identity type are equal.
\end{proof}

\section{Universes}
Assume a functor $\Ty^* : \Ctx\op \to \Set$, and write $\Gamma \sez T \type^*$ for $T \in \Ty^*(\Gamma)$.
\begin{definition}
	We say that a CwF \textdef{supports a universe} for $\Ty^*$ if it is closed under the following rules:
	\begin{equation}
		\inference{\Gamma \ctx}{\Gamma \sez \uni{}^* \type}{} \qquad
		\inference{\Gamma \sez T : \uni{}^*}{\Gamma \sez \El\,T \type^*}{} \qquad
		\inference{\Gamma \sez T \type^*}{\Gamma \sez \tycode T : \uni{}^*}{}
	\end{equation}
	where $\El$ and $\tycode \loch$ are inverses and all three operators commute with substitution:
	\begin{equation}
		\uni{}^*[\sigma] = \uni{}^*, \qquad
		(\El\,T)[\sigma] = \El(T[\sigma]), \qquad
		\tycode T[\sigma] = \tycode{T[\sigma]}.
	\end{equation}
	Note that here, we are switching between term substitution and *-type substitution.
\end{definition}
We assume that the metatheory has Grothendieck universes, i.e. there is a chain
\begin{equation}
	\Set_0 \in \Set_1 \in \Set_2 \in \ldots
\end{equation}
such that every $\Set_k$ is a model of ZF set theory. We say that a type $T \in \Ty(\Gamma)$ has level $k$ if $T \dsub \gamma \in \Set_k$ for every $\gamma : \DSub W \Gamma$. Let $\Ty_k(\Gamma)$ be the set of all level $k$ types over $\Gamma$; this constitutes a functor $\Ty_k : \widehat \catW \op \to \Set_{k+1}$. Write $\Gamma \sez T \type_k$ for $T \in \Ty_k(\Gamma)$.
\begin{proposition}\label{thm:unipsh}
	Every presheaf category (over a base category of level 0) supports a universe $\uni k$ for $\Ty_k$ that is itself of level $k+1$.
\end{proposition}
\begin{proof}
	Given $\gamma : \DSub W \Gamma$, we set $\uni k \dsub \gamma = \set{\dtycode T}{\yoneda W \sez T \type_k} \cong \Ty_k(\yoneda W) \in \Set_{k+1}$. Again, $\dtycode \loch$ is just a label which we add for readability. We set $\dtycode T \psub \vfi = \dtycode{T \sub \vfi}$. Naturality in $\Gamma$ is immediate, as the definition of $\uni k$ does not refer to either $\Gamma$ or $\gamma$.
	
	Given $\Gamma \sez T \type_k$, we define $\Gamma \sez \tycode T : \uni k$ by $\tycode T \dsub \gamma = \dtycode{T \sub \gamma}$, which satisfies
	\begin{equation}
		\tycode T \dsub \gamma \psub \vfi = \dtycode{T \sub \gamma} \psub \vfi = \dtycode{T \sub \gamma \sub \vfi} = \tycode T \dsub{\gamma \vfi}.
	\end{equation}
	
	Write $\dEl$ for the inverse of $\dtycode \loch$. It satisfies $\dEl\,A \sub \vfi = \dEl(A \psub \vfi)$. Given $\Gamma \sez A : \uni k$, we define $\Gamma \sez \El\,A \type_k$ by $\El\,A \dsub \gamma = \dEl(A \dsub \gamma) \dsub \id$. Given $\vfi : V \to W$ and $W \Dsez a : \El\,A \dsub \gamma$, we set $a \psub \vfi^{\El\,A} : \El\,A \dsub{\gamma \vfi}$ equal to $a \psub \vfi^{\dEl(A \dsub \gamma)} : \dEl(A \dsub \gamma) \dsub \vfi$. This is well-typed, because
	\begin{equation}
		\dEl(A \dsub \gamma) \dsub \vfi
		= \dEl(A \dsub \gamma) \sub \vfi \dsub \id
		= \dEl(A \dsub{\gamma \vfi}) \dsub \id
		= \El\,A\dsub{\gamma\vfi}.
	\end{equation}
	To see that $\El\,\tycode T = T$:
	\begin{equation}
		\El\,\tycode T \dsub \gamma
		= \dEl (\tycode T \dsub \gamma) \dsub \id
		= \dEl (\dtycode{T \sub \gamma}) \dsub \id = T \dsub \gamma 
	\end{equation}
	One can check that the substitution operations of $\El\,\tycode T$ and $T$ also match.
	
	Conversely, we show that $\tycode{\El\,A} = A$. To that end, we unpack both completely by applying $\dEl(\loch \dsub \gamma) \dsub \vfi$:
	\begin{align}
		\dEl(\tycode{\El\,A} \dsub \gamma) \dsub \vfi
		&= \dEl(\dtycode{\El\,A \sub \gamma}) \dsub \vfi
		= \El\,A \dsub{\gamma \vfi}
		= \dEl(A \dsub{\gamma \vfi}) \dsub \id, \\
		\dEl(A \dsub \gamma) \dsub \vfi
		&= \dEl(A \dsub \gamma) \sub \vfi \dsub \id
		= \dEl(A \dsub{\gamma \vfi}) \dsub \id. \qedhere
	\end{align}
\end{proof}
We say that a type $T \in \Ty(\Gamma)$ is a \textdef{proposition} if for every $\gamma : \DSub W \Gamma$, we have $T \dsub \gamma \subseteq \accol \star$. We denote this as $T \in \Prop(\Gamma)$ or $\Gamma \sez T \prop$.
\begin{proposition}
	Every presheaf category (over a base category of level 0) supports a universe $\Prop$ of propositions that is itself of level 0.
\end{proposition}
\begin{proof}
	Completely analogous.
\end{proof}
One easily shows that $\Prop$ is closed under $\top$, $\bot$, $\wedge$ and $\vee$. It also clearly contains the identity types. There is an absurd eliminator for $\bot$ and we can construct systems to eliminate proofs of $\vee$.

\section{Glueing}
\begin{definition}
	A CwF \textdef{supports glueing} if it is closed under the following rules:
	\begin{equation*}
		\inference{
			\Gamma \sez P \prop \\
			\Gamma.P \sez T \type \\
			\Gamma.P \sez f : T \to A[\pi] \\
			\Gamma \sez A \type
		}{\Gamma \sez \Gluesys{A}{\Gluesysclauseb{P}{T}{f}} \type}{}, \qquad
		\inference{
			\Gamma \sez \Gluesys{A}{\Gluesysclauseb{P}{T}{f}} \type \\
			\Gamma \sez a : A \\
			\Gamma.P \sez t : T \\
			\Gamma.P \sez ft = a[\pi] : T
		}{\Gamma \sez \gluesys{a}{\sysclauseb{P}{t}} : \Gluesys{A}{\Gluesysclauseb{P}{T}{f}}}{},
	\end{equation*}
	\begin{equation*}
		\inference{
			\Gamma \sez b : \Gluesys{A}{\Gluesysclauseb{P}{T}{f}}
		}{\Gamma \sez \ungluesys{\sysclauseb P f} b : A}{},
	\end{equation*}
	naturally in $\Gamma$, such that
	\begin{align*}
		\Gluesys{A}{\Gluesysclauseb{\top}{T}{f}} &= T[\cwfpair \id \star], \\
		\gluesys{a}{\sysclauseb{\top}{t}} &= t[\cwfpair \id \star], \\
		\ungluesys{\sysclauseb \top f} b &= f[\cwfpair \id \star]\,b, \\
		\ungluesys{\sysclauseb P f} (\gluesys{a}{\sysclauseb{P}{t}}) &= a, \\
		\gluesys{\ungluesys{\sysclauseb P f} b}{\sysclauseb{P}{b[\pi]}} &= b.
	\end{align*}
\end{definition}
\begin{proposition}
	Every presheaf category supports glueing.
\end{proposition}
\begin{proof}
	We assume given the prerequisites of the type former. Write $G = \Gluesys{A}{\Gluesysclauseb{P}{T}{f}}$.
	\begin{description}
		\item[The type] We define $G \dsub \gamma$ by case distinction on $P \dsub \gamma$:
		\begin{enumerate}
			\item If $P \dsub \gamma = \accol \star$, then we set $G \dsub \gamma = T \dsub{\gamma, \star}$.
			\item If $P \dsub \gamma = \eset$, we let $P \dsub \gamma$ be the set of pairs $(a \mapsfrom t)$ where $a : A \dsub \gamma$ and $\yoneda W.P[\gamma] \sez t : T[\gamma \subext]$, such that for every $\vfi$ for which $P \dsub{\gamma \vfi} = \accol \star$, the application $f \dsub{\gamma \vfi, \star} \cdot t \dsub{\vfi, \star}$ is equal to $a \psub \vfi$.
		\end{enumerate}
		Given $g : G \dsub \gamma$ and $\vfi : \PSub V W$, we need to define $g \psub \vfi$.
		\begin{enumerate}
			\item If $P \dsub \gamma = P \dsub{\gamma \vfi} = \accol \star$, then we use the definition from $T$.
			\item If $P \dsub \gamma = P \dsub{\gamma \vfi} = \eset$, then we set $(a \mapsfrom t) \psub \vfi = (a \psub \vfi \mapsfrom t[\vfi \subext])$.
			\item If $P \dsub \gamma = \eset$ and $P \dsub{\gamma \vfi} = \accol \star$, then we set $(a \mapsfrom t) \psub \vfi = t \dsub{\vfi, \star}$.
		\end{enumerate}
		One can check that this definition preserves composition and identity, and that this entire construction is natural in $\Gamma$.
		
		\item[The constructor] Write $g = \gluesys{a}{\sysclauseb{P}{t}}$. We define $g \dsub \gamma$ by case distinction on $P \dsub \gamma$:
		\begin{enumerate}
			\item If $P \dsub \gamma = \accol \star$, then we set $g \dsub \gamma = t \dsub{\gamma, \star}$.
			\item If $P \dsub \gamma = \eset$, then we set $g \dsub \gamma = (a \dsub \gamma \mapsfrom t \sub{\gamma \subext})$.
		\end{enumerate}
		By case distinction, it is easy to check that this is natural in the domain $W$ of $\gamma$. Naturality in $\Gamma$ is straightforward.
		
		\item[The eliminator] Write $u = \ungluesys{\sysclauseb P f} b$. We define $u \dsub \gamma$ by case distinction on $P \dsub \gamma$:
		\begin{enumerate}
			\item If $P \dsub \gamma = \accol \star$, then we set $u \dsub \gamma = f \dsub{\gamma, \star} \cdot b \dsub{\gamma}$.
			\item If $P \dsub \gamma = \eset$, then $b \dsub \gamma$ is of the form $(a \mapsfrom t)$ and we set $u \dsub \gamma = a$.
		\end{enumerate}
		Naturality in the domain $W$ of $\gamma$ is evident when we consider non-cross-case restrictions. Naturality for cross-case restrictions is asserted by the condition on pairs $(a \mapsfrom t)$. Again, naturality in $\Gamma$ is straightforward.
		
		\item[The $\beta$-rule] Pick $\gamma : \DSub W \Gamma$. Write $g = \gluesys{a}{\sysclauseb{P}{t}}$.
		\begin{enumerate}
			\item If $P \dsub \gamma = \accol \star$, then we get
			\begin{equation}
				\mathrm{LHS} \dsub \gamma
				= f \dsub{\gamma, \star} \cdot g \dsub{\gamma}
				= f \dsub{\gamma, \star} \cdot t \dsub{\gamma, \star}
				= (f\,t) \dsub{\gamma, \star} = a \dsub{\gamma}
			\end{equation}
			by the premise of the $\glue$ rule.
			\item If $P \dsub \gamma = \eset$, then we have $g \dsub \gamma = (a \dsub \gamma \mapsfrom t \sub{\gamma \subext})$ and $\unglue$ simply extracts the first component.
		\end{enumerate}
		
		\item[The $\eta$-rule] Pick $\gamma : \DSub W \Gamma$. Write $u = \ungluesys{\sysclauseb P f} b$.
		\begin{enumerate}
			\item If $P \dsub \gamma = \accol \star$, then we have $\mathrm{LHS} \dsub \gamma = b \sub \pi \dsub{\gamma, \star} = b \dsub \gamma$.
			\item If $P \dsub \gamma = \eset$, then $b \dsub \gamma$ has the form $(a \mapsfrom t)$ and we get
			\begin{equation}
				\mathrm{LHS} \dsub \gamma
				= (u \dsub \gamma \mapsfrom b \sub \pi \sub{\gamma \subext})
				= (a \mapsfrom b \sub{\gamma \pi})
				=^{(\dagger)} (a \mapsfrom t).
			\end{equation}
			The last step $(\dagger)$ is less than trivial. We show that $\yoneda W.P[\gamma] \sez b \sub{\gamma \pi} = t : T$. Pick any $(\vfi, \star) : \DSub V {(\yoneda W.P[\gamma])}$. If you manage to pick one, then $P \dsub{\gamma \vfi} = \accol \star$, so $b \sub{\gamma \pi} \dsub{\vfi, \star} = b \dsub{\gamma \vfi} = (a \mapsfrom t) \psub \vfi = t \dsub{\vfi, \star}$. \qedhere
		\end{enumerate} 
	\end{description}
\end{proof}

\section{Welding}
\begin{definition}
	A CwF \textdef{supports welding} if it is closed under the following rules:
	\begin{equation}
		\inference{
			\Gamma \sez P \prop \\
			\Gamma, \var p : P \sez T \type \\
			\Gamma, \var p : P \sez f : A[\pi] \to T \\
			\Gamma \sez A \type
		}{\Gamma \sez \Weldsys{A}{\Weldsysclauseb{\var p : P}{T}{f}} \type}{}, \qquad
		\inference{
			\Gamma \sez \Weldsys{A}{\Weldsysclauseb{\var p : P}{T}{f}} \type \\
			\Gamma \sez a : A
		}{\Gamma \sez \weldsys{\sysclauseb{\var p : P}{f}} a : \Weldsys{A}{\Weldsysclauseb{\var p : P}{T}{f}}}{},
	\end{equation}
	\begin{equation}
		\inference{
			\Gamma, \var y : \Weldsys{A}{\Weldsysclauseb{\var p : P}{T}{f}} \sez C \type \\
			\Gamma, \var p : P, \var y : T \sez d : C[\wknvar p \subext] \\
			\Gamma, \var x : A \sez c : C[\wknvar x, \weldsys{\sysclauseb{\var p : P[\wknvar x \subext]}{f[\wknvar x \subext]}} \var x/\var y] \\
			\Gamma, \var p : P, \var x : A[\wknvar p] \sez d[\wknvar x, f[\wknvar x]\,\var x/\var y] = c[\wknvar p \subext] : C[\wknvar p \wknvar x, f[\wknvar x]\,\var x/\var y]\\
			\Gamma \sez b : \Weldsys{A}{\Weldsysclauseb{\var p : P}{T}{f}}
		}{\Gamma \sez \ind_\Weld(\var y.C, \sys{\sysclauseb{\var p : P}{\var y.d}}, \var x.c, b) : C[\id, b/\var y]}{}
	\end{equation}
	naturally in $\Gamma$, such that
	\begin{align*}
		\Weldsys{A}{\Weldsysclauseb{\var p : \top}{T}{f}} &= T[\id, \star/\var p], \\
		\weldsys{\sysclauseb{\var p : \top}{f}} a &= f[\id, \star/\var p] a, \\
		\ind_\Weld(\var y.C, \sys{\sysclauseb{\var p : \top}{\var y.d}}, \var x.c, b) &= d[\id, \star / \var p, b / \var y], \\
		\ind_\Weld(\var y.C, \sys{\sysclauseb{\var p : P}{\var y.d}}, \var x.c, \weldsys{\sysclauseb{\var p : P}{f}} a) &= c[\id, a/\var x].
	\end{align*}
\end{definition}
\begin{proposition}
	Every presheaf category supports welding.
\end{proposition}
\begin{proof}
	We assume given the prerequisites of the type former. Write $\Omega = \Weldsys{A}{\Weldsysclauseb{\var p : P}{T}{f}}$.
	\begin{description}
		\item[The type] We define $W \dsub \gamma$ by case distinction on $P \dsub \gamma$:
		\begin{enumerate}
			\item If $P \dsub \gamma = \accol \star$, then we set $\Omega \dsub \gamma = T \dsub{\gamma, \star/\var p}$.
			\item If $P \dsub \gamma = \eset$, then we set $\Omega \dsub \gamma = \set{\dweld\,a}{a : A \dsub \gamma} \cong A \dsub \gamma$. Once more, $\dweld$ is a meaningless label that we add for readability.
		\end{enumerate}
		Given $w : \Omega \dsub \gamma$ and $\vfi : \PSub V W$, we need to define $w \psub \vfi$.
		\begin{enumerate}
			\item If $P \dsub \gamma = P \dsub{\gamma \vfi} = \accol \star$, then we use the definition from $T$.
			\item If $P \dsub \gamma = P \dsub{\gamma \vfi} = \eset$, then we set $(\dweld\,a) \psub \vfi = \dweld (a \psub \vfi)$
			\item If $P \dsub \gamma = \eset$ and $P \dsub{\gamma \vfi} = \accol \star$, then we set $(\dweld\,a) \psub \vfi = f \dsub{\gamma \vfi, \star / \var p} \cdot a \psub \vfi$.
		\end{enumerate}
		One can check that this definition preserves composition and identity, and that this entire construction is natural in $\Gamma$.
		
		\item[The constructor] Write $w = \weldsys{\sysclauseb{\var p : P}{f}} a$.
		\begin{enumerate}
			\item If $P \dsub \gamma = \accol \star$, then we set $w \dsub \gamma = f \dsub{\gamma, \star / \var p} \cdot a \dsub \gamma$.
			\item If $P \dsub \gamma = \eset$, then we set $w \dsub \gamma = \dweld(a \dsub \gamma)$.
		\end{enumerate}
		This is easily checked to be natural in $\Gamma$ and the domain $W$ of $\gamma$.
		
		\item[The eliminator] Write $z = \ind_\Weld(\var y.C, \sys{\sysclauseb{\var p : P}{\var y.d}}, \var x.c, b)$.
		\begin{enumerate}
			\item If $P \dsub \gamma = \accol \star$, then $b \dsub \gamma : T \dsub{\gamma, \star / \var p}$, and we can set $z \dsub \gamma = d \dsub{\gamma, \star / \var p, b \dsub \gamma / \var y}$.
			\item If $P \dsub \gamma = \eset$, then $\dunweld(b \dsub \gamma) : A \dsub \gamma$ and we can set $z \dsub \gamma = c \dsub{\gamma, \dunweld(b \dsub \gamma) / \var x}$, where $\dunweld$ removes the $\dweld$ label.
		\end{enumerate}
		We need to show that this is natural in the domain $W$ of $\gamma$, which is only difficult in the cross-case-scenario. So pick $\vfi : \PSub V W$ such that $P \dsub \gamma = \eset$ and $P \dsub{\gamma \vfi} = \accol \star$. We need to show that $z \dsub \gamma \psub \vfi = z \dsub{\gamma \vfi}$. Write $b \dsub \gamma = \dweld\,a$. We have
		\begin{align}
			z \dsub \gamma \psub \vfi
			&= c \dsub{\gamma, a / \var x} \psub \vfi
			= c \dsub{\gamma \vfi, a \psub \vfi / \var x}, \\ \nn
			z \dsub{\gamma \vfi}
			&= d \dsub{\gamma \vfi, \star / \var p, b \dsub{\gamma \vfi} / \var y}
			= d \dsub{\gamma \vfi, \star / \var p, (\dweld\,a) \psub \vfi / \var y} \\ \nn
			&= d \dsub{\gamma \vfi, \star / \var p, f \dsub{\gamma \vfi, \star / \var p} \cdot a \psub \vfi / \var y}.
		\end{align}
		From the premises, we know that $d[\wknvar x, f[\wknvar x]\,\var x/\var y] = c[\wknvar p \subext]$. Applying $\loch \dsub{\gamma \vfi, \star / \var p, a \psub \vfi / \var x}$ to this equation yields
		\begin{equation}
			d \dsub{\gamma \vfi, \star / \var p, f \dsub{\gamma \vfi, \star / \var p} \cdot a \psub \vfi / \var y} = c \dsub{\gamma \vfi, a \psub \vfi / \var x}.
		\end{equation}
		
		\item[The $\beta$-rule] Write $w = \weldsys{\sysclauseb{\var p : P}{f}} a$.
		\begin{enumerate}
			\item If $P \dsub \gamma = \accol \star$, then
			\begin{align}
				\mathrm{LHS} \dsub \gamma
				&= d \dsub{\gamma, \star / \var p, w \dsub \gamma / \var y}
				= d \dsub{\gamma, \star / \var p, f \dsub{\gamma, \star / \var p} \cdot a \dsub \gamma / \var y}, \\ \nn
				\mathrm{RHS} \dsub \gamma
				&= c \dsub{\gamma, a \dsub \gamma / \var x}.
			\end{align}
			Again, the premises give us this equality.
			\item If $P \dsub \gamma = \eset$, then the equality is trivial. \qedhere
		\end{enumerate}
	\end{description}
\end{proof}

\chapter{Internalizing transformations of semantics}\label{ch:cwf-morphisms}
Given categories with families (CwFs) $\cat C$ and $\cat D$, we can consider functors $F : \cat C \to \cat D$ that sufficiently preserve the CwF structure to preserve semantical truth (though not necessarily falsehood) of judgements. Such functors will be called \textbf{morphisms of CwFs}.

In \cref{sec:cwf-transform} of this chapter, we are concerned with how we can internalize a CwF morphism and even more interestingly, how we can internalize a natural transformation between CwF morphisms. That is, we want to answer the question: What inference rules become meaningful when we know of the existence of (natural transformations between) CwF morphisms? Finally, we consider adjoint CwF morphisms, which of course give rise to unit and co-unit natural transformations.

In \cref{sec:psh-transform}, we delve deeper and study the implications of functors, natural transformations and adjunctions between categories $\catV$ and $\catW$ for the CwFs $\widehat{\catV}$ and $\widehat{\catW}$.

Throughout the chapter, we will need to annotate symbols like $\Ty$, $\sez$ and $\Dsez$ with the CwF that we are talking about.

\section{Categories with families}\label{sec:cwf-transform}

\subsection{Morphisms of CwFs}
\begin{definition}
	A \textbf{(weak) morphism of CwFs} $F : \cat C \to \cat D$ consists of:
	\begin{enumerate}
		\item A functor $F_\Ctx : \cat C \to \cat D$,
		\item A natural transformation $F_\Ty : \Ty_{\cat C} \to \Ty_{\cat D} \circ F_\Ctx$,
		\item A natural transformation $F_\Tm : \Tm_{\cat C} \to \Tm_{\cat D} \circ F_{\int}$, where $F_{\int} : \int_{\cat C} \Ty_{\cat C} \to \int_{\cat D} \Ty_{\cat D}$ is easily constructed from $F_\Ctx$ and $F_\Ty$,
		\item such that $F_\Ctx ()$ is terminal,
		\item such that $(F_\Ctx \pi, F_\Tm \xi) : F_\Ctx(\Gamma.T) \to (F_\Ctx \Gamma) . (F_\Ty T)$ is an isomorphism.
	\end{enumerate}
	The images of a context $\Gamma$, a substitution $\sigma$, a type $T$ and a term $t$ are also denoted $F\Gamma$, $F \sigma$, $FT$ and $\ftrtm F t$ respectively. We choose to denote the action of terms differently because CwF morphisms act very differently on types and on terms of the universe: in general $\ftrtm F A$ will be quite different from $F(\El\,A)$.
	
	Given $\sigma : \Delta \to F \Gamma$ and $\Delta \sez t : (FT)[\sigma]$, we write $(\sigma, t)_F := (F\pi, \ftrtm F \xi)\inv (\sigma, t) : \Delta \to F(\Gamma.T)$. In particular, $(\pi, \xi)_F = (F\pi, \ftrtm F \xi)\inv$. In variable notation, we will write $(\sigma, t/\ftrvar F x)_F : \Delta \to F(\Gamma, \var x : T)$
	
	The canonical map to $F()$ will be denoted $()_F$.
	
	\bigskip
	
	A morphism of CwFs is called \textdef{strict} if
	\begin{enumerate}
		\setcounter{enumi}{3}
		\item $F_\Ctx () = ()$,
		\item $F_\Ctx(\Gamma.T) = (F_\Ctx \Gamma).(F_\Ty T)$, $F_\Ctx \pi = \pi$ and $F_\Tm \xi = \xi$.
	\end{enumerate}
\end{definition}
A morphism of CwFs $F : \cat C \to \cat D$ is easy to internalize:
\begin{equation}
	\inference{\Gamma \sez_\cat C \Ctx}{F\Gamma \sez_\cat D \Ctx}{} \qquad
	\inference{\Gamma \sez_\cat C T \type}{F\Gamma \sez_\cat D FC \type}{} \qquad
	\inference{\Gamma \sez_\cat C t : T}{F\Gamma \sez_\cat D \ftrtm F t : FT}{}
	\qquad
	\inference{\sigma : \Delta \to F \Gamma \qquad \Delta \sez t : (FT)[\sigma]}{(\sigma, t/\ftrvar F x)_F : \Delta \to F(\Gamma, \var x : T)}{}
\end{equation}

with equations for
substitution:
\begin{equation}
	F\id = \id, \qquad
	F(\tau \sigma) = (F\tau) (F\sigma), \qquad
	F(T[\sigma]) = (FT)[F\sigma], \qquad
	\ftrtm F {(t[\sigma])} = (\ftrtm F t)[F\sigma],
\end{equation}
and pairing and projecting:
\begin{equation}
	F \pi \circ (\sigma, t)_F = \sigma, \qquad
	(\ftrtm F \xi) [(\sigma, t)_F] = t, \qquad
	(F \pi \circ \rho, \ftrtm F \xi[\rho])_F = \rho,
\end{equation}
\begin{equation}
	(\sigma, t)_F \circ \rho = (\sigma \rho, t[\rho])_F, \qquad
	F(\cwfpair \sigma t) = (F \sigma, \ftrtm F t)_F,
\end{equation}
or in variable notation:
\begin{equation}
	F \wknvar x \circ (\sigma, t / \ftrvar F x)_F = \sigma, \qquad
	(\ftrvar F x)[(\sigma, t / \ftrvar F x)_F] = t, \qquad
	(F \wknvar x \circ \rho, \ftrvar F x[\rho]/\ftrvar F x)_F = \rho,
\end{equation}
\begin{equation}
	(\sigma, t/\ftrvar F x)_F \circ \rho = (\sigma \rho, t[\rho]/\ftrvar F x), \qquad
	F(\sigma, t / \var x) = (F \sigma, \ftrtm F t / \ftrvar F x)_F.
\end{equation}
Note that $G(\sigma, t)_F = (G\sigma, \ftrtm G t)_{GF}$ because
\begin{align}
	GF\pi \circ G(\sigma, t)_F
	&= G(F\pi \circ (\sigma, t)_F)
	= G\sigma, \\
	(\ftrtm{GF}{\xi})[G(\sigma, t)_F]
	&= \ftrtm{G}{\paren{(\ftrtm F \xi)[(\sigma, t)_F]}}
	= \ftrtm G t.
\end{align}

If $F$ is a \textdef{strict} CwF morphism, then we get equations for
context formation:
\begin{equation}
	F () = (), \qquad F(\Gamma.T) = F\Gamma.FT,
\end{equation}
and pairing and projecting:
\begin{equation}
	F \pi = \pi, \qquad
	\ftrtm F \xi = \xi, \qquad
	F(\cwfpair \sigma t) = (F \sigma, \ftrtm F t).
\end{equation}
When using variable notation, the equation $\ftrtm F \xi = \xi$ inspires us to write $F(\Gamma, \var x : T) = (F\Gamma, \ftrtm{F}{\var x} : FT)$. Here, $\ftrtm F {\var x}$ is to be regarded as an atomic variable name, which happens to be equal to the compound term of the same notation. Then we get $F \wknvar x = \wkn{\ftrtm F x}$.

\begin{proposition}
	If $G$ is a morphism of CwFs, and $\zeta : F \cong G$, then $F$ is a morphism of CwFs.
\end{proposition}
\begin{proof}
	Given a type $\Gamma \sez T \type$, we define $F\Gamma \sez FT \type$ by $FT = (GT)[\zeta]$.
	
	Given a term $\Gamma \sez t : T$, we define $F\Gamma \sez \ftrtm F t : FT$ by $\ftrtm F t = (\ftrtm G t)[\zeta]$.
	
	Clearly, $F()$ is terminal as $F() \cong G()$.
	
	The substitution $(F\pi, \ftrtm F \xi) = (F\pi, (\ftrtm G \xi)[\zeta]) : F(\Gamma.T) \to F \Gamma.F T = F \Gamma.(GT)[\zeta]$ is an isomorphism, because the following diagram commutes and the other trajectory consists of isomorphisms:
	\begin{equation}
		\xymatrix{
			F(\Gamma.T) \ar[r]_{(F\pi, \ftrtm F \xi)}
				\ar@/^{2em}/[rr]^{(F\pi, (\ftrtm G \xi)[\zeta])}
				\ar[d]_{\zeta}
			& F\Gamma.FT \ar@{=}[r]
			& F\Gamma.(GT)[\zeta] \ar[d]^{\zeta \subext} \\
			G(\Gamma.T) \ar[rr]_{(G\pi, \ftrtm G \xi)}
			&& G \Gamma.GT
		}
	\end{equation}
\end{proof}
It will be a corollary of the results in the following section, that a functor is a morphism of CwFs in at most one way up to isomorphism.

\subsection{Natural transformations of CwFs}
In this section, we consider morphisms of CwFs $F, G : \cat C \to \cat D$ and a natural transformation $\nu : F \to G$ between the underlying functors. It is clear that for any context $\Gamma$, we get a substitution between its respective images:
\begin{equation}
	\inference{\Gamma \sez_\cat C \Ctx}{\nu : F\Gamma \xrightarrow{\cat D} G\Gamma}{}
\end{equation}
Just like we did for $\pi$ and $\xi$, we will omit the index $\Gamma$ on $\nu$.
\begin{proposition}\label{thm:nattrans-function}
	We have an operation $\nu_\loch(\loch)$ for applying $\nu$ to terms:
	\begin{equation}
		\inference{
			\Gamma \sez_{\cat C} T \type \qquad
			\sigma : \Delta \xrightarrow{\cat D} F \Gamma \qquad
			\Delta \sez_{\cat D} t : (FT)[\sigma]
		}{\Delta \sez_{\cat D} \nu_\sigma(t) : (GT)[\nu \sigma]}{}.
	\end{equation}
	\begin{enumerate}
		\item This operation is natural in $\Delta$, i.e. $\nu_\sigma(t)[\tau] = \nu_{\sigma \tau}(t[\tau])$.
		\item It is also natural in $\Gamma$, i.e. if $\rho : \Gamma' \xrightarrow{\cat C} \Gamma$ and $\sigma : \Delta \xrightarrow{\cat D} F\Gamma'$, then $\nu_{F\rho \circ \sigma}(t) = \nu_\sigma(t)$. For this reason, we will write $\nu(t)$ for $\nu_\sigma(t)$.
		\item For a context $\Gamma, \var x : T$, we have $\nu = (\nu \circ F \wknvar x, \nu(\ftrtm F \var x) / \ftrtm G \var x)_G : F(\Gamma, \var x : T) \xrightarrow{\cat D} G( \Gamma, \var x : T)$.
		\item We have $\nu(\ftrtm F {t'}) = \ftrtm G {t'}[\nu]$.
		\item We have $(\nu \mu)(t) = \nu(\mu(t))$ and $\id(t) = t$.
		\item We have $\ftrtm{R}{(\nu(t))} = (R\nu)(\ftrtm R t)$.
	\end{enumerate}
\end{proposition}
\begin{proof}
	The unpairing of $\nu$ requires that $F(\Gamma, \var x : T) \sez \ftrvar G x [\nu] = \nu_{F \wknvar x}(\ftrvar F x) : (GT)[\nu \circ F \wknvar x]$. Now $t = (\ftrvar F x)[(\sigma, t/\ftrvar F x)_F]$, so naturality requires that we define $\nu_\sigma(t) = (\ftrvar G x)[\nu (\sigma, t / \ftrvar F x)_F]$.
	\begin{enumerate}
		\item This is easily seen to be natural in $\Delta$.
		\item For naturality in $\Gamma$:
		\begin{align}
			\nu_{F \rho \circ \sigma}(t)
			= (\ftrtm G \var x)[\nu(F \rho \circ \sigma, t / \ftrtm F \var x)_F]
			&= (\ftrtm G \var x)[\nu \circ F(\rho+) \circ (\sigma, t / \ftrtm F \var x)_F] \\ \nn
			&= (\ftrtm G \var x)[G(\rho+) \circ \nu \circ (\sigma, t / \ftrtm F \var x)_F]
			= (\ftrtm G \var x)[\nu(\sigma, t / \ftrtm F \var x)_F]
			= \nu_\sigma(t).
		\end{align}
		\item We have $G \wknvar x \circ \nu = \nu \circ F \wknvar x$ and $\ftrvar G x [\nu] = \ftrvar G x[\nu (F \wknvar x, \ftrvar F x / \ftrvar F x)_F] = \nu_{F\wknvar x}(\ftrvar F x)$.
		\item $\nu_\id(\ftrtm F {t'}) = (\ftrtm G \var x)[\nu (\id, \ftrtm F {t'} / \ftrtm F \var x)_F] = (\ftrtm G \var x)[\nu \circ F(\id, t' / \var x)] = (\ftrtm G \var x)[G(\id, t' / \var x) \circ \nu] = \ftrtm{G}{t'}[\nu]$.
		\item Assume $\mu : E \to F$ and $\nu : F \to G$. We have
		\begin{align}
			(\nu \mu)(t)
			&= \ftrtm G \var x[\nu \mu(\sigma, t / \ftrtm E \var x)_E]
			= \ftrtm G \var x[\nu (\mu \circ E \wknvar x, \mu(\ftrvar E x)/\ftrvar F x)_F(\sigma, t / \ftrtm E \var x)_E] \\ \nn
			&= \ftrtm G \var x[\nu (\mu \sigma, \mu(t) / \ftrvar F x)_F] = \nu(\mu(t)).
		\end{align}
		\item This is immediate from the definition. \qedhere
	\end{enumerate}
\end{proof}

\subsection{Adjoint morphisms of CwFs}
In this section, we consider functors $L : \cat C \to \cat D$ and $R : \cat D \to \cat C$, where $L$ may be and $R$ is a morphism of CwFs, such that $\alpha : L \dashv R$. Then $LR : \cat D \to \cat D$ will be a comonad and $RL : \cat C \to \cat C$ will be a monad. Of course, we have unit and co-unit natural transformations
\begin{equation}
	\eps : LR \to \Id_{\cat D}, \qquad \eta : \Id_{\cat C} \to RL,
\end{equation}
such that $\eps L \circ L \eta = \id_L$ and $R\eps \circ \eta R = \id_R$. The isomorphism $\alpha : \Hom(L\loch, \loch) \cong \Hom(\loch, R\loch)$ can be retrieved from unit and co-unit:
\begin{equation}
	\alpha(\sigma) = R\sigma \circ \eta, \qquad
	\alpha\inv(\tau) = \eps \circ L\tau,
\end{equation}
and vice versa:
\begin{equation}
	\eta = \alpha(\id), \qquad \eps = \alpha\inv(\id).
\end{equation}
Finally, $\alpha$ is natural in the following sense:
\begin{equation}
	\alpha(\tau \circ \sigma \circ L\rho) = R\tau \circ \alpha(\sigma) \circ \rho, \qquad
	\alpha\inv(R\tau \circ \sigma \circ \rho) = \tau \circ \alpha\inv(\sigma) \circ L\rho.
\end{equation}
If $L$ is a CwF morphism, then \cref{thm:nattrans-function} gives us functions $\eta : T \to (RLT)[\eta]$ and $\eps : LRT \to T[\eps]$. Moreover, $\eta(t) = (\ftrtm{RL}{t})[\eta]$ and $\eps(\ftrtm{LR}{t}) = t[\eps]$.

\begin{proposition}\label{thm:adjunction-rules}
	Assume $\Gamma \sez_{\cat D} T \type$. We have mutually inverse rules:
	\begin{equation}
		\inference{
			\sigma : L\Delta \xrightarrow{\cat D} \Gamma \\
			L\Delta \sez_{\cat D} t : T[\sigma]
		}{\Delta \sez_{\cat C} \alpha_\sigma(t) : (RT)[\alpha(\sigma)]}{}
		\qquad \qquad
		\inference{
			\tau : \Delta \xrightarrow{\cat D} R\Gamma \\
			\Delta \sez_{\cat C} t' : (RT)[\tau]
		}{L\Delta \sez_{\cat D} \alpha\inv_{\alpha\inv(\tau)}(t') : T[\alpha\inv(\tau)]}{}
	\end{equation}
	\begin{enumerate}
		\item These operations are natural in $\Delta$, i.e. $\alpha_\sigma(t)[\tau] = \alpha_{\sigma \circ L \tau}(t[L\tau])$.
		\item These operations are natural in $\Gamma$, i.e if $\rho : \Gamma' \to \Gamma$ and $\sigma : L\Delta \xrightarrow{\cat D} \Gamma'$, then $\alpha_{\rho \sigma}(t) = \alpha_\sigma(t)$. For this reason, we will omit the subscript on $\alpha$.
		\item $\alpha(t) = (\ftrtm R t)[\eta]$.
		\item If $L$ is a CwF morphism, then $\alpha\inv(t') = \eps(\ftrtm L {t'})$. In general, we have $\alpha\inv(t') = \xi[\eps \circ L((\eta, t')_R)]$.
	\end{enumerate}
\end{proposition}
\begin{proof}
	We have $\alpha : (L\Delta \to (\Gamma, \var x : T)) \cong (\Delta \to R(\Gamma, {\var x} : T))$. We will define $\alpha_\sigma(t)$ and $\alpha_{\alpha\inv(\tau)}(t')$ by the (equivalent) equations:
	\begin{equation}
		\alpha(\sigma, t / \var x) = (\alpha(\sigma), \alpha_\sigma(t) / \ftrtm R \var x)_R,
		\qquad
		\alpha\inv((\tau, t' / \ftrtm R \var x)_R) = (\alpha\inv(\tau), \alpha\inv_{\alpha\inv(\tau)}(t') / \var x).
	\end{equation}
	The first components of these equations are correct:
	\begin{align}
		R \wknvar x \circ \alpha(\sigma, t / \var x)
		&= \alpha(\wknvar x \circ (\sigma, t / \var x))
		= \alpha(\sigma), \\
		\wknvar x \circ \alpha\inv((\tau, t'/\ftrtm R \var x)_R)
		&= \alpha\inv(R \wknvar x \circ (\tau, t' / \ftrtm R \var x)_R)
		= \alpha\inv(\tau).
	\end{align}
	\begin{enumerate}
		\item This follows from naturality of the adjunction $\alpha$.
		\item We have $\alpha_{\rho \sigma}(t) = (\ftrvar R x)[\alpha(\rho \sigma, t / \var x)] = (\ftrvar R x)[R(\rho +) \circ \alpha(\sigma, t / \var x)] = \ftrtm{R}{(\var x[\rho +])}[\alpha(\sigma, t / \var x)] = \ftrvar R x[\alpha(\sigma, t/\var x)] = \alpha_\sigma(t)$. Then naturality for $\alpha\inv$ holds because it is inverse to $\alpha$.
		\item We have $\alpha_\sigma(t) = \ftrtm R \var x[\alpha(\sigma, t / \var x)] = \ftrtm R \var x[(R\sigma, \ftrtm R t / \ftrtm R \var x)_R \circ \eta] = \ftrtm R t[\eta]$.
		\item We have 
			$\alpha\inv_{\alpha\inv(\tau)}(t')
			= \alpha\inv_{\alpha\inv(\eta)}(t')
			= \var x[\alpha\inv((\eta, t' / \ftrtm R \var x)_R)]
			= \var x[\eps \circ L((\eta, t' / \ftrtm R \var x)_R)]$. If $L$ is a weak CwF morphism, then this reduces further to $\eps(\ftrtm{LR} \var x[(L\eta, \ftrtm{L}{t'}/\ftrtm{LR} \var x)_{LR}]) = \eps(\ftrtm{L}{t'})$. \qedhere
	\end{enumerate}
\end{proof}
\begin{corollary}
	We have naturality rules as for ordinary adjunctions:
	\begin{equation}
		\begin{array}{r c l c r c l}
			\alpha(\tau \circ \sigma \circ L\rho)
			&=& R\tau \circ \alpha(\sigma) \circ \rho
			&\qquad \qquad&
			\alpha\inv(R\tau \circ \sigma \circ \rho)
			&=& \tau \circ \alpha\inv(\sigma) \circ L\rho
			\\
			\alpha(t[\sigma][L\rho])
			&=& (\ftrtm R t)[\alpha(\sigma)][\rho]
			&&
			\alpha\inv((\ftrtm R t)[\sigma][\rho])
			&=& t[\alpha\inv(\sigma)][L\rho]
			\\
			\alpha(\nu(s[L \rho])) &=& (R\nu)(\alpha(s))[\rho]
			&&
			\alpha\inv((R\nu)s[\rho]) &=& \nu(\alpha\inv(s)[L \rho])
			\\
			\alpha(\nu(\mu(\ftrtm L r))) &=& (R\nu)(\alpha(\mu)(r))
			&&
			\alpha\inv((R\nu)(\mu(r))) &=& \nu(\alpha\inv(\mu)(\ftrtm L r))
		\end{array}
	\end{equation}
	where $\rho, \sigma, \tau$ denote substitutions, $s, t$ denote terms and $\mu, \nu$ denote natural transformations.
\end{corollary}
\begin{proof}
	Each equation on the right is equivalent to its counterpart on the left. The first equation on the left is old news. The other equations follow from $\alpha(t) = (\ftrtm R t)[\eta]$.
\end{proof}

\section{Presheaf models}\label{sec:psh-transform}
In this section, we study the implications of functors, natural transformations and adjunctions between categories $\catV$ and $\catW$ for the CwFs $\widehat{\catV}$ and $\widehat{\catW}$.

\subsection{Lifting functors}\label{sec:lifting-functors}
A functor $F : \catV \to \catW$ gives rise to a functor $\fpsh F : \widehat{\catW} \to \widehat{\catV} : \Gamma \mapsto \Gamma \circ F$.
\begin{theorem}\label{thm:fpsh-strict-cwf-morphism}
	For any functor $F : \catV \to \catW$, the functor $\fpsh F$ is a strict morphism of CwFs.
\end{theorem}
\begin{proof}
Throughout the proof, it is useful to think of $\fpsh F$ as being right adjoint to $F$. For that reason we will again add an ignorable label $\fpshadj F$ for readability.
\begin{enumerate}
	\item A context $\Gamma \in \widehat{\catW}$ is mapped to the context $\fpsh F \Gamma = \Gamma \circ F \in \Psh(\catV)$. A substitution $\sigma : \Delta \to \Gamma$ is mapped to a substitution $\fpsh F \sigma : \fpsh F \Delta \to \fpsh F \Gamma$ by functoriality of composition. Unpacking this and adding labels, we get:
	\begin{itemize}
		\item $(\DSub{V}{\fpsh F \Gamma}) = \set{\fpshadj F(\gamma)}{\gamma : \DSub{FV}{\Gamma}}$,
		\item For $\vfi : \PSub{V'}{V}$ and $\gamma : \DSub{V}{\fpsh F \Gamma}$, we have $\fpshadj F(\gamma) \circ \vfi = \fpshadj F(\gamma \circ F \vfi) : \DSub{V'}{\fpsh F \Gamma}$.
		\item For $\sigma : \Delta \to \Gamma$, we get $\fpsh F \sigma \circ \fpshadj F(\delta) = \fpshadj F(\sigma \circ \delta)$.
	\end{itemize}
	
	\item We easily construct a functor $\int_{\catV}\fpsh F \Gamma \to \int_{\catW} \Gamma$ sending $(V, \gamma)$ to $(FV, \gamma)$. Precomposing with this functor, yields a map $\Ty(\Gamma) \to \Ty(\fpsh F \Gamma) : T \mapsto \fpsh F T$.
	Let us unpack and label this construction. Given $(\Gamma \sez_{\widehat{\catW}} T \type)$, the type $(\fpsh F \Gamma \sez_{\widehat{\catV}} \fpsh F T \type)$ is defined by:
	\begin{itemize}
		\item Terms $V \Dsez \fpshadj F(t) : (\fpsh F T) \dsub{\fpshadj F(\gamma)}$ are obtained by labelling terms $FV \Dsez_{\widehat{\catW}} t : T \dsub \gamma$.
		\item $\fpshadj F(t) \psub \vfi = \fpshadj F(t \psub{F \vfi})$.
	\end{itemize}
	This is natural in $\Gamma$, because
	\begin{align}
		\fpsh F(T \sub \sigma) \dsub{\fpshadj F(\gamma)} &\cong T \sub \sigma \dsub \gamma = T \dsub{\sigma \gamma} \cong (\fpsh F T) \dsub{\fpshadj F(\sigma \gamma)} = (\fpsh F T) \sub{\fpsh F \sigma} \dsub{\fpshadj F(\gamma)},
	\end{align}
	where the isomorphisms become equalities if we ignore labeling.
	
	\item Given $\Gamma \sez_{\widehat{\catW}} t : T$, we define $\fpsh F \Gamma \sez_{\widehat{\catV}} \ftrtm{\fpsh F}{t} : \fpsh F T$ by setting $(\ftrtm{\fpsh F}{t})\dsub{\fpshadj F(\gamma)}$ equal to $\fpshadj F(t \dsub \gamma)$. This is natural in the domain $V$ of $\fpshadj F(\gamma)$:
	\begin{equation}
		(\ftrtm{\fpsh F} t) \dsub{\fpshadj F(\gamma)} \psub \vfi
		= \fpshadj F(t \dsub \gamma) \psub{\vfi}
		= \fpshadj F(t \dsub{\gamma \circ F\vfi})
		= (\ftrtm{\fpsh F} t) \dsub{\fpshadj F(\gamma \circ F \vfi)}
		= (\ftrtm{\fpsh F} t) \dsub{\fpshadj F(\gamma) \circ \vfi}.
	\end{equation}
	It is also natural in $\Gamma$ by the same reasoning as for types.
	
	\item The terminal presheaf over $\catW$ is automatically mapped to the terminal presheaf over $\catV$.
	
	\item It is easily checked that comprehension, $\pi$, $\xi$ and $(\cwfpair \loch \loch)$ are preserved on the nose if we assume that $\fpshadj F(\gamma, t) = (\fpshadj F(\gamma), \fpshadj F(t))$, e.g. by ignoring labels.\qedhere
\end{enumerate}
\end{proof}
\begin{proposition}\label{thm:lifted-preserves-types}
	For any functor $F : \catV \to \catW$, the morphism of CwFs $\fpsh F$ preserves all operators related to $\Sigma$-types, $\Weld$-types and identity types on the nose.
\end{proposition}
\begin{proof}
	The defining terms of each of these types, are built from of other defining terms. This is in constrast with $\Pi$- and $\Glue$-types, where defining terms also contain non-defining terms, and the universe, whose defining terms even contain types. As $\fpsh F$ merely reshuffles defining terms, it respects each of these operations (ignoring labels).
\end{proof}

The reason we use a different notation for terms and for types ($\ftrtm{\fpsh F} t$ versus ${\fpsh F} T$) is to avoid confusion when it comes to encoding and decoding types: $\fpsh F$ acts very differently on types and on terms of the universe. To begin with, ${\fpsh F} \uni{}$ will typically not be the universe of $\widehat{\catV}$. Indeed, its primitive terms still originate as types over the primitive contexts of $\catW$. So if $\Gamma \sez_{\widehat{\catW}} A : \uni{}$, then $\ftrtm{\fpsh F} A$ is a representation in $\widehat{\catV}$ of a type from $\widehat{\catW}$. In contrast, $\fpsh F(\El\,A)$ is a type in $\widehat{\catV}$. Put differently still, when applied to an element of the universe, ${\fpsh F}$ reorganizes the presheaf structure of the universe. When applied to a type $T$, ${\fpsh F}$ reorganizes the presheaf structure of $T$.

\begin{example}\label{ex:presheaf-sets}
	Let $\pointcat$ be the terminal category with just a single object $()$ and only the identity morphism. It is easy to see that $\widehat{\pointcat} \cong \Set$. Meanwhile, let $\RGcat$ be the base category with objects $()$ and $(\var i : \IE)$ and maps non-freely generated by $() : (\var i : \IE) \to ()$ and $(0/\var i), (1/\var i) : () \to (\var i : \IE)$, such that $\widehat \RGcat$ is the category of reflexive graphs. We have a functor $F : \pointcat \to \RGcat$ sending $()$ to $()$. This constitutes a functor $\fpsh F : \widehat \RGcat \to \widehat \pointcat$ sending a reflexive graph to its set of nodes. It is not surprising that this functor is sufficiently well-behaved to be a CwF morphism.
	
	This example also illustrates that the universe is not preserved. The nodes of the universe in $\widehat \RGcat$ are reflexive graphs. Thus $\fpsh F \uni{}$ will be the set of reflexive graphs. Then if $\Gamma \sez_{\widehat \RGcat} A : \uni{}$ in $\widehat \RGcat$, then $\fpsh F \Gamma \sez_{\widehat \pointcat} \ftrtm{\fpsh F}{A} : \fpsh F \uni{}$ sends nodes of $\Gamma$ to reflexive graphs. The edges \emph{between} types are forgotten. However, the type $\fpsh F \Gamma \sez_{\widehat \pointcat} \fpsh F (\El\,A) \type$ is simply the dependent set of nodes of $\El\,A$. Here, the edges \emph{within} $\El\,A$ are also forgotten.
	
	In general, from $\var X : \uni{} \sez \El\,\var X \type$, we can deduce $\ftrtm{\fpsh F}{\var X} : \fpsh F \uni{} \sez \fpsh F(\El\,\var X) \type$.
\end{example}

\subsection{Lifting natural transformations}
Assume we have functors $F, G : \catV \to \catW$ and a natural transformation $\nu : F \to G$. Then we get morphisms of CwFs $\fpsh F, \fpsh G : \Psh(\catW) \to \Psh(\catV)$ and, since presheaves are contravariant, a natural transformation $\fpsh \nu : \fpsh G \to \fpsh F$.

Let us see how $\fpsh \nu$ works. Pick $\fpshadj G(\gamma) : \DSub{V}{\fpsh G \Gamma}$. Then $\gamma : \DSub{GV}{\Gamma}$ and $\gamma \circ \nu : \DSub{FV}{\Gamma}$. Now $\fpsh \nu \circ \fpshadj G(\gamma) = \fpshadj F(\gamma \circ \nu)$. This is natural because
\begin{equation}
	\fpsh \nu \circ (\fpshadj G(\gamma) \circ \vfi)
	= \fpsh \nu \circ \fpshadj G(\gamma \circ G\vfi)
	= \fpshadj F(\gamma \circ G\vfi \circ \nu)
	= \fpshadj F(\gamma \circ \nu \circ F\vfi)
	= \fpshadj F(\gamma \circ \nu) \circ \vfi
	= (\fpsh \nu \circ \fpshadj F(\gamma)) \circ \vfi.
\end{equation}

\subsection{Lifting adjunctions}
\begin{proposition}\label{thm:lifting-adjunctions}
	Assume we have functors $L : \catV \to \catW$ and $R : \catW \to \catV$ such that $\alpha : L \dashv R$ with unit $\eta : \Id \to RL$ and co-unit $\eps : LR \to \Id$. Then we get $\fpsh L : \widehat \catW \to \widehat \catV$ and $\fpsh R : \widehat \catV \to \widehat \catW$ and we have $\fpsh \alpha : \fpsh L \dashv \fpsh R$ with unit $\fpsh \eps : \Id \to \fpsh R \fpsh L$ and co-unit $\fpsh \eta : \fpsh L \fpsh R \to \Id$. Moreover, $\fpsh L \circ \yoneda \cong \yoneda \circ R : \catW \to \widehat \catV$.
\end{proposition}
So $\fpsh \loch$ is an operation that takes the right adjoint of a functor and immediately extends it to the entire presheaf category.
\begin{proof}
	To prove the adjunction, it is sufficient to prove that $\fpsh \eta \circ \fpsh L \fpsh \eps = \id_{\fpsh L}$ and $\fpsh R \fpsh \eta \circ \fpsh \eps = \id_{\fpsh R}$. We need to assume that $\fpshadj{\Id} = \id$ and $\fpshadj{FG} = \fpshadj{G} \circ \fpshadj{F}$.
	
	Pick $\fpshadj L(\gamma) : \DSub{V}{\fpsh L \Gamma}$. Then
	\begin{align}
		\fpsh \eta \circ \fpsh L \fpsh \eps \circ \fpshadj L(\gamma)
		&= \fpsh \eta \circ \fpshadj L(\fpsh \eps \circ \gamma)
		= \fpsh \eta \circ \fpshadj L(\fpshadj{R} (\fpshadj{L}(\gamma \circ \eps))) \\ \nn
		&= \fpshadj{L}(\gamma \circ \eps) \circ \eta
		= \fpshadj{L}(\gamma \circ \eps \circ L\eta)
		= \fpshadj{L}(\gamma).
	\end{align}
	
	Similarly, pick $\fpshadj R(\delta) : \DSub{W}{\fpsh R \Gamma}$. Then
	\begin{align}
		\fpsh R \fpsh \eta \circ \fpsh \eps \circ \fpshadj R(\delta)
		&= \fpsh R \fpsh \eta \circ \fpshadj R(\fpshadj L(\fpshadj R(\delta) \circ \eps))
		= \fpsh R \fpsh \eta \circ \fpshadj R(\fpshadj L(\fpshadj R(\delta \circ R\eps))) \\ \nn
		&= \fpshadj R(\fpsh \eta \circ \fpshadj L(\fpshadj R(\delta \circ R\eps)))
		= \fpshadj R(\delta \circ R \eps \circ \eta) = \fpshadj R(\delta).
	\end{align}
	To see the isomorphism:
	\begin{equation}
		(\DSub V {\fpsh L \yoneda W})
		\cong (\DSub{LV}{\yoneda W})
		= (\PSub{LV}{W})
		\cong (\PSub{V}{RW})
		= (\DSub{V}{\yoneda RW}).
	\end{equation}
	This is clearly natural.
\end{proof}

\subsection{The left adjoint to a lifted functor}\label{sec:left-adjoint-to-lifted-functor}
Assume a functor $F : \catV \to \catW$. Under reasonable conditions, one finds a general construction of a \emph{functor} $\lpsh F : \widehat \catV \to \widehat \catW$ that is left adjoint to $\fpsh F : \widehat \catW \to \widehat \catV$ \cite[00VC]{stacks-project17}. Although we will need such a left adjoint at some point, the general construction is overly complicated for our needs. Therefore, we will construct that functor ad hoc when we need it.
For now, we simply assume its existence and prove a lemma and some bad news:
\begin{lemma}\label{thm:yoneda-and-left-adjoint}
	Suppose we have a functor $\lpsh F : \widehat \catV \to \widehat \catW$ and a functor $F : \catV \to \catW$, such that $\lpsh F \dashv \fpsh F$. Then $\lpsh F \circ \yoneda \cong \yoneda \circ F$.
\end{lemma}
\begin{proof}
	We have a chain of isomorphisms, natural in $W$ and $\Gamma$:
	\begin{equation}
		(\lpsh F \yoneda W \to \Gamma) \cong (\yoneda W \to \fpsh F \Gamma) \cong (\DSub{F W}{\Gamma}) \cong (\yoneda F W \to \Gamma).
	\end{equation}
	Call the composite of these isomorphisms $f$. Then we have $f(\id_{\lpsh F \yoneda}) : \yoneda F \to \lpsh F \yoneda$ and $f\inv(\id_{\yoneda F}) : \lpsh F \yoneda \to \yoneda F$. By naturality in $\Gamma$, we have:
	\begin{equation}
		f(\id) \circ f\inv(\id) = f\inv(f(\id) \circ \id) = f\inv(f(\id)) = \id,
	\end{equation}
	\begin{equation}
		f\inv(\id) \circ f(\id) = f(f\inv(\id) \circ \id) = f(f\inv(\id)) = \id. \qedhere
	\end{equation}
\end{proof}
\begin{proposition}\label{thm:left-adjoint-not-cwf}
	The functor $\lpsh F : \widehat \catV \to \widehat \catW$ is not in general a morphism of CwFs.
\end{proposition}
\begin{proof}
	Consider the only functor $G : \RGcat \to \pointcat$; it sends $()$ and $(\var i : \IE)$ to $()$. Then $\fpsh G : \widehat \pointcat \to \widehat \RGcat$ sends a set $S$ to the discrete reflexive graph on $S$.
	
	Its left adjoint $\lpsh G$ sends a reflexive graph $\Gamma$ to its colimit. That is, it identifies all edge-connected nodes. Now consider a type $\Gamma \sez_{\widehat \RGcat} T \type$:
	\begin{equation}
		\xymatrix{
			\gamma \ar@{-}[d] &&& t \ar@{-}[dl] \ar@{-}[dr]
			\\
			\gamma' && t' && t''
		}
	\end{equation}
	That is: $\Gamma$ contains two nodes and an edge connecting them (as well as constant edges). $T$ has one node above $\gamma$ and two nodes above $\gamma'$ and they are connected as shown.
	
	There are two substitutions from $\yoneda O$, namely $\gamma$ and $\gamma' : \yoneda O \to \Gamma$. Clearly, $\lpsh G(T[\gamma]) \neq \lpsh G(T[\gamma'])$. But $\lpsh G (\gamma) = \lpsh G (\gamma')$. Thus, $\lpsh G$ cannot preserve substitution in the sense that $\lpsh G(T[\sigma])$ is equal to $\lpsh GT[\lpsh G \sigma]$.
\end{proof}

\subsection{The right adjoint to a lifted functor}
In this section, we look at the general construction of the right adjoint to a lifted functor. This section is not a prerequisite for \cref{part:paramdtt}.
\begin{definition}
	Given a functor $F : \catV \to \catW$, we define $\rpsh F : \widehat \catV \to \widehat \catW$ by
	\begin{equation}
		(\DSub{W}{\rpsh F \Gamma}) := (\fpsh F \yoneda W \to \Gamma),
	\end{equation}
	which is obviously contravariant in $W$. As usual, we treat this equality as though it were a non-trivial isomorphism: if $\sigma : \fpsh F \yoneda W \to \Gamma$, then we write $\rpshadj F(\sigma) : \DSub{W}{\rpsh F \Gamma}$.
\end{definition}
Thus, we have $\rpshadj F(\sigma) \circ \vfi = \rpshadj F(\sigma \circ \fpsh F \vfi)$ and $\rpsh F \tau \circ \rpshadj F(\sigma) = \rpshadj F(\tau \circ \sigma)$. (Recall that we do not write $\yoneda$ when applied to morphisms, treating primitive substitutions as substitutions.)
\begin{proposition}
	The functor $\rpsh F : \widehat \catV \to \widehat \catW$ is right adjoint to $\fpsh F : \widehat \catW \to \widehat \catV$.
\end{proposition}
\begin{proof}
	We construct unit and co-unit and prove the necessary equalities.
	\begin{description}
		\item[$\eta$] Pick $\Gamma \in \widehat \catW$ and $\gamma : \DSub W \Gamma$. We need to construct $\eta \circ \gamma : \DSub{W}{\rpsh F \fpsh F \Gamma}$. Since we have $\fpsh F \gamma : \fpsh F \yoneda W \to \fpsh F \Gamma$, we can set $\eta \circ \gamma := \rpshadj F(\fpsh F \gamma)$.
		To see that this is natural in $W$, pick $\vfi : \PSub{W'}{W}$. Then $\rpshadj F(\fpsh F (\gamma \vfi)) = \rpshadj F(\fpsh F \gamma \circ \fpsh F \vfi) = \rpshadj F(\fpsh F \gamma) \vfi$.
		
		\item[$\eps$] Pick $\Gamma \in \widehat \catV$ and $\fpshadj F(\rpshadj F(\sigma)) : \DSub{V}{\fpsh F \rpsh F \Gamma}$. We need to construct $\eps \circ \fpshadj F(\rpshadj F(\sigma)) : \DSub V \Gamma$. We have $\rpshadj F(\sigma) : \DSub{FV}{\rpsh F \Gamma}$, i.e. $\sigma : \fpsh F \yoneda F V \to \Gamma$. From $\id_{FV} : \DSub{FV}{\yoneda FV}$, we get $\fpshadj F(\id_{FV}) : \DSub{V}{\fpsh F \yoneda F V}$, so we can set $\eps \circ \fpshadj F(\rpshadj F(\sigma)) := \sigma \circ \fpshadj F(\id_{FV}) : \DSub V \Gamma$. To see that this is natural, pick $\vfi : \PSub{V'}{V}$. Then we have
		\begin{align}
			\sigma \circ \fpshadj F(\id_{FV}) \circ \vfi
			&= \sigma \circ \fpshadj F(F \vfi)
			= \sigma \circ \fpsh F F \vfi \circ \fpshadj F(\id_{FV'})
		\end{align}
		while
		\begin{equation}
			\fpshadj F (\rpshadj F(\sigma)) \circ \vfi
			= \fpshadj F(\rpshadj F(\sigma) \circ F \vfi)
			= \fpshadj F(\rpshadj F(\sigma \circ \fpsh F F \vfi)).
		\end{equation}
		
		\item[$\rpsh F \eps \circ \eta \rpsh F = \id$] Pick $\rpshadj F(\sigma) : \DSub{W}{\rpsh F \Gamma}$. We have
		\begin{align}
			\rpsh F \eps \circ \eta \rpsh F \circ \rpshadj F(\sigma)
			&= \rpsh F \eps \circ \rpshadj F(\fpsh F \rpshadj F(\sigma))
			= \rpshadj F(\eps \circ \fpsh F \rpshadj F(\sigma)).
		\end{align}
		So it remains to show that $\sigma = \eps \circ \fpsh F \rpshadj F(\sigma) : \fpsh F \yoneda W \to \Gamma$. So pick $\fpshadj F(\vfi) : \DSub{V}{\fpsh F \yoneda W}$. Then we have
		\begin{align}
			\eps \circ \fpsh F \rpshadj F(\sigma) \circ \fpshadj F(\vfi)
			&= \eps \circ \fpshadj F(\rpshadj F(\sigma) \circ \vfi)
			= \eps \circ \fpshadj F(\rpshadj F(\sigma \circ \fpsh F \vfi)) \\ \nn
			&= \sigma \circ \fpsh F \vfi \circ \fpshadj F(\id_{FV}) = \sigma \circ \fpshadj F(\vfi),
		\end{align}
		as required.
		
		\item[$\eps \fpsh F \circ \fpsh F \eta = \id$] Pick $\fpshadj F(\gamma) : \DSub{V}{\fpsh F \Gamma}$. We have
		\begin{align}
			\eps \fpsh F \circ \fpsh F \eta \circ \fpshadj F(\gamma)
			&= \eps \fpsh F \circ \fpshadj F(\eta \circ \gamma)
			= \eps \fpsh F \circ \fpshadj F(\rpshadj F(\fpsh F \gamma)) \\ \nn
			&= \fpsh F \gamma \circ \fpshadj F(\id_{FV}) = \fpshadj F (\gamma). \qedhere
		\end{align}
	\end{description}
\end{proof}

\begin{proposition}\label{thm:rpsh-cwf-morphism}
	The functor $\rpsh F : \widehat \catV \to \widehat \catW$ is a (weak) morphism of CwFs.
\end{proposition}
\begin{proof}
	\begin{enumerate}
		\item The action of $\rpsh F$ on contexts and substitutions is of course $\rpsh F$ itself.
		\item Pick $\Gamma \sez_{\widehat \catV} T \type$. We define $\rpsh F \Gamma \sez_{\widehat \catW} \rpsh F T \type$.
		\begin{itemize}
			\item The defining terms $W \Dsez_{\widehat \catW} \rpshadj F(t) : \rpsh F T \dsub{\rpshadj F(\sigma)}$ are labelled terms $\fpsh F \yoneda W \sez_{\widehat \catV} t : T[\sigma]$.
			\item Given $\vfi : \PSub{W'}{W}$, we define $\rpshadj F(t) \psub \vfi := \rpshadj F(t[\fpsh F \vfi])$.
		\end{itemize}
		We need to prove that this is natural in $\Gamma$, i.e. that $\rpsh F(T[\tau]) = \rpsh F T [\rpsh F \tau]$. We have
		\begin{align*}
			&W \Dsez_{\widehat \catW} \rpshadj F(t) : \rpsh F (T[\tau]) \dsub{\rpshadj F(\sigma)} \\ \nn
			&\Leftrightarrow
			\fpsh F \yoneda W \sez_{\widehat \catV} t : (T[\tau])[\sigma] \\ \nn
			&\Leftrightarrow
			\fpsh F \yoneda W \sez_{\widehat \catV} t : T[\tau \sigma] \\ \nn
			&\Leftrightarrow
			W \Dsez_{\widehat \catW} \rpshadj F(t) : \rpsh F T \dsub{\rpshadj F(\tau \sigma)} \\ \nn
			&\Leftrightarrow
			W \Dsez_{\widehat \catW} \rpshadj F(t) : \rpsh F T [\rpsh F \tau] \dsub{\rpshadj F(\sigma)}.
		\end{align*}
		
		\item Pick $\Gamma \sez_{\widehat \catV} t : T$. We define $\rpsh F \Gamma \sez_{\widehat \catW} \ftrtm{\rpsh F}{t} : \rpsh F T$ by
		\begin{equation}
			\ftrtm{\rpsh F}{t} \dsub{\rpshadj F(\sigma)} = \rpshadj F(t[\sigma]).
		\end{equation}
		This is natural, since
		\begin{equation}
			\ftrtm{\rpsh F}{t} \dsub{\rpshadj F(\sigma)} \psub{\vfi}
			= \rpshadj F(t[\sigma]) \psub \vfi
			= \rpshadj F(t[\sigma \circ \fpsh F \vfi])
			= \ftrtm{\rpsh F}{t} \dsub{\rpshadj F(\sigma \circ \fpsh F \vfi)}
			= \ftrtm{\rpsh F}{t} \dsub{\rpshadj F(\sigma) \circ \vfi}.
		\end{equation}
		
		\item We show that $\rpsh F ()$ is terminal. We have
		\begin{equation}
			(\DSub{W}{\rpsh F ()}) = (\fpsh F \yoneda W \to ()),
		\end{equation}
		which is a singleton.
		
		\item We show that $(\rpsh F \pi, \ftrtm{\rpsh F}{\xi}) : \rpsh F (\Gamma.T) \to \rpsh F \Gamma.\rpsh F T$ is invertible. We will first check what it does. Pick $\rpshadj F(\tau) : \DSub{W}{\rpsh F (\Gamma.T)}$. Then $\tau : \fpsh F \yoneda W \to \Gamma.T$ is of the form $(\sigma, t)$, where $\sigma = \pi \tau$ and $t = \xi[\tau]$. We have
		\begin{align}
			(\rpsh F \pi, \ftrtm{\rpsh F}{\xi}) \circ \rpshadj F(\sigma, t)
			&= (\rpsh F \pi \circ \rpshadj F(\sigma, t), \ftrtm{\rpsh F}{\xi}[\rpshadj F(\sigma, t)]) \\ \nn
			&= (\rpshadj F(\pi \circ (\sigma, t)), \rpshadj F(\xi[\sigma, t]))
			= (\rpshadj F(\sigma), \rpshadj F(t)).
		\end{align}
		This can clearly be inverted by mapping $(\rpshadj F(\sigma), \rpshadj F(t))$ to $\rpshadj F(\sigma, t)$. \qedhere
	\end{enumerate}
\end{proof}
To the extent possible, we will avoid inspecting the definition of $\rpsh F$.

\begin{proposition}\label{thm:right-adjoint-preserves-types}
	For any functor $F : \catV \to \catW$, the morphism of CwFs $\rpsh F$ preserves all operators related to $\Sigma$-types and identity types up to isomorphism.
\end{proposition}
\begin{proof}
	\begin{description}
		\item[$\Sigma$-types] Assume $\Gamma \sez A \type$ and $\Gamma.A \sez B \type$. We show that $\rpsh F (\Sigma A B)$ is isomorphic to $\Sigma(\rpsh F A)((\rpsh F B)[(\pi, \xi)_{\rpsh F}])$.
		
		We have correspondences (i.e. isomorphisms between the sets whose element relation is denoted using these judgement notations) as follows:
		\begin{equation*}
			\xymatrix{
				&{W \Dsez \rpshadj F (a, b) : (\rpsh F (\Sigma A B))\dsub{\rpshadj F (\sigma)}} \ar@{-}[d] \\
				&{\fpsh F \yoneda W \sez (a, b) : (\Sigma A B)[\sigma]} \ar@{-}[ld] \ar@{-}[d]
				\\
				\fpsh F \yoneda W \sez a : A[\sigma] \ar@{-}[dd]
				& \fpsh F \yoneda W \sez b : B[(\sigma, a)] \ar@{-}[d]
				\\
				& W \Dsez \rpshadj F(b) : B \dsub{\rpshadj F (\sigma, a)} \ar@{-}[d]
				\\
				W \Dsez \rpshadj F(a) : (\rpsh F A) \dsub{\rpshadj F(\sigma)} \ar@{-}[rd]
				& W \Dsez \rpshadj F(b) : (\rpsh F B) \sub{(\pi, \xi)_{\rpsh F}} \dsub{\rpshadj F (\sigma), \rpshadj F (a)} \ar@{-}[d]
				\\
				& W \Dsez \paren{\rpshadj F (a), \rpshadj F (b)} : \paren{\Sigma(\rpsh F A)((\rpsh F B)[(\pi, \xi)_{\rpsh F}])} \dsub{\rpshadj F (\sigma)}
			}
		\end{equation*}
		One can show that the pairing and projection operations are also respected.
		
		\item[Identity types] Assume $\Gamma \sez a, b : A$. We can show that $\rpsh F(a \idtp A b)$ is isomorphic to $\ftrtm{\rpsh F}{a} \idtp{\rpsh F A} \ftrtm{\rpsh F}{b}$. We have
		\begin{equation*}
			\xymatrix{
				W \Dsez \rpshadj F(p) : \rpsh F (a \idtp A b) \dsub{\rpshadj F (\sigma)} \ar@{-}[d]
				\\
				\fpsh F \yoneda W \sez p : (a \idtp A b)[\sigma] \ar@{-}[d]
				\\
				\fpsh F \yoneda W \sez a[\sigma] = b[\sigma] : A[\sigma] \ar@{-}[d]
				\\
				W \Dsez (\ftrtm{\rpsh F}{a})\dsub{\rpshadj F (\sigma)} = (\ftrtm{\rpsh F}{b})\dsub{\rpshadj F (\sigma)} : (\rpsh F A)\dsub{\rpshadj F (\sigma)} \ar@{-}[d]
				\\
				W \Dsez q : \paren{\ftrtm{\rpsh F}{a} \idtp{\rpsh F A} \ftrtm{\rpsh F}{b}} \dsub{\rpshadj F (\sigma)}
			}
		\end{equation*}
		
		\item[$\Weld$-types]
	\end{description}
\end{proof}

\begin{proposition}\label{thm:right-adjoint-no-weld}
	The right adjoint to a lifted functors does not always preserve $\Weld$-types up to isomorphism.
\end{proposition}
\begin{proof}
	Consider the functor $\cohdisc : \pointcat \to \RGcat$ sending $()$ to $()$. Then $\lpsh \cohdisc$ sends a set to its discrete reflexive graph, $\fpsh \cohdisc$ sends a graph to its set of nodes, and $\rpsh \cohdisc$ sends a set to its codiscrete reflexive graph.
	
	Now define a context $\Gamma \in \widehat{\pointcat}$ by setting $(\DSub{()}{\Gamma}) = \accol{\gamma_0, \gamma_1}$, a type $\Gamma \sez A \type$ by setting $A \dsub{\gamma_0} = \eset$ and $A \dsub{\gamma_1} = \accol{a_1}$, and a proposition $\Gamma \sez P \prop$ by setting $P \dsub{\gamma_0} = \accol{\star}$ and $P \dsub{\gamma_1} = \eset$. Finally, define $\Gamma.P \sez T$ by setting $T \dsub{\gamma_0, \star} = \accol{t_0}$. Define $\Gamma.P \sez f : A[\pi] \to T$ from the absurd. Let $\Omega = \Weldsys{A}{\Weldsysclauseb{P}{T}{f}}$. Then we have
	\begin{align*}
		\Omega \dsub{\gamma_0} &= T \dsub{\gamma_0, \star} = \accol{t_0} \\
		\Omega \dsub{\gamma_1} &= \set{\dweld~a}{a : A \dsub{\gamma_1}} = \accol{\dweld~a_1}.
	\end{align*}
	Hence, we can define $\Gamma \sez w : \Omega$ by $w \dsub{\gamma_0} = t_0$ and $w \dsub{\gamma_1} = \dweld~a_1$. Then we also get
	\begin{equation}
		\rpsh \cohdisc \Gamma \sez \ftrtm{\rpsh \cohdisc}{w} : \rpsh \cohdisc \Omega.
	\end{equation}
	
	Meanwhile, define $\Omega' = \Weldsys{\rpsh \cohdisc A}{\Weldsysclauseb{\rpsh \cohdisc P}{(\rpsh \cohdisc T)[(\pi, \xi)_{\rpsh \cohdisc}]}{\lambda((\ftrtm{\rpsh \cohdisc}{\ap~f})[(\pi, \xi)_{\rpsh \cohdisc}])}}$. Now suppose $\Omega \cong \Omega'$, then in particular we have some term
	\begin{equation}
		\rpsh \cohdisc \Gamma \sez w' : \Omega'.
	\end{equation}
	Now there is a defining substitution $\gamma : \DSub{(\var i : \IE)}{\rpsh \cohdisc \Gamma}$ which is an edge from `$\gamma_0$' to `$\gamma_1$'. Since $(\rpsh \cohdisc P) \dsub \gamma$ is false, we have
	\begin{equation}
		\Omega' \dsub{\gamma} = \set{\dweld~a}{a : \rpsh \cohdisc A} = \eset,
	\end{equation}
	so the object $w' \dsub \gamma$ cannot exist.
	
	The idea of this proof is that we built a non-codiscrete graph by welding together codiscrete graphs.
\end{proof}

\part{A Model of Parametric Dependent Type Theory in Bridge/Path Cubical Sets}\label{part:paramdtt}
\chapter{Bridge/path cubical sets}\label{ch:bpcubecat}
In this chapter, we move from the general presheaf model to the category $\widehat{\cubecat}$ of cubical sets and the category $\widehat{\bpcubecat}$ of bridge/path cubical sets. Throughout, we assume the existence of an infinite alphabet $\aleph$ of variable names, as well as a function $\fresh : \Pi(A \subseteq \aleph).(\aleph \setminus A)$ that picks a fresh variable for a given set of variables $A$.

\section{The category of cubes}\label{sec:cubecat}
\label{sec:cube}
In this section, we define the category of cubes $\cubecat$; presheaves over this category are called \textdef{cubical sets}. The reason we are interested in cubical sets, is that they generalize reflexive graphs: they contain points and edges, but also edges between edges (squares), edges between squares (cubes), etc. Imagine we have a type $\IE$ that contains two points, connected by an edge (as opposed to a bridge or a path). Then an $n$-dimensional cube is a value that ranges over $n$ variables of type $\IE$. For this reason, we define a \textdef{cube} as a finite set $W \subseteq \aleph$. We write $()$ for the 0-dimensional cube, $(W, \ctxedge{\var i})$ for $W \uplus \accol{\var i}$, implying that $\var i \not\in W$, and similarly $(V, W)$ for $V \uplus W$.

A face map $\vfi : \PSub V W$ assigns to every variable $\var i \in W$ a value $\var i \psub \vfi$ which is either 0, 1 (up to $n - 1$ for $n$-ary parametricity), or a variable in $V$. We use common substitution notation to denote face maps. So $\var i \psub{()} = \var i$, while $\var i \psub{\vfi, t / \var i} = t$. We write $\vfi = (\psi, f(\var i)/\var i \in W') : \PSub V W$ to denote that $\var i \psub \vfi = f(\var i)$ for all $\var i \in W' \subseteq W$. To emphasize that a face map does not use a variable $\var i$, we write $\vfi = (\psi, \facewkn{\var i})$. We make sure that different clauses in the same substitution do not conflict; hence we need no precedence rules.

Then a cubical set $\Gamma$ contains, for every cube $W$, a set of $|W|$-dimensional cubes $\DSub W \Gamma$. Every such cube has $2^{|W|}$ vertices, extractable using the face maps $(\epsilon(\var i)/\var i \in M) : \PSub \eset W$ for all $\epsilon : W \to \accol{0, 1}$. By substituting 0 or 1 for only some variables, we obtain the sides of a cube. We can create flat (\textdef{degenerate}/\textdef{constant}) cubes by not using variables, e.g. $(\facewkn{\var i}) : \PSub {(M, \ctxedge{\var i})} M$. Cubes also have diagonals, e.g. $(\var i/\var j) : \PSub{(\ctxedge{\var i})}{(\ctxedge{\var i, \var j})}$.

It is easy to see that $\cubecat$ is closed under finite cartesian products; namely $V \times W = (V, W)$. Because the Yoneda-embedding preserves limits such as cartesian products, we have $\yoneda W \cong (\yoneda \IE)^W$.

This is not the only useful definition. For example, \cite{moulin-param3} and \cite{huber} consider cubes which have no diagonals, while \cite{cubical} uses cubes that have `connections', a way of adding a dimension by folding open a line to become a square with two adjacent constant sides.

\section{The category of bridge/path cubes}
\label{sec:bpcube}
The novel category of bridge/path cubes $\bpcubecat$ is similar to $\cubecat$ but its cubes have two flavours of dimensions: bridge dimensions and path dimensions. So a \textdef{bridge/path cube} $W$ is a pair $(W_\IB, W_\IP)$, where $W_\IB$ and $W_\IP$ are disjoint subsets of $\aleph$. We write $()$ for the 0-dimensional cube, $(W, \ctxbrid{\var i})$ for $(W_\IB \uplus \accol{\var i}, W_\IP)$ and $(W, \ctxpath{\var i})$ for $(W_\IB, W_\IP \uplus \accol{\var i})$.

A face map $\vfi : \PSub V W$ assigns to every bridge variable $(\ctxbrid{\var i}) \in W$ either 0, 1 or a bridge variable from $V$, and to every path variable $(\ctxpath{\var i}) \in W$ either 0, 1, or a path or bridge variable from $V$. We will sometimes add a superscript to make the status of a variable clear, e.g. $(\var j^\IB / \var i^\IP) : \PSub{(\var j : \IB)}{(\var i : \IP)}$.

Then a bridge/path cubical set $\Gamma$ contains, for every bridge/path cube $W$, a set of cubes with $|W_\IB|$ bridge dimensions and $|W_\IP|$ path dimensons. Again we can extract vertices and faces and we can create flat cubes by introducing bridge or path dimensions. We can weaken paths to bridges, e.g. $(\var j/\var i) : \PSub{(W, \ctxbrid{\var j})}{(W, \ctxpath{\var i})}$. Finally, we can extract diagonals, but if the cube of which we take the diagonal, has at least one bridge dimension, then the diagonal has to be a bridge, e.g. $(\var j^\IB/\var i^\IP) : \PSub{(W, \ctxbrid{\var j})}{(W, \ctxpath{\var i}, \ctxbrid{\var j})}$.

\section{The cohesive structure of $\widehat{\bpcubecat}$ over $\widehat{\cubecat}$}
In this section, we construct a chain of five adjoint functors (all but the leftmost one morphisms of CwFs) between $\widehat{\bpcubecat}$ and $\widehat{\cubecat}$. By composing each one with its adjoint, these give rise to a chain of four adjoint endofunctos (all but the leftmost one endomorphisms of CwFs) on $\widehat{\bpcubecat}$.

We are not interested in the adjoint quintuple between $\widehat{\bpcubecat}$ and $\widehat{\cubecat}$ per se, but making a detour along them reveals a structure similar to what is studied in cohesive type theory, and may also be beneficial in order to understand intuitively what is going on in the rest of this text.

\subsection{Cohesion}
Let $\cat S$ be a category whose objects are some notion of spaces. Then a notion of \emph{cohesion} on objects of $\cat S$ gives rise to a category $\cat C$ of cohesive spaces and a \textbf{forgetful functor} $\cohfget : \cat C \to \cat S$ which maps a cohesive space $C$ to the underlying space $UC$, forgetting its cohesive structure.

Typically, if $\cohfget : \cat C \to \cat S$ appeals to the intuition about cohesion, then it is part of an adjoint quadruple of functors\footnote{More often, these are denoted $\Pi \dashv \Delta \dashv U \dashv \nabla$, but some of these symbols are already heavily in use in type theory.}
\begin{equation}
	\cohpi \dashv \cohdisc \dashv \cohfget \dashv \cohcodisc.
\end{equation}
Here, the \textbf{discrete functor} $\cohdisc$ equips a space $S$ with a discrete cohesive structure, i.e. in the cohesive space $\cohdisc S$, nothing is stuck together. As such, a cohesive map $\cohdisc S \to C$ amounts to a map $S \to UC$.

Dually, the \textbf{codiscrete functor} $\cohcodisc$ equips a space $S$ with a codiscrete cohesive structure, sticking everything together. As such, a cohesive map $C \to \cohcodisc S$ amounts to a map $UC \to S$.

Finally, the functor $\cohpi$ maps a cohesive space $C$ to its space of cohesively connected components. A map $C \to \cohdisc S$ will necessarily be constant on cohesive components, as $\cohdisc S$ is discrete, and hence amounts to a map $\cohpi C \to S$.

Typically, the composites $\cohpi \cohdisc, \cohfget \cohdisc, \cohfget \cohcodisc : \cat S \to \cat S$ will be isomorphic to the identity functor, the latter two even equal. Indeed: if we equip a space with discrete cohesion and then contract components, we have done essentially nothing. If we equip a space with a discrete or codiscrete cohesion, and then forget it again, we have litterally done nothing.

The composites $\shp = \cohdisc \cohpi$, $\flat = \cohdisc \cohfget$ and $\sharp = \cohcodisc \cohfget : \cat C \to \cat C$ form a more interesting adjoint triple $\shp \dashv \flat \dashv \sharp$ of endofunctors on $\cat C$. The \textbf{shape} functor $\shp$ contracts cohesive components. The \textbf{flat} functor $\flat$ removes the existing cohesion in favour of the discrete one, and the \textbf{sharp} functor $\sharp$ removes it in favour of the codiscrete one. If the adjoint triple on $\cat S$ is indeed as described above, then we can show
\begin{equation}
	\begin{array}{c c c c}
		\shp \shp \cong \shp &
		\shp \flat \cong \flat &&\\
		\flat \shp = \shp &
		\flat \flat = \flat &
		\flat \sharp = \flat &\\&
		\sharp \flat = \sharp &
		\sharp \sharp = \sharp,
	\end{array}
\end{equation}
Moreover, $\shp \dashv \flat$ will have (essentially) the same co-unit $\kappa : \flat \to \Id$ as $\cohdisc \dashv \cohfget$ and the same unit $\varsigma : \Id \to \shp$ as $\cohpi \dashv \cohdisc$. The adjunction $\flat \dashv \sharp$ will have the same unit $\kappa : \flat \to \Id$ as $\cohdisc \dashv \cohfget$ and the same co-unit $\iota : \Id \to \sharp$ as $\cohfget \dashv \cohcodisc$. For more information, see e.g. \cite{adjoint-logic}.

\begin{example}
	Let $\cat C = \catTop$ (the category of topological spaces) and $\cat S = \Set$. Then $\cohfget : \catTop \to \Set$ maps a topological space $(X, \tau)$ to the underlying set $X$. The discrete functor $\cohdisc$ equips a set $X$ with its discrete topology $2^X$ and $\cohcodisc$ equips it with the codiscrete topology $(\eset, X)$. Finally, $\cohpi$ maps a topological space to its set of connected components.
\end{example}
\begin{example}
	Let $\cat C = \Cat$ and $\cat S = \Set$. Then we can take $\cohfget \cat A = \Obj~\cat A$, make $\cohdisc X$ the discrete category on $X$ with only identity morphisms, $\cohcodisc X$ the codiscrete category on $X$ where every Hom-set is a singleton, and $\cohpi \cat A$ the set of zigzag-connected components of $\cat A$, i.e. $\cohpi \cat A = \Obj(\cat A)/\Hom$.
\end{example}
\begin{example}
	Let $\cat C = \widehat \RGcat$, the category of reflexive graphs, and $\cat S = \Set$. Then we can let $\cohfget$ map a reflexive graph $\Gamma$ to its set of nodes $\cohfget \Gamma = (\DSub{()} \Gamma)$; $\cohdisc$ will map a set $X$ to the discrete reflexive graph with only constant edges (i.e. $(\DSub{()}{\cohdisc X}) = (\DSub{(\var i : \IE)}{\cohdisc X}) = X$), and $\cohcodisc$ maps a set $X$ to the codiscrete reflexive graph with a unique edge between any two points (i.e. $(\DSub{()}{\cohcodisc X}) = X$ and $(\DSub{(\var i : \IE)}{\cohcodisc X}) = X \times X$). Finally, $\cohpi$ maps a graph $\Gamma$ to its set $\cohpi \Gamma$ of edge-connected components.
\end{example}
This last example is interesting, because we know that $\cohfget : \widehat{\RGcat} \to \widehat{\pointcat} \cong \Set$ is a morphism of CwFs, arising as $\cohfget = \fpsh F$ with $F : \pointcat \to \RGcat$ the unique functor mapping $()$ to $()$ (see \cref{ex:presheaf-sets}). The functor $\cohdisc$ is also a morphism of CwFs, arising as $\cohdisc = \fpsh G$, with $G$ the unique functor $\RGcat \to \pointcat$.

\subsection{The cohesive structure of $\widehat{\bpcubecat}$ over $\widehat{\cubecat}$, intuitively}
In the remainder of this section, we establish a chain of no less than five adjoint functors between $\widehat{\bpcubecat}$ and $\widehat{\cubecat}$. The reason we have more than in other situations, is that bridge/path cubical sets can be seen as cohesive cubical sets in two ways: we can either view cubical sets as bridge-only cubical sets, in which case the forgetful functor $\cohfget : \widehat{\bpcubecat} \to \widehat{\cubecat}$ forgets the cohesive structure given by the paths; or we can view cubical sets as path-only cubical sets, in which case the forgetful functor $\cohpaths : \widehat{\bpcubecat} \to \widehat{\cubecat}$ forgets the cohesion given by the bridges.

The chain of functors we obtain is the following:
\begin{equation}
	\cohpi \dashv \cohdisc \dashv \cohfget \dashv \cohcodisc \dashv \cohpaths, \qquad
	\cohpi, \cohfget, \cohpaths : \widehat{\bpcubecat} \to \widehat{\cubecat}, \qquad
	\cohdisc, \cohcodisc : \widehat{\cubecat} \to \widehat{\bpcubecat}.
\end{equation}
and it can likely be extended by a sixth functor on the right, that we currently have no use for. The (cohesion-as-paths) forgetful functor $\cohfget$ maps a bridge/path cubical set $\Gamma$ to the cubical set $\cohfget \Gamma$ made up of its bridges, forgetting which bridges are in fact paths. The (cohesion-as-paths) discrete functor $\cohdisc$ introduces a discrete path relation, the bridges of $\cohdisc \Gamma$ are the edges of $\Gamma$, whereas the paths of $\cohdisc \Gamma$ are all constant. The (cohesion-as-paths) codiscrete functor $\cohcodisc$ introduces a path relation which is codiscrete in the sense that there are as many paths as there can be: every bridge is a path. So a bridge in $\cohcodisc \Gamma$ is the same as an edge in $\Gamma$, and a path in $\cohcodisc \Gamma$ is also the same as an edge in $\Gamma$. Note that $\cohcodisc$ is also the cohesion-as-bridges discrete functor: viewing $\Gamma$ as a path-only cubical set, it equips $\cohcodisc \Gamma$ with the fewest bridges possible: only when there is a path, there will also be a bridge. The paths functor $\cohpaths$, which is the cohesion-as-bridges forgetful functor, maps a bridge/path cubical set $\Gamma$ to its cubical set of paths $\cohpaths \Gamma$. Finally, $\cohpi$ divides out a bridge/path cubical set by its path relation, obtaining a bridge-only cubical set.

\subsection{The cohesive structure of $\bpcubecat$ over $\cubecat$}
We saw in \cref{thm:lifting-adjunctions} that if we have adjoint functors $L \dashv R$ on the base categories, then we obtain functors $\fpsh L \dashv \fpsh R$ on the presheaf categories and moreover $\fpsh L$ extends $R$. So $\fpsh \loch$ takes the right adjoint of a functor and at the same time extends it from the category of primitive contexts, to the entire presheaf category. In this sense, it is a good idea to start by defining the functors
\begin{equation}
	\cohpi \dashv \cohdisc \dashv \cohfget \dashv \cohcodisc, \qquad
	\cohpi, \cohfget : \bpcubecat \to \cubecat, \qquad
	\cohdisc, \cohcodisc : \cubecat \to \bpcubecat,
\end{equation}
on the base categories. We define these functors as in \cref{fig:cohesion-base}.
This may not be entirely intuitive. The key here is that every path dimension can be weakened to a bridge dimension. Thus, two adjacent vertices of a bridge/path cube are always connected by a bridge, and only connected by a path if they are adjacent along a path dimension. The $\cohpi$ functor leaves bridges alone (converting them to edges), but contracts paths. The $\cohdisc$ functor turns edges into bridges, but does not produce paths. The $\cohfget$ functor keeps bridges (converting them to edges) and forgets paths, but remembers the bridges they weaken to. The $\cohcodisc$ functor turns edges into paths, which are also bridges.
\begin{figure}
	\begin{equation*}
		\begin{array}{c || c c c | c c c c}
			& () & (W, \ctxbrid{\var i}) & (W, \ctxpath{\var i})
			& () & (\vfi, \var j^\IB /\var i^\IB) & (\vfi, \var j^\IB /\var i^\IP) & (\vfi, \var j^\IP /\var i^\IP) \\
			\hline \hline
			\cohpi
			& () & (\cohpi W, \ctxedge{\var i}) & \cohpi W 
			& () & (\cohpi \vfi, \var j^\IE / \var i^\IE) & \cohpi \vfi & \cohpi \vfi \\
			\cohfget
			& () & (\cohfget W, \ctxedge{\var i}) & (\cohfget W, \ctxedge{\var i})
			& () & (\cohfget \vfi, \var j^\IE / \var i^\IE) & (\cohfget \vfi, \var j^\IE / \var i^\IE) & (\cohfget \vfi, \var j^\IE / \var i^\IE) \\
			\hline
			\shp
			& () & (\shp W, \ctxbrid{\var i}) & \shp W
			& () & (\shp \vfi, \var j^\IB / \var i^\IB) & \shp \vfi & \shp \vfi \\
			\flat
			& () & (\flat W, \ctxbrid{\var i}) & (\flat W, \ctxbrid{\var i})
			& () & (\flat \vfi, \var j^\IB/\var i^\IB) & (\flat \vfi, \var j^\IB/\var i^\IB) & (\flat \vfi, \var j^\IB/\var i^\IB) \\
			\sharp
			& () & (\sharp W, \ctxpath{\var i}) & (\sharp W, \ctxpath{\var i})
			& () & (\sharp \vfi, \var j^\IP / \var i^\IP) & (\sharp \vfi, \var j^\IP / \var i^\IP) & (\sharp \vfi, \var j^\IP / \var i^\IP) \\
			\hline
			\varsigma : \Id \to \shp
			& () & (\varsigma_W, \var i^\IB / \var i^\IB) & (\varsigma_W, \facewkn{\var i^\IP})\\
			\kappa : \flat \to \Id
			& () & (\kappa_W, \var i^\IB / \var i^\IB) & (\kappa_W, \var i^\IB / \var i^\IP) &&& \circlearrowright \\
			\iota : \Id \to \sharp
			& () & (\iota_W, \var i^\IB / \var i^\IP) & (\iota_W, \var i^\IP / \var i^\IP)
		\end{array}
	\end{equation*}
	\begin{equation*}
		\begin{array}{c || c c | c c}
			& () & (W, \ctxedge{\var i})
			& () & (\vfi, \var j^\IE / \var i^\IE) \\
			\hline \hline
			\cohdisc & () & (\cohdisc W, \ctxbrid{\var i}) 
			& () & (\cohdisc \vfi, \var j^\IB / \var i^\IB) \\
			\cohcodisc & () & (\cohcodisc W, \ctxpath{\var i}) 
			& () & (\cohcodisc \vfi, \var j^\IP / \var i^\IP)
		\end{array}
	\end{equation*}
	\caption{The cohesive structure of $\bpcubecat$ over $\cubecat$.}
	\label{fig:cohesion-base}
\end{figure}

Note also that $\cohpaths$ cannot be defined this way as it does not map all primitive contexts to primitive contexts. For example, $(\ctxbrid{\var i})$ consists of two points connected by a bridge. The only paths are constant. Hence, forgetting the bridge structure yields two loose points, which together do not form a cube.

\begin{lemma}\label{thm:uniqueness-of-nattrans-base}
	Let $\catV, \catW \in \accol{\cubecat, \bpcubecat}$ and let $F, G : \catV \to \catW$ be composites of the functors $\cohpi$, $\cohdisc$, $\cohfget$ and $\cohcodisc$. Then all natural transformations $F \to G$ are equal.
\end{lemma}
\begin{proof}
	Let $\nu : F \to G$ be a natural transformation. We show that $\nu$ is completely determined. Since $F$ and $G$ preserve products, $\nu$ is determined by its action on single-variable contexts. Now if $\Gamma$ has a single variable $\var i$, then either $G\Gamma = ()$ in which case $\nu_\Gamma : F \Gamma \to ()$ is determined, or $G\Gamma$ also contains $\var i$ as its only variable. Now $\var i \psub{\nu_\Gamma} \neq 0$ because then the following diagram could not commute:
\begin{equation}
	\xymatrix{
		() \ar[r]^{\nu_{()}} \ar[d]_{F(1/\var i)}
		& () \ar[d]^{G(1/\var i) = (1 / \var i)}
		\\
		F\Gamma \ar[r]^{\nu_\Gamma}
		& G\Gamma
	}
\end{equation}
Similarly, $\var i \psub{\nu_\Gamma} \neq 1$. Then $F\Gamma$ must contain a variable, implying that it contains only the variable $\var i$ and $\var i \psub{\nu_\Gamma} = \var i$.
\end{proof}

\begin{proposition}[The cohesive structure of $\bpcubecat$ over $\cubecat$]\label{thm:cohesion-base}
	These four functors are adjoint: $\cohpi \dashv \cohdisc \dashv \cohfget \dashv \cohcodisc$. On the $\cubecat$-side, we have
	\begin{equation}
		\cohpi\cohdisc = \Id \quad \dashv \quad
		\cohfget\cohdisc = \Id \quad \dashv \quad
		\cohfget\cohcodisc = \Id \quad 
		: \quad \cubecat \to \cubecat.
	\end{equation}
	On the $\bpcubecat$-side, we write
	\begin{equation}
		\shp := \cohdisc\cohpi \quad \dashv \quad
		\flat := \cohdisc\cohfget \quad \dashv \quad
		\sharp := \cohcodisc\cohfget \quad
		: \quad \bpcubecat \to \bpcubecat.
	\end{equation}
	By consequence, we have
	\begin{equation}
		\begin{array}{c c c c}
			\shp \shp = \shp &
			\shp \flat = \flat &&\\
			\flat \shp = \shp &
			\flat \flat = \flat &
			\flat \sharp = \flat &\\&
			\sharp \flat = \sharp &
			\sharp \sharp = \sharp.
		\end{array}
	\end{equation}
	The following table lists the units and co-units of all adjunctions involved:
	\begin{equation}
		\begin{array}{c || c | c | c || c | c}
			& \cohpi \dashv \cohdisc & \cohdisc \dashv \cohfget & \cohfget \dashv \cohcodisc & \shp \dashv \flat & \flat \dashv \sharp \\
			\hline
			\text{unit} & \varsigma : \Id \to \shp & \id : \Id \to \Id & \iota : \Id \to \sharp & \varsigma : \Id \to \shp & \iota : \Id \to \sharp \\
			\hline
			\text{co-unit} & \id : \Id \to \Id & \kappa : \flat \to \Id & \id : \Id \to \Id & \kappa : \flat \to \Id & \kappa : \flat \to \Id
		\end{array}
	\end{equation}
	The functors $\shp$, $\flat$ and $\sharp$ and the natural transformations $\varsigma$, $\kappa$ and $\iota$ are given in \cref{fig:cohesion-base}. Finally, the following natural transformations are all the identity:
	\begin{equation}\label{eq:cohesion-base-identities}
		\begin{array}{c c c c | c c c}
			\cohpi \to \cohpi & \cohdisc \to \cohdisc & \cohfget \to \cohfget & \cohcodisc \to \cohcodisc & \shp \to \shp & \flat \to \flat & \sharp \to \sharp \\ \hline
			\cohpi \varsigma & \varsigma \cohdisc &&&
			\shp \varsigma = \varsigma \shp & \varsigma \flat \qquad{} \\
			& \kappa \cohdisc & \cohfget \kappa &&
			{} \qquad \kappa \shp & \kappa \flat = \flat \kappa & \sharp \kappa \qquad{} \\
			&& \cohfget \iota & \iota \cohcodisc &
			& {}\qquad \flat \iota & \sharp \iota = \iota \sharp
		\end{array}
	\end{equation}
\end{proposition}
\begin{proof}
	The equalities are immediate from \cref{thm:uniqueness-of-nattrans-base}.
\end{proof}
\begin{proof}
	Each of the transformations in \cref{eq:cohesion-base-identities} is easily seen to be the identity by inspecting the definitions in \cref{fig:cohesion-base}. In order to prove that $L \dashv R$ with unit $\eta$ and co-unit $\eps$, it suffices to check that $\eps L \circ L \eta = \id : L \to L$ and $R \eps \circ \eta R = \id : R \to R$, which also follows from \cref{thm:uniqueness-of-nattrans-base}. In other words, the mere existence of well-typed candidates for the unit and co-unit is sufficient to conclude adjointness.
\end{proof}

\subsection{The cohesive structure of $\widehat{\bpcubecat}$ over $\widehat{\cubecat}$, formally}
We now define functors and natural transformations of the same notation by
\begin{equation}
	\cohdisc := \fpsh \cohpi, \quad
	\cohcodisc := \fpsh \cohfget \quad
	: \widehat{\cubecat} \to \widehat{\bpcubecat},
\end{equation}
\begin{equation}
	\cohfget := \fpsh \cohdisc, \quad
	\cohpaths := \fpsh \cohcodisc \quad
	: \widehat{\bpcubecat} \to \widehat{\cubecat},
\end{equation}
\begin{equation}
	\flat := \fpsh \shp = \cohdisc \cohfget, \quad
	\sharp := \fpsh \flat = \cohcodisc \cohfget, \quad
	\coshp := \fpsh \sharp = \cohcodisc \cohpaths \quad : \widehat{\bpcubecat} \to \widehat{\bpcubecat}.
\end{equation}
\begin{equation}
	\kappa := \fpsh \varsigma : \flat \to \Id, \qquad
	\iota := \fpsh \kappa : \Id \to \sharp, \qquad
	\vartheta := \fpsh \iota : \coshp \to \Id.
\end{equation}
From \cref{sec:left-adjoint-to-lifted-functor}, we know that there is a further left adjoint $\cohpi \dashv \cohdisc$, which we should not expect to be a morphism of CwFs. Indeed, the proof of \cref{thm:left-adjoint-not-cwf} is easily adapted to show the contrary. We postpone its construction to \cref{sec:def-cohpi}; however, by \cref{thm:yoneda-and-left-adjoint} it will satisfy the property that $\cohpi \circ \yoneda \cong \yoneda \circ \cohpi$, which we can use to characterize its behaviour. We take a moment to see how each of these functors behaves:
\begin{description}
	\item[$\cohpi$] We know that $W$-cubes $\gamma : \DSub W \Gamma$ correspond to substitutions $\gamma : \yoneda W \to \Gamma$ and hence give rise to substitutions $\cohpi \gamma : \cohpi(\yoneda W) \to \cohpi \Gamma$, which in turn correspond to $\cohpi W$-cubes $\DSub{\cohpi W}{\cohpi \Gamma}$. So a bridge $\DSub{(\ctxbrid{\var i})}{\Gamma}$ is turned into an edge $\DSub{(\ctxedge{\var i})}{\cohpi \Gamma}$, whereas a path $\DSub{(\ctxpath{\var i})}{\Gamma}$ is contracted to a point $\DSub{()}{\cohpi \Gamma}$. Simply put, $\cohpi$ contracts paths to points.
	\item[$\cohdisc$] A $W$-cube $\fpshadj{\cohpi}(\gamma) : \DSub{W}{\cohdisc \Gamma}$ is a $\cohpi W$-cube $\gamma : \DSub{\cohpi W}{\Gamma}$. So a bridge $\DSub{(\ctxbrid{\var i})}{\cohdisc \Gamma}$ is the same as an edge $\DSub{(\ctxedge{\var i})}{\Gamma}$ and a path $\DSub{(\ctxpath{\var i})}{\cohdisc \Gamma}$ is the same as a point $\DSub{()}{\Gamma}$, which in turn is the same as a point $\DSub{()}{\cohdisc \Gamma}$, showing that there are only constant paths.
	
	Viewed differently, using that $\cohdisc \circ \yoneda = \yoneda \circ \cohdisc$ by \cref{thm:lifting-adjunctions}, we can say that an edge $\DSub{(\ctxedge{\var i})}{\Gamma}$ gives rise to a bridge $\DSub{(\ctxbrid{\var i})}{\cohdisc \Gamma}$, while there is nothing that gives rise to (non-trivial) paths.
	\item[$\cohfget$] A $W$-cube $\fpshadj \cohdisc(\gamma) : \DSub{W}{\cohfget \Gamma}$ is the same as a brdige $\gamma : \DSub{\cohdisc W}{\Gamma}$. So an edge in $\cohfget \Gamma$ is a bridge in $\Gamma$. Alternatively, we can say that any bridge and any path in $\Gamma$ gives rise to an edge in $\cohfget \Gamma$.
	\item[$\cohcodisc$] A bridge in $\cohcodisc \Gamma$ is the same as an edge in $\Gamma$. A path in $\cohcodisc \Gamma$ is also the same as an edge in $\Gamma$. Alternatively, we can say that an edge in $\Gamma$ gives rise to a path in $\cohcodisc \Gamma$, which can then also be weakened to a bridge.
	\item[$\cohpaths$] An edge in $\cohpaths \Gamma$ is the same as a path in $\Gamma$. The alternative formulation --- a path in $\Gamma$ gives rise to an edge in $\cohpaths \Gamma$; a bridge in $\Gamma$ is forgotten --- cannot be formalized as in the previous cases, because $\cohpaths$ was not defined for primitive contexts and hence the property $\cohpaths \circ \yoneda = \yoneda \circ (\ldots)$ cannot be formulated
\end{description}
\begin{proposition}[The cohesive structure of $\widehat{\bpcubecat}$ over $\widehat{\cubecat}$]\label{thm:cohesion-psh}
	These five functors are adjoint: $\cohpi \dashv \cohdisc \dashv \cohfget \dashv \cohcodisc \dashv \cohpaths$. On the $\widehat{\cubecat}$-side, we have
	\begin{equation}
		\bar \shp := \cohpi\cohdisc \cong \Id \quad \dashv \quad
		\cohfget\cohdisc = \Id \quad \dashv \quad
		\cohfget\cohcodisc = \Id \quad \dashv \quad 
		\cohpaths\cohcodisc = \Id \quad
		: \quad \widehat{\cubecat} \to \widehat{\cubecat}.
	\end{equation}
	On the $\widehat{\bpcubecat}$-side, we write
	\begin{equation}
		\shp := \cohdisc\cohpi \quad \dashv \quad
		\flat := \cohdisc\cohfget \quad \dashv \quad
		\sharp := \cohcodisc\cohfget \quad \dashv \quad
		\coshp := \cohcodisc\cohpaths \quad
		: \quad \widehat{\bpcubecat} \to \widehat{\bpcubecat}.
	\end{equation}
	By consequence, we have
	\begin{equation}
		\begin{array}{c c c c c}
			\shp \shp \cong \shp &
			\shp \flat \cong \flat &&\\
			\flat \shp = \shp &
			\flat \flat = \flat &
			\flat \sharp = \flat &\\&
			\sharp \flat = \sharp &
			\sharp \sharp = \sharp &
			\sharp \coshp = \coshp \\&&
			\coshp \sharp = \sharp &
			\coshp \coshp = \coshp,
		\end{array}
	\end{equation}
	
	The following tables lists the units and co-units of all adjunctions involved:
	\begin{equation}
		\begin{array}{c || c | c | c | c}
			& \cohpi \dashv \cohdisc & \cohdisc \dashv \cohfget & \cohfget \dashv \cohcodisc & \cohcodisc \dashv \cohpaths \\
			\hline
			\text{unit} & \varsigma : \Id \to \shp & \id : \Id \to \Id & \iota : \Id \to \sharp & \id : \Id \to \Id \\
			\hline
			\text{co-unit} & \bar \varsigma\inv : \bar \shp \cong \Id & \kappa : \flat \to \Id & \id : \Id \to \Id & \vartheta : \coshp \to \Id
		\end{array}
	\end{equation}
	\begin{equation}
		\begin{array}{c || c || c | c | c}
			& \bar \shp \dashv \Id & \shp \dashv \flat & \flat \dashv \sharp & \sharp \dashv \coshp \\
			\hline
			\text{unit} & \bar \varsigma : \Id \cong \bar \shp & \varsigma : \Id \to \shp & \iota : \Id \to \sharp & \iota : \Id \to \sharp \\
			\hline
			\text{co-unit} & \bar \varsigma \inv : \bar \shp \cong \Id & \kappa \circ (\varsigma \flat)\inv : \shp \flat \to \Id & \kappa : \flat \to \Id & \vartheta : \coshp \to \Id
		\end{array}
	\end{equation}
	Finally, the following natural transformations are all (compatible with) the identity:
	\begin{equation}\label{eq:cohesion-psh-identities}
		\begin{array}{c c c c c | c c c c}
			\cohpi \to \cohpi & \cohdisc \to \cohdisc & \cohfget \to \cohfget & \cohcodisc \to \cohcodisc & \cohpaths \to \cohpaths & \shp \to \shp & \flat \to \flat & \sharp \to \sharp & \coshp \to \coshp \\ \hline
			(\cohpi \varsigma) & (\varsigma \cohdisc) &&&&
			(\shp \varsigma = \varsigma \shp) & (\varsigma \flat) \qquad{} \\
			& \kappa \cohdisc & \cohfget \kappa &&&
			{} \qquad \kappa \shp & \kappa \flat = \flat \kappa & \sharp \kappa \qquad{} \\
			&& \cohfget \iota & \iota \cohcodisc &&
			& {}\qquad \flat \iota & \sharp \iota = \iota \sharp & \iota \coshp \qquad{} \\
			&&& \vartheta \cohcodisc & \cohpaths \vartheta &
			&& {} \qquad \vartheta \sharp & \vartheta \coshp = \coshp \vartheta
		\end{array}
	\end{equation}
	The ones involving $\kappa$, $\iota$ and $\vartheta$ are actually equal, while for $\varsigma$ we have
	\begin{equation*}
		\cohpi \varsigma = \bar \varsigma \cohpi : \cohpi \cong \cohpi \shp, \qquad
		\varsigma \cohdisc = \cohdisc \bar \varsigma : \cohdisc \cong \shp \cohdisc, \qquad
		\shp \varsigma = \varsigma \shp = \cohdisc \bar \varsigma \cohpi : \shp \shp \cong \shp, \qquad
		\varsigma \flat = \cohdisc \bar \varsigma \cohfget : \shp \flat \cong \flat.
	\end{equation*}
\end{proposition}
\begin{lemma}\label{thm:unlift-nattrans}
	A natural transformation $\nu : \fpsh F \to H : \widehat \catV \to \widehat \catW$ whose domain is a lifted functor, is fully determined by $\nu \yoneda F : \fpsh F \yoneda F \to H \yoneda F : \catW \to \widehat \catW$. If $H = \fpsh G$, then $\nu = \fpsh{\tilde \nu}$ for some $\tilde \nu : G \to F : \catW \to \catV$. If $H = \lpsh K \dashv \fpsh K$ for some $K : \catV \to \catW$, then $\nu$ corresponds to a natural transformation $\mu : \Id \to KF : \catW \to \catW$.
\end{lemma}
\begin{proof}
	Pick a presheaf $\Gamma \in \widehat \catV$ and a defining substitution $\fpshadj F (\gamma) : \DSub{W}{\fpsh F \Gamma}$. Then the following diagram commutes:
	\begin{equation}
		\xymatrix{
			& \fpsh F \yoneda F W \ar[r]^{\nu \yoneda FW} \ar[d]^{\fpsh F \gamma}
			& H \yoneda F W \ar[d]^{H \gamma} 
			\\
			W \ar@{=>}[ru]^{\fpshadj F (\id)} \ar@{=>}[r]_{\fpshadj F (\gamma)}
			& \fpsh F \Gamma \ar[r]_{\nu \Gamma}
			& H\Gamma,
		}
	\end{equation}
	showing that $\nu_\Gamma \circ \fpshadj F(\gamma)$ is determined by $\nu \yoneda F W$.
	\begin{itemize}
		\item If $H = \fpsh G$, then we define $\tilde \nu W = \fpshadj G\inv(\nu \yoneda F W \circ \fpshadj F(\id)) : \PSub{GW}{FW}$. Note that if $\nu = \fpsh \mu$, then we would find
		\begin{equation}
			\tilde \nu W = \fpshadj G\inv(\fpsh \mu \yoneda F W \circ \fpshadj F(\id))
			= \fpshadj G\inv(\fpshadj G(\id \circ \mu W)) = \mu W : \PSub{GW}{FW}.
		\end{equation}
		In general, this is a natural transformation because if we have $\vfi : \PSub V W$, then\footnote{Remember that we write $\gamma : \yoneda W \to \Gamma$ when $\gamma : \DSub W \Gamma$, implying that we do not write $\yoneda$ when applied to a morphism.}
		\begin{align*}
			F \vfi \circ \tilde \nu V
			&= \fpshadj G \inv (\fpsh G (F \vfi) \circ \nu \yoneda F V \circ \fpshadj F (\id)) \\
			&= \fpshadj G \inv (\nu \yoneda F W \circ \fpsh F (F \vfi) \circ \fpshadj F (\id)) \\
			&= \fpshadj G \inv (\nu \yoneda F W \circ \fpshadj F (F \vfi)) \\
			&= \fpshadj G \inv (\nu \yoneda F W \circ \fpshadj F (\id) \circ \vfi) \\
			&= \fpshadj G \inv (\nu \yoneda F W \circ \fpshadj F (\id)) \circ G \vfi
			= \tilde \nu W \circ G \vfi.
		\end{align*}
		Moreover, the above diagram shows that $\nu = \fpsh{\tilde \nu}$ because
		\begin{equation}
			\nu \circ \fpshadj F(\gamma) = \fpsh G \gamma \circ \nu \yoneda F W \circ \fpshadj F(\id) = \fpsh G \gamma \circ \fpshadj G (\tilde \nu) = \fpshadj G (\gamma \circ \tilde \nu) = \fpsh{\tilde \nu} \circ \fpshadj G(\gamma).
		\end{equation}
		
		\item If $H = \lpsh K \dashv \fpsh K$, then we show that natural transformations $\fpsh F \to \lpsh K : \widehat \catV \to \widehat \catW$ correspond to natural transformations $\Id \to KF : \catW \to \catW$. We already know that $\nu : \fpsh F \to \lpsh K$ is determined by $\nu \yoneda F : \fpsh F \yoneda F \to \lpsh K \yoneda F$, and \cref{thm:yoneda-and-left-adjoint} tells us that $\zeta : \lpsh K \yoneda F \cong \yoneda K F$.
		
		Given $\nu$, we now define $\mu : \Id \to KF$ by $\mu W = \zeta W \circ \nu \yoneda F W \circ \fpshadj F(\id_W) \in (\DSub{W}{\yoneda KF W}) = (\PSub{W}{KFW})$. Conversely, given $\mu : \Id \to KF$, we define $\nu : \fpsh F \to \lpsh K$ by setting for every $\fpshadj F(\gamma) : \DSub{W}{\fpsh F \Gamma}$ (i.e. $\gamma : \DSub{FW}{\Gamma}$), the composition $\nu \Gamma \circ \fpshadj F(\gamma)$ equal to $\lpsh K \gamma \circ (\zeta W)\inv \circ \mu W : \DSub W {\lpsh K \Gamma}$. These operations are inverse:
		\begin{align*}
			\lpsh K \gamma \circ (\zeta W)\inv \circ \mu W
			&= \lpsh K \gamma \circ (\zeta W)\inv \circ \zeta W \circ \nu \yoneda F W \circ \fpshadj F(\id_W) \\
			&= \lpsh K \gamma \circ \nu \yoneda F W \circ \fpshadj F(\id_W)
			= \nu \Gamma \circ \fpsh F \gamma \circ \fpshadj F(\id_W)
			= \nu \Gamma \circ \fpshadj F(\gamma). \\
			\zeta W \circ \nu \yoneda F W \circ \fpshadj F(\id_W)
			&= \zeta W \circ \lpsh K\,\id_W \circ (\zeta W)\inv \circ \mu W = \mu W. \qedhere
		\end{align*}
	\end{itemize}
\end{proof}
\begin{corollary}\label{thm:uniqueness-of-nattrans-psh}
	Let $\catV, \catW \in \accol{\cubecat, \bpcubecat}$ and $F, G : \widehat \catV \to \widehat \catW$. Then all natural transformations from $F$ to $G$ are equal in each of the following cases:
	\begin{enumerate}
		\item If $F$ and $G$ are composites of $\cohdisc$, $\cohfget$, $\cohcodisc$ and $\cohpaths$;
		\item If $F$ is a composite of $\cohdisc$, $\cohfget$, $\cohcodisc$ and $\cohpaths$, and $G$ is a composite of $\cohpi$, $\cohdisc$, $\cohfget$ and $\cohcodisc$;
		\item If $F$ factors as $LP$, where $L \dashv R$ and one of the previous cases apply to $P$ and $RG$.
	\end{enumerate}
\end{corollary}
\begin{proof}
	Pick $\nu : F \to G$.
	\begin{enumerate}
		\item Then both $F$ and $G$ are lifted so that $\nu = \fpsh{\tilde \nu}$, and $\tilde \nu$ is completely determined by \cref{thm:uniqueness-of-nattrans-base}.
		\item Then $F$ is lifted and $G$ is left adjoint to a lifted functor, so that $\nu$ corresponds to a natural transformation of primitive contexts, which is uniquely determined because of \cref{thm:uniqueness-of-nattrans-base}.
		\item Natural transformations $LP \to G$ are in bijection with natural transformations $P \to RG$ because $L \dashv R$. \qedhere
	\end{enumerate}
\end{proof}
\begin{lemma}\label{thm:unicity-of-left-adjoint}
	Assume $L_1, L_2 : \catV \to \catW$ and $R : \catW \to \catV$ such that $\alpha_i : L_i \dashv R$ with unit $\eta_i : \Id \to R L_i$ and co-unit $\eps_i : L_i R \to \Id$. Then there is a natural isomorphism $\zeta : L_1 \cong L_2$ such that $R \zeta \circ \eta_1 = \eta_2$ and $\eps_1 = \eps_2 \circ \zeta R$.
\end{lemma}
\begin{proof}
	We set $\zeta = \eps_1 L_2 \circ L_1 \eta_2$ and $\zeta\inv = \eps_2 L_1 \circ L_2 \eta_1$. We show that $\zeta \circ \zeta\inv = \id$; the other equation holds by symmetry of the indices. Observe the commutative diagram:
	\begin{equation}
		\xymatrix{
			L_2 \ar[rr]^{L_2 \eta_1} \ar[d]_{L_2 \eta_2}
			\ar@/^{2em}/[rrrr]^{\zeta}
			&& L_2 R L_1 \ar[rr]^{\eps_2 L_1} \ar[d]_{L_2 R L_1 \eta_2}
			&& L_1 \ar[d]_{L_1 \eta_2}
			\ar@/^{2em}/[dd]^{\zeta\inv}
			\\
			L_2 R L_2 \ar[rr]^{L_2 \eta_1 R L_2} \ar[rrd]_{\id}
			&& L_2 R L_1 R L_2 \ar[rr]^{\eps_2 L_1 R L_2} \ar[d]^{L_2 R \eps_1 L_2}
			&& L_1 R L_2 \ar[d]_{\eps_1 L_2}
			\\
			&& L_2 R L_2 \ar[rr]_{\eps_2 L_2}
			&& L_2
		}
	\end{equation}
	The top right square applies naturality of $\eps_2 L_1$ to $L_1 \eta_2$. The top left square still holds after removing $L_2$ on the left and is then an instance of naturality of $\eta_1$. The lower right square still holds after removing $L_2$ on the right and is then an instance of naturality of $\eps_2$. The lower right triangle commutes because $R \eps_1 \circ \eta_1 R = \id$. Finally, the entire left-lower side composes to the identity. We have
	\begin{align*}
		R\zeta \circ \eta_1
		&= R \eps_1 L_2 \circ R L_1 \eta_2 \circ \eta_1
		= R \eps_1 L_2 \circ \eta_1 RL_2 \circ \eta_2 = \eta_2, \\
		\eps_2 \circ \zeta R
		&= \eps_2 \circ \eps_1 L_2 R \circ L_1 \eta_2 R
		= \eps_1 \circ L_1 R \eps_2 \circ L_1 \eta_2 R = \eps_1. \qedhere
	\end{align*}
\end{proof}
\begin{proof}[Proof of \cref{thm:cohesion-psh}]
	The adjunctions $\cohdisc \dashv \cohfget \dashv \cohcodisc \dashv \cohpaths$ follow from \cref{thm:lifting-adjunctions}. For now, we just assume $\cohpi$ to be some left adjoint to $\cohdisc$.
	
	The fact that $\cohfget \cohdisc = \Id$, $\cohfget\cohcodisc = \Id$ and $\cohpaths\cohcodisc = \Id$ follows from the fact that $\fpsh \loch$ swaps composition and preserves identity.
	
	It is clear that both $\bar \shp := \cohpi \cohdisc$ and $\Id$ are left adjoint to $\cohfget \cohdisc = \Id$. Then \cref{thm:unicity-of-left-adjoint} below gives us, after filling in the identity in various places, an isomorphism $\bar \varsigma : \Id \cong \bar \shp$ which is the unit of $\bar \shp \dashv \Id$, while the co-unit is $\bar \varsigma \inv$.
	
	The equalities and isomorphisms are obvious.
	
	We \emph{define} $\varsigma$ as the unit of $\cohpi \dashv \cohdisc$.
	The rest of the theorem now follows from \ref{thm:uniqueness-of-nattrans-psh}.
\end{proof}

\subsection{Characterizing cohesive adjunctions}
We currently have various cohesion-based ways of manipulating terms: we can apply functors, turning ($t \mapsto \ftrtm F t$), we can apply natural transformations ($t \mapsto \nu(t)$), we can instead substitute with natural transformations ($t \mapsto t[\nu]$) and apply adjunctions ($t \mapsto \alpha(t)$). We have various equations telling us how these relate, but altogether it becomes hard to tell whether terms are equal. For this reason, we will at least try to write the relevant adjunctions in terms of the other constructions.
\begin{proposition}
	For any contexts $\Gamma$ and $\Delta$, we have the following diagrams, in which all arrows are invertible:
	\begin{equation}
		\xymatrix{
			(\shp \Gamma \to \Delta)
				\ar[rr]^{\alpha_{\shp \dashv \flat}}
				\ar@{<-}[rd]_{\kappa \circ \loch}
			&&
			(\Gamma \to \flat \Delta)
				\ar@{<-}[ld]^{\loch \circ \varsigma}
			\\&
			(\shp \Gamma \to \flat \Delta)
		\\
			(\flat \Gamma \to \Delta)
				\ar[rr]^{\alpha_{\flat \dashv \sharp}}
				\ar[rd]_{\iota \circ \loch}
			&&
			(\Gamma \to \sharp \Delta)
				\ar[ld]^{\loch \circ \kappa}
			\\&
			(\flat \Gamma \to \sharp \Delta)
		\\
			(\sharp \Gamma \to \Delta)
				\ar[rr]^{\alpha_{\sharp \dashv \coshp}}
				\ar@{<-}[rd]_{\vartheta \circ \loch}
			&&
			(\Gamma \to \coshp \Delta)
				\ar@{<-}[ld]^{\loch \circ \iota}
			\\&
			(\sharp \Gamma \to \coshp \Delta)
		}
	\end{equation}
\end{proposition}
\begin{proof}
	For all diagonal arrows except $\kappa \circ \loch$, we can again take the adjunction isomorphism since $\flat$, $\sharp$ and $\coshp$ are idempotent. One can check that these boil down to composition with a (co)-unit, e.g. for $\sigma : \flat \Gamma \to \Delta$ we have
	\begin{equation}
		\alpha_{\flat \dashv \sharp}^{\flat \Gamma, \Delta}(\sigma)
		= \sharp \sigma \circ \iota \flat \Gamma
		= \iota \Delta \circ \sigma.
	\end{equation}
	For $\sigma : \sharp \Gamma \to \coshp \Delta$, we have
	\begin{equation}
		\alpha_{\sharp \dashv \coshp}^{\Gamma, \coshp \Delta}(\sigma)
		= \coshp \sigma \circ \iota \Gamma = \sigma \circ \iota \Gamma
	\end{equation}
	because $\coshp \sigma = \vartheta \coshp \Delta \circ \coshp \sigma = \sigma \circ \vartheta \sharp \Gamma = \sigma$. For the arrow $\kappa \circ \loch$, we use a composition of isomorphisms
	\begin{equation}
		(\shp \Gamma \to \flat \Delta) \xrightarrow{(\alpha_{\shp \dashv \flat}^{\shp \Gamma, \Delta})\inv} (\shp \shp \Gamma \to \flat \Delta) \xrightarrow{\loch \circ \varsigma \shp \Gamma} (\shp \Gamma \to \Delta).
	\end{equation}
	A substitution $\sigma : \shp \Gamma \to \flat \Delta$ is then mapped to
	\begin{equation}
		(\alpha_{\shp \dashv \flat}^{\shp \Gamma, \Delta})\inv(\sigma) \circ \varsigma \shp \Gamma
		= \kappa \Delta \circ (\varsigma \flat \Delta)\inv \circ \shp \sigma \circ \varsigma \shp \Gamma
		= \kappa \Delta \circ (\varsigma \flat \Delta)\inv \circ \varsigma \flat \Delta \circ \sigma
		= \kappa \Delta \circ \sigma.
	\end{equation}
	Finally, one can check that the diagrams commute, e.g. if $\sigma : \shp \Gamma \to \flat \Delta$ then
	\begin{equation}
		\alpha_{\shp \dashv \flat}^{\Gamma, \Delta}(\kappa \Delta \circ \sigma)
		= \flat \kappa \Delta \circ \flat \sigma \circ \varsigma \Gamma
		= \flat \sigma \circ \varsigma \Gamma = \sigma \circ \varsigma \Gamma
	\end{equation}
	because $\flat \sigma = \kappa \flat \Delta \circ \flat \sigma = \sigma \circ \kappa \shp \Gamma = \sigma$.
\end{proof}
\begin{notation}
	When it exists, we write $\tau \setminus \sigma$ for the unique substitution such that $\tau \circ (\tau \setminus \sigma) = \sigma$, and $\sigma / \tau$ for the unique substitution such that $(\sigma / \tau) \circ \tau = \sigma$. Uniqueness implies that $\tau \setminus (\tau \circ \sigma) = \sigma$ and $(\sigma \circ \tau) \setminus \tau = \sigma$ if the left hand side exists.
	
	Similarly, when it exists, we write $\nu\inv(t)$ for the unique term such that $\nu(\nu\inv(t)) = t$ and $t[\nu]\inv$ for the unique term such that $t[\nu]\inv[\nu] = t$. Uniqueness implies that $\nu\inv(\nu(t)) = t$ and $t[\nu][\nu]\inv = t$.

	The above theorem then justifies the following notations:
	\begin{equation}
		\begin{array}{r | c | c | c}
			& \shp \dashv \flat & \flat \dashv \sharp & \sharp \dashv \coshp
			\\ \hline
			\alpha(\sigma)
			& (\kappa \setminus \sigma) \circ \varsigma
			& (\iota \circ \sigma) / \kappa
			& (\vartheta \setminus \sigma) \circ \iota
			\\
			\alpha\inv(\tau)
			& \kappa \circ (\tau / \varsigma)
			& \iota \setminus (\tau \circ \kappa)
			& \vartheta \circ (\tau / \iota)
			\\
			\alpha(t)
			& \kappa\inv(t)[\varsigma]
			& \iota(t)[\kappa]\inv
			& \vartheta\inv(t)[\iota]
			\\
			\alpha\inv(u)
			& \kappa(t[\varsigma]\inv)
			& \iota\inv(t[\kappa])
			& \vartheta(t[\iota]\inv)
		\end{array}
	\end{equation}
	Note that terms correspond to substitutions to an extended context, which are then subject to the diagrams above.
\end{notation}

\chapter{Discreteness}\label{ch:discreteness}
It is common in categorical models of dependent type theory to designate a certain class of morphisms $\homclass F$ in the category of contexts, typically called \textbf{fibrations}, and to require that for any type $\Gamma \sez T \type$, the morphism $\pi : \Gamma.T \to \Gamma$ is a fibration. Types satisfying this criterion are then called \textbf{fibrant}. A context $\Gamma$ is called \textbf{fibrant} if the map $\Gamma \to ()$ is a fibration.

Typically, the fibrations can be characterized using a lifting property with respect to another class of morphisms $\homclass H$ which we will call \textbf{horn inclusions}. If $\eta : \Lambda \to \Delta$ is a horn inclusion, then we call a map $\sigma : \Lambda \to \Gamma$ a \textbf{horn} in $\Gamma$ and if $\sigma$ factors as $\tau \eta$, then $\tau : \Delta \to \Gamma$ is called a \textbf{filler} of $\sigma$.

Now the fibrations are usually those morphisms $\rho : \Gamma' \to \Gamma$ such that any horn $\sigma$ in $\Gamma'$ which has a filler $\tau$ in $\Gamma$ (meaning that $\rho\sigma$ has a filler $\tau$), also has a (compatible) filler in $\Gamma'$, i.e. commutative squares like the following have a diagonal:
\begin{equation}
	\xymatrix{
		\Lambda \ar[r]^\sigma \ar[d]_{\eta \in \homclass H}
		& \Gamma' \ar[d]^{\rho}
		\\
		\Delta \ar[r]_{\tau} \ar@{.>}[ru]
		&
		\Gamma
	}
\end{equation}
In this text, we will call the fibrant types and contexts \textbf{discrete} and we will speak of \textbf{discrete maps} instead of fibrations, because this better reflects the idea behind what we are doing.

In many models, only fibrant contexts are used. However, we will also consider non-discrete contexts because our modalities do not preserve discreteness.

\section{Definition}
\begin{definition}
	Let $\XX$ stand for $\IB$, $\IP$ or $\IE$.
	We say that a defining substitution $\gamma : \DSub{(W, \ctxline{\var i})}{\Gamma}$ or a defining term $(W, \ctxline{\var i}) \Dsez t : T \dsub \gamma$ is \textdef{degenerate in $\var i$} if it factors over $(\facewkn{\var i}) : \PSub{(W, \ctxline{\var i})}{W}$.
\end{definition}
The notion of degeneracy is thus meaningful in the CwFs $\widehat \cubecat$ and $\widehat \bpcubecat$. Thinking of $\var i$ as a variable, this means that $\gamma$ and $t$ do not refer to $\var i$. Thinking of $\var i$ as a dimension, this means that $\gamma$ and $t$ are flat in dimension $\var i$. Note that $t$ can only be degenerate in $\var i$ if $\gamma$ is.
\begin{corollary}
	For a defining substitution $\gamma : \DSub{(W, \ctxline{\var i})}{\Gamma}$ or a defining term $(W, \ctxline{\var i}) \Dsez t : T \dsub \gamma$, the following are equivalent:
	\begin{enumerate}
		\item $\gamma$/$t$ is degenerate in $\var i$,
		\item $\gamma = \gamma \circ (0/\var i, \facewkn{\var i})$; $t = t \psub{0/\var i, \facewkn{\var i}}$,
		\item $\gamma = \gamma \circ (1/\var i, \facewkn{\var i})$; $t = t \psub{1/\var i, \facewkn{\var i}}$. \qed
	\end{enumerate}
\end{corollary}
\begin{definition}
	We call a context \textbf{discrete} if all of its cubes are degenerate in every path dimension.
	
	We call a map $\rho : \Gamma' \to \Gamma$ \textbf{discrete} if every defining substitution $\gamma$ of $\Gamma'$ is degenerate in every path dimension in which $\rho \circ \gamma$ is degenerate.
	
	We call a type $\Gamma \sez T \type$ \textbf{discrete} (denoted $\Gamma \sez T \dtype$) if every defining term $t : T \dsub \gamma$ is degenerate in every path dimension in which $\gamma$ is degenerate.
\end{definition}
\begin{proposition}
	A type $\Gamma \sez T \type$ is discrete if and only if $\pi : \Gamma.T \to \Gamma$ is discrete.
\end{proposition}
\begin{proof}
	\begin{itemize}
		\item[$\Rightarrow$] Assume that $T$ is discrete. Pick $(\gamma, t) : \DSub{(W, \ctxpath{\var i})}{\Gamma.T}$ such that $\pi \circ (\gamma, t) = \gamma$ is degenerate in $\var i$. Then $t$ is degenerate in $\var i$ by discreteness of $T$ and so is $(\gamma, t)$.
		\item[$\Leftarrow$] Assume that $\pi$ is discrete. Pick $t : T \dsub \gamma$ where $\gamma$ is degenerate in $\var i$. Then $(\gamma, t)$ is degenerate in $\var i$ since $\pi(\gamma, t) = \gamma$, and hence $t$ is degenerate in $\var i$. \qedhere
	\end{itemize}
\end{proof}
\begin{proposition}
	A context $\Gamma$ is discrete if and only if $\Gamma \to ()$ is discrete.
\end{proposition}
\begin{proof}
	Note that every defining substitution of $()$ is degenerate in every dimension. This proves the claim.
\end{proof}
\begin{proposition}
	A map $\rho : \Gamma' \to \Gamma$ is discrete if and only if it has the lifting property with respect to all horn inclusions $(\facewkn{\var i}) : \yoneda(W, \ctxpath{\var i}) \to \yoneda W$.
\end{proposition}
\begin{proof}
	\begin{itemize}
		\item[$\Rightarrow$] Suppose that $\rho$ is discrete and consider a square
		\begin{equation}
			\xymatrix{
				\yoneda(W, \ctxpath{\var i}) \ar[r]^{\gamma'} \ar[d]_{(\facewkn{\var i})}
				& \Gamma' \ar[d]^{\rho}
				\\
				\yoneda W \ar[r]_{\gamma}
				&
				\Gamma.
			}
		\end{equation}
		Then the defining substitution $\rho \circ \gamma' : \DSub{(W, \ctxpath{\var i})}{\Gamma}$ clearly factors over $(\facewkn{\var i})$ so that it is degenerate in $\var i$. By degeneracy of $\rho$, the same holds for $\gamma'$, yielding the required diagonal.
		
		\item[$\Leftarrow$] Suppose that $\rho$ has the lifting property and take $\gamma' : \DSub{(W, \ctxpath{\var i})}{\Gamma'}$ such that $\rho \circ \gamma'$ is degenerate in $\var i$. This gives us a square as above, which has a diagonal, showing that $\gamma'$ is degenerate. \qedhere
	\end{itemize}
\end{proof}
\begin{example}
	In the examples, we will develop the content of this chapter for the CwF $\widehat{\RGcat}$ of reflexive graphs. This will fail when we come to product types, which is the reason why we choose to work with presheaves over $\widehat{\bpcubecat}$ which contain not only points, paths and bridges, but also coherence cubes. An alternative is to require edges (bridges and paths) to be proof-irrelevant in the style of \cite{dtt-parametricity}, but this property cannot be satisfied by the universe. Remember that $\RGcat$ has objects $()$ and $(\ctxedge{\var i})$ and the same morphisms between them as we find in $\cubecat$.
	
	We call an edge $p : \DSub{(\ctxedge{\var i})}{\Gamma}$ \textbf{degenerate} if it is the constant edge on some point $x : \DSub{()}{\Gamma}$, i.e.\ $p = x \circ (\facewkn{\var i})$. This point is uniquely determined, as it must be equal to the edge's source and target. So we can say that $p$ is degenerate iff $p = p \circ (0/\var i) \circ (\facewkn{\var i})$ iff $p \circ (1/\var i) \circ (\facewkn{\var i})$.
	
	We call a context (reflexive graph) \textbf{discrete} if all of its edges are degenerate.
	
	We call a map $\rho : \Gamma' \to \Gamma$ \textbf{discrete} if every edge $\gamma : \DSub{(\ctxedge{\var i})}{\Gamma}$ for which $\rho \circ \gamma$ is degenerate, is itself degenerate.
	
	We call a type $\Gamma \sez T \type$ \textbf{discrete} if every defining edge $(\ctxedge{\var i}) \Dsez t : T \dsub \gamma$ over a degenerate edge $\gamma$, is degenerate.
	
	One can prove:
	\begin{itemize}
		\item A type $\Gamma \sez T \type$ is discrete if and only if $\pi : \Gamma.T \to \Gamma$ is discrete.
		\item A context $\Gamma$ is discrete if and only if $\Gamma \to ()$ is discrete.
		\item A map $\rho : \Gamma' \to \Gamma$ is discrete if it has the lifting property with respect to the horn inclusion $(\facewkn{\var i}) : \yoneda(\ctxedge{\var i}) \to \yoneda()$.
	\end{itemize}
\end{example}

\section{A model with only discrete types}
\begin{theorem}
	If we define $\DTy(\Gamma)$ to be the set of all \emph{discrete} types over $\Gamma$, then we obtain a new CwF $\bpdisc$ which also supports dependent products, dependent sums and identity types.
\end{theorem}
We prove this theorem in several parts. In the examples, we will try and fail to prove the corresponding theorem for $\RGcat$: we do have a CwF $\RGdisc$ and it supports dependent sums and identity types, but we will fail to prove that it supports dependent products.

\subsection{The category with families $\bpdisc$}
\begin{lemma}
	$\bpdisc$ is a well-defined category with families (see \cref{def:cwf}).
\end{lemma}
\begin{proof}
	The only thing we need to prove in order to show this is that $\DTy$ still has a morphism part, i.e. that discreteness of types is preserved under substitution.

	So pick a substitution $\sigma : \Delta \to \Gamma$ and a discrete type $\Gamma \sez T \dtype$. Take a defining term $(W, \ctxpath{\var i}) \Dsez t : T \sub{\sigma} \dsub \delta$ and assume that $\delta$ is degenerate along $\var i$. We have to prove that $t$ is, too. But if $\delta$ factors over $(\facewkn{\var i})$, then so does $\sigma \delta$, and therefore also $t : T \dsub{\sigma \delta}$ by discreteness of $T$. Since restriction for $T[\sigma]$ is inherited from $T$, the term $t$ is also degenerate as a defining term of $T[\sigma]$.
\end{proof}
\begin{example}
	Show that $\RGdisc$ is a well-defined category with families.
\end{example}

\subsection{Dependent sums}
\begin{lemma}
	The category with families $\bpdisc$ supports dependent sums.
\end{lemma}
\begin{proof}
	Take a context $\Gamma$ and discrete types $\Gamma \sez A \dtype$ and $\Gamma.A \sez B \dtype$. It suffices to show that $\Sigma A B$ is discrete. Pick a defining term $(W, \ctxpath{\var i}) \Dsez (a, b) : \Sigma A B \dsub \gamma$ where $\gamma$ is degenerate along $\var i$. Then by discreteness of $A$, $a$ is degenerate along $\var i$ and so is $(\gamma, a) : \DSub{(W, \ctxpath{\var i})}{\Gamma.A}$. Then by discreteness of $B$, $b$ is also degenerate along $\var i$ and hence so is $(a, b)$.
\end{proof}

\subsection{Dependent products}
\begin{lemma}
	The category with families $\bpdisc$ supports dependent products.
\end{lemma}
In fact, the proof of this lemma proves something stronger:
\begin{lemma}\label{thm:disc-prod}
	Given a context $\Gamma$, an arbitrary type $\Gamma \sez A \type$ and a discrete type $\Gamma.A \sez B \dtype$, the type $\Pi A B$ is discrete.
\end{lemma}
\begin{example}
	In order to clarify the idea behind the proof for $\bpdisc$, we will first (vainly) try to prove the same lemma for $\RGdisc$.
	
	Pick a context $\Gamma$, a type $\Gamma \sez A \type$ and a discrete type $\Gamma.A \sez B \dtype$. It suffices to show that $\Pi A B$ is discrete. Pick an edge $(\ctxedge{\var i}) \Dsez h : (\Pi A B) \dsub{\gamma (\facewkn{\var i})}$ over the constant edge at point $\gamma : \DSub{()}{\Gamma}$. Write
	\begin{equation}
		() \Dsez f := h \psub{0/\var i}, g := h \psub{1 / \var i} : (\Pi A B) \dsub \gamma.
	\end{equation}
	In order to have a visual representation, assume that $A$ looks like the upper diagram here; then the image of $h$ is the lower one (degenerate edges are hidden in both diagrams). Every cell projects to the cell from $\Gamma$ shown on its left, or a constant edge of it:
	\begin{equation}
		\xymatrix{
			a \ar@{-}[rr]^{a_1}
			&& a' \ar@{-}[rr]^{a_2}
			&& a''
			\\
			f \cdot a
				\ar@{-}[rr]^{f \psub{\facewkn{\var i}} \cdot a_1} 
				\ar@{-}[dd]^(.7){h \cdot (a \psub{\facewkn{\var i}})}
			  	\ar@{-}[rrdd]^{h \cdot a_1}
			&& f \cdot a'
				\ar@{-}[rr]^{f \psub{\facewkn{\var i}} \cdot a_2}
				\ar@{-}[dd]^(.7){h \cdot (a' \psub{\facewkn{\var i}})}
			  	\ar@{-}[rrdd]^{h \cdot a_2}
			&& f \cdot a'' 
				\ar@{-}[dd]^(.7){h \cdot (a'' \psub{\facewkn{\var i}})}
			\\ \\
			g \cdot a
				\ar@{-}[rr]_{g \psub{\facewkn{\var i}} \cdot a_1}
			&& g \cdot a'
				\ar@{-}[rr]_{g \psub{\facewkn{\var i}} \cdot a_2}
			&& g \cdot a''
		}
	\end{equation}
	We try to show that $h$ is degenerate by showing that $f \psub{\facewkn{\var i}} = h$. Recall that a defining $W \Dsez k : (Pi A B)\dsub \gamma$ is fully determined if we know all $k \psub \vfi \cdot a$ for all $\vfi : \PSub V W$ and all $V \Dsez a : A \dsub{\gamma \vfi}$. There are five candidates for $\vfi$:
	\begin{itemize}
		\item[$(0/\var i)$] We have $f \psub{\facewkn{\var i}} \psub{0/\var i} = f = h \psub{0/\var i}$.
		\item[$(0/\var i, \facewkn{\var i})$] This follows by further restriction by $(\facewkn{\var i})$.
		\item[$(1/\var i)$] Take a node $() \Dsez a : A \dsub \gamma$. We need to show that $f \psub{\facewkn{\var i}} \psub{1/\var i} \cdot a = h \psub{1/\var i} \cdot a$, i.e.\ $f \cdot a = g \cdot a$. To this end, consider $h \cdot (a \psub{\facewkn{\var i}})$. The fact that $\loch \cdot \loch$ commutes with restriction, guarantees that this is an edge from $f \cdot a$ to $g \cdot a$. However, it has type $(\ctxedge{\var i}) \Dsez h \cdot (a \psub{\facewkn{\var i}}) : B [(\gamma (\facewkn{\var i})) \subext] \dsub{\id, a \psub{\facewkn{\var i}}} = B \dsub{(\gamma, a)(\facewkn{\var i})}$. So it clearly lives over a degenerate edge and hence it is degenerate, implying that $f \cdot a = g \cdot a$.
		\item[$(1/\var i, \facewkn{\var i})$] Take an edge $(\ctxedge{\var i}) \Dsez a : A \dsub{\gamma (\facewkn{\var i})}$. We need to show that $f \psub{\facewkn{\var i}} \psub{1/\var i, \facewkn{\var i}} \cdot a = h \psub{1/\var i, \facewkn{\var i}} \cdot a$, i.e. $f \psub{\facewkn{\var i}} \cdot a = g \psub{\facewkn{\var i}} \cdot a$. This is an equality of edges. If we could introduce an additional variable $\var j$, we could apply the same technique as for the source extractor $(0 / \var i)$, considering a square from $f \psub{\facewkn{\var i}} \cdot a$ to $g \psub{\facewkn{\var i}} \cdot a$ that would have to be degenerate in one dimension. However, the squares would cause functions to have more components, which we would have to prove equal, requiring a third variable. This is why we chose to start from a model $\bpcubecat$ in which we can have arbitrarily many dimension variables. However, in $\RGdisc$, we cannot proceed.
		\item[$\id$] Take an edge $(\ctxedge{\var i}) \Dsez a : A \dsub{\gamma (\facewkn{\var i})}$. We need to show that $f \psub{\facewkn{\var i}} \cdot a = h \cdot a$. Now $h \cdot a$ is going to be the diagonal of the square we fancied in the previous clause. If we know that this degenerate, then the diagonal is equal to the sides. However, in $\RGdisc$, we cannot proceed.
	\end{itemize}	 
\end{example}
\begin{proof}
	Pick a context $\Gamma$, a type $\Gamma \sez A \type$ and a discrete type $\Gamma.A \sez B \dtype$. It suffices to show that $\Pi A B$ is discrete. Pick $(W, \ctxpath{\var i}) \Dsez h : (\Pi A B) \dsub{\gamma(\facewkn{\var i})}$. We name the $\var i$-source $f := h \psub{0/\var i}$ and the $\var i$-target $g := h \psub{1 / \var i}$. We have $W \Dsez f, g : (\Pi A B) \dsub \gamma$.
	
	We show that $h$ is degenerate along $\var i$ by showing that $f \psub{\facewkn{\var i}} = h$. So pick some $\vfi : \PSub{V}{(W, \ctxpath{\var i})}$. We make a case distinction by inspecting $\var i \psub{\vfi} \in \accol{0, 1} \uplus V$:
	\begin{description}
		\item[$\var i \psub \vfi = 0$] Then $\vfi = (0/\var i)\psi$ for some $\psi : \PSub V W$. Then we have $f \psub{\facewkn{\var i}} \psub{\vfi} = f \psub \psi$ and $h \psub \vfi = f \psub \psi$.
		
		\item[$\var i \psub \vfi = 1$] Then $\vfi = (1/\var i)\psi$ for some $\psi : \PSub V W$. Then we have $f \psub{\facewkn{\var i}} \psub{\vfi} = f \psub \psi$ and $h \psub \vfi = g \psub \psi$ and we need to show that $V \Dsez f \psub \psi \cdot a = g \psub \psi \cdot a : B \dsub{\psi, a}$ for every $V \Dsez a : A \dsub \psi$.
		
		Without loss of generality, we may assume that $\var i \not\in V$.
		Then we have a path $(V, \ctxpath{\var i}) \Dsez h \psub{\psi, \var i/\var i} \cdot (a \psub{\facewkn{\var i}}) : B \dsub{(\psi, a)(\facewkn{\var i})}$, with source $f \psub \psi$ and target $g \psub \psi$. Moreover, by discreteness of $B$, it is degenerate along $\var i$, implying that source and target are equal.
		
		\item[$\var i \psub \vfi \in V$] Without loss of generality, we may assume that $(W, \ctxpath{\var i})$ and $V$ are disjoint. Write $\var k = \var i \psub \vfi$ (and note that $\var k$ may be either a bridge or a path variable). Then $\vfi$ factors as $(\psi, \var i / \var i)(\var k / \var i) = (\var k / \var i)(\psi, \var k / \var k)$ for some $\psi : \PSub V W$. We have $f \psub{\facewkn{\var i}} \psub{\vfi} = f \psub \psi$ and $h \psub \vfi = h \psub{\var k / \var i} \psub \psi$. We have to show that $V \Dsez f \psub \psi \cdot a = h \psub{\psi, \var i / \var i} \psub{\var k / \var i} \cdot a : B \dsub{\psi, a}$ for all $V \Dsez a : A \dsub \psi$.
		
		Again, we have a path $(V, \ctxpath{\var i}) \Dsez h \psub{\psi, \var i/\var i} \cdot (a \psub{\facewkn{\var i}}) : B \dsub{(\psi, a)(\facewkn{\var i})}$ with $\var i$-source $f \psub \psi \cdot a$ and $(\var k/\var i)$-diagonal $h \psub{\psi, \var i / \var i} \psub{\var k / \var i} \cdot a$. This path is again degenerate in $\var i$, showing that the source and the diagonal are equal. \qedhere
	\end{description}
\end{proof}

\subsection{Identity types and propositions}
\begin{lemma}
	The category with families $\bpdisc$ supports identity types.
\end{lemma}
We even have a stronger result:
\begin{lemma}
	Propositions are discrete. \qed
\end{lemma}

\subsection{Glueing}
\begin{lemma}
	The category with families $\bpdisc$ supports glueing.
\end{lemma}
\begin{proof}
	Suppose we have $\Gamma \sez A \dtype$, $\Gamma \sez P \prop$, $\Gamma.P \sez T \dtype$ and $\Gamma.P \sez f : T \to A[\pi]$. It suffices to show that $G = \Gluesys{A}{\Gluesysclauseb{P}{T}{f}}$ is discrete. So pick $(W, \ctxpath{\var i}) \Dsez b : G \dsub{\gamma}$ where $\gamma$ is degenerate along $\var i$.
	
	If $P \dsub \gamma = \accol \star$, then $b \psub{0/\var i, \facewkn{\var i}}^G = b \psub{0/\var i, \facewkn{\var i}}^{T[\id, \star]} = b$ by discreteness of $T$.
	
	If $P \dsub \gamma = \eset$, then $b$ is of the form $(a \mapsfrom t)$. Then $(a \mapsfrom t) \psub{0/\var i, \facewkn{\var i}} = (a \psub{0/\var i, \facewkn{\var i}} \mapsfrom t[(0/\var i, \facewkn{\var i}) \subext]) = (a \mapsfrom t[(0/\var i, \facewkn{\var i}) \subext])$ by discreteness of $A$. Finally, discreteness of $\Pi P T$ shows that
	\begin{equation}
		t[(0/\var i, \facewkn{\var i}) \subext]
		= \dap (\dlambda (t[(0/\var i, \facewkn{\var i}) \subext]))
		= \dap((\dlambda t) \psub{0/\var i, \facewkn{\var i}})
		= \dap (\dlambda t) = t. \qedhere
	\end{equation}
\end{proof}

\subsection{Welding}
\begin{lemma}
	The category with families $\bpdisc$ supports welding.
\end{lemma}
\begin{proof}
	Suppose we have $\Gamma \sez A \dtype$, $\Gamma \sez P \prop$, $\Gamma.P \sez T \dtype$ and $\Gamma.P \sez f : A[\pi] \to T$. It suffices to show that $\Omega = \Weldsys{A}{\Weldsysclauseb P T f}$ is discrete. So pick $(W, \ctxpath{\var i}) \Dsez w : \Omega \dsub \gamma$ where $\gamma$ is degenerate along $\var i$.
	
	If $P \dsub \gamma = \accol \star$, then $w \psub{0/\var i, \facewkn{\var i}}^\Omega = w \psub{0/\var i, \facewkn{\var i}}^{T[\id, \star]} = w$ by discreteness of $T$.
	
	If $P \dsub \gamma = \eset$, then $w \psub{0/\var i, \facewkn{\var i}}^\Omega = w \psub{0/\var i, \facewkn{\var i}}^{A} = w$ by discreteness of $A$.
\end{proof}

\section{Discreteness and cohesion}
In this section, we consider the interaction between discreteness and cohesion. In the first subsection, we characterize discrete contexts as those that are in the image of the discrete functor $\cohdisc$, or equivalently in the image of $\flat$. In the second one, we show that $\coshp$ preserves discreteness. In the rest of the section we are concerned with making things discrete by quotienting out paths. For a context $\Gamma$, we will define a discrete context $\quotshp \Gamma$ and a substitution $\inquotshp : \Gamma \to \quotshp \Gamma$ into it. This will enable us to finally construct $\cohpi$. For a type $\Gamma \sez T \type$, we will define a discrete type $\Gamma \sez \quotshp T \dtype$ and a mapping $\hatinquotshp : \Tm(\Gamma, T) \to \Tm(\Gamma, \quotshp T)$. This is a prerequisite for defining existential types.
\begin{remark}
	Note that in models of HoTT, this process of forcing a type to be fibrant is ill-behaved in the sense that it does not commute with substitution: we will typically not have $\quotshp (T[\sigma]) = (\quotshp T)[\sigma]$. We will show that our shape operator does commute with substitution. This is likely related to the fact that all our horn inclusions are epimorphisms (levelwise surjective presheaf maps), so that forcing something to be discrete is an operation that transforms presheaves locally, not globally. Put differently: if a type is fibrant in HoTT, we obtain transport functions which allow us to move things around and derive a contradiction (see \cite{nlab:fibrant-replacement} for details). Discrete types however do not provide any transport or composition operations.
\end{remark}
\subsection{Discrete contexts and the discrete functor}
\begin{proposition}\label{thm:discrete-contexts-and-cohesion}
	For a context $\Gamma \in \Psh(\cat P)$, the following are equivalent:
	\begin{enumerate}
		\item $\Gamma$ is discrete,
		\item $\Gamma$ is isomorphic to $\cohdisc \Theta$ for some cubical set $\Theta \in \Psh(\cat Q)$,
		\item The substitution $\kappa : \flat \Gamma \to \Gamma$ is an isomorphism.
	\end{enumerate}
\end{proposition}
\begin{proof}
	\begin{description}
		\item[$1 \Rightarrow 3$.] Assume that $\Gamma$ is discrete. We show that $\kappa : \flat \Gamma \to \Gamma$ is an isomorphism. Pick a defining substitution $\gamma : \DSub W \Gamma$. Because $\Gamma$ is discrete, $\gamma$ is degenerate in every path dimension, i.e. it factors over $\varsigma : W \to \shp W$, say $\gamma = \gamma' \varsigma$. Then we have $\fpshadj{\shp}(\gamma') : \DSub{W}{\flat \Gamma}$ and moreover $\kappa \circ \fpshadj{\shp}(\gamma') = \fpsh \varsigma \circ \fpshadj{\shp}(\gamma') = \gamma' \circ \varsigma = \gamma$.
		
		\item[$3 \Rightarrow 2$.] Note that $\flat \Gamma = \cohdisc \cohfget \Gamma$.
		
		\item[$2 \Rightarrow 1$.] It suffices to prove that $\cohdisc \Theta$ is discrete. Pick some $\fpshadj \cohpi(\theta) : \DSub{(W, \ctxpath{\var i})}{\cohdisc \Theta}$. Then we have $\theta : \DSub{\cohpi (W, \ctxpath{\var i}) = \cohpi W}{\Theta}$ and hence $\fpshadj \cohpi(\theta) : \DSub W {\cohdisc \theta}$. Moreover, $\fpshadj \cohpi(\theta) \circ (\facewkn{\var i}) = \fpshadj \cohpi(\theta \circ \cohpi(\facewkn{\var i})) = \fpshadj \cohpi(\theta)$, showing that the picked defining substitution is degenerate along $\var i$. \qedhere
	\end{description}
\end{proof}

\subsection{The $\coshp$ functor preserves discreteness}
\begin{lemma}\label{thm:coshp-preserves-discreteness}
	For any discrete type $\Gamma \sez T \dtype$, the type $\coshp \Gamma \sez \coshp T \dtype$ is also discrete.
\end{lemma}
\begin{proof}
	Pick a defining term $(W, \ctxpath{\var i}) \Dsez \fpshadj \sharp(t) : (\coshp T) \dsub{\fpshadj \sharp(\gamma) (\facewkn{\var i})}$; we will show that it is degenerate. Note that $\fpshadj \sharp(\gamma) (\facewkn{\var i}) = \fpshadj \sharp (\gamma \circ \sharp (\facewkn{\var i})) = \fpshadj \sharp (\gamma (\facewkn{\var i}))$. Hence, we have $\sharp (W, \ctxpath{\var i}) = (\sharp W, \ctxpath{\var i}) \Dsez t : T \dsub{\gamma (\facewkn{\var i})}$. By discreteness of $T$, $t$ factors over $(\facewkn{\var i})$, i.e.\ $t = t' \psub{\facewkn{\var i}}$. Then $\fpshadj \sharp(t) = \fpshadj \sharp(t' \psub{\facewkn{\var i}}) = \fpshadj \sharp(t' \psub{\sharp (\facewkn{\var i})}) = \fpshadj \sharp(t) \psub{\facewkn{\var i}}$.
\end{proof}

\subsection{Equivalence relations on presheaves}
Before we can define $\quotshp T$, we need a little bit of theory on equivalence relations on (dependent) presheaves. We will then be able to define $\quotshp T$ straightforwardly as a quotient of $T$.
\begin{definition}
	An \textbf{equivalence relation $E$ on a presheaf} $\Gamma \in \widehat \catW$ consists of:
	\begin{itemize}
		\item For every $W \in \catW$, an equivalence relation $E_W$ on $(\DSub W \Gamma)$,
		\item So that if $E_W(\gamma, \gamma')$ and $\vfi : \PSub V W$, then $E_V(\gamma \vfi, \gamma' \vfi)$.
	\end{itemize}
	We will denote this as $E \eqrel \Gamma$.
	
	Similarly, an \textbf{equivalence relation $E$ on a dependent presheaf} $(\Gamma \sez T \type)$ consists of:
	\begin{itemize}
		\item For every $W \in \catW$ and every $\gamma : \DSub W \Gamma$, an equivalence relation $E \dsub \gamma$ on $T \dsub \gamma$,
		\item So that if $E \dsub \gamma (s, t)$ and $\vfi : V \to W$, then $E \dsub{\gamma \vfi}(s \psub \vfi, t \psub \vfi)$.
	\end{itemize}
	We will denote this as $\Gamma \sez E \eqrel T$.
\end{definition}
Given a presheaf map $\sigma : \Delta \to \Gamma$ and an equivalence relation $\Gamma \sez E \eqrel T$, we can easily define its substitution $\Delta \sez E[\sigma] \eqrel T[\sigma]$, by setting $E[\sigma]\dsub \delta = E\dsub{\sigma \delta}$. Substitution of equivalence relations obviously respects identity and composition.
\begin{lemma}
	The intersection of arbitrarily many equivalence relations on a given (dependent) presheaf, is again an equivalence relation on that presheaf. \qed
\end{lemma}
\begin{lemma}[Substitution of equivalence relations, has a right adjoint]
	Given a substitution $\sigma : \Delta \to \Gamma$ and equivalence relations $\Delta \sez F \eqrel T[\sigma]$ and $\Gamma \sez E \eqrel T$, there is an equivalence relation $\Gamma \sez \forall_\sigma F \eqrel T$ such that $E[\sigma] \subseteq F$ if and only if $E \subseteq \forall_\sigma F$.
\end{lemma}
The notation $\forall$ is related to the notation of the product type. Indeed, the product type is right adjoint to weakening, which is a special case of substitution.
\begin{proof}
	Given $W \Dsez x, y : T \dsub \gamma$, we set $(\forall_\sigma F) \dsub \gamma (x, y)$ if and only if for every face map $\vfi : \PSub V W$ and every $\delta : \DSub V \Delta$ such that $\sigma \delta = \gamma \vfi$, we have $F \dsub \delta(x \psub \vfi, y \psub \vfi)$. The quantification over $V$ guarantees that equivalence is preserved under restriction.
	
	We now show that $E[\sigma] \subseteq F$ if and only if $E \subseteq \forall_\sigma F$.
	\begin{itemize}
		\item[$\Rightarrow$] Assume that $E[\sigma] \subseteq F$. Pick $W \Dsez x, y : T \dsub \gamma$ such that $E \dsub \gamma (x, y)$. We show that $\forall_\sigma F \dsub \gamma(x, y)$. For any $\vfi : \PSub V W$ we have $E \dsub{\gamma \vfi}(x \psub \vfi, y \psub \vfi)$ and hence for any $\delta : \DSub V \Delta$ such that $\sigma \delta = \gamma \vfi$, we have $E[\sigma] \dsub{\delta}(x \psub \vfi, y \psub \vfi)$, implying $F \dsub \delta(x \psub \vfi, y \psub \vfi)$.
		
		\item[$\Leftarrow$] Assume that $E \subseteq \forall_\sigma F$. Pick $W \Dsez x, y : T[\sigma] \dsub \delta$ and assume that $E[\sigma]\dsub \delta(x, y)$. We show that $F \dsub \delta(x, y)$. Clearly, we have $\forall_\sigma F \dsub{\sigma \delta}(x, y)$. Instantiating $\vfi$ with $\id$, we can conclude $F \dsub{\delta}(x, y)$. \qedhere
	\end{itemize}
\end{proof}
\begin{lemma}[Applying a lifted functor to an equivalence relation, has a right adjoint]
	Assume a functor $K : \catV \to \catW$ and equivalence relations $\Gamma \sez_{\widehat \catW} E \eqrel T$ and $\fpsh K \Gamma \sez_{\widehat \catV} F \eqrel \fpsh K T$. There is an equivalence relation $\Gamma \sez_{\widehat \catW} \forall_{K} F \eqrel T$ such that $\fpsh K E \subseteq F$ if and only if $E \subseteq \forall_{K} F$. Here, $\fpsh K E$ is defined by $\fpsh K E \dsub{\fpshadj K (\gamma)}(\fpshadj K(x), \fpshadj K(y)) = E \dsub \gamma(x, y)$.
\end{lemma}
\begin{proof}
	The idea is entirely the same. Given $W \Dsez_{\widehat \catW} x, y : T \dsub \gamma$, we set $\forall_K F \dsub \gamma(x, y)$ if and only if for every $V \in \cat V$ and every face map $\vfi : \PSub{KV}{W}$, we have $F \dsub{\fpshadj K(\gamma \vfi)}(\fpshadj K(x \psub \vfi), \fpshadj K(y \psub \vfi))$. The quantification over $V$ guaranties that equivalence is preserved under restriction.
	
	We now show that $\fpsh K E \subseteq F$ if and only if $E \subseteq \forall_K F$.
	\begin{itemize}
		\item[$\Rightarrow$] Assume that $\fpsh K E \subseteq F$. Pick $W \Dsez_{\widehat \catW} x, y : T \dsub \gamma$ and assume that $E \dsub \gamma(x, y)$. We show that $\forall_K F \dsub \gamma(x, y)$. For any $\vfi : \PSub{KV}{W}$, we have $E \dsub{\gamma \vfi}(x \psub \vfi, y \psub \vfi)$, i.e. $\fpsh K E \dsub{\fpshadj K(\gamma \vfi)}(\fpshadj K(x \psub \vfi), \fpshadj K(y \psub \vfi))$, implying $F \dsub{\fpshadj K(\gamma \vfi)}(\fpshadj K(x \psub \vfi), \fpshadj K(y \psub \vfi))$.
		
		\item[$\Leftarrow$] Assume that $E \subseteq \forall_{\fpsh K} F$. Pick $V \Dsez_{\widehat \catV} \fpshadj K(x), \fpshadj K(y) : \fpsh K T \dsub{\fpshadj K(\gamma)}$ such that $\fpsh K E \dsub{\fpshadj K(\gamma)}(\fpshadj K(x), \fpshadj K(y))$, i.e.\ $E \dsub \gamma(x, y)$. Then $\forall_K F \dsub \gamma(x, y)$. Instantiating $\vfi = \id : \PSub{KV}{KV}$, we have $F \dsub{\fpshadj K(\gamma)}(\fpshadj K(x), \fpshadj K(y))$. \qedhere
	\end{itemize}
\end{proof}
\begin{definition}
	If $E \eqrel \Gamma$, then we define the context $\Gamma/E$ by setting $(\DSub{W}{\Gamma/E}) = (\DSub{W}{\Gamma})/E_W$ and $\overline\gamma \circ \vfi = \overline{\gamma \circ \vfi}$, which is well-defined by virtue of the second bullet in the definition of an equivalence relation.
	
	If $\Gamma \sez E \eqrel T$, then we define $\Gamma \sez T/E \type$ by setting $(T/E) \dsub \gamma = T \dsub \gamma/E \dsub{\gamma}$ and $\overline{x} \psub \vfi = \overline{x \psub \vfi}$.
\end{definition}
One easily checks that $(T/E)[\sigma] = T[\sigma]/E[\sigma]$.

\subsection{Discretizing contexts and the functor $\cohpi$}\label{sec:def-cohpi}
In this section, our aim is to construct for any context $\Gamma$, a discrete context $\quotshp \Gamma$ with a substitution $\inquotshp : \Gamma \to \quotshp \Gamma$ such that any substitution $\tau : \Gamma \to \Gamma'$ to a discrete context $\Gamma'$, factors uniquely over $\inquotshp$. The effect of applying $\quotshp$ will be that we are contracting every path to a point (and more generally, that we are contracting every cell in all its path dimensions). When we postcompose with $\cohfget$, we obtain our desired left adjoint $\cohpi = \cohfget \quotshp$ of $\cohdisc$.

\subsubsection{Discretizing contexts}
Our approach is quite straightforward: we simply divide out the least equivalence relation that makes the quotient discrete. Recall that a context $\Gamma$ is discrete iff for every $\gamma : \DSub{(W, \ctxpath{\var i})}{\Gamma}$, we have that $\gamma = \gamma (0 / \var i, \facewkn{\var i})$.
\begin{definition}\label{def:ctxshp}
	Let the \textbf{shape equivalence relation} $\sheq$ on $\Gamma$ be the least equivalence relation such that $\sheq(\gamma, \gamma(0 / \var i, \facewkn{\var i}))$ for any $\gamma : \DSub{(W, \ctxpath{\var i})}{\Gamma}$. Then we define the \textbf{shape quotient} of $\Gamma$ as $\quotshp \Gamma = \Gamma/\sheq$. Given a substitution $\sigma : \Gamma \to \Gamma'$, we define $\quotshp \sigma : \quotshp \Gamma \to \quotshp \Gamma'$ by setting $\quotshp \sigma \circ \overline \gamma = \overline{\sigma \circ \gamma}$. This constitutes a functor $\quotshp : \widehat{\bpcubecat} \to \widehat{\bpcubecat}$.
	
	We define a natural transformation $\inquotshp : \Id \to \quotshp$ by $\inquotshp \circ \gamma := \overline \gamma$.
\end{definition}
It is easy to see that any substitution $\tau$ from $\Gamma$ into a discrete context, factors uniquely over $\inquotshp : \Gamma \to \quotshp \Gamma$. In particular, $\quotshp \sigma$ is well-defined.
\begin{lemma}\label{thm:kappa-injective}
	For any context $\Gamma \in \widehat \bpcubecat$, the substitution $\kappa : \flat \Gamma \to \Gamma$ is an injective presheaf map, meaning that $\kappa \circ \loch : (\DSub{W}{\flat \Theta}) \to (\DSub{W}{\Theta})$ is injective for every $W$.
\end{lemma}
\begin{proof}
	Given $\fpshadj \shp(\gamma) : \DSub{W}{\flat \Gamma}$, we have $\kappa \circ \fpshadj \shp(\gamma) = \gamma \varsigma$, where $\varsigma : \PSub{W}{\shp W}$ is easily seen to have a right inverse.
\end{proof}
\begin{lemma}
	Any substitution $\tau : \Gamma \to \Theta$ from a discrete context $\Gamma$ to any context $\Theta$, factors uniquely over $\kappa : \flat \Theta \to \Theta$. Hence, $\quotshp$ is left adjoint to $\flat$.
\end{lemma}
\begin{proof}
	\begin{description}
		\item[Existence.] Pick $\gamma : \DSub{W}{\Gamma}$. Since $\Gamma$ is discrete, $\gamma$ factors over $\varsigma : \PSub{W}{\shp W}$ as $\gamma = \gamma' \varsigma$. Then we have $\tau \gamma' : \DSub{\shp W}{\Theta}$ and hence $\fpshadj \shp(\tau \gamma') : \DSub{W}{\flat \Theta}$. So we define $\tau' : \Gamma \to \flat \Theta$ by setting $\tau' \gamma' \varsigma = \fpshadj \shp(\tau \gamma')$. To see that this is natural:
		\begin{equation}
			\tau' \gamma' \varsigma \vfi = \tau' \gamma' (\shp \vfi) \varsigma = \fpshadj \shp(\tau \gamma' (\shp \vfi)) = \fpshadj \shp(\tau \gamma') \vfi.
		\end{equation}
		
		\item[Uniqueness.] This follows from the fact that $\kappa : \flat \Theta \to \Theta$ is an injective presheaf map, see \cref{thm:kappa-injective}.
		
		\item[Adjointness.] Substitutions $\quotshp \Gamma \to \Theta$ factor uniquely (and thus naturally) over $\kappa : \flat \Theta \to \Theta$ and are thus in natural correspondence with substitutions $\quotshp \Gamma \to \flat \Theta$. On the other hand, substitutions $\Gamma \to \flat \Theta$ factor uniquely (and thus naturally) over $\inquotshp : \Gamma \to \quotshp \Gamma$ and are thus also in natural correspondence with substitutions $\quotshp \Gamma \to \flat \Theta$. This proves the adjunction. \qedhere
	\end{description}
\end{proof}
Since $\kappa$ is an injective presheaf map and $\inquotshp$ is clearly a surjective presheaf map, we use the following notations. If $\sigma : \Gamma \to \flat \Theta$, then we write $\kappa \sigma \inquotshp\inv$ for $\alpha\inv_{\quotshp \dashv \flat}(\sigma) : \quotshp \Gamma \to \Theta$. Conversely, if $\tau : \quotshp \Gamma \to \Theta$, we write $\kappa\inv \tau \inquotshp$ for $\alpha_{\quotshp \dashv \flat}(\tau) : \Gamma \to \flat \Theta$. Thus, we have unit $\kappa\inv \inquotshp : \Id \to \flat \quotshp$ and co-unit $\kappa \inquotshp \inv : \quotshp \flat \to \Id$.

\subsubsection{The functor $\cohpi$}
As $\quotshp$ is left adjoint to $\flat$, we could define $\shp$ as $\quotshp$. However, this does not give us a decomposiiton $\shp = \cohdisc \cohpi$ or the property $\flat \shp = \shp$. Instead, we define $\cohpi := \cohfget \quotshp \dashv \flat \cohcodisc = \cohdisc$. By consequence, we have $\shp = \cohdisc \cohpi = \flat \quotshp \dashv \flat \sharp = \flat$. Since both $\shp$ and $\quotshp$ are now left adjoint to $\flat$, we have $\kappa \quotshp : \shp \cong \quotshp$.
By \cref{thm:uniqueness-of-nattrans-psh}, we have $\varsigma = (\kappa \quotshp)\inv \inquotshp : \Id \to \shp$ and $\bar \varsigma = \cohfget \varsigma \cohdisc : \Id \to \bar \shp$.
We will maximally avoid to inspect the definition of $\cohpi$; hence we will avoid explicit use of $\quotshp$ and $\inquotshp$.

\subsection{Discretizing types}
If $\shp$ were a morphism of CwFs, then from a type $\Gamma \sez T \type$, we could define a type $\Gamma \sez (\shp T)[\varsigma] \dtype$, but unfortunately this is meaningless. Because we need that operation nonetheless, we will define it explicitly in this section. The approach is the same as for contexts: we simply obtain $\Gamma \sez \quotshp T \dtype$ from $T$ by dividing out the least equivalence relation $\sheq^T$ that makes $T$ discrete. The main obstacle is that we want this operation to commute with substitution, i.e. that $\sheq^T$ commutes with substitution. Here, once more, we will need the existence of coherence squares, as is evident from the following example, where we fail to prove the same result for the category of reflexive graphs $\widehat{\RGcat}$.

\subsubsection{The shape equivalence relation}
Before we proceed $\widehat \bpcubecat$, we will try to define the shape operation in $\widehat \RGcat$.
\begin{example}\label{eg:shape-equivalence-relation}
	Given any type $\Gamma \sez T \type$, the \textbf{shape equivalence relation} $\sheq^T$ is the smallest equivalence relation such that $\sheq^T \dsub{\gamma (\facewkn{\var i})}(p, p \psub{0 / \var i, \facewkn{\var i}})$ for every edge $(W, \ctxedge{\var i}) \Dsez p : T \dsub{\gamma (\facewkn{\var i})}$.
	
	Again, dividing out $\sheq^T$ is precisely what is needed to make $T$ discrete. We now try to prove the following (false) claim: The shape equivalence relation respects substitution: $\sheq^T[\sigma] = \sheq^{T[\sigma]}$.
	\begin{proof}[Non-proof]
		Pick a substitution $\sigma : \Delta \to \Gamma$ and a type $\Gamma \sez T \type$. We try to prove both inclusions.
		\begin{itemize}
			\item[$\supseteq$] It suffices to show that $\sheq^T[\sigma]$ satisfies the defining property of $\sheq^{T[\sigma]}$. Pick an edge $(\ctxedge{\var i}) \Dsez p : T[\sigma] \dsub{\delta (\facewkn{\var i})}$. We have to show that $\sheq^T[\sigma]\dsub{\delta (\facewkn{\var i})} (p, p \psub{0 / \var i, \facewkn{\var i}})$. After composing $\sigma$ and $\delta (\facewkn{\var i})$, this follows immediately from the definition of $\sheq^T$.
			
			\item[$\subseteq$] We will try to prove the equivalent statement that $\sheq^T \subseteq \forall_\sigma \sheq^{T[\sigma]}$. It suffices to show that the right hand side satisfies the defining property of $\sheq^T$. Pick an edge $(\ctxedge{\var i}) \Dsez p : T \dsub{\gamma (\facewkn{\var i})}$. In order to show that $\forall_\sigma \sheq^{T[\sigma]} \dsub{\gamma (\facewkn{\var i})} (p, p \psub{0 / \var i, \facewkn{\var i}})$, we need to show for every $\vfi : \PSub{V}{(\ctxedge{\var i})}$ and every $\delta : \DSub{V}{\Delta}$ such that $\sigma \delta = \gamma (\facewkn{\var i}) \vfi$, that $\sheq^{T[\sigma]}\dsub{\delta}(p \psub \vfi, p \psub{0 / \var i, \facewkn{\var i}} \psub \vfi)$. In the case where $V = (\ctxedge{\var i})$, a problem arises because we do not know that $\delta$ is degenerate. In fact, we can give a counterexample. \qedhere
		\end{itemize}
	\end{proof}
	\begin{proof}[Counterexample]
		Let $\Gamma \cong ()$ and $\Delta \cong \yoneda(\ctxedge{\var i})$: a reflexive graph with two nodes $\delta, \delta' : \DSub{()}{\Delta}$ and a single non-trivial edge $\delta_1$ (\cref{fig:shape-equivalence-relation}). Let $\sigma$ be the unique substitution $\Delta \to \Gamma$. Consider the type $\Gamma \sez T \type$ consisting of two nodes $x$ and $y$ connected by two non-trivial edges $p$ and $q$. We get the setup shown in \cref{fig:shape-equivalence-relation}, which we briefly discuss here.
		
		Both $\delta$ and $\delta'$ are mapped to $\gamma$ under $\sigma$; hence in $T[\sigma]$, they both get a copy of $x$ and $y$. The degenerate edges $\delta \psub{\facewkn{\var i}}$ and $\delta'\psub{\facewkn{\var i}}$, as well as the edge $\delta_1$, are mapped to $\gamma \psub{\facewkn{\var i}}$. Hence in $T[\sigma]$, both $p$ and $q$ are tripled. The degenerate edges, too, are tripled, but you see only one copy of them, as the other two are still degenerate.
		
		All four edges in $T$ live above the degenerate edge $\gamma \psub{\facewkn{\var i}}$. Hence, when dividing out $\sheq^T$, they are all contracted to the degenerate edge at their source. Only a point remains.
		
		In $T[\sigma]$, only the vertical edges and the constant ones, live above degenerate edges in $\Delta$. Hence, only those are contracted. The horizontal and diagonal edges are preserved.
		
		In $(T/\sheq^T)[\sigma]$ (which is easily checked to be equal to $T[\sigma]/\sheq^T[\sigma]$), by contrast, we only have a single horizontal edge. Indeed: we get two copies of $\overline x$, and three copies of its constant edge, two of which are still degenerate.
		
		This is an example where $\sheq^T[\sigma] \neq \sheq^{T[\sigma]}$.
	\end{proof}
	The situation would have been different, had we had coherence squares. Indeed, in that case, we would have constant squares on $p$ and $q$ in $T$, living above the constant square at $\gamma$. These would produce squares filling up the front and back of $T[\sigma]$, living above the constant square at $\delta_1$. We could make $\sheq^X$ contract not just edges above degenerate edges, but also squares living above (partially) degenerate squares. Then both filling squares, as well as their diagonals, would be contracted and we would end up with just a single horizontal edge in $T[\sigma]/\sheq^{T[\sigma]}$.
\end{example}
\begin{figure}[htb]
	\begin{equation*}
		\xymatrix{
				\Gamma =
				& {\gamma}
				&&& \Delta =
				& {\delta} \ar@{-}[rr]^{\delta_1}
				&& {\delta'}
				\\
				& x_\gamma \ar@{-}@/_{1em}/[dd]_p \ar@{-}@/^{1em}/[dd]^q
				&&& & x_\delta \ar@{-}[rr]
					\ar@{-}@/_{1em}/[dd] \ar@{-}@/^{1em}/[dd]
					\ar@{-}@/_{1em}/[ddrr] \ar@{-}@/^{1em}/[ddrr]
				&& x_{\delta'} \ar@{-}@/_{1em}/[dd] \ar@{-}@/^{1em}/[dd]
				\\
				T = & &&& T[\sigma] =
				\\
				& y_\gamma
				&&& & y_\delta \ar@{-}[rr]
				&& y_{\delta'}
				\\
				T/\sheq^T =
				& \overline x_\gamma
				&&& T[\sigma]/\sheq^{T[\sigma]} =
				& \overline{x}_\delta
					\ar@{-}@/^{1.5em}/[rr]
					\ar@{-}@/^/[rr]
					\ar@{-}@/_/[rr]
					\ar@{-}@/_{1.5em}/[rr]
				&& \overline{x}_{\delta'}
				\\
				& &&& (T/\sheq^T)[\sigma] =
				& {\overline x_\delta} \ar@{-}[rr]
				&& \overline x_{\delta'}
		}
	\end{equation*}
	\caption{Setup from the counterexample in \cref{eg:shape-equivalence-relation}. Degenerate edges are not shown, and nodes of types are indexed with the context nodes they live above, in order to distinguish duplicates.}
	\label{fig:shape-equivalence-relation}
\end{figure}
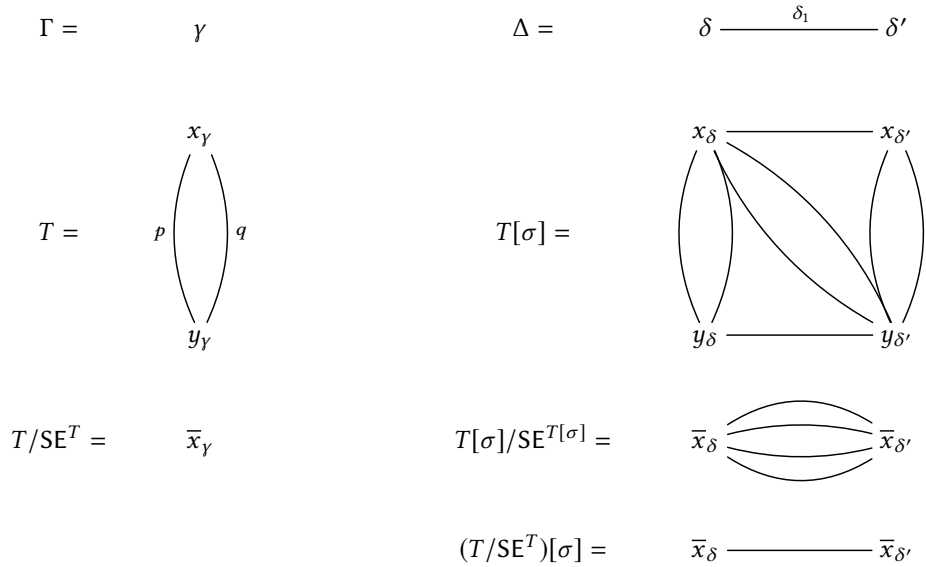
\begin{remark}\label{remark:annotated-restriction}
	Note, in \cref{fig:shape-equivalence-relation}, that the horizontal edges of $T[\sigma]$ arise from reflexive edges in $T$, yet they are themselves not reflexive. Therefore, it is important to distinguish between $x \psub{\facewkn{\var i}}^T$ and $x \psub{\facewkn{\var i}}^{T[\sigma]}$. Every edge that can be written as $x \psub{\facewkn{\var i}}^{T[\sigma]}$ can also be written as $x \psub{\facewkn{\var i}}^T$, but the converse does not hold as exhibited by the horizontal edges.
\end{remark}
\begin{definition}
	Given any type $\Gamma \sez T \type$, the \textbf{shape equivalence relation} $\sheq^T$ is the smallest equivalence relation on $T$ such that for any $(W, \ctxpath{\var i}) \Dsez p : T \dsub{\gamma(\facewkn{\var i})}$, we have $\sheq^T(p, p \psub{0/\var i, \facewkn{\var i}})$.
\end{definition}
\begin{lemma}
	The shape equivalence relation respects substitution: $\sheq^T[\sigma] = \sheq^{T[\sigma]}$.
\end{lemma}
\begin{proof}
	Pick a substitution $\sigma : \Delta \to \Gamma$ and a type $\Gamma \sez T \type$. We prove both inclusions.
	\begin{itemize}
		\item[$\supseteq$] It suffices to show that $\sheq^T[\sigma]$ satisfies the defining property of $\sheq^{T[\sigma]}$. Pick a path $(W, \ctxpath{\var i}) \Dsez p : T[\sigma] \dsub{\delta (\facewkn{\var i})}$. We have to show that $\sheq^T[\sigma]\dsub{\delta (\facewkn{\var i})} (p, p \psub{0 / \var i, \facewkn{\var i}})$. After composing $\sigma$ and $\delta (\facewkn{\var i})$, this follows immediately from the definition of $\sheq^T$.
		
		\item[$\subseteq$] We prove the equivalent statement that $\sheq^T \subseteq \forall_\sigma \sheq^{T[\sigma]}$. It suffices to show that the right hand side satisfies the defining property of $\sheq^T$. Pick a path $(W, \ctxpath{\var i}) \Dsez p : T \dsub{\gamma (\facewkn{\var i})}$. In order to show that $\forall_\sigma \sheq^{T[\sigma]} \dsub{\gamma (\facewkn{\var i})} (p, p \psub{0 / \var i, \facewkn{\var i}})$, we need to show for every $\vfi : \PSub{V}{(W, \ctxpath{\var i})}$ and every $\delta : \DSub{V}{\Delta}$ such that $\sigma \delta = \gamma (\facewkn{\var i}) \vfi$, that $\sheq^{T[\sigma]}\dsub{\delta}(p \psub \vfi, p \psub{0 / \var i, \facewkn{\var i}} \psub \vfi)$. We make a case distinction based on $\var i \psub \vfi$.
		\begin{description}
			\item[$\var i \psub \vfi = 0$] Then $\vfi = (0/\var i)\psi$ for some $\psi : \PSub V W$. Then we have to prove $\sheq^{T[\sigma]}\dsub \delta (p \psub{0/\var i} \psub \psi, p \psub{0/\var i} \psub \psi)$ which holds by reflexivity.
			
			\item[$\var i \psub \vfi = 1$] Then $\vfi = (1/\var i)\psi$ for some $\psi : \PSub V W$. Then we have to prove $\sheq^{T[\sigma]}\dsub \delta (p \psub{1/\var i} \psub \psi, p \psub{0/\var i} \psub \psi)$.
			Without loss of generality, we may assume that $\var i \not\in V$.
			Then we have a path $(V, \ctxpath{\var i}) \Dsez p \psub{\psi, \var i / \var i} : T\dsub{\gamma (\facewkn{\var i})(\psi, \var i/ \var i)}$. Now we have
			\begin{equation}
				\gamma (\facewkn{\var i})(\psi, \var i/ \var i) = \gamma (\facewkn{\var i}) (1/\var i) \psi (\facewkn{\var i}) = \sigma \delta (\facewkn{\var i}).
			\end{equation}
			Hence, we have $(V, \ctxpath{\var i}) \Dsez p \psub{\psi, \var i / \var i} : T[\sigma]\dsub{\delta(\facewkn{\var i})}$. Applying the definition of $\sheq^{T[\sigma]}$, we have $\sheq^{T[\sigma]} \dsub{\delta(\facewkn{\var i})} (p \psub{\psi, \var i / \var i}, p \psub{\psi, 0 / \var i, \facewkn{\var i}})$. Subsequently restricting by $(1/\var i) : \PSub{V}{(V, \ctxpath{\var i})}$ yields the desired result.
			
			\item[$\var i \psub \vfi \in V$] Without loss of generality, we may assume that $(W, \ctxpath{\var i})$ and $V$ are disjoint. Write $\var k = \var i \psub \vfi$ (and note that $\var k$ may be either a bridge or a path variable). Then $\vfi$ factors as $(\psi, \var i / \var i)(\var k / \var i) = (\var k / \var i)(\psi, \var k / \var k)$ for some $\psi : \PSub V W$. We have to prove $\sheq^{T[\sigma]} \dsub \delta(p \psub{\psi, \var i / \var i} \psub{\var k / \var i}, p \psub{\psi, \var 0 / \var i, \facewkn{\var i}} \psub{\var k/ \var i})$. This follows by restricting $\sheq^{T[\sigma]} \dsub{\delta(\facewkn{\var i})} (p \psub{\psi, \var i / \var i}, p \psub{\psi, 0 / \var i, \facewkn{\var i}})$, derived above, by $(\var k / \var i)$. \qedhere
		\end{description}
	\end{itemize}
\end{proof}
\begin{lemma}
	For any type $\Gamma \sez T \type$, we have $\sharp \sheq^T \subseteq \sheq^{\sharp T}$ and $\coshp \sheq^T = \sheq^{\coshp T}$.
\end{lemma}
This is in line with the intuition that $\sharp T$ has more paths than $T$, whereas $\coshp T$ has the same path relation as $T$.
\begin{proof}
	We first prove $\sharp \sheq^T \subseteq \sheq^{\sharp T}$.
	\begin{itemize}
		\item[$\subseteq$] We prove the equivalent statement that $\sheq^T \subseteq \forall_\flat \sheq^{\sharp T}$. It suffices to show that the right hand side satisfies the defining property of $\sheq^T$. Pick a path $(W, \ctxpath{\var i}) \Dsez p : T \dsub{\gamma (\facewkn{\var i})}$. In order to show that $\forall_\flat \sheq^{\sharp T} \dsub{\gamma (\facewkn{\var i})} (p, p \psub{0 / \var i, \facewkn{\var i}})$, we need to show for every $\vfi : \PSub{\flat V}{(W, \ctxpath{\var i})}$ that $\sheq^{\sharp T}\dsub{\fpshadj \flat (\gamma (\facewkn{\var i}) \vfi)}(\fpshadj \flat(p \psub \vfi), \fpshadj \flat(p \psub{0 / \var i, \facewkn{\var i}} \psub \vfi))$. We make a case distinction based on $\var i \psub \vfi$.
		\begin{description}
			\item[$\var i \psub \vfi = 0$] Then $\vfi = (0/\var i)\psi$ for some $\psi : \PSub {\flat V} W$. Then we have to prove $\sheq^{\sharp T}\dsub{\fpshadj \flat (\gamma \psi)}(\fpshadj \flat(p \psub{0/\var i} \psub \psi), \fpshadj \flat(p \psub{0 / \var i} \psub \psi))$ which holds by reflexivity.
			
			\item[$\var i \psub \vfi = 1$] Then $\vfi = (1/\var i)\psi$ for some $\psi : \PSub {\flat V} W$. Then we have to prove $\sheq^{\sharp T}\dsub{\fpshadj \flat (\gamma \psi)}(\fpshadj \flat(p \psub{1/\var i} \psub \psi), \fpshadj \flat(p \psub{0 / \var i} \psub \psi))$.
			Without loss of generality, we may assume that $\var i \not\in V$.
			Then we have a bridge $(\flat V, \ctxbrid{\var i}) \Dsez p \psub{\psi, \var i / \var i} : T\dsub{\gamma \psi}$ and hence a path $(V, \ctxpath{\var i}) \Dsez \fpshadj \flat(p \psub{\psi, \var i / \var i}) : (\sharp T) \dsub{\fpshadj \flat(\gamma \psi)}$. Applying the definition of $\sheq^{\sharp T}$, we have $\sheq^{\sharp T} \dsub{\fpshadj \flat(\gamma \psi)} (\fpshadj \flat(p \psub{\psi, \var i / \var i}), \fpshadj \flat (p \psub{\psi, 0 / \var i, \facewkn{\var i}}))$. Subsequently restricting by $(1/\var i) : \PSub{V}{(V, \ctxpath{\var i})}$ yields the desired result.
			
			\item[$\var i \psub \vfi \in V$] Analogous.
		\end{description}
	\end{itemize}
	We now prove $\coshp \sheq^T = \sheq^{\coshp T}$ by proving both inclusions.
	\begin{itemize}
		\item[$\supseteq$] It suffices to show that $\coshp \sheq^T$ satisfies the defining property of $\sheq^{\coshp T}$. Pick a path $(W, \ctxpath{\var i}) \Dsez \fpshadj \sharp(p) : \coshp T \dsub{\fpshadj \sharp(\gamma) \circ (\facewkn{\var i})}$. We have to show that $\coshp \sheq^T \dsub{\fpshadj \sharp(\gamma) \circ (\facewkn{\var i})} (\fpshadj \sharp(p), \fpshadj \sharp(p) \psub{0/\var i, \facewkn{\var i}})$, i.e. $\sheq^T \dsub{\gamma (\facewkn{\var i})} (p, p \psub{0/\var i, \facewkn{\var i}})$. Note that $p$ has type $(\sharp W, \ctxpath{\var i}) \Dsez p : T \dsub{\gamma (\facewkn{\var i})}$, i.e.\ it is a path. So this follows immediately from the definition of $\sheq^T$. (We could not prove the inclusion $\sharp \sheq^T \supseteq \sheq^{\sharp T}$ because we would have to apply $\flat$ to the primitive context, finding that $p$ is only a bridge.)
		\item[$\subseteq$] The proof for $\sharp$ can be copied almost verbatim. \qedhere
	\end{itemize}
\end{proof}

\subsubsection{The shape of a type}\label{sec:tyshp}
\begin{definition}
	Given a type $\Gamma \sez T \type$, we define the discrete type $\Gamma \sez \quotshp T \dtype$ as $\quotshp T = T / \sheq^T$.
\end{definition}
This definition commutes with substitution, as $(\quotshp T)[\sigma] = (T/\sheq^T)[\sigma] = T[\sigma]/\sheq^T[\sigma] = T[\sigma]/\sheq^{T[\sigma]} = \quotshp(T[\sigma])$.
\begin{proposition}\label{thm:hatinquotshp}
	Given $\Gamma \sez T \type$, we have
	\begin{equation}
		\begin{array}{l l l}
			\Gamma \sez \hatinquotshp : T \to \quotshp T, \\
			\sharp \Gamma \sez \sharp \hatinquotshp : \sharp T \to \sharp \quotshp T, &
			\qquad &
			\sharp \Gamma \sez (\hatinquotshp \sharp \hatinquotshp\inv) : \sharp \quotshp T \to \quotshp \sharp T, \\
			\coshp \Gamma \sez \coshp \hatinquotshp : \coshp T \to \coshp \quotshp T, &
			\qquad &
			\coshp \Gamma \sez \coshp \quotshp T = \quotshp \coshp T \type.
		\end{array}
	\end{equation}
	naturally in $\Gamma$.
	We have commutative diagrams
	\begin{equation}
		\xymatrix{
			\sharp T \ar[d]_{\sharp \hatinquotshp} \ar[r]^{\hatinquotshp}
			& {\quotshp \sharp T} \\
			{\sharp \quotshp T} \ar[ru]_{(\hatinquotshp \sharp \hatinquotshp\inv)}
		} \qquad
		\xymatrix{
			\coshp T \ar[d]_{\coshp \hatinquotshp} \ar[r]^{\hatinquotshp}
			& {\quotshp \coshp T} \\
			{\coshp \quotshp T} \ar@{=}[ru]
		}
	\end{equation}
	If $T$ is discrete, then $\hatinquotshp$, $\sharp \hatinquotshp$ and $\coshp \hatinquotshp$ are also invertible.
\end{proposition}
\begin{proof}
	Recall from \cref{sec:psh-pi-types} that a function $\Gamma \sez f : \Pi A B$ is fully determined if we know $W \Dsez f \dsub \gamma \cdot a : B \dsub{\gamma, a}$ for every $W$, $\gamma : \DSub W \Gamma$ and $W \Dsez a : A \dsub \gamma$.
	
	We set $\hatinquotshp \dsub \gamma \cdot t := \overline t$. This is well-defined because $\hatinquotshp \dsub{\gamma \vfi} \cdot (t \psub \vfi) = \overline{t \psub \vfi} = \overline t \psub \vfi = (\hatinquotshp \dsub \gamma \cdot t) \psub \vfi$.
	
	We define $\sharp \hatinquotshp := \lambda(\ftrtm{\sharp}{(\ap\,\hatinquotshp)})$ and $\coshp \hatinquotshp := \lambda(\ftrtm{\coshp}{(\ap\,\hatinquotshp)})$. Then we have (twice using the fact that labels can be ignored)
	\begin{align*}
		(\sharp \hatinquotshp) \dsub{\fpshadj \flat(\gamma)} \cdot \fpshadj \flat(t)
		&= \ftrtm{\sharp}{(\ap\,\hatinquotshp)} \dsub{\fpshadj \flat(\gamma), \fpshadj \flat(t)}
		= \ftrtm{\sharp}{(\ap\,\hatinquotshp)} \dsub{\fpshadj \flat(\gamma, t)}
		= \fpshadj \flat(\ap\,\hatinquotshp \dsub{\gamma, t}) \\
		&= \fpshadj \flat(\hatinquotshp \dsub \gamma \cdot t)
		= \fpshadj \flat(\overline t) = \overline{\fpshadj \flat(t)},
	\end{align*}
	i.e. $(\sharp \hatinquotshp) \dsub \gamma \cdot t = \overline t$. Similarly, we find $(\coshp \hatinquotshp) \dsub \gamma \cdot t = \overline t$.
	
	Note that $\sharp \quotshp T = \sharp(T/\sheq^T) = \sharp T / \sharp \sheq^T$ and $\quotshp \sharp T = \sharp T / \sheq^{\sharp T}$. Since $\sharp \sheq^T \subseteq \sheq^{\sharp T}$, we can define $(\hatinquotshp \sharp \hatinquotshp\inv) \dsub \gamma \cdot \overline t := \overline t$. Then the first commuting diagram is clear.
	
	Also note that $\coshp \quotshp T = \coshp(T/\sheq^T) = \coshp T/\coshp \sheq^T = \coshp T/\sheq^{\coshp T} = \quotshp \coshp T$. The second commuting diagram is then also clear.
	
	Now suppose that $T$ is discrete. We show that $\tmshp\loch$ is an isomorphism by showing that every equivalence class of $\sheq^T$ is a singleton. This is equivalent to saying that $\sheq^T$ is the equality relation. Clearly, the equality relation is the weakest of all equivalence relations, so it suffices to show that the equality relation satisfies the defining property of $\sheq^T$. But that is precisely the statement that $T$ is discrete.
	
	Furthermore, if $\sheq^T$ is the equality relation, then so are $\sharp \sheq^T$ and $\coshp \sheq^T$. Hence, $\sharp \hatinquotshp$ and $\coshp \hatinquotshp$ will also be invertible.
\end{proof}
\begin{remark}
	Substitutions can be applied to the functions $\sharp \hatinquotshp$ and $\coshp \hatinquotshp$ from \cref{thm:hatinquotshp}, moving them to non-$\sharp$ or non-$\coshp$ contexts. We will omit those substitutions, writing e.g. $\Delta \sez \sharp \hatinquotshp : (\sharp T)[\sigma] \to (\sharp \quotshp T)[\sigma]$.
\end{remark}
\begin{lemma}\label{thm:elim-quotshp}
	For discrete types $T$ living in the appropriate context, we have \emph{invertible} rules
	\begin{equation}
		\binference{
			\Gamma, \var x : \quotshp S \sez t : T
		}{
			\Gamma, \var x : S \sez
			t[\wknvar x, \hatinquotshp(\var x)/\var x]
			: T[\wknvar x, \hatinquotshp(\var x)/\var x]
		}{}.
	\end{equation}
	\begin{equation}
		\binference{
			\Gamma, \ftrtm \sharp {\var x} : (\sharp \quotshp S) [\sigma] \sez t : T
		}{
			\Gamma, \ftrtm \sharp {\var x} : (\sharp S) [\sigma] \sez
			t[\wkn{\ftrtm \sharp {\var x}}, (\sharp \hatinquotshp)(\ftrtm \sharp {\var x})/\ftrtm \sharp {\var x}]
			: T[\wkn{\ftrtm \sharp {\var x}}, (\sharp \hatinquotshp)(\ftrtm \sharp {\var x})/\ftrtm \sharp {\var x}]
		}{}.
	\end{equation}
	\begin{equation}
		\binference{
			\Gamma, \ftrtm \coshp {\var x} : (\coshp \quotshp S) [\sigma] \sez t : T
		}{
			\Gamma, \ftrtm \coshp {\var x} : (\coshp S) [\sigma] \sez
			t[\wkn{\ftrtm \coshp {\var x}}, (\coshp \hatinquotshp)(\ftrtm \coshp {\var x})/\ftrtm \sharp {\var x}]
			: T[\wkn{\ftrtm \coshp {\var x}}, (\coshp \hatinquotshp)(\ftrtm \coshp {\var x})/\ftrtm \coshp {\var x}]
		}{}.
	\end{equation}
	The downward direction is each time a straightforward instance of substitution and hence natural in $\Gamma$. The inverse is then automatically also natural in $\Gamma$.
\end{lemma}
\begin{proof}
	\begin{description}
		\item[Rule 1] To show that the first rule is invertible, we pick a term $\Gamma, \var x : S \sez u : T[\wknvar x, \hatinquotshp(\var x)/\var x]$ and show that it factors over $(\wknvar x, \hatinquotshp(\var x)/ \var x) : (\Gamma, \var x : S) \to (\Gamma, \var x : \quotshp S)$. So pick a path $(W, \ctxpath{\var i}) \Dsez s : S \dsub{\gamma(\facewkn{\var i})}$ that becomes degenerate in $\quotshp S$. We have to show that $u \dsub{\gamma (\facewkn{\var i}), s} = u \dsub{\gamma (\facewkn{\var i}), s \psub{0/\var i, \facewkn{\var i}}}$. Note that both live in $T \dsub{\gamma (\facewkn{\var i}), \hatinquotshp(s)}$.
	
		Now, the defining substitution $(\gamma (\facewkn{\var i}), \hatinquotshp(s))$ is degenerate in $\var i$ because this is obvious for the first component and then degeneracy of the second component follows from discreteness of $\quotshp S$. Hence, $u \dsub{\gamma (\facewkn{\var i}), s}$ is degenerate as a defining term of type $T$, meaning that
		\begin{equation}
			u \dsub{\gamma (\facewkn{\var i}), s} = u \dsub{\gamma (\facewkn{\var i}), s} \psub{0/\var i, \facewkn{\var i}} = u \dsub{\gamma (\facewkn{\var i}), s \psub{0/\var i, \facewkn{\var i}}}.
		\end{equation}
		\item[Rule 2] Note that $(\sharp \quotshp S)[\sigma] = (\sharp(S/\sheq^S))[\sigma] = (\sharp S/\sharp \sheq^S)[\sigma] = (\sharp S)[\sigma]/(\sharp \sheq^S)[\sigma]$.
		To show that the second rule is invertible, we pick a term
		$\Gamma, \ftrtm \sharp {\var x} : (\sharp S) [\sigma] \sez u : T[\wkn{\ftrtm \sharp {\var x}}, (\sharp \hatinquotshp)(\ftrtm \sharp {\var x})/\ftrtm \sharp {\var x}]$ and show that it factors over $(\wkn{\ftrtm \sharp {\var x}}, (\sharp \hatinquotshp)(\ftrtm \sharp {\var x})/\ftrtm \sharp {\var x}) : (\Gamma, \ftrtm \sharp {\var x} : (\sharp S) [\sigma]) \to (\Gamma, \ftrtm \sharp {\var x} : (\sharp \quotshp S) [\sigma])$. To that end, we need to show that whenever $(\sharp \sheq^{S}) [\sigma] \dsub \gamma(r, s)$, we also have $u \dsub{\gamma, r} = u \dsub{\gamma, s}$. Let us write $U \dsub \gamma(r, s)$ for $u \dsub{\gamma, r} = u \dsub{\gamma, s}$. This is easily seen to be an equivalence relation on $(\sharp S)[\sigma]$. So we need to prove $(\sharp \sheq^{S}) [\sigma] \subseteq U$, or equivalently $\sheq^S \subseteq \forall_\flat \forall_\sigma U$.
		
		Let $\Delta$ be the context of $S$, i.e. $\Delta \sez S \type$ and $\sigma : \Gamma \to \sharp \Delta$. It is sufficient to show that $\forall_\flat \forall_\sigma U$ satisfies the defining property of $\sheq^S$. So pick a path $(W, \ctxpath{\var i}) \Dsez p : S \dsub{\delta(\facewkn{\var i})}$. Write $q = p \psub{0/\var i, \facewkn{\var i}}$. We have to show $\forall_\flat \forall_\sigma U \dsub{\delta(\facewkn{\var i})}(p, q)$. So pick $\vfi : \PSub{\flat V}{(W, \ctxpath{\var i})}$; then we have to show $\forall_\sigma U \dsub{\fpshadj \flat(\delta (\facewkn{\var i}) \vfi)} (\fpshadj \flat(p \psub \vfi), \fpshadj \flat(q \psub \vfi))$.
		
		Without loss of generality, we may assume that $V$ and $(W, \ctxpath{\var i})$ are disjoint. Write $k = \var i \psub \vfi \in V \uplus \accol{0, 1}$. Then $\vfi$ factors as $(\psi, \var i^\IB / \var i^\IP)(k / \var i^\IB)$ for some $\psi : \PSub{\flat V}{W}$. Because $\forall_\sigma U$ respects restriction by $(k / \var i^\IP)$ and because $\flat (k / \var i^\IP) = (k / \var i^\IB)$, it is then sufficient to show that
		\begin{equation}
			\forall_\sigma U \dsub{\fpshadj \flat(\delta (\facewkn{\var i}) (\psi, \var i^\IB / \var i^\IP))} (\fpshadj \flat (p \psub{\psi, \var i^\IB / \var i^\IP}), \fpshadj \flat (q \psub{\psi, \var i^\IB / \var i^\IP}))
		\end{equation}
		which simplifies to
		\begin{equation}
			\forall_\sigma U\dsub{\fpshadj \flat(\delta \psi)(\facewkn{\var i^\IP})}(\fpshadj \flat(p \psub{\psi, \var i^\IB / \var i^\IP}), \fpshadj \flat(q \psub{\psi, \var i^\IB / \var i^\IP})).
		\end{equation}
		Write
		\begin{align*}
			\delta' &:= \fpshadj \flat(\delta \psi) : \DSub{V}{\sharp \Delta}, \\
			(V, \ctxpath{\var i}) \Dsez p' &:= \fpshadj \flat(p \psub{\psi, \var i^\IB / \var i^\IP}) : (\sharp S) \dsub{\delta' (\facewkn{\var i})}, \\
			(V, \ctxpath{\var i}) \Dsez q' &:= \fpshadj \flat(q \psub{\psi, \var i^\IB / \var i^\IP}) : (\sharp S) \dsub{\delta' (\facewkn{\var i})},
		\end{align*}
		which satisfies
		\begin{align}
			(V, \ctxpath{\var i}) \Dsez \overline{p'} &= \overline{q'} : (\sharp \quotshp S) \dsub{\delta' (\facewkn{\var i})}, \label{eq:pf-left-quotshp-1}\\
			(V, \ctxpath{\var i}) \Dsez q' &= p' \psub{0 / \var i, \facewkn{\var i}} : (\sharp \quotshp S) \dsub{\delta' (\facewkn{\var i})}. \label{eq:pf-left-quotshp-2}
		\end{align}
		Then we can further simplify to $\forall_\sigma U \dsub{\delta' (\facewkn{\var i})}(p', q')$.
		
		So pick $\chi : \PSub{Y}{(V, \ctxpath{\var i})}$ and $\gamma : \DSub Y \Gamma$ so that $\sigma \gamma = \delta' (\facewkn{\var i}) \chi$. We have to prove $U \dsub{\gamma}(p' \psub \chi, q' \psub \chi)$. Again, without loss of generality, we may assume that $Y$ and $(V, \ctxpath{\var i})$ are disjoint. Then again, $\chi$ factors as $(\omega, \var i / \var i)(j / \var i)$ for some $\omega : \PSub Y V$, where $j = \var i \psub \chi$. We claim that it then suffices to show that $U \dsub{\gamma (\facewkn{\var i})}(p' \psub{\omega, \var i / \var i}, q' \psub{\omega, \var i / \var i})$. First, note that this is well-typed, i.e.
		\begin{equation}
			(Y, \ctxpath{\var i}) \Dsez p' \psub{\omega, \var i / \var i}, q' \psub{\omega, \var i / \var i} : (\sharp S)[\sigma]\dsub{\gamma(\facewkn{\var i})}
		\end{equation}
		because $\delta'(\facewkn{\var i})(\omega, \var i/\var i) = \delta'(\facewkn{\var i})(\omega, \var i/\var i)(j / \var i)(\facewkn{\var i}) = \delta'(\facewkn{\var i})\chi(\facewkn{\var i}) = \gamma(\facewkn{\var i})$. Second, if we further restrict the anticipated result by $(j / \var i)$, then we do obtain $U \dsub \gamma(p' \psub \chi, q' \psub \chi)$.
		
		So it remains to prove that $U \dsub{\gamma (\facewkn{\var i})}(p' \psub{\omega, \var i / \var i}, q' \psub{\omega, \var i / \var i})$, i.e.
		\begin{equation}
			(Y, \ctxpath{\var i}) \Dsez u \dsub{\gamma (\facewkn{\var i}), p' \psub{\omega, \var i / \var i}} = u \dsub{\gamma (\facewkn{\var i}), q' \psub{\omega, \var i / \var i}} : T\dsub{\gamma(\facewkn{\var i}), \overline{p' \psub{\omega, \var i / \var i}}},
		\end{equation}
		which is well-typed by \cref{eq:pf-left-quotshp-1}. The combination of \cref{eq:pf-left-quotshp-1} and \cref{eq:pf-left-quotshp-2} tells us that $\overline{p' \psub{\omega, \var i / \var i}}$ is degenerate in $\var i$. Hence, by discreteness of $T$, we have
		\begin{equation}
			u \dsub{\gamma (\facewkn{\var i}), p' \psub{\omega, \var i / \var i}}
			= u \dsub{\gamma (\facewkn{\var i}), p' \psub{\omega, \var i / \var i}} \psub{0/\var i, \facewkn{\var i}}
			= u \dsub{\gamma (\facewkn{\var i}), q' \psub{\omega, \var i / \var i}}.
		\end{equation}
		
		\item[Rule 3] Since $(\coshp \quotshp S)[\sigma] = (\quotshp \coshp S)[\sigma] = \quotshp((\coshp S)[\sigma])$, and $\coshp \hatinquotshp = \hatinquotshp$, the third rule is a special case of the first rule. \qedhere
	\end{description}
\end{proof}

\section{Universes of discrete types}
In \cref{sec:uniNDD} we give a straightforward definition of a sequence of universes that classify discrete types. Unfortunately, these universes are themselves not discrete, so that they do not contain their lower-level counterparts. In \cref{sec:uniDD-discussion} we discuss the problem and define a hierarchy of discrete universes of discrete types.
As of this point, we will write $\uniPsh_\ell$ for the standard presheaf universe $\uni \ell$.

\subsection{Non-discrete universes of discrete types}\label{sec:uniNDD}
In any presheaf model, we have a hierarchy of universes $\uniPsh_{\ell}$ such that
\begin{equation}
	\inference{\Gamma \ctx}{\Gamma \sez \uniPsh_\ell \type_{\ell+1}}{}, \qquad
	\binference{\Gamma \sez A : \uniPsh_\ell}{\Gamma \sez \El\,A \type_\ell}{}.
\end{equation}
In this section, we will devise a sequence of universes $\uniNDD_\ell$ such that
\begin{equation}
	\inference{\Gamma \ctx}{\Gamma \sez \uniNDD \type_{\ell+1}}{}, \qquad
	\binference{\Gamma \sez A : \uniNDD_\ell}{\Gamma \sez \El\,A \dtype_\ell}{},
\end{equation}
that is: $\uniNDD_\ell$ classifies discrete types of level $\ell$, but it is itself non-discrete. In \cref{sec:uniDD-discussion}, we will devise a universe that is itself discrete, and that in an unusual way classifies all discrete types.
\begin{proposition}
	The CwF $\widehat{\bpcubecat}$ supports a universe for $\DTy_\ell$, the functor that maps a context $\Gamma$ to its set of discrete level $\ell$ types $\Gamma \sez T \dtype_\ell$.
\end{proposition}
\begin{proof}
	Given $\gamma : \DSub W \Gamma$, we define $\uniNDD_\ell \dsub \gamma := \set{\dtycode T}{\yoneda W \sez T \dtype_\ell}$. This makes $\uniNDD_\ell$ a dependent subpresheaf of $\uniPsh_\ell$. We use the same construction for encoding and decoding types (see \cref{thm:unipsh} on page \pageref{thm:unipsh}). The only thing we have to show is that a type $(\Gamma \sez T \type_\ell)$ is discrete if and only if its encoding $(\Gamma \sez \tycode T : \uniPsh_\ell)$ is a term of $\uniNDD_\ell$.
	\begin{itemize}
		\item[$\Rightarrow$] If $\Gamma \sez T \dtype_\ell$, then $\tycode T \dsub \gamma = \dtycode{T [\gamma]}$, and clearly $\yoneda W \sez T [\gamma] \dtype$ is discrete.
		
		\item[$\Leftarrow$] Assume $\Gamma \sez A : \uniNDD_\ell$. We show that $\Gamma \sez \El\,A \type$ is a discrete type, so pick a path $(W, \ctxpath{\var i}) \Dsez p : (\El\,A) \dsub{\gamma (\facewkn{\var i})}$. We need to show that $p = p \psub{0/\var i, \facewkn{\var i}}^{\El\,A}$. We have
		\begin{align*}
			p \psub{0/\var i, \facewkn{\var i}}^{\El\,A}
			&= p \psub{0/\var i, \facewkn{\var i}}^{\dEl(A \dsub{\gamma (\facewkn{\var i})})}
			= p \psub{0/\var i, \facewkn{\var i}}^{\dEl(A \dsub{\gamma}) [\facewkn{\var i}]},
		\end{align*}
		and so we need to prove $(W, \ctxpath{\var i}) \Dsez p = p \psub{0/\var i, \facewkn{\var i}} : \dEl(A \dsub{\gamma}) \dsub{\facewkn{\var i}}$. But $\dEl(A \dsub \gamma)$ is discrete by construction of $\uniNDD_\ell$ and $(\facewkn{\var i}) : \DSub{(W, \ctxpath{\var i})}{\yoneda W}$ is degenerate in $\var i$, so that this equality indeed holds. \qedhere
	\end{itemize}
\end{proof}

\subsection{Discrete universes of discrete types}\label{sec:uniDD-discussion}
Let us have a look at the structure of $\uniNDD$ (ignoring universe levels for a moment):
\begin{itemize}
	\item A \textbf{point} in $\uniNDD$ is a discrete type $\yoneda() \sez T \dtype$. Since $\yoneda()$ is the empty context, this effectively means that points in $\uniNDD$ are discrete closed types, as one would expect. Differently put, for every shape $W$, there is only one cube $\bullet : \DSub{W}{\yoneda()}$ and thus all $W$-shaped cubes $W \Dsez t : T \dsub{\bullet}$ have the same status; essentially $T$ has the structure of a non-dependent presheaf.
	\item A \textbf{path} in $\uniNDD$ is a discrete type $\yoneda(\ctxpath{\var i}) \sez T \dtype$. For every shape $W$, the presheaf $\yoneda(\ctxpath{\var i})$ contains fully degenerate $W$-cubes $(\facewkn{W}, 0/\var i), (\facewkn{W}, 1/\var i) : \DSub{W}{\yoneda(\ctxpath{\var i})}$. As these cubes are fully degenerate, all $W$-cubes of $T$ above them, must also be degenerate in all path dimensions (as $T$ is discrete). So $T$ contains two discrete, closed types $T[0/\var i]$ and $T[1/\var i]$.
	
	Moreover, for every shape $W \not\ni \var i$, we have a cube $(\facewkn W) : \DSub{(W, \ctxpath{\var i})}{\yoneda(\ctxpath{\var i})}$ that is degenerate in all dimensions but $\var i$. We can think of this as the constant cube on the path $\id : \DSub{(\ctxpath{\var i})}{\yoneda(\ctxpath{\var i})}$. Above it live heterogeneous higher paths (degenerate in all path dimensions but $\var i$) that connect a $W$-cube of $A$ with a $W$-cube of $B$. We get a similar setup of heterogeneous higher bridges from $(\facewkn W, \var i^\IB / \var i^\IP) : \DSub{(W, \ctxbrid{\var i})}{\yoneda(\ctxpath{\var i})}$. Finally, the face map $(\var i^\IB/\var i^\IP) : \PSub{(W, \ctxbrid{\var i})}{(W, \ctxpath{\var i})}$ allows us to find under every heterogeneous path, a heterogeneous bridge.
	
	Thus, bluntly put, a path from $A$ to $B$ in $\uniNDD$ consists of:
	\begin{itemize}
		\item A (discrete) notion of heterogeneous paths with source in $A$ and target in $B$,
		\item A (discrete) notion of heterogeneous bridges with source in $A$ and target in $B$,
		\item An operation that gives us a heterogeneous bridge under every heterogeneous path.
	\end{itemize}
	
	\item A \textbf{bridge} in $\uniNDD$ is a discrete type $\yoneda(\ctxbrid{\var i}) \sez T \type$. The presheaf $\yoneda(\ctxbrid{\var i})$ has everything that $\yoneda(\ctxpath{\var i})$ has, except for the interesting path. A similar analysis as above, shows that a bridge from $A$ to $B$ in $\uniNDD$ is quite simply a (discrete) notion of heterogeneous bridges from $A$ to $B$.
\end{itemize}
Now let us think a moment about what we want:
\begin{itemize}
	\item The \textbf{points} seem to be all right: we want them to be discrete closed types.
	\item A \textbf{path} in the universe should always be degenerate, if we want the universe to be a discrete closed type.
	\item In order to understand what a \textbf{bridge} should be, let us have a look at parametric functions. A function $f : \forall(X : \uni{}).\El\,X$ (which we know does not exist, but this choice of type keeps the example simple) is supposed to map related types $X$ and $Y$ to heterogeneously equal values $fX : \El\,X$ and $fY : \El\,Y$. Since bridges were invented as an abstraction of relations, and paths as some sort of pre-equality, we can reformulate this: The function $f$ should map bridges from $X$ to $Y$ to heterogeneous paths from $fX$ to $fY$. Well, then a bridge from $X$ to $Y$ will certainly have to provide a notion of heterogeneous paths between $\El\,X$ and $\El\,Y$!
	
	On the other hand, consider the (non-parametric) type $\Sigma(X : \uni{}).\El\,X$. What is a bridge between $(X, x)$ and $(Y, y)$ in this type? We should expect it to be a bridge from $X$ to $Y$ and a heterogeneous bridge from $x$ to $y$. This shows that bridges in the universe should also provide a notion of bridges.
\end{itemize}
To conclude: we want $\uni{}$ to be a type whose paths are constant, and whose bridges are the paths from $\uniNDD$, i.e. terms $(\ctxbrid{\vec{\var j}}, \ctxpath{\vec{\var i}}) \Dsez A : \uniDD$ should correspond to terms $(\ctxpath{\vec {\var j}}) \Dsez A' : \uniNDD$. So we define it that way:
\begin{definition}
	We define the \textbf{discrete universe of discrete level $\ell$ types} $\sez \uniDD_\ell \dtype_{\ell+1}$ as $\uniDD_\ell = \cohdisc \cohpaths \uniNDD_\ell = \flat \coshp \uniNDD_\ell$.
\end{definition}
Note that $\sharp \shp \dashv \flat \coshp$ and that $\sharp \shp (\ctxbrid{\vec{\var j}}, \ctxpath{\vec{\var i}}) = (\ctxpath{\vec {\var j}})$.

There is a minor issue with the above definition: we want $\uniDD_\ell$ to exist in any context. We can simply define $\Gamma \sez \uniDD_\ell \dtype_{\ell+1}$ as $\uniDD_\ell = (\flat \coshp \uniNDD_\ell)[\bullet]$. Note that $\uniNDD_\ell = \uniNDD_\ell[\bullet]$, so this does not destroy any information.

The universes $\uniPsh_\ell$ and $\uniNDD_\ell$ have a decoding operation $\El$ and an inverse encoding operation $\tycode \loch$ that allow us to turn terms of the universe into types and vice versa. Moreover, the operators for $\uniNDD_\ell$ are simply those of $\uniPsh_\ell$ restricted to $\uniNDD_\ell$ (for $\El$) or to discrete types (for $\tycode \loch$). For $\uniDD$, the situation is different:
\begin{proposition}
	We have mutually inverse rules
	\begin{equation}
		\inference{
			\Gamma \sez A : \uniDD_\ell
		}{\sharp \shp \Gamma \sez \ElDD~A \dtype_\ell}{} \qquad
		\inference{
			\sharp \shp \Gamma \sez T \dtype_\ell
		}{\Gamma \sez \tycodeDD T : \uniDD_\ell}{}
	\end{equation}
	that are natural in $\Gamma$, i.e. $(\ElDD\,A)[\sharp \shp \sigma] = \ElDD(A[\sigma])$.
	Moreover, $(\ElDD\,A)[\iota \varsigma] = \El\,\vartheta(\kappa(A))$.
\end{proposition}
\begin{proof}
	We use that $\uniDD_\ell = \flat \coshp \uniNDD_\ell$.

	We set $\ElDD~A = \El~\alpha_{\sharp \dashv \coshp}\inv(\alpha_{\shp \dashv \flat}\inv(A)) = \El~\vartheta(\kappa(A[\varsigma\inv])[\iota]\inv)$. Then the inverse is given by $\tycodeDD T = \kappa\inv(\vartheta\inv(\tycode T)[\iota])[\varsigma]$.
\end{proof}

\chapter{Semantics of ParamDTT}\label{ch:paramdtt}
In this chapter, we finally interpret the inference rules of ParamDTT in the category with families $\widehat{\bpcubecat}$ of bridge/path cubical sets. We start with some auxiliary lemmas, then give the meta-type of the interpretation function, followed by interpretations for the core typing rules, the typing rules related to internal parametricity, and the typing rules related to $\Nat$ and $\Size$.

\section{Some lemmas}\label{sec:uniDD-lemmas}
\begin{lemma}
	For discrete types $T$ in the relevant contexts, we have invertible rules:
	\begin{equation}
		\binference{\shp \Gamma \sez t : T}{\Gamma \sez t[\varsigma] : T[\varsigma]}{}, \qquad
		\binference{\sharp \shp \Gamma \sez t : T}{\sharp \Gamma \sez t[\sharp \varsigma] : T[\sharp \varsigma]}{}.
	\end{equation}
\end{lemma}
\begin{proof}
	\begin{description}
		\item[Rule 1] Recall that we have $\kappa \quotshp : \shp \cong \quotshp$ and $\varsigma = (\kappa \quotshp)\inv \inquotshp$. Thus, it is sufficient to prove
		\begin{equation}
			\binference{\quotshp \Gamma \sez t' : T'}{\Gamma \sez t'[\inquotshp] : T'[\inquotshp]}{},
		\end{equation}
		after which we can pick $T' = T[(\kappa \quotshp)\inv]$ and $t' = t[(\kappa \quotshp)\inv]$. A proof of this is analogous to but simpler than the proof of the first rule in \cref{thm:elim-quotshp}.
		
		\item[Rule 2] Since $\sharp \flat = \sharp$ and $\sharp \kappa = \id$, we have $\sharp \shp = \sharp \quotshp$ and $\sharp \varsigma = \sharp \inquotshp$. Thus, we need to prove
		\begin{equation}
			\binference{\sharp \quotshp \Gamma \sez t : T}{\sharp \Gamma \sez t[\sharp \inquotshp] : T[\sharp \inquotshp]}{}.
		\end{equation}
		 A proof of this is analogous to but simpler than the proof of the second rule in \cref{thm:elim-quotshp}. \qedhere
	\end{description}
\end{proof}
\begin{lemma}
	For discrete types $\shp\Gamma \sez T \dtype$, we have an invertible substitution
	\begin{equation}
		(\shp \pi, \xi[\varsigma]\inv) : \shp(\Gamma.T[\varsigma]) \cong (\shp \Gamma).T.
	\end{equation}
	We will abbreviate it as $\subext \varsigma\inv$ and the inverse as $\subext \varsigma$. We have $\subext \varsigma \circ \varsigma \subext = \varsigma$.
\end{lemma}
\begin{proof}
	We have a commutative diagram
	\begin{equation}
		\xymatrix{
			\shp(\Gamma.T[\varsigma])
				\ar[rr]^{(\shp \pi, \xi[\varsigma]\inv)}
			&& (\shp \Gamma).T
			\\
			& \Gamma.T[\varsigma]
				\ar[lu]_{\varsigma}
				\ar[ru]^{\varsigma \subext}
				\ar[ld]^{\inquotshp}
				\ar[rd]_{\inquotshp \subext}
			\\
			{\quotshp(\Gamma.T[\varsigma])}
				\ar[rr]^{(\quotshp \pi, \xi[\inquotshp]\inv)}
				\ar[uu]^{(\kappa \quotshp)\inv}_{\wr}
			&& (\quotshp \Gamma).T[(\kappa \quotshp)\inv].
				\ar[uu]_{(\kappa \quotshp)\inv \subext}^{\wr}
		}
	\end{equation}
	The left and right triangles commute because $\varsigma = (\kappa \quotshp)\inv \inquotshp$. The upper triangle commutes because $(\shp \pi, \xi[\varsigma]\inv) \varsigma = (\shp \pi \circ \varsigma, \xi[\varsigma]\inv[\varsigma]) = (\varsigma \pi, \xi) = \varsigma \subext$. The lower one commutes by similar reasoning. The square commutes because
	\begin{align*}
		\kappa \quotshp \subext \circ (\shp \pi, \xi[\varsigma]\inv)
		&= \kappa \quotshp \subext \circ (\flat \quotshp \pi, \xi[\inquotshp]\inv[\kappa \quotshp])
		= (\kappa \quotshp \circ \flat \quotshp \pi, \xi[\inquotshp]\inv[\kappa \quotshp]) \\
		&= (\quotshp \pi \circ \kappa \quotshp, \xi[\inquotshp]\inv[\kappa \quotshp])
		= (\quotshp \pi, \xi[\inquotshp]\inv) \circ \kappa \quotshp.
	\end{align*}
	So in order to prove the theorem, it is sufficient to show that the lower arrow is invertible. It maps $\overline{(\gamma, t)} : \DSub W {\quotshp(\Gamma.T[\varsigma])}$ to
	\begin{equation}
		(\quotshp \pi, \xi[\inquotshp]\inv) \circ \overline{(\gamma, t)}
		= (\quotshp \pi \circ \overline{(\gamma, t)}, \xi[\inquotshp]\inv \dsub{\overline{(\gamma, t)}})
		= (\overline \gamma, \xi \dsub{\gamma, t}) = (\overline \gamma, t) : \DSub{W}{(\quotshp \Gamma).T[(\kappa \quotshp)\inv]}.
	\end{equation}
	So we have to show that we can do the converse. So we have to show that if we have $\gamma, \gamma' : \DSub W \Gamma$ such that $\sheq^\Gamma_W(\gamma, \gamma')$ (i.e.\ $\overline \gamma = \overline{\gamma'}$), and $t : T[\varsigma]\dsub{\gamma} = T[(\kappa \quotshp)\inv] \dsub{\overline \gamma}$, then $\overline{(\gamma, t)} = \overline{(\gamma', t)}$. We will prove a stronger statement, namely that $\sheq^\Gamma \subseteq E$, where we say $E_W(\gamma, \gamma')$ when $\sheq^\Gamma(\gamma, \gamma')$ and for every $\vfi : \PSub{V}{W}$ and every $t : T[(\kappa \quotshp)\inv]\dsub{\overline{\gamma \vfi}}$, we have $\overline{(\gamma\vfi, t)} = \overline{(\gamma'\vfi, t)}$. Because $E$ is an equivalence relation on $\Gamma$, it suffices to prove that $E$ satisfies the defining property of $\Gamma$.
	
	So pick a path $\gamma : \DSub{(W, \ctxpath{\var i})}{\Gamma}$. We have to prove $E(\gamma, \gamma (0/\var i, \facewkn{\var i}))$. Pick some $\vfi : \PSub{V}{(W, \ctxpath{\var i})}$. As usual, we can decompose $\vfi = (\psi, \var i/\var i)(k/\var i)$ for some $\psi : \PSub V W$ and $k \in V \uplus \accol{0, 1}$. Pick $t : T[(\kappa \quotshp)\inv]\dsub{\overline{\gamma \vfi}}$. We have
	\begin{equation}
		\overline{(\gamma(\psi, \var i/\var i), t \psub{\facewkn{\var i}})}
		= \overline{(\gamma(\psi, \var i/\var i), t \psub{\facewkn{\var i}})} (0/\var i, \facewkn{\var i})
		= \overline{(\gamma (0/\var i, \facewkn{\var i}) (\psi, \var i/\var i), t \psub{\facewkn{\var i}})}
	\end{equation}
	where the first equality holds by definition of $\sheq$, and the second one follows from calculating with substitutions. Restricting by $(k/\var i)$ yields the desired result.
\end{proof}

\section{Meta-type of the interpretation function}
\textbf{Contexts} $\Gamma \ctx$ are interpreted to bridge/path cubical sets $\interp \Gamma \ctx$.

\textbf{Types} $\Gamma \judty{T}$ are interpreted to discrete types $\sharp \interp \Gamma \sez \interp T_\Ty \dtype$.

\textbf{Terms} $\Gamma \judty t$ are interpreted as terms $\interp \Gamma \sez \interp t : \interp T [\iota]$.

\textbf{Definitional equality} is interpreted as equality of interpretations.

In the paper, the promotion of an element of the universe to a type, is not reflected syntactically. For that reason, we need a different interpretation function for types and for terms. However, to keep things simpler here, we will add a syntactical reminder $\El$ of the term-to-type promotion, allowing us to omit the index $\Ty$.

\section{Core typing rules}
\subsection{Contexts}
Context formation rules are interpreted as follows:
\begin{equation}
	\interp{
		\inference{}{\ctx}{c-em}	
	} =
	\inference{}{\ctx}{}
\end{equation}
\begin{equation}
	\interp{
		\inference{\Gamma \judty T}{\Gamma, \ctxvar \mu x T \ctx}{c-ext}
	} =
	\inference{
		\inference{
			\inference{
				\sharp \interp \Gamma \sez \interp T \dtype \qquad
				\mu \in \accol{\Id, \sharp, \coshp}
			}{\sharp \interp \Gamma \sez \mu \interp T \type}{}
		}{\interp \Gamma \sez (\mu \interp T)[\iota] \type}{}
	}{\interp \Gamma, \ftrvar \mu x : (\mu \interp T)[\iota] \ctx}{}
\end{equation}
The variable rule
\begin{equation}
	\interp{
	\inference{
		\Gamma \ctx \quad
		(\ctxvar \mu x T) \in \Gamma \quad
		\mu \leq \idmod
	}{\Gamma \judtm x T}{t-var}
	}
\end{equation}
is interpreted through a combination of weakening and the following rules:
\begin{equation}
	\inference{
		\interp \Gamma, \var x : \interp T [\iota] \ctx
	}{\interp \Gamma, \var x : \interp T [\iota] \sez \var x : \interp T [\iota \wknvar x]}{}
	, \qquad
	\inference{
		\interp \Gamma, \var x : (\coshp \interp T) [\iota] \ctx
	}{\interp \Gamma, \var x : (\coshp \interp T) [\iota] \sez \vartheta(\var x) : \interp T [\iota \wknvar x]}{}
\end{equation}
The second one would normally have type $\interp T [\vartheta \iota \wknvar x]$, but we have $\vartheta \sharp = \id$, so we may remove $\vartheta : \sharp \interp \Gamma \to \sharp \interp \Gamma$.
\begin{lemma}\label{thm:leftflat}\label{thm:leftflat2}
	For any syntactic context $\Gamma$, we have $\sharp \interp \Gamma = \sharp \interp{\sharp \setminus \Gamma}$ and equivalently $\flat \interp \Gamma = \flat \interp{\sharp \setminus \Gamma}$. The substitution $\kappa : \flat \interp \Gamma \cong \interp{\sharp \setminus \Gamma}$ is an isomorphism.
\end{lemma}
\begin{proof}
	We prove this by induction on the length of the context.
	\begin{description}
		\item[Empty context] We have $\sharp \setminus () = ()$; hence $\flat \interp{\sharp \setminus ()} = \flat \interp{()} = \flat () = ()$. Then $\kappa : () \to ()$ is the only substitution of that type and it is indeed an isomorphism.
		
		\item[Pointwise extension] We have
		\begin{align*}
			\flat \interp{\Gamma, \ctxptw x T}
			&= \flat (\interp \Gamma, \ftrvar \coshp x : (\coshp \interp T)[\iota])
			= \flat \interp \Gamma, \ftrvar {\flat\coshp} x : \flat \coshp \interp T, \\
			\flat \interp{\sharp \setminus (\Gamma, \ctxptw x T)}
			&= \flat (\interp{\sharp \setminus \Gamma}, \ftrvar \coshp x : (\coshp \interp T)[\iota])
			= \flat \interp{\sharp \setminus \Gamma}, \ftrvar {\flat\coshp} x : \flat \coshp \interp T,
		\end{align*}
		which is equal by virtue of the induction hypothesis. The context
		\begin{equation}
			\interp{\sharp \setminus (\Gamma, \ctxptw x T)}
			 = \interp{\sharp \setminus \Gamma}, \ftrvar \coshp x : (\coshp \interp T)[\iota]
		\end{equation}
		is discrete because $\interp{\sharp \setminus \Gamma} \cong \flat \interp \Gamma$ by the induction hypothesis, $\interp T$ is discrete and $\coshp$ preserves discreteness (\cref{thm:coshp-preserves-discreteness}). Hence, $\kappa$ is an isomorphism for this context.
		
		\item[Continuous extension] We have
		\begin{align*}
			\flat \interp{\Gamma, \ctxctu x T}
			&= \flat (\interp \Gamma, \var x : \interp T[\iota])
			= \flat \interp \Gamma, \ftrvar \flat x : \flat \interp T, \\
			\flat \interp{\sharp \setminus (\Gamma, \ctxctu x T)}
			&= \flat (\interp{\sharp \setminus \Gamma}, \var x : \interp T[\iota])
			= \flat \interp{\sharp \setminus \Gamma}, \ftrvar \flat x : \flat \interp T,
		\end{align*}
		which is equal by virtue of the induction hypothesis. The context
		\begin{equation}
			\interp{\sharp \setminus (\Gamma, \ctxctu x T)}
			 = \interp{\sharp \setminus \Gamma}, \var x : \interp T[\iota]
		\end{equation}
		is discrete because $\interp{\sharp \setminus \Gamma} \cong \flat \interp \Gamma$ by the induction hypothesis and $\interp T$ is discrete. Hence, $\kappa$ is an isomorphism for this context.
		
		\item[Parametric extension] We have
		\begin{align*}
			\flat \interp{\Gamma, \ctxpar x T}
			&= \flat (\interp \Gamma, \ftrvar \sharp x : (\sharp \interp T)[\iota])
			= \flat \interp \Gamma, \ftrvar{\flat\sharp} x : \flat \sharp \interp T,
			= \flat \interp \Gamma, \ftrvar \flat x : \flat \interp T, \\
			\flat \interp{\sharp \setminus (\Gamma, \ctxpar x T)}
			&= \flat \interp{(\sharp \setminus \Gamma), \ctxctu x T)}
			= \flat (\interp{\sharp \setminus \Gamma}, \var x : \interp T[\iota])
			= \flat \interp{\sharp \setminus \Gamma}, \ftrvar \flat x : \flat \interp T,
		\end{align*}
		which is equal by virtue of the induction hypothesis. The context
		\begin{equation}
			\interp{\sharp \setminus (\Gamma, \ctxpar x T)}
			= \interp{(\sharp \setminus \Gamma), \ctxctu x T)}
			 = \interp{\sharp \setminus \Gamma}, \var x : \interp T[\iota]
		\end{equation}
		is discrete because $\interp{\sharp \setminus \Gamma} \cong \flat \interp \Gamma$ by the induction hypothesis and $\interp T$ is discrete. Hence, $\kappa$ is an isomorphism for this context.
		
		\item[Interval extensions] These will be special cases of the above.
		\item[Face predicate extension] See the addendum in \cref{sec:face-predicates}. \qedhere
	\end{description}
\end{proof}
\begin{lemma}\label{thm:leftsharp}\label{thm:leftsharp2}
	For any syntactic context $\Gamma$, we have $\interp{\coshp \setminus \Gamma} = \sharp \interp \Gamma$.
\end{lemma}
\begin{proof}
	We prove this by induction on the length of the context.
	\begin{description}
		\item[Empty context] We have $\coshp \setminus () = ()$; hence $\interp{\coshp \setminus ()} = \interp{()} = () = \sharp ()$.
		
		\item[Pointwise extension] We have
		\begin{align*}
			\sharp \interp{\Gamma, \ctxptw x T}
			&= \sharp (\interp \Gamma, \ftrvar \coshp x : (\coshp \interp T)[\iota])
			= \sharp \interp \Gamma, \ftrvar{\sharp \coshp}{x} : \sharp \coshp \interp T
			= \sharp \interp \Gamma, \ftrvar{\coshp}{x} : \coshp \interp T, \\
			\interp{\coshp \setminus (\Gamma, \ctxptw x T)}
			&= \interp{(\coshp \setminus \Gamma), \ctxptw x T}
			= \interp{\coshp \setminus \Gamma}, \ftrvar \coshp x : (\coshp \interp T)[\iota]
			= \sharp \interp{\Gamma}, \ftrvar \coshp x : \coshp \interp T,
		\end{align*}
		where in the last step we used the induction hypothesis and the fact that $\iota \sharp = \id : \sharp \interp \Gamma \to \sharp \interp \Gamma$.
		
		\item[Continuous extension] We have
		\begin{align*}
			\sharp \interp{\Gamma, \ctxctu x T}
			&= \sharp (\interp \Gamma, \var x : \interp T[\iota])
			= \sharp \interp \Gamma, \ftrvar \sharp x : \sharp \interp T, \\
			\interp{\coshp \setminus (\Gamma, \ctxctu x T)}
			= \interp{(\coshp \setminus \Gamma), \ctxpar x T}
			&= \interp{\coshp \setminus \Gamma}, \ftrvar \sharp x : (\sharp \interp T)[\iota]
			= \sharp \interp \Gamma, \ftrvar \sharp x : \sharp \interp T.
		\end{align*}
		
		\item[Parametric extension] We have
		\begin{align*}
			\sharp \interp{\Gamma, \ctxpar x T}
			&= \sharp(\interp \Gamma, \ftrvar \sharp x : (\sharp \interp T)[\iota])
			= \sharp \interp \Gamma, \ftrvar \sharp x : \sharp \interp T, \\
			\interp{\coshp \setminus (\Gamma, \ctxpar x T)}
			= \interp{(\coshp \setminus \Gamma), \ctxpar x T}
			&= \interp{\coshp \setminus \Gamma}, \ftrvar \sharp x : (\sharp \interp T)[\iota]
			= \sharp \interp \Gamma, \ftrvar \sharp x : \sharp \interp T.
		\end{align*}
		
		\item[Interval extensions] These will be special cases of the above.
		
		\item[Face predicate extension] See the addendum in \cref{sec:face-predicates}. \qedhere
	\end{description}
\end{proof}

\subsection{Universes}
We have
\begin{equation}
	\interp{
		\inference{
			\Gamma \ctx \qquad
			\ell \in \IN
		}{\Gamma \judtm{\uni \ell}{\El\,\uni{\ell+1}}}{t-Uni}
	} =
	\inference{
		\interp \Gamma \ctx \qquad
		\ell \in \IN
	}{\interp \Gamma \sez \tycodeDD{\uniDD_\ell} : \uniDD_{\ell + 1}}{}
\end{equation}
\begin{equation}
	\interp{
		\inference{
			\Gamma \judtm T {\El\,\uni k} \qquad
			k \leq \ell \in \IN
		}{\Gamma \judtm{T}{\El\,\uni \ell}}{t-lift}
	} =
	\inference{
		\interp \Gamma \sez \interp T : \uniDD_k \qquad
		k \leq \ell \in \IN
	}{\interp \Gamma \sez \interp T : \uniDD_\ell}{}
\end{equation}
\begin{equation}
	\interp{
		\inference{\leftflat\Gamma \judtm A {\uni{\ell}}}{\Gamma \judty{\El\,A}}{ty}
	} =
	\inference{
		\inference{
			\interp{\sharp \setminus \Gamma} \sez \interp A : \uniDD_\ell
		}{\sharp \shp \interp{\sharp \setminus \Gamma} \sez \ElDD \interp A \dtype}{}
	}{\sharp \interp{\sharp \setminus \Gamma} = \sharp \interp \Gamma \sez (\ElDD \interp A)[\sharp \varsigma] \dtype}{}
\end{equation}
In particular, we have
\begin{equation}
	\interp{\El\,\uni \ell} = (\ElDD \interp{\uni \ell})[\sharp \varsigma]
	= (\ElDD \tycodeDD{\uniDD_\ell})[\sharp \varsigma]
	= \uniDD_\ell[\sharp \varsigma] = \uniDD_\ell,
\end{equation}
so that it is justified that we simply put $\uniDD_k$ on several occasions where we should have used $\interp{\El\,\uni k}$.
\begin{remark}
	In the paper, we defined $\interp{\El\,A}$ as $\El~\vartheta(\ftrtm{\sharp}{\interp A})$. Note that we have
	\begin{align*}
		(\ElDD \interp A)[\sharp \varsigma]
		&= \El\,\vartheta(\kappa(\interp A[\varsigma]\inv)[\iota]\inv)[\sharp \varsigma]
		= \El\,\vartheta(\kappa(\interp A[\varsigma]\inv)[\varsigma][\iota]\inv)
		= \El\,\vartheta(\kappa(\interp A)[\iota]\inv).
	\end{align*}
	Moreover, $\sharp \interp \Gamma \sez \kappa(\interp A)[\iota]\inv = \ftrtm{\sharp}{\interp A} : \coshp \uniNDD_\ell$ because $(\ftrtm{\sharp}{\interp A})[\iota] = \iota(\interp A) = \kappa(\interp A)$ where the last step uses $\iota = \kappa : \flat \coshp \to \coshp$.
\end{remark}

\subsection{Substitution}
We have syntactic substitution rule
\begin{equation}
	\inference{
		\Gamma, \ctxvar \mu x T, \Delta \sez J \qquad
		\mu \setminus \Gamma \judtm t T
	}{\Gamma, \Delta[\varclr t/\varclr x] \sez J[\varclr t/\varclr x]}{subst}
\end{equation}
which can be shown to be admissible by induction on the derivation of $J$. The idea behind this is a combination of the general idea of substitution, and the fact that we use $\mu \setminus \Gamma \judtm t T$ to express something that would more intuitively look like $\Gamma \sez \ctxvar \mu t T$. In fact, we have the following result:
\begin{lemma}
	For $\mu \in \accol{\coshp, \idmod, \sharp}$, we have
	\begin{equation}
		\inference{
			\interp{\mu \setminus \Gamma} \sez t : T [\iota]
		}{\interp \Gamma \sez \forsub \mu t : (\mu T)[\iota]}{}
	\end{equation}
	where $\forsub \idmod t = t$, $\forsub \sharp t = (\ftrtm \sharp t)[\iota]$ and $\forsub \coshp t = (\ftrtm \coshp t)[\iota]$.
\end{lemma}
\begin{proof}
	The idea is that $\mu \setminus \loch$ is left adjoint to $\mu \circ \loch$. In the model, we see this formally as $\interp{\coshp \setminus \Gamma} = \sharp \interp \Gamma$, $\interp{\idmod \setminus \Gamma} = \interp \Gamma$ and $\interp{\sharp \setminus \Gamma} \cong \flat \interp \Gamma$. So in each case, we can simply use the adjunction in the model.
	\begin{description}
		\item[$\mu = \idmod$] Then $\interp \Gamma \sez t : T[\iota]$ by the premise.
		
		\item[$\mu = \sharp$] Then we have
		\begin{equation}
			\inference{
			\inference{
				\interp{\sharp \setminus \Gamma} \sez t : T [\iota]
			}{\flat \interp \Gamma \sez t [\kappa] : T [\iota \kappa]}{}
			}{\interp \Gamma \sez \iota(t [\kappa])[\kappa]\inv : (\sharp T)[\iota]}{},
		\end{equation}
		i.e.\ we first apply the isomorphism $\kappa : \flat \interp \Gamma \cong \interp{\sharp \setminus \Gamma}$ and then the adjunction $\iota(\loch)[\kappa]\inv : \flat \dashv \sharp$. The second step is well-typed because $\iota \circ (\iota \kappa) / \kappa = \iota : \interp \Gamma \to \sharp \interp \Gamma = \sharp \interp{\sharp \setminus \Gamma}$, or more meaningfully
		\begin{equation}
			\iota \sharp \interp \Gamma \circ (\iota \interp{\sharp \setminus \Gamma} \circ \kappa \interp{\sharp \setminus \Gamma}) / \kappa \interp{\Gamma} = \iota \interp \Gamma.
		\end{equation}
		Indeed, $\iota \sharp = \id$ and $(\iota \kappa) \interp{\sharp \setminus \Gamma} = \iota \flat \interp{\sharp \setminus \Gamma} = \iota \flat \interp \Gamma = \iota \interp \Gamma \circ \kappa \interp \Gamma$.
		
		However, the resulting term is a bit obscure. We can instead do
		\begin{equation}
			\inference{
			\inference{
				\interp{\sharp \setminus \Gamma} \sez t : T[\iota]
			}{\sharp \interp{\Gamma} \sez \ftrtm{\sharp}{t} : \sharp T}{}
			}{\interp \Gamma \sez (\ftrtm{\sharp}{t})[\iota] : (\sharp T)[\iota]}{}.
		\end{equation}
		In the first step, we used that $\sharp \interp{\sharp \setminus \Gamma} = \sharp \interp \Gamma$ and $\sharp \iota = \id$. As it happens, $\iota(t [\kappa])[\kappa]\inv = (\ftrtm{\sharp}{t})[\iota]$, or more precisely
		\begin{equation}
			\iota(t [\kappa \interp{\sharp \setminus \Gamma}])[\kappa \interp \Gamma]\inv = (\ftrtm{\sharp}{t})[\iota \interp \Gamma]
		\end{equation}
		because
		\begin{equation}
			(\ftrtm{\sharp}{t})[\iota \interp \Gamma][\kappa \interp \Gamma]
			= (\ftrtm{\sharp}{t})[\iota \interp{\sharp \setminus \Gamma}][\kappa \interp{\sharp \setminus \Gamma}]
			= \iota(t)[\kappa \interp {\sharp \setminus \Gamma}]
			= \iota(t[\kappa \interp {\sharp \setminus \Gamma}]).
		\end{equation}
		So we conclude $\forsub \sharp t = (\ftrtm{\sharp}{t})[\iota]$.
		
		\item[$\mu = \coshp$] We can apply the adjunction $\vartheta\inv(\loch)[\iota] : \sharp \dashv \coshp$:
		\begin{equation}
			\inference{
				\sharp \interp \Gamma \sez t : T
			}{\interp \Gamma \sez \vartheta\inv(t)[\iota] : (\coshp T)[\iota]}{}
		\end{equation}
		In the premise, we can omit $[\iota]$ on $T$ because $\iota \sharp \interp \Gamma = \id$. The conclusion should normally have type $(\coshp T)[\vartheta \setminus \iota]$. However, $\vartheta \sharp \interp \Gamma = \id$, so we are left with just $\iota$. Note that $\vartheta\inv(t) = \ftrtm{\coshp}{t}$ because $\vartheta (\ftrtm{\coshp}{t}) = t[\vartheta]$ and $\vartheta \sharp \interp \Gamma = \id$.
		So we conclude $\forsub \coshp t = (\ftrtm{\coshp}{t})[\iota]$. \qedhere
	\end{description}
\end{proof}
We assume the following without proof:\footnote{One could argue that the presence of a conjecture in this technical report, implies that we have not proven soundness of the type system. However, in practice, any mathematical proof will wipe some tedious details under the carpet, when the added value of figuring them out is outweighed by the work required to do so. Note also that, would this conjecture be false, most of the model remains intact.}
\begin{conjecture}
	The interpretation of the substitution rule, corresponding to the syntactic admissibility proof, is given by the substitution $(\id, \forsub \mu t) : \interp \Gamma \to \interp{\Gamma, \ctxvar \mu x T}$.
\end{conjecture}
\begin{lemma}
	Let $\Gamma' = (\Gamma, \ctxvar \mu x {\El~A})$ be a syntactic context. Then we have $\mu \setminus \Gamma' \judtm{x}{\El~A}$. The interpretation of $x$ satisfies $\forsub \mu \interp x = \ftrvar \mu x$.
\end{lemma}
\begin{proof}
	For $\mu = \idmod$, this is trivial.
	
	For $\mu = \sharp$, we have
	\begin{equation}
		\forsub \sharp \interp x = \forsub \sharp \var x = (\ftrvar \sharp x)[\iota].
	\end{equation}
	Here, $\iota$ has type $\interp{\Gamma'} \to \sharp \interp{\Gamma'}$, but $\ftrvar \sharp x$ is already in a sharp type in $\interp{\Gamma'}$ and $\iota \sharp = \id$, so we can omit it and have $\forsub \sharp \interp x = \ftrvar \sharp x$.
	
	For $\mu = \coshp$, we have
	\begin{equation}
		\forsub \coshp \interp x = \forsub \coshp \vartheta(\ftrvar \coshp x) = \ftrtm{\coshp}{(\vartheta(\ftrvar \coshp x))}[\iota] = \ftrvar \coshp x [\iota] = \ftrvar \coshp x,
	\end{equation}
	because $\iota \coshp = \id$.
\end{proof}

\subsection{Definitional equality} As definitional equality is interpreted as equality, it is evidently consistent to assume that this is an equivalence relation and a congruence. The conversion rule is also obvious.

\subsection{Quantification}
\subsubsection{Continuous quantification}
In general, write $\subext \iota = (\wknvar x, \iota(\var x)/\ftrvar \sharp x) : (\Gamma, \var x : T) \to (\Gamma, \ftrvar \sharp x : \sharp T[\iota])$. If $\Gamma = \sharp \Delta$, then this becomes $\subext \iota : (\sharp \Delta, \var x : T) \to (\sharp \Delta, \ftrvar \sharp x : \sharp T) = \sharp (\Delta, \var x : T[\iota])$. For the continuous quantifiers, we have
\begin{align*}
	&\interp{
		\inference{
			\Gamma \judtm{A}{\El\,\uni \ell} \qquad
			\Gamma, \ctxctu x {\El\,A} \judtm{B}{\El\,\uni \ell}
		}{\Gamma \judtm{\prodctu x A.B}{\El\,\uni \ell}}{t-$\Pi$}\quad
	} = \nn\\ &
	\inference{
	\inference{
		\inference{
			\interp \Gamma \sez \interp A : \uniDD_\ell
		}{\sharp \shp \interp \Gamma \sez \ElDD \interp A \dtype_\ell}{}
		\qquad
		\inference{
		\inference{
		\inference{
			\interp \Gamma, \var x : (\ElDD \interp A)[\sharp \varsigma][\iota] \sez \interp B : \uniDD_\ell
		}{\sharp \shp \paren{\interp \Gamma, \var x : (\ElDD \interp A)[\sharp \varsigma][\iota]} \sez \ElDD \interp B \dtype_\ell}{}
		}{\sharp \paren{\shp \interp \Gamma, \var x : (\ElDD \interp A)[\iota]} \sez \ElDD \interp B[\sharp(\subext \varsigma)] \dtype_\ell}{}
		}{\sharp \shp \interp \Gamma, \var x : \ElDD \interp A \sez \ElDD \interp B[\sharp(\subext \varsigma)][\subext \iota] \dtype_\ell}{}
	}{\sharp \shp \Gamma \sez \Pi(\var x : \ElDD \interp A).\ElDD \interp B[\sharp(\subext \varsigma)][\subext \iota] \dtype_\ell}{}
	}{\Gamma \sez \tycodeDD{\Pi(\var x : \ElDD \interp A).(\ElDD \interp B[\sharp(\subext \varsigma)][\subext \iota])} : \uniDD_\ell}{}
\end{align*}
and similar for $\Sigma$. Note that
\begin{align}
	\interp{\El~\prodctu x A.B}
	&= \ElDD \interp{\prodctu x A.B} [\sharp \varsigma]
	= \Pi(\var x : \ElDD \interp A).(\ElDD \interp B[\sharp(\subext \varsigma)][\subext \iota])~[\sharp \varsigma] \\
	&= \Pi(\var x : \ElDD \interp A [\sharp \varsigma]).(\ElDD \interp B[\sharp(\subext \varsigma)][\subext \iota][\sharp \varsigma \subext]) \nn\\
	&= \Pi(\var x : \ElDD \interp A [\sharp \varsigma]).(\ElDD \interp B[\sharp(\subext \varsigma \circ \varsigma \subext)][\subext \iota]) \nn\\
	&= \Pi(\var x : \ElDD \interp A [\sharp \varsigma]).(\ElDD \interp B[\sharp \varsigma][\subext \iota])
	= \Pi(\var x : \interp{\El\,A}).(\interp{\El\,B} [\subext \iota]). \nn
\end{align}
If we further apply $[\iota]$, we find
\begin{equation}
	\interp{\El~\prodctu x A.B}[\iota] = \Pi(\var x : \interp{\El~A}[\iota]).(\interp{\El~B}[\subext \iota][\iota \subext]) = \Pi(\var x : \interp{\El~A}[\iota]).(\interp{\El~B}[\iota]).
\end{equation}
\subsubsection{Parametric quantification}
For the existential type $\Sigmapar$, we have
\begin{align*}
	& \interp{
		\inference{
			\Gamma \judtm{A}{\El\,\uni \ell} \qquad
			\Gamma, \ctxctu x {\El~A} \judtm{B}{\El\,\uni \ell}
		}{\Gamma \judtm{\sumpar x A.B}{\El\,\uni \ell}}{t-$\Sigma$}
	} = \nn\\ &
	\inference{
	\inference{
	\inference{
		\inference{
		\inference{
			\interp \Gamma \sez \interp A : \uniDD_\ell
		}{\sharp \shp \interp \Gamma \sez \ElDD \interp A \dtype_\ell}{}
		}{\sharp \shp \interp \Gamma \sez \sharp \ElDD \interp A \type_\ell}{}
		\qquad
		\inference{
		\inference{
			\interp \Gamma, \var x : (\ElDD \interp A)[\sharp \varsigma][\iota] \sez \interp B : \uniDD_\ell
		}{\sharp \shp \paren{\interp \Gamma, \var x : (\ElDD \interp A)[\sharp \varsigma][\iota]} \sez \ElDD \interp B \dtype_\ell}{}
		}{\sharp \shp \interp \Gamma, \ftrvar \sharp x : \sharp \ElDD \interp A \sez \ElDD \interp B [\sharp(\subext\varsigma)] \dtype_\ell}{}
	}{\sharp \shp \Gamma \sez \Sigma(\ftrvar \sharp x : \sharp \ElDD \interp A).(\ElDD \interp B [\sharp (\subext \varsigma)]) \type_\ell}{}
	}{\sharp \shp \Gamma \sez \quotshp \Sigma(\ftrvar \sharp x : \sharp \ElDD \interp A).(\ElDD \interp B [\sharp (\subext \varsigma)]) \dtype_\ell}{}
	}{\Gamma \sez \tycodeDD{\quotshp \Sigma(\ftrvar \sharp x : \sharp \ElDD \interp A).(\ElDD \interp B [\sharp (\subext \varsigma)])} : \uniDD_\ell}{}
\end{align*}
Then we have
\begin{align*}
	\interp{\El~\sumpar x A.B}
	&= \ElDD \interp{\sumpar x A.B} [\sharp \varsigma]
	= \paren{\quotshp \Sigma(\ftrvar \sharp x : \sharp \ElDD \interp A).(\ElDD \interp B [\sharp (\subext \varsigma)])} [\sharp \varsigma] \\
	&= \quotshp \Sigma(\ftrvar \sharp x : (\sharp \ElDD \interp A)[\sharp \varsigma]).(\ElDD \interp B [\sharp (\subext \varsigma)] [\sharp (\varsigma \subext)]) \\
	&= \quotshp \Sigma(\ftrvar \sharp x : \sharp(\ElDD \interp A[\sharp \varsigma])).(\ElDD \interp B [\sharp \varsigma])
	= \quotshp \Sigma(\ftrvar \sharp x : \sharp \interp{\El\,A}).\interp{\El\,B}.
\end{align*}
The universal type $\Pipar$ is defined similarly, but without throwing in $\quotshp$:
\begin{equation}
	\interp{\prodpar x A.B} = \tycodeDD{\Pi(\ftrvar \sharp x : \sharp \ElDD \interp A).(\ElDD \interp B [\sharp(\subext \varsigma)])}.
\end{equation}
Then we have
\begin{equation}
	\interp{\El~\prodpar x A.B} = \Pi(\ftrvar \sharp x : \sharp \interp{\El\,A}).\interp{\El\,B}.
\end{equation}
If we further apply $[\iota]$, we find
\begin{equation}
	\interp{\El~\prodpar x A.B} [\iota] = \Pi(\ftrvar \sharp x : (\sharp \interp{\El\,A})[\iota]).(\interp{\El\,B}[\iota \subext]) = \Pi(\ftrvar \sharp x : (\sharp \interp{\El\,A})[\iota]).(\interp{\El\,B}[\iota]).
\end{equation}
In the second step, we use that $\iota \subext = (\iota \wknvar x, \var x/\var x)$ is equal to $\iota = (\iota \wknvar x, \iota(\var x)/\var x)$ because $\var x$ has a sharp type and $\iota\sharp = \id$.

\subsubsection{Pointwise quantification}
We generalize the $\subext \nu = (\wknvar x, \nu(\var x)/\var x)$ notation to any natural transformation $\nu$ between morphisms of CwFs. We have
\begin{align*}
	&\interp{
		\inference{
			\Gamma \judtm{A}{\El\,\uni \ell} \qquad
			\Gamma, \ctxptw x {\El\,A} \judtm{B}{\El\,\uni \ell}
		}{\Gamma \judtm{\prodptw x A.B}{\El\,\uni \ell}}{t-$\Pi$}
	} = \nn\\ &
	\inference{
	\inference{
		\inference{
		\inference{
			\interp \Gamma \sez \interp A : \uniDD_\ell
		}{\sharp \shp \interp \Gamma \sez \ElDD \interp A \dtype_\ell}{}
		}{\sharp \shp \interp \Gamma \sez \coshp \ElDD \interp A \dtype_\ell}{}
		\qquad
		\inference{
		\inference{
			\interp \Gamma, \ftrvar \coshp x : (\coshp(\ElDD \interp A [\sharp \varsigma]))[\iota] \sez \interp B : \uniDD_\ell
		}{\sharp \shp \paren{\interp \Gamma, \ftrvar \coshp x : (\coshp(\ElDD \interp A [\sharp \varsigma]))[\iota]} \sez \ElDD \interp B \dtype_\ell}{}
		}{\sharp \shp \interp \Gamma, \ftrvar \coshp x : \coshp \ElDD \interp A \sez \ElDD \interp B [\sharp(\subext \varsigma)] \dtype_\ell}{}
	}{\sharp \shp \interp \Gamma \sez \Pi(\ftrvar \coshp x : \coshp \ElDD \interp A).(\ElDD \interp B [\sharp(\subext \varsigma)]) \dtype_\ell}{}
	}{\interp \Gamma \sez \tycodeDD{\Pi(\ftrvar \coshp x : \coshp \ElDD \interp A).(\ElDD \interp B [\sharp(\subext \varsigma)])} : \uniDD_\ell}{}
\end{align*}
and similar for $\Sigma$. The step where we use $\sharp(\subext \varsigma)$ is a bit obscure. We have
\begin{equation}
	\subext \varsigma :
	\paren{\shp \interp \Gamma, \ftrvar \coshp x : (\coshp(\ElDD \interp A))[\iota]}
	\to
	\shp \paren{\interp \Gamma, \ftrvar \coshp x : (\coshp(\ElDD \interp A [\sharp \varsigma]))[\iota]}
\end{equation}
because
\begin{equation}
	(\coshp(\ElDD \interp A))[\iota][\varsigma]
	= (\coshp(\ElDD \interp A))[\sharp \varsigma][\iota]
	= (\coshp(\ElDD \interp A [\sharp \varsigma]))[\iota].
\end{equation}
Now if we apply $\sharp$ on the domain of $\subext \varsigma$, then $\sharp \iota = \id$ disappears, and $\coshp$ absorbs $\sharp$ on its left.

We will have
\begin{equation}
	\interp{\El~\prodptw x A.B} = \Pi(\ftrvar \coshp x : \coshp \interp{\El~A}).\interp{\El~B}
\end{equation}
and similar for $\Sigma$. If we further apply $[\iota]$, we find
\begin{equation}
	\interp{\El~\prodptw x A.B}[\iota] = \Pi(\ftrvar \coshp x : (\coshp \interp{\El~A})[\iota]).(\interp{\El~B}[\iota \subext]) = \Pi(\ftrvar \coshp x : (\coshp \interp{\El~A})[\iota]).(\interp{\El~B}[\iota]),
\end{equation}
where $\iota \subext = \iota$ because $\var x$ has a type in the image of $\coshp$ and $\iota \coshp = \id$.

So in general, we see that
\begin{equation}
	\interp{\El~\prodvar \mu x A.B}[\iota] = \Pi(\ftrvar \mu x : (\mu \interp{\El~A})[\iota]).(\interp{\El~B}[\iota]).
\end{equation}

\subsection{Functions}
Abstraction is interpreted as
\begin{equation}
	\interp{
		\inference{
			\Gamma, \ctxvar \mu x {\El~A} \judtm b {\El~B}
		}{\Gamma \judtm{\lamannotvar \mu x A.b}{\El~\prodvar \mu x A.B}}{t-$\lambda$}
	} =
	\inference{
		\interp \Gamma, \ftrvar \mu x : (\mu \interp{\El~A})[\iota] \sez \interp b : \interp{\El~B}[\iota]
	}{\interp \Gamma \sez \lambda \var x . \interp b : \Pi(\ftrvar \mu x : (\mu \interp{\El~A})[\iota]).(\interp{\El~B}[\iota])}{}.
\end{equation}

Application is interpreted as
\begin{align*}
	&\interp{
		\inference{
			\Gamma \judtm{f}{\El~\prodvar \mu x A.B} \qquad
			\mu \setminus \Gamma \judtm{a}{\El~A}
		}{\Gamma \judtm{\apvar \mu f a}{\El~B[\varclr a / \varclr x]}}{t-ap}
	} = \nn\\ &
	\inference{
		\interp \Gamma \sez \interp f : \Pi(\ftrvar \mu x : (\mu \interp{\El~A})[\iota]).(\interp{\El~B}[\iota]) \qquad
		\inference{
			\interp{\mu \setminus \Gamma} \sez \interp a : \interp{\El~A}[\iota]
		}{\interp \Gamma \sez \forsub \mu {\interp a} : (\mu \interp{\El~A})[\iota]}{}
	}{\interp \Gamma \sez \interp f (\forsub \mu {\interp a}) : (\mu \interp{\El~A})[\iota][\id, \forsub \mu \interp a / \ftrvar \mu x]}{}
\end{align*}

The $\beta$-rule looks like this:
\begin{align*}
	& \interp{
		\inference{
			\Gamma, \ctxvar \mu x {\El~A} \judtm b {\El~B} \qquad
			\mu \setminus \Gamma \judtm{a}{\El~A}
		}{\Gamma \judtmeq{(\lamannotvar \mu x A.b)\apvar \mu{}{a}}{b[\varclr a / \varclr x]}{\El~B[\varclr a / \varclr x]}}{}
	} = \nn\\ &
	\inference{
		\interp \Gamma, \ftrvar \mu x : (\mu \interp{\El~A})[\iota] \sez \interp b : \interp{\El\,B}[\iota] \qquad
		\inference{
			\interp{\mu \setminus \Gamma} \sez \interp a : \interp{\El~A}[\iota]
		}{\interp \Gamma \sez \forsub \mu \interp a : (\mu \interp{\El~A})[\iota]}{}
	}{\interp \Gamma \sez (\lambda \ftrvar \mu x . \interp b) (\forsub \mu \interp a) = \interp b [\id, \forsub \mu \interp a / \ftrvar \mu x] : \interp{\El~B}[\iota][\id, \forsub \mu \interp a / \ftrvar \mu x]}{}
\end{align*}
and follows from \cref{def:pi-types}. The $\eta$-rule is:
\begin{align*}
	& \interp{
		\inference{
			\Gamma \judtm{f}{\El~\prodvar \mu x A.B}
		}{\Gamma \judtmeq{\lamannotvar \mu x A.\apvar \mu f x}{f}{\El~\prodvar \mu x A.B}}{}
	} = \nn\\ &
	\inference{
		\interp \Gamma \sez \interp f : \Pi(\ftrvar \mu x : \interp{\El~A}[\iota]).(\interp{\El~B}[\iota])
	}{\interp \Gamma \sez \lambda \ftrvar \mu x . (\interp f [\wkn{\ftrvar \mu x}]) (\forsub \mu \interp x) = \interp f : \Pi(\ftrvar \mu x : \interp{\El~A}[\iota]).(\interp{\El~B}[\iota])}{}
\end{align*}
This rule also follows from \cref{def:pi-types}, because $\forsub \mu \interp x = \ftrvar \mu x$.

\subsection{Pairs}
For pair formation, we have quite straightforwardly:
\begin{align*}
	& \interp{
		\inference{
			\Gamma \judty{\El~\sumvar \mu x A.B} \qquad
			\mu \setminus \Gamma \judtm a {\El~A} \qquad
			\Gamma \judtm b {\El~B[\varclr a/\varclr x]}
		}{\Gamma \judtm{\pairvar \mu a b}{\El~\sumvar \mu x A.B}}{t-pair}
	} = \nn\\ &
	\inference{
		\inference{
			\interp{\mu \setminus \Gamma} \sez \interp a : \interp{\El~A}[\iota]
		}{\interp{\Gamma} \sez \forsub \mu \interp a : (\mu \interp{\El~A})[\iota]}{}
		\qquad
		\interp \Gamma \sez \interp b : \interp{\El~B}[\iota][\id, \forsub \mu \interp a/\ftrvar \mu x]
	}{\interp \Gamma \sez (\forsub \mu \interp a, \interp b) : \Sigma(\ftrvar \mu x : (\mu \interp{\El~A})[\iota]).(\interp{\El~B}[\iota])}{}
\end{align*}
For $\mu = \sharp$, we have to further apply $\hatinquotshp$.

The type is included in the syntactical rule to ensure that $B$ is actually a well-defined type. We also need to know that in the model, but we did not write it explicitly in the interpretation. Note that we are not in fact using the interpretation of the existence of the $\Sigma$-type; rather, we use that if the $\Sigma$-type exists, then admissibly $B$ is a type.

\subsubsection{Projections for continuous pairs}
Instead of interpreting the eliminator, we interpret the first and second projections:
\begin{equation}
	\interp{
		\inference{
			\Gamma \judtm{p}{\sumctu x A.B}
		}{\Gamma \judtm{\fst~p}{A}}{}
	} =
	\inference{
		\interp \Gamma \sez \interp p : \Sigma(\var x : \interp{\El~A}[\iota]).(\interp{\El~B}[\iota])
	}{\interp \Gamma \sez \fst \interp p : \interp{\El~A}[\iota]}{},
\end{equation}
\begin{equation}
	\interp{
		\inference{
			\Gamma \judtm{p}{\sumctu x A.B}
		}{\Gamma \judtm{\snd~p}{B[\fst~p/x]}}{}
	} =
	\inference{
		\interp \Gamma \sez \interp p : \Sigma(\var x : \interp{\El~A}[\iota]).(\interp{\El~B}[\iota])
	}{\interp \Gamma \sez \snd \interp p : \interp{\El~B}[\iota][\id, \fst \interp p / \var x]}{}.
\end{equation}
Then $\beta$- and $\eta$-rules lift from the model.

\subsubsection{Projections for pointwise pairs}
\begin{equation}
	\interp{
		\inference{
			\sharp \setminus \Gamma \judtm{p}{\El~\sumptw x A.B}
		}{\Gamma \judtm{\fstptw p}{\El~A}}{}
	} =
	\inference{
	\inference{
	\inference{
		\interp{\sharp \setminus \Gamma} \sez \interp p : \Sigma(\ftrvar \coshp x : (\coshp \interp{\El~A})[\iota]).(\interp{\El~B}[\iota])
	}{\interp \Gamma \sez \forsub \sharp \interp p : \Sigma(\ftrvar \coshp x : (\coshp \interp{\El~A})[\iota]).(\sharp \interp{\El~B}[\iota])}{}
	}{\interp \Gamma \sez \fst (\forsub \sharp \interp p) : (\coshp \interp{\El~A})[\iota]}{}
	}{\interp \Gamma \sez \vartheta(\fst (\forsub \sharp \interp p)) : \interp{\El~A}[\iota]}{},
\end{equation}
\begin{equation}
	\interp{
		\inference{
			\Gamma \judtm{p}{\El~\sumptw x A.B}
		}{\Gamma \judtm{\sndptw p}{\El~B[\fstptw p / x]}}{}
	} =
	\inference{
		\interp \Gamma \sez \interp p : \Sigma(\ftrvar \coshp x : (\coshp \interp{\El~A})[\iota]).(\interp{\El~B}[\iota])
	}{\interp \Gamma \sez \snd \interp p : \interp{\El~B}[\iota][\id, \fst \interp p / \ftrvar \coshp x]}{}
\end{equation}
One can show that $\fst \interp p$ is the appropriate term to appear in the substitution, for the conclusion to be well-typed.

\subsubsection{Elimination of parametric pairs}
We have to interpret the rule
\begin{equation}
		\inference{
			\Gamma, \ctxvar \nu z {\El~\sumpar x A.B} \judty{\El~C} \\
			\Gamma, \ctxpar{x}{\El~A}, \ctxvar{\nu}{y}{\El~B} \judtm{c}{\El~C[\varclr{\pairpar x y}/\varclr z]} \\
			\nu \setminus \Gamma \judtm{p}{\El~\sumpar x A.B}
		}{\Gamma \judtm{\ind^\nu_{\Sigmapar}(\varclr z.\parclr C, \parclr x.\varclr y.c, \varclr p)}{\El~C[\varclr p / \varclr z]}}{t-indpair}.
\end{equation}
Thanks to the $\forsub \nu$ operator, it is sufficient to interpret
\begin{equation}
	\inference{
		\Gamma, \ctxvar \nu z {\El~\sumpar x A.B} \judty{\El~C} \\
		\Gamma, \ctxpar{x}{\El~A}, \ctxvar{\nu}{y}{\El~B} \judtm{c}{\El~C[\varclr{\pairpar x y}/\varclr z]}
	}{\Gamma, \ctxvar \nu z {\El~\sumpar x A.B} \judtm{\ind^\nu_{\Sigmapar}(\varclr z.\parclr C, \parclr x.\varclr y.c, \varclr z)}{\El~C}}{t-indpair}.
\end{equation}
We have
\begin{equation*}
	\inference{
	\inference{
	\inference{
		\interp \Gamma, \ftrvar \sharp x : (\sharp \interp{\El~A})[\iota], \ftrvar \nu y : (\nu \interp{\El~B})[\iota] \sez c : \interp{\El~C}[\iota][\wkn{\ftrvar \sharp x} \wkn{\ftrvar \nu y}, (\nu \hatinquotshp)(\ftrvar \sharp x, \ftrvar \nu y) / \ftrvar \nu z]
	}{\interp \Gamma, \ftrvar \nu z : \Sigma(\ftrvar{\sharp} x : (\sharp \interp{\El~A})[\iota]).((\nu \interp{\El~B})[\iota]) \sez c [\wkn{\ftrvar \nu z}, \fst~\ftrvar \nu z / \ftrvar \sharp x, \snd~\ftrvar \nu z/\ftrvar \nu y] : \interp{\El~C}[\iota][\wkn{\ftrvar \nu z}, (\nu \hatinquotshp)(\ftrvar \nu z)/\ftrvar \nu z]}{}
	}{\interp \Gamma, \ftrvar \nu z : \paren{\nu \Sigma(\ftrvar \sharp x : \sharp \interp{\El~A}).\interp{\El~B}} [\iota] \sez c [\wkn{\ftrvar \nu z}, \fst~\ftrvar \nu z / \ftrvar \sharp x, \snd~\ftrvar \nu z/\ftrvar \nu y] : \interp{\El~C}[\iota][\wkn{\ftrvar \nu z}, (\nu \hatinquotshp)(\ftrvar \nu z)/\ftrvar \nu z]}{}
	}{\interp \Gamma, \ftrvar \nu z : \paren{\nu \quotshp \Sigma(\ftrvar \sharp x : \sharp \interp{\El~A}).\interp{\El~B}} [\iota] \sez c [\wkn{\ftrvar \nu z}, \fst~\ftrvar \nu z / \ftrvar \sharp x, \snd~\ftrvar \nu z/\ftrvar \nu y] [\wkn{\ftrvar \nu z}, (\nu \hatinquotshp)(\ftrvar \nu z)/\ftrvar \nu z]\inv : \interp{\El~C}[\iota]}{}
\end{equation*}
This is best read bottom-up. In the last step, we get rid of $\quotshp$ using \cref{thm:elim-quotshp}. In the middle, we simply rewrite the context using that $\Sigma$ commutes with lifted functors such as $\nu$, and that $\iota \sharp = \id$ to turn $[\iota \subext]$ into $[\iota]$ on the $\Sigma$-type's codomain. Above, we split up $\ftrvar \nu z$ in its components.

In order to see that the substitution in the type of the premise is correct, note that we have $\interp{  \nu \setminus (\Gamma, \ctxpar x {\El~A}, \ctxvar \nu y {\El~B}) \judtm{\pairpar x y}{\sumpar x A.B}  }$. Interpreting this and applying $\forsub \nu$, one finds $(\nu \hatinquotshp)(\ftrvar \sharp x, \ftrvar \nu y)$ after working through some tedious case distinctions. 

\subsection{Identity types}\label{sec:idtp-semantics}
We have
\begin{align*}
	& \interp{
		\inference{
			\Gamma \judtm{A}{\El~\uni \ell} \qquad
			\Gamma \judtm{a, b}{\El~A}
		}{\Gamma \judtm{a \idtp A b}{\El~\uni \ell}}{t-Id}
	} = \nn\\ &
	\inference{
	\inference{
		\inference{
			\interp \Gamma \sez \interp A : \uniDD_\ell
		}{\sharp \shp \interp \Gamma \sez \ElDD \interp A \dtype_\ell}{} \qquad
		\inference{
		\inference{
			\interp \Gamma \sez \interp a, \interp b : \ElDD \interp A [\sharp \varsigma] [\iota]
		}{\shp \interp \Gamma \sez \interp a [\varsigma]\inv, \interp b[\varsigma]\inv : \ElDD \interp A[\iota]}{}
		}{\sharp \shp \interp \Gamma \sez \ftrtm{\sharp}{(\interp a [\varsigma]\inv)}, \ftrtm{\sharp}{(\interp b [\varsigma]\inv)} : \sharp \ElDD \interp A}{}
	}{\sharp \shp \interp \Gamma \sez \ftrtm{\sharp}{(\interp a [\varsigma]\inv)} \idtp{\sharp \ElDD \interp A} \ftrtm{\sharp}{(\interp b [\varsigma]\inv)} \dtype_\ell}{}
	}{\interp \Gamma \sez \tycodeDD{\ftrtm{\sharp}{(\interp a [\varsigma]\inv)} \idtp{\sharp \ElDD \interp A} \ftrtm{\sharp}{(\interp b [\varsigma]\inv)}} : \uniDD_\ell}{}
\end{align*}
Observe:
\begin{align*}
	\interp{\El~a \idtp A b}
	&= \paren{\ftrtm{\sharp}{(\interp a [\varsigma]\inv)} \idtp{\sharp \ElDD \interp A} \ftrtm{\sharp}{(\interp b [\varsigma]\inv)}} [\sharp \varsigma]
	= \paren{\ftrtm{\sharp}{\interp a} \idtp{\sharp \interp{\El\,A}} \ftrtm{\sharp}{\interp b}}.
\end{align*}
If we further apply $[\iota]$, we get
\begin{equation}
	\interp{\El~a \idtp A b}[\iota]
	= \paren{ \ftrtm{\sharp}{\interp a}[\iota] \idtp{(\sharp \interp{\El\,A})[\iota]} \ftrtm{\sharp}{\interp b}[\iota] }.
\end{equation}

For reflexivity, we have
\begin{equation}
	\interp{
		\inference{
			\leftflat \Gamma \judtm a {\El~A}
		}{\Gamma \judtm{\refl\,\parclr a}{\El~a \idtp A a}}{t-refl}
	} =
	\inference{
	\inference{
		\interp{\sharp \setminus \Gamma} \sez \interp a : \interp{\El\,A} [\iota]
	}{\interp \Gamma \sez \forsub \sharp \interp a : (\sharp \interp{\El~A})[\iota]}{}
	}{\interp \Gamma \sez \refl~(\forsub \sharp \interp a) : \forsub \sharp \interp a \idtp{(\sharp \interp{\El~A})[\iota]} \forsub \sharp \interp a}{}
\end{equation}
where $\forsub \sharp t = \ftrtm \sharp t [\iota]$.

We also need to interpret the $\J$-rule:
\begin{equation}
	\inference{
		\leftflat \Gamma \judtm{a, b}{\El~A} \qquad
		\Gamma, \ctxpar{y}{\El~A}, \ctxvar \nu {w}{\El~a \idtp A y} \judty {\El~C} \\
		\nu \setminus \Gamma \judtm{e}{\El~a \idtp A b} \qquad
		\Gamma \judtm{c}{\El~C[\parclr a/\parclr y, \varclr{\refl~a}/\varclr w]}
	}{\Gamma \judtm{
		\J^\nu(\parclr a, \parclr b, \parclr y.\varclr w.\parclr C, \varclr e, c)
	}{\El~C[\parclr b/\parclr y, \varclr e/\varclr w]}}{t-J}
\end{equation}
First of all, note that if $\nu \in \accol{\coshp, \idmod, \sharp}$, then $\nu \interp{\El~a \idtp A b} = \interp{\El~a \idtp A b}$, because
\begin{equation}
	\nu \interp{\El~a \idtp A b}
	= \nu \paren{\ftrtm{\sharp}{\interp a} \idtp{\sharp \interp{\El\,A}} \ftrtm{\sharp}{\interp b}}
	= \paren{\ftrtm{\nu \sharp}{\interp a} \idtp{\nu \sharp \interp{\El\,A}} \ftrtm{\nu \sharp}{\interp b}}
	= \paren{\ftrtm{\sharp}{\interp a} \idtp{\sharp \interp{\El\,A}} \ftrtm{\sharp}{\interp b}}.
\end{equation}
Hence, we may assume that $\nu = \idmod$. We then have
\begin{equation}
	\inference{
		\interp \Gamma \sez \forsub \sharp \interp a, \forsub \sharp \interp b : (\sharp \interp{\El~A})[\iota] \\
		\interp \Gamma, \ftrvar \sharp y : (\sharp \interp{\El~A})[\iota], \var w : \forsub \sharp \interp a \idtp{(\sharp \interp{\El~A})[\iota]} \ftrvar \sharp y \sez \interp{\El~C}[\iota] \dtype \\
		\interp \Gamma \sez \interp e : \forsub \sharp \interp a \idtp{(\sharp \interp{\El~A})[\iota]} \forsub \sharp \interp b \\
		\interp \Gamma \sez \interp c : \interp{\El~C}[\iota][\id, \forsub \sharp \interp a / \ftrvar \sharp y, \refl~(\forsub \sharp \interp a) / \var w]
	}{\interp \Gamma \sez \J(\forsub \sharp \interp a, \forsub \sharp \interp b, \ftrvar \sharp y.\var w.\interp{\El~C}[\iota], \interp e, \interp c) : \interp{\El~C}[\iota][\id, \forsub \sharp \interp b / \ftrvar \sharp y, \interp e / \var w]}{}
\end{equation}

\subsubsection{The reflection rule}
\begin{remark}[Erratum]
	The original version of \cite{bpcubicalsets}, claimed to prove the reflection rule:
	\begin{equation}
		\inference{
			\Gamma \judtm{a, b}{\El~A} \qquad
			\Gamma \judtm{e}{\El~ a \idtp A b}
		}{\Gamma \judtmeq a b {\El~A}}{t-rflct}.
	\end{equation}
	\begin{proof}[Erroneous proof]
		Indeed, we have
		\begin{equation}
			\inference{
				\interp \Gamma \sez \interp e : \ftrtm \sharp{\interp a}[\iota] \idtp{\sharp (\ElDD \interp{A} [\sharp \varsigma]) [\iota]} \ftrtm \sharp{\interp b}[\iota]
			}{\interp \Gamma \sez \ftrtm \sharp{\interp a}[\iota] = \ftrtm \sharp{\interp b}[\iota] : \sharp(\ElDD \interp{A} [\sharp \varsigma]) [\iota]}{}
		\end{equation}
		Now for any $\gamma : \DSub{W}{\interp \Gamma}$, we have
		\begin{equation}
			\ftrtm \sharp{\interp a}[\iota]\dsub \gamma
			= \ftrtm \sharp{\interp a}\dsub{\fpshadj \flat(\gamma \kappa)}
			= \fpshadj \flat(\interp a \dsub{\gamma \kappa}).
		\end{equation}
		Hence, we can conclude that for any $\gamma$, we have $\flat W \Dsez \interp a \dsub{\gamma \kappa} = \interp b \dsub{\gamma \kappa} : \ElDD \interp{A} [\sharp \varsigma] [\iota] \dsub{\gamma \kappa}$. This means that $\interp a$ and $\interp b$ are equal on bridges, but maybe not on paths.
		
		The type $\ElDD \interp{A} [\sharp \varsigma] [\iota] \dsub{\gamma \kappa}$ we wrote there is correct, because
		\begin{equation}
			\fpshadj \flat(\iota \interp \Gamma \circ \gamma \circ \kappa W)
			= \fpshadj \flat(\fpshadj \flat(\gamma \circ \kappa W \circ \kappa \flat W))
			= \fpshadj \flat(\fpshadj \flat(\gamma \circ \kappa W))
			= \fpshadj \flat(\gamma \circ \kappa W)
			= \iota \interp \Gamma \circ \gamma,
		\end{equation}
		so that if $\flat W \Dsez t : \ElDD \interp A [\sharp \varsigma] \dsub{\iota \gamma \kappa}$, then $W \Dsez \fpshadj \flat(t) : \sharp(\ElDD \interp A [\sharp \varsigma]) \dsub{\iota \gamma}$.
		
		By discreteness, we can form
		\begin{equation}\label{eq:error-reflection}
			\inference{
				\interp \Gamma \sez \interp a, \interp b : \ElDD \interp A [\sharp \varsigma][\iota]
			}{\shp \interp \Gamma \sez \interp a [\varsigma]\inv, \interp b[\varsigma]\inv : \ElDD \interp A [\iota]}{}.
		\end{equation}
		Clearly if we can show $\interp a [\varsigma]\inv = \interp b[\varsigma]\inv$, then $\interp a = \interp b$. This is essentially saying that if $\interp a$ and $\interp b$ act the same way on bridges, then they are equal, so we should be almost there.
		
		Pick $\fpshadj \shp(\overline \gamma) : \DSub{W}{\shp \Gamma = \flat \quotshp \Gamma}$, i.e. $\overline \gamma : \DSub{\shp W}{\quotshp \Gamma}$, i.e. $\gamma : \DSub{\shp W}{\Gamma}$. Recall that $\varsigma = (\kappa \quotshp)\inv \circ \inquotshp$. Now
		\begin{equation}
			\kappa \quotshp \Gamma \circ \fpshadj \shp(\overline \gamma) = \overline \gamma \circ \varsigma W = \inquotshp(\gamma) \circ \varsigma W
		\end{equation}
		so we have $\varsigma \Gamma \circ \gamma \circ \varsigma W = (\kappa \quotshp \Gamma)\inv \circ \inquotshp \gamma \circ \varsigma W = \fpshadj \shp(\overline \gamma)$. Hence
		\begin{align*}
			\interp a [\varsigma]\inv \dsub{\fpshadj \shp(\overline \gamma)}
			= \interp a \dsub{\gamma \circ \varsigma W} &= \interp{a} \dsub{\gamma \circ \kappa \shp W} \psub{\varsigma W} \\ \nn 
			&= \interp b \dsub{\gamma \circ \kappa \shp W} \psub{\varsigma W} = \ldots = \interp b [\varsigma]\inv \dsub{\fpshadj \shp(\overline \gamma)}.
		\end{align*}
		Because $\fpshadj \shp (\overline \gamma)$ is a fully general defining subsitution of $\shp \Gamma$, we can conclude that $\interp a [\varsigma]\inv = \interp b[\varsigma]\inv$ and hence $\interp a = \interp b$.
	\end{proof}
	The error is in \cref{eq:error-reflection}. The idea there is that $\sharp \varsigma \circ \iota = \iota \circ \varsigma$, so that $\interp a$ and $\interp b$ would a type of the form $X[\varsigma]$ and the application of $[\varsigma]\inv$ would be valid. However, when we make some implicit arguments explicit we see that the premise is:
	\begin{equation}
		\interp \Gamma \sez \interp a, \interp b : \ElDD \interp A [\sharp \varsigma \interp{\sharp \setminus \Gamma}][\iota \interp \Gamma], %
	\end{equation}
	so that even though $\sharp \varsigma \circ \iota = \iota \circ \varsigma$ (as an equation of natural transformations), the equation cannot be applied here as the natural transformations have been instantiated on different contexts. The context $\interp{\sharp \setminus \Gamma} \cong \flat \interp \Gamma$ is discrete, i.e. all its paths are constant. By mixing it up with the context $\interp \Gamma$, we ended up assuming that $\interp \Gamma$ is discrete (which need not be the case if it contains parametric variables). Since we already knew that $\interp a$ and $\interp b$ are equal on bridges, we could then conclude that they are completely equal, as there would be only trivial paths.
	
	We see two ways to fix this erratum. Either we actually move back to the context $\sharp \setminus \Gamma$ in the reflection rule's conclusion, making it weaker and breaking the proof of function extensionality for parametric arguments. Or we conjecture the principle of pathhood irrelevance --- that any bridge can be a path in at most one way --- so that equality in $\sharp \setminus \Gamma$ entails equality in $\Gamma$.
\end{remark}

\paragraph{Solution 1: A more careful reflection rule}
\begin{lemma}\label{thm:careful-reflection}
	The model supports the following reflection rule:
	\begin{equation}
		\inference{
			\sharp \setminus \Gamma \judtm{a, b}{\El~A} \qquad
			\Gamma \judtm{e}{\El~ a \idtp A b}
		}{\sharp \setminus \Gamma \judtmeq a b {\El~A}}{}. %
	\end{equation}
\end{lemma}
\begin{proof}
	We have
	\begin{equation}
		\inference{
		\inference{
		\inference{
		\inference{
		\inference{
			\interp \Gamma \sez \interp e : \ftrtm \sharp{\interp a}[\iota \interp \Gamma] \idtp{\sharp (\ElDD \interp{A} [\sharp \varsigma \interp{\sharp \setminus \Gamma}]) [\iota \interp \Gamma]} \ftrtm \sharp{\interp b}[\iota \interp \Gamma]
		}{\interp \Gamma \sez \ftrtm \sharp{\interp a}[\iota \interp \Gamma] = \ftrtm \sharp{\interp b}[\iota \interp \Gamma] : \sharp(\ElDD \interp{A} [\sharp \varsigma \interp{\sharp \setminus \Gamma}]) [\iota \interp \Gamma]}{}
		}{\flat \interp \Gamma \sez \ftrtm \flat{\interp a} = \ftrtm \flat{\interp b} : \flat (\ElDD \interp{A} [\sharp \varsigma \interp{\sharp \setminus \Gamma}][\iota \interp{\sharp \setminus \Gamma}])}{}
		}{\flat \interp \Gamma \sez \kappa(\ftrtm \flat{\interp a}) = \kappa(\ftrtm \flat{\interp b}) : \ElDD \interp{A} [\sharp \varsigma \interp{\sharp \setminus \Gamma}] [\iota \interp{\sharp \setminus \Gamma}] [\kappa \interp{\sharp \setminus \Gamma}]}{}
		}{\flat \interp{\sharp \setminus \Gamma} \sez \interp a [\kappa \interp{\sharp \setminus \Gamma}] = \interp b [\kappa \interp{\sharp \setminus \Gamma}] : \ElDD \interp{A} [\sharp \varsigma \interp{\sharp \setminus \Gamma}] [\iota \interp{\sharp \setminus \Gamma}] [\kappa \interp{\sharp \setminus \Gamma}]}{}
		}{\interp{\sharp \setminus \Gamma} \sez \interp a = \interp b : \ElDD \interp{A} [\sharp \varsigma \interp{\sharp \setminus \Gamma}] [\iota \interp{\sharp \setminus \Gamma}]}{} %
	\end{equation}
\end{proof}

\paragraph{Solution 2: Pathhood irrelevance}
\begin{definition}
	A type $T$ has irrelevant pathhood when the weakening of paths $(W, \ctxpath{\var i}) \Dsez t : T \dsub \gamma$ to bridges $(W, \ctxbrid{\var i}) \Dsez t \psub{\var i/\var i} : T \dsub{\gamma (\var i/\var i)}$ is an injection.
\end{definition}
\begin{lemma}
	If $T$ is pathhood irrelevant, then the operation
	\begin{equation}
		\inference{
			\Gamma \sez t : T
		}{\flat \Gamma \sez t[\kappa] : T[\kappa]}{} %
	\end{equation}
	is injective.
\end{lemma}
\begin{conjecture}
	The interpretation of any type that can be constructed in ParamDTT, has irrelevant pathhood.
\end{conjecture}
\begin{proof}[Justification]
	Pathhood irrelevance is preserved by all functors used in the model (in particular, by all modalities). The discrete universe of discrete types, has irrelevant pathhood by construction. Of the other types, the $\Weld$ type is the most dangerous one, but as we only allow propositions of the form $\idpr i j$ internally, one cannot use $\Weld$ to identify bridges without identifying the corresponding paths.
\end{proof}
Then the model supports the following rule:
\begin{equation}
		\inference{
			\Gamma \judtm{a, b}{\El~A} \qquad
			\sharp \setminus \Gamma \judtmeq a b {\El~A}
		}{\Gamma \judtmeq a b {\El~A}}{}, %
\end{equation}
which we can combine with Solution 1 to obtain the original reflection rule.

\subsubsection{Function extensionality}
Using the reflection rule, we can derive function extensionality internally:
\begin{equation}
	\inference{
	\inference{
	\inference{
		\inference{
			\leftflat \Gamma \judtm{f, g}{\El~\prodvar \mu x A.B}
		}{(\sharp \setminus \Gamma), \ctxvar \mu x A \judtm{\apvar \mu f x, \apvar \mu g x}{\El~B}}{1}
		\qquad
		\inference{
		\inference{
			\Gamma \judtm{p}{\El~\prodvar \mu x A.{\apvar \mu f x} \idtp{B} {\apvar \mu g x}}
		}{\sharp \setminus \Gamma \judtm{p}{\El~\prodvar \mu x A.{\apvar \mu f x} \idtp{B} {\apvar \mu g x}}}{2}
		}{(\sharp \setminus \Gamma), \ctxvar \mu x A \judtm{\apvar \mu p x}{{\apvar \mu f x} \idtp{B} {\apvar \mu g x}}}{3}
	}{(\sharp \setminus \Gamma), \ctxvar \mu x A \judtmeq{\apvar \mu f x}{\apvar \mu g x}{\El~B}}{4}
	}{\sharp \setminus \Gamma \judtmeq{f}{g}{\El~\prodvar \mu x A.B}}{5}
	}{\Gamma \judtm{\refl~\parclr{f}}{\El~f \idtp{\prodvar \mu x A.B} g}}{6}
\end{equation}
Here we used (1) weakening and application, (2) weakening of variances, (3) weakening and application, (4) the reflection rule, (5) $\lambda$-abstraction and the $\eta$-rule and (6) reflexivity and conversion.
\begin{remark}
	If we use the more careful reflection rule (\cref{thm:careful-reflection}), then after (4) we end up with $\ctxvar {\sharp \setminus \mu} x A$ instead of $\ctxvar \mu x A$, so that we only obtain function extensionality for functions of a modality of the form $\sharp \setminus \mu$, i.e. only for pointwise and continuous functions.
\end{remark}

\subsubsection{Uniqueness of identity proofs}
The model supports uniqueness of identity proofs:
\begin{equation}
	\inference{
		\Gamma \judtm{e, e'}{a \idtp A b}
	}{\Gamma \judtmeq{e}{e'}{a \idtp A b}}{t=-UIP}.
\end{equation}
To prove this, we need to show
\begin{equation}
	\inference{
		\interp \Gamma \sez \interp e, \interp{e'} : {\ftrtm{\sharp}{\interp a}}\idtp{(\sharp \interp{\El~A})[\iota]}{\ftrtm{\sharp}{\interp b}}
	}{\interp \Gamma \sez \interp e = \interp{e'} : {\ftrtm{\sharp}{\interp a}}\idtp{(\sharp \interp{\El~A})[\iota]}{\ftrtm{\sharp}{\interp b}}}{}
\end{equation}
But of course, for any $\gamma : \DSub{W}{\Gamma}$, we have $\interp e \dsub \gamma = \star = \interp{e'} \dsub \gamma$.

\section{Internal parametricity: glueing and welding}
\subsection{The interval}
We interpret the interval as a type $\Gamma \sez \IX \dtype$ that exists in any context $\Gamma$ and is natural in $\Gamma$, i.e.\ it is a closed type. Hence, $\IX \dsub \gamma$ will not depend on $\gamma$. Instead, for $\gamma : \DSub W \Gamma$, we set $\IX \dsub \gamma = (\PSub W {(\ctxbrid{\var i})})$.
\begin{lemma}
	The type $\Gamma \sez \IX \dtype$ is discrete.
\end{lemma}
\begin{proof}
	Pick a term $(W, \ctxpath{\var j}) \Dsez t : \IX \dsub{\gamma(\facewkn{\var j})}$. Then $t$ is a primitive substitution $t : \PSub{(W, \ctxpath{\var j})}{(\ctxbrid{\var i})}$ which necessarily factors over $(\facewkn{\var j})$.
\end{proof}
The interval can be seen as a type:
\begin{equation}
	\interp{\Gamma \judty \IX} = \sharp \interp \Gamma \sez \IX \dtype,
\end{equation}
or as an element of the universe:
\begin{equation}
	\interp{\Gamma \judtm{\IX}{\uni 0}} = \interp \Gamma \sez \tycodeDD \IX : \uniDD_0,
\end{equation}
and then $\interp{\El\,\IX} = \IX [\sharp \varsigma] = \IX$. The terms $\interp{\Gamma \judtm{0, 1}{\IX}}$ are modelled by $\interp \Gamma \sez 0, 1 : \IX$ where $W \Dsez 0 \dsub \gamma = (0/\var i, \facewkn W) : \IX \dsub \gamma$ and similar for 1. All other rules regarding the interval are straightforwardly interpreted now that we know that $\IX$ is semantically a type like any other.

\subsection{Face predicates and face unifiers}\label{sec:face-predicates}
\subsubsection{The discrete universe of propositions}
If we had an internal face predicate judgement $\Gamma \sez \parclr P ~\name{fpred}$, analogous to the type judgement $\Gamma \judty T$, then the most obvious interpretation would be $\sharp \interp \Gamma \sez \interp P \prop$. However, in order to satisfy $\interp{\coshp \setminus \Delta} = \sharp \interp \Delta$ for contexts $\Delta$ that contain face predicates, we only want to consider face predicates that absorb $\sharp$. One can show that these take the form $\sharp \interp \Gamma \sez \sharp P \prop$, where $\flat \interp \Gamma \sez P \prop$. The latter corresponds to $\flat \interp \Gamma \sez \tycode P : \Prop$, which in turn corresponds to $\sharp \interp \Gamma \sez \iota(\tycode P)[\kappa]\inv : \sharp \Prop$. So whereas $\uniDD_\ell$ was defined as $\flat \coshp \uniPsh_\ell$, we define the discrete universe of propositions $\PropD = \flat \coshp (\sharp \Prop) = \flat \Prop$. We have
\begin{equation}
	\inference{
	\inference{
	\inference{
	\inference{
	\inference{
		\Gamma \sez P : \PropD = \flat \coshp \sharp \Prop
	}{\shp \Gamma \sez \kappa(P[\varsigma]\inv) : \coshp \sharp \Prop}{}
	}{\sharp \shp \Gamma \sez \vartheta(\kappa(P[\varsigma]\inv)[\iota]\inv) : \sharp \Prop}{}
	}{\flat \sharp \shp \Gamma = \shp \Gamma \sez \iota\inv(\vartheta(\kappa(P[\varsigma]\inv)[\iota]\inv)[\kappa]) : \Prop}{}
	}{\shp \Gamma \sez \El~\iota\inv(\vartheta(\kappa(P[\varsigma]\inv)[\iota]\inv)[\kappa]) \prop}{}
	}{\sharp \shp \Gamma \sez \sharp \El~\iota\inv(\vartheta(\kappa(P[\varsigma]\inv)[\iota]\inv)[\kappa]) \prop}{}
\end{equation}
We took a significant detour here in order to emphasize the parallel with $\uniDD_\ell$. We could have more simply done
\begin{equation}
	\inference{
	\inference{
	\inference{
		\Gamma \sez P : \PropD = \flat \Prop
	}{\shp \Gamma \sez \kappa(P[\varsigma]\inv) : \Prop}{}
	}{\shp \Gamma \sez \El~\kappa(P[\varsigma]\inv) \prop}{}
	}{\sharp \shp \Gamma \sez \sharp \El~\kappa(P[\varsigma]\inv) \prop}{}
\end{equation}
We show that these are equal. Making the interesting part of the former term more precise, we get:
\begin{align*}
	\iota\inv((\vartheta \sharp)((\kappa\coshp \sharp)(P[\varsigma \Gamma]\inv)[\iota \shp \Gamma]\inv)[\kappa \sharp \shp \Gamma])
	&= \iota\inv((\vartheta \sharp)((\kappa\coshp \sharp)(P[\varsigma \Gamma]\inv)[\iota \shp \Gamma]\inv)[\iota \shp \Gamma]) \\
	&= \iota\inv((\vartheta \sharp)((\kappa\coshp \sharp)(P[\varsigma \Gamma]\inv))) \\
	&= \iota\inv((\vartheta \sharp \circ \kappa\coshp \sharp)(P[\varsigma \Gamma]\inv)) \\
	&= \iota\inv((\iota \circ \kappa)(P[\varsigma \Gamma]\inv))
	= \kappa(P [\varsigma \Gamma]\inv),
\end{align*}
which is the corresponding part of the latter term. We set $\ElD~P = \El~\kappa(P[\varsigma]\inv)$ and inversely $\tycodeD P = \kappa\inv(\tycode P)[\varsigma]$.

\subsubsection{The face predicate formers}
We interpret $\interp{\Gamma \judty \IF} = (\sharp \interp \Gamma \sez \PropD \dtype)$, giving meaning to the face predicate judgement. The identity predicate is interpreted as:
\begin{equation}
	\interp{
		\inference{\Gamma \judtm{i, j}{\IX}}{\Gamma \judtm{\idpr i j}{\IF}}{f-eq}
	} =
	\inference{
	\inference{
	\inference{
		\interp \Gamma \sez \interp i, \interp j : \IX
	}{\shp \interp \Gamma \sez \interp i [\varsigma]\inv, \interp j [\varsigma]\inv : \IX}{}
	}{\shp \interp \Gamma \sez \interp i [\varsigma]\inv \idtp{\IX} \interp j [\varsigma]\inv \prop}{}
	}{\interp \Gamma \sez \tycodeD{\interp i [\varsigma]\inv \idtp{\IX} \interp j [\varsigma]\inv} : \PropD}{}.
\end{equation}
Other connectives are interpreted simply by decoding and encoding, e.g.
\begin{equation}
	\interp{
		\inference{\Gamma \judtm{P, Q}{\IF}}{\Gamma \judtm{P \wedge Q}{\IF}}{f-$\wedge$}
	} =
	\inference{
	\inference{
	\inference{
		\interp \Gamma \sez \interp P, \interp Q : \PropD
	}{\shp \interp \Gamma \sez \ElD \interp P, \ElD \interp Q \prop}{}
	}{\shp \interp \Gamma \sez \ElD \interp P \wedge \ElD \interp Q \prop}{}
	}{\interp \Gamma \sez \tycodeD{\ElD \interp P \wedge \ElD \interp Q} : \PropD}{}
\end{equation}

\subsubsection{Context extension with a face predicate}
\begin{equation}
	\interp{
		\inference{
			\Gamma \ctx \qquad
			\leftflat \Gamma \judtm P \IF
		}{\Gamma, \ctxface P \ctx}{c-f}
	} =
	\inference{
		\interp \Gamma \ctx \qquad
		\inference{
		\inference{
		\inference{
		\inference{
			\interp{\sharp \setminus \Gamma} \sez \interp P : \PropD
		}{\shp \interp{\sharp \setminus \Gamma} \sez \ElD \interp P \prop}{}
		}{\sharp \shp \interp{\sharp \setminus \Gamma} \sez \sharp \ElD \interp P \prop}{}
		}{\sharp \interp \Gamma \sez (\sharp \ElD \interp P)[\sharp \varsigma] \prop}{}
		}{\interp \Gamma \sez (\sharp \ElD \interp P)[\sharp \varsigma][\iota] \prop}{}
	}{\interp \Gamma, \_ : (\sharp \ElD \interp P)[\sharp \varsigma][\iota] \ctx}{}
\end{equation}
Note that we have
\begin{equation}
	(\sharp \ElD \interp P)[\sharp \varsigma]
	= \sharp ((\ElD \interp P)[\varsigma])
	= \sharp \El~\kappa(\interp P).
\end{equation}
By analogy to types, we will denote this as $\interp{\El~P}$, even though $\El~P$ does not occur in the syntax. We can then write extended contexts as $\interp \Gamma, \_ : \interp{\El~P}[\iota]$.

\begin{proof}[Addendum to the proof of \cref{thm:leftflat}]
	The fact that $\flat \interp{\Gamma, \ctxface P} = \flat \interp{\sharp \setminus (\Gamma, \ctxface P)}$ follows trivially from $\flat \interp \Gamma = \flat \interp{\sharp \setminus \Gamma}$. Since propositions are discrete, extending the context with a proposition preserves its discreteness and hence the fact that $\kappa$ for that context is an isomorphism.
\end{proof}
\begin{proof}[Addendum to the proof of \cref{thm:leftsharp}]
	The fact that $\interp{\coshp \setminus(\Gamma, \ctxface P)} = \sharp\interp{\Gamma, \ctxface P}$ follows from
	\begin{equation}
		\sharp(\interp{\El~P}[\iota]) = \sharp \interp{\El~P} = \interp{\El~P} = \interp{\El~P}[\iota\sharp]. \qedhere
	\end{equation}
\end{proof}

\subsubsection{Face unifiers} The use of face unifiers is motivated from a computational perspective and is a bit unpractical semantically. Since $\Prop$ is the subobject quantifier of $\widehat \bpcubecat$, extending a context with a proposition amounts to taking a subobject of the context. Every syntactic face unifier $\sigma : \Delta \to \Gamma$ has an interpretation $\interp \sigma : \interp \Delta \to \interp \Gamma$. One can show that the union of the images of all interpretations of all face unifiers to a context $\Gamma$, is equal to all of $\interp \Gamma$. Hence, checking whether something works under all face unifiers, amounts to checking whether it works. One can also show that $P \Rightarrow Q$ means $\interp P \subseteq \interp Q$. Then the rule
\begin{equation}
	\inference{
		\Gamma \judtm{P, Q}{\IF} \qquad
		P \Leftrightarrow Q
	}{\Gamma \judtmeq P Q \IF}{f=}
\end{equation}
is trivial. We also have
\begin{align*}
	& \interp{
		\inference{
			\Gamma \judtm{i, j}{\IX} \qquad
			\top \Rightarrow (\idpr i j)
		}{\Gamma \judtmeq i j \IX}{i=-f}
	} = 
	\inference{
	\inference{
	\inference{
	\inference{
		\interp \Gamma \sez \top \subseteq \tycodeD{(\interp i [\varsigma]\inv \idtp{\IX} \interp j [\varsigma]\inv)} : \PropD
	}{\shp \interp \Gamma \sez \top \subseteq (\interp i [\varsigma]\inv \idtp{\IX} \interp j [\varsigma]\inv) \prop}{}
	}{\shp \interp \Gamma \sez \star : \interp i [\varsigma]\inv \idtp{\IX} \interp j [\varsigma]\inv}{}
	}{\shp \interp \Gamma \sez \interp i [\varsigma]\inv = \interp j [\varsigma]\inv : \IX}{}
	}{\interp \Gamma \sez \interp i = \interp j : \IX}{}
\end{align*}

\subsection{Systems} The interpretation of systems is straightforward.

\subsection{Welding}
We interpret
\begin{equation}
	\inference{
		\Gamma \judtm{P}{\IF} \qquad
		\Gamma, \ctxface P \judtm{T}{\El~\uni \ell} \qquad
		\Gamma \judtm{A}{\El~\uni \ell} \\
		\leftsharp \Gamma, \ctxface P \judtm{f}{\El~A \to T}
	}{\Gamma \judtm{\Weldsys{A}{\Weldsysclause{P}{T}{f}}}{\El~\uni \ell}}{t-Weld}
\end{equation}
as
\begin{equation*}
	\inference{
	\inference{
		\left\{
		\begin{matrix}
			\inference{
				\interp \Gamma \sez \interp P : \PropD
			}{\sharp \shp \interp \Gamma \sez \sharp \ElD \interp P \prop}{}
			&
			\inference{
			\inference{
			\inference{
				\interp \Gamma, \var p : (\sharp \ElD \interp P)[\sharp \varsigma][\iota] \sez \interp T : \uniDD_\ell 
			}{\sharp \shp \paren{\interp \Gamma, \var p : (\sharp \ElD \interp P)[\sharp \varsigma][\iota]} \sez \ElDD \interp T : \dtype_\ell}{}
			}{\sharp \paren{\shp \interp \Gamma, \var p : (\sharp \ElD \interp P)[\iota]} \sez \ElDD \interp T [\sharp (\subext \varsigma)] \dtype_\ell}{}
			}{\sharp \shp \interp \Gamma, \var p : \sharp \ElD \interp P \sez \ElDD \interp T [\sharp (\subext \varsigma)] \dtype_\ell}{} \\
			\qquad \\
			\inference{
				\interp \Gamma \sez \interp A : \uniDD_\ell
			}{\sharp \shp \interp \Gamma \sez \ElDD \interp A \dtype_\ell}{}
			&
			\inference{
				\sharp \interp \Gamma, \var p : \sharp \ElD \interp P [\sharp \varsigma] \sez \interp f : \ElDD \interp A [\sharp \varsigma][\wknvar p] \to \ElDD \interp T [\sharp \varsigma]
			}{\sharp \shp \interp \Gamma, \var p : \sharp \ElD \interp P \sez \interp f [\varsigma \subext]\inv : \ElDD \interp A [\wknvar p] \to \ElDD \interp T[\sharp(\subext \varsigma)]}{}
		\end{matrix}
		\right.
	}{\sharp \shp \interp \Gamma \sez \Weldsys{\ElDD \interp A}{\Weldsysclauseb{\sharp \ElD \interp P}{\ElDD \interp T [\sharp (\subext \varsigma)]}{\interp f [\varsigma \subext]\inv}} \dtype_\ell}{}
	}{\interp \Gamma \sez \tycodeDD{\Weldsys{\ElDD \interp A}{\Weldsysclauseb{\sharp \ElD \interp P}{\ElDD \interp T [\sharp (\subext \varsigma)]}{\interp f [\varsigma \subext]\inv}}} : \uniDD_\ell}{}.
\end{equation*}
We have
\begin{align*}
	\interp{\El~\Weldsys{A}{\Weldsysclause{P}{T}{f}}}
	&= \ElDD \interp{\Weldsys{A}{\Weldsysclause{P}{T}{f}}} [\sharp \varsigma] \\
	&= \Weldsys{\interp{\El~A}}{\Weldsysclauseb{\interp{\El~P}}{\interp{\El~T}}{\interp f}}.
\end{align*}

The constructor
\begin{equation}
	\inference{
		\Gamma \judty{\El~\Weldsys{A}{\Weldsysclause{P}{T}{f}}} \qquad
		\Gamma \judtm{a}{\El~A}
	}{\Gamma \judtm{\weldsys{\sysclause P f} a}{\El~\Weldsys{A}{\Weldsysclause{P}{T}{f}}}}{t-weld}
\end{equation}
becomes
\begin{equation}
	\inference{
		\interp \Gamma \sez \interp a : \interp{\El~A} [\iota]
	}{\interp \Gamma \sez \weldsys{\sysclauseb{\interp{\El~P} [\iota]}{\interp f} [\iota]} \interp a : \Weldsys{\interp{\El~A}}{\Weldsysclauseb{\interp{\El~P}}{\interp{\El~T}}{\interp f}}[\iota]}{}.
\end{equation}

For the eliminator
\begin{equation}
	\inference{
		\Gamma, \ctxvar \nu {y}{\El~\Weldtp P A T f} \judty{\El~C} \qquad
		\Gamma, \ctxface P, \ctxvar \nu y {\El~T} \judtm{d}{\El~C} \\
		\Gamma, \ctxvar \nu x {\El~A} \judtm{c}{\El~C[\varclr{\weldtm P f x} / \varclr y]} \qquad
		\Gamma, \ctxface P, \ctxvar \nu x {\El~A} \judtmeq{c}{d[\varclr{f\,x} / \varclr y]}{\El~C[\varclr{f\,x} / \varclr y]} \\
		\nu \setminus \Gamma \judtm{b}{\El~\Weldtp P A T f}
	}{\Gamma \judtm{
		\ind^\nu_\Weld(\varclr y.\parclr C, \sys{\sysclause{P}{\varclr y.d}}, \varclr x.c, \varclr b)	
	}{\El~C[\varclr b/\varclr y]}}{t-indweld}
\end{equation}
first note that
\begin{align*}
	\nu \interp{\El~\Weldsys{A}{\Weldsysclause{P}{T}{f}}}
	&= \nu \Weldsys{\interp{\El~A}}{\Weldsysclauseb{\interp{\El~P}}{\interp{\El~T}}{\interp f}} \\
	&= \Weldsys{\nu \interp{\El~A}}{\Weldsysclauseb{\interp{\El~P}}{\nu \interp{\El~T}}{\lambda \ftrvar \nu x.\ftrtm \nu {(\interp f \var x)}}}
\end{align*}
because lifted functors preserve $\Weld$ and $\nu \interp{\El~P} = \interp{\El~P}$ for $\nu \in \accol{\coshp, \idmod, \sharp}$. Taking that into account, all of this boils down to straightforward use of the eliminator of the $\Weld$-type for presheaves.

\subsection{Glueing}
We similarly get
\begin{equation}
	\interp{\El~\Gluesys{A}{\Gluesysclause{P}{T}{f}}}
	= \Gluesys{\interp{\El~A}}{\Gluesysclauseb{\interp{\El~P}}{\interp{\El~T}}{\interp f}}.
\end{equation}

The constructor
\begin{equation}
	\inference{
		\Gamma \judty{\El~\Gluetp P A T f} \qquad
		\Gamma, \ctxface P \judtm t {\El~T} \qquad
		\Gamma \judtm a {\El~A} \qquad
		\Gamma, \ctxface P \judtmeq{f\,t}{a}{\El~A}
	}{\Gamma \judtm{\gluetm P a t}{\El~\Gluetp P A T f}}{t-glue}
\end{equation}
becomes
\begin{equation}
	\inference{
		\interp \Gamma, \var p : \interp{\El~P}[\iota] \sez \interp t : \interp{\El~T} \\
		\interp \Gamma \sez \interp a : \interp{\El~A} \\
		\interp \Gamma, \var p : \interp{\El~P}[\iota] \sez \interp f [\iota] \interp t = \interp a [\wknvar p] : \interp{\El~A}[\iota]
	}{\interp \Gamma \sez \gluesys{\interp a}{\sysclauseb{\interp{\El~P}[\iota]}{\interp t}} : \Gluesys{\interp{\El~A}}{\Gluesysclauseb{\interp{\El~P}}{\interp{\El~T}}{\interp f}}[\iota]}{}.
\end{equation}

The eliminator
\begin{equation}
	\inference{
		\Gamma \judtm{b}{\El~\Gluetp P A T f}
	}{\Gamma \judtm{\ungluetm P f b}{\El~A}}{t-unglue}
\end{equation}
becomes
\begin{equation}
	\inference{
		\interp \Gamma \sez \interp b : \Gluesys{\interp{\El~A}}{\Gluesysclauseb{\interp{\El~P}}{\interp{\El~T}}{\interp f}}[\iota]
	}{\interp \Gamma \sez \ungluesys{\sysclauseb{\interp{\El~P} [\iota]}{\interp f [\iota]}} \interp b : \interp{\El~A}[\iota]}{}.
\end{equation}

\subsection{The path degeneracy axiom}
We will interpret the path degeneracy axiom
\begin{equation}
	\inference{
		\Gamma \judty A \qquad
		\leftflat{\Gamma} \judtm{p}{\Pipar(i : \IX).A}
	}{\Gamma \judtm{\degaxof p}{{p}\idtp{\Pipar(i : \IX).A}{\paren{\lamannotpar{i}{\IX}\appar p 0}}}}{t-degax}
\end{equation}
via the stronger rule
\begin{equation}
	\inference{
		\Gamma \judtm{p}{\prodpar i \IX.A}
	}{\Gamma \judtmeq{p}{\lamannotpar i \IX.\appar p 0}{\prodpar i \IX.A}}{}
\end{equation}
after which the axiom follows by reflexivity. For simplicity, we put the variables in the context. We have
\begin{equation}
	\inference{
	\inference{
		\interp \Gamma, \ftrvar \sharp i : \sharp \IX \sez \interp a : \ElDD \interp A [\sharp \varsigma][\iota][\wkn{\ftrvar \sharp i}]
	}{\shp \paren{\interp \Gamma, \ftrvar \sharp i : \sharp \IX} \sez \interp a [\varsigma]\inv : \ElDD \interp A [\iota][\shp \wkn{\ftrvar \sharp i}]}{}
	}{\shp \interp \Gamma \sez \interp a [\varsigma]\inv [\shp \wkn{\ftrvar \sharp i}]\inv : \ElDD \interp A [\iota]}{}
\end{equation}
because $\shp \wkn{\ftrvar \sharp i} : \shp \paren{\interp \Gamma, \ftrvar \sharp i : \sharp \IX} \cong \shp \interp \Gamma$ can be shown to be an isomorphism. Since these operations are invertible, we have
\begin{equation}
	\interp a = \interp a [\varsigma]\inv [\shp \wkn{\ftrvar \sharp i}]\inv [\varsigma] [\wkn{\ftrvar \sharp i}]
\end{equation}
and the right hand side is clearly invariant under $\loch[\id, \ftrtm \sharp 0/\ftrvar \sharp i]$.

To see in general that $\shp \wkn{\ftrvar \sharp i} : \shp(\Gamma, \ftrvar \sharp i : \IX) \cong \shp \Gamma$ is an isomorphism, pick $\overline{(\gamma, \fpshadj \flat(\vfi))} : \DSub{W}{\shp(\Gamma, \ftrvar \sharp i : \IX)}$. We show that $\overline{(\gamma, \fpshadj \flat(\vfi))} = \overline{(\gamma, \ftrtm \sharp 0)}$. We have $\fpshadj \flat(\vfi) : \PSub{W}{\sharp (\ctxbrid{\var i})}$ (there is some benevolent abuse of notation involved here, related to the fact that $\IX$ is a closed type) and hence $\vfi : \PSub{\flat W}{(\ctxbrid{\var i})}$. Now $\vfi$ factors as $\vfi = (\facewkn{\flat W})(k / \var i)$. Similarly, one can show that $\ftrtm \sharp 0 = \fpshadj \flat((\facewkn{\flat W})(0/\var i))$.

By discreteness of $\shp$, we have
\begin{equation}
	\overline{(\gamma (\facewkn{\var i}), \fpshadj \flat(\facewkn{\flat W}))}
	= \overline{(\gamma (\facewkn{\var i}), \fpshadj \flat(\facewkn{\flat W}))}(0/\var i, \facewkn{\var i})
	= \overline{(\gamma (\facewkn{\var i}), \ftrtm \sharp 0(\facewkn{\var i}))}
	= \overline{(\gamma, \ftrtm \sharp 0)} (\facewkn{\var i}).
\end{equation}
Restricting both sides by $(k/\var i)$, we find what we wanted to prove.

\section{Sizes and natural numbers}
\subsection{The natural numbers}
In any presheaf category, we can define a closed type $\Nat$ by setting $\Nat \dsub \gamma = \IN$ and $n \psub \vfi = n$. We have
\begin{equation*}
	\inference{\Gamma \ctx}{\Gamma \sez \Nat \type_0}{}, \qquad
	\inference{\Gamma \ctx}{\Gamma \sez 0 : \Nat}{}, \qquad
	\inference{\Gamma \sez n : \Nat}{\Gamma \sez \suc\,n : \Nat}{} \qquad
	\inference{
		\Gamma, \var m : \Nat \sez C \type \\
		\Gamma \sez c_0 : C[\id, 0/\var m] \\
		\Gamma, \var m : \Nat, \var c : C \sez c_\suc : C[\id,\suc\,\var m/\var m] \\
		\Gamma \sez n : \Nat
	}{\Gamma \sez \ind_\Nat(\var m.C, c_0, \var m.\var c.c_\suc, n) : C[\id, n/\var m]}{}
\end{equation*}
and all these operators are natural in $\Gamma$ \emph{and} are respected by lifted functors, e.g.\ $\fpsh F \Nat = \Nat$ and $\ftrtm{\fpsh F}{(\suc\,n)} = \suc(\ftrtm{\fpsh F}{n})$.

In $\widehat{\bpcubecat}$, $\Nat$ is a discrete type since $n \psub{0/\var i, \facewkn{\var i}} = n$ since in general $n \psub \vfi = n$. We can now interpret the inference rules for natural numbers:
\begin{equation}
	\interp{
		\inference{
			\Gamma \ctx
		}{\Gamma \judtm{\Nat}{\uni 0}}{t-Nat}
	} =
	\inference{
	\inference{
		\interp \Gamma \ctx
	}{\sharp \shp \interp \Gamma \sez \Nat \dtype_0}{}
	}{\interp \Gamma \sez \tycodeDD \Nat : \uniDD_0}{}.
\end{equation}
We have $\interp{\El~\Nat} = \Nat$.
\begin{equation}
	\interp{
		\inference{
			\Gamma \ctx
		}{\Gamma \judtm{0}{\El~\Nat}}{t-0}
	} =
	\inference{
		\interp \Gamma \ctx
	}{\interp \Gamma \sez 0 : \Nat}{},
	\qquad \qquad
	\interp{
		\inference{
			\Gamma \judtm{n}{\El~\Nat}
		}{\Gamma \judtm{\suc\,n}{\El~\Nat}}{t-s}
	} =
	\inference{
		\interp \Gamma \sez \interp n : \Nat
	}{\interp \Gamma \sez \suc \interp n : \Nat}{}.
\end{equation}
The induction principle
\begin{equation}
	\inference{
		\Gamma, \ctxvar \nu m {\El~\Nat} \judty {\El~C} \qquad
		\Gamma \judtm{c_0}{\El~C[\varclr 0/\varclr m]} \\
		\Gamma, \ctxvar \nu m {\El~\Nat}, \ctxctu c {\El~C} \judtm{c_\suc}{\El~C[\varclr{\suc\,m}/\varclr m]} \\
		\nu \setminus \Gamma \judtm n {\El~\Nat}
	}{
		\Gamma \judtm{\ind^\nu_\Nat(\varclr m.\parclr C, c_0, \varclr m.c.c_\suc, \varclr n)
	}{\El~C[\varclr n / \varclr m]}}{t-indnat}
\end{equation}
is interpreted as
\begin{equation}
	\inference{
		\interp \Gamma, \var m : \Nat \sez \interp{\El~C}[\iota] \dtype \qquad
		\interp \Gamma \sez \interp{c_0} : \interp{\El~C}[\iota][\id, 0/\var m] \\
		\interp \Gamma, \var m : \Nat, \var c : \interp{\El~C}[\iota] \sez \interp{c_\suc} : \interp{\El~C}[\iota][\suc\,\var m/\var m] \\
		\interp \Gamma \sez \forsub \nu \interp n : \Nat
	}{\interp \Gamma \sez \ind_\Nat(\var m.\interp{\El~C}[\iota], \interp{c_0}, \var m.\var c.\interp{c_\suc}, \forsub \nu \interp n)}{},
\end{equation}
where we use extensively that $\nu \Nat = \Nat$.

\subsection{Sizes}$~$

\medskip

\subsubsection{In the model}

\begin{proposition}
	We have a discrete type $\Size$ with the following inference rules:
	\begin{equation}
		\inference{\Gamma \ctx}{\Gamma \sez \Size \dtype_0}{} \qquad
		\inference{\Gamma \ctx}{\Gamma \sez \szero : \Size}{} \qquad
		\inference{\Gamma \sez n : \Size}{\Gamma \sez \ssuc n : \Size}{} \qquad
		\inference{
			\shp \Gamma \sez P \prop \qquad
			\Gamma, \var p : P[\varsigma] \sez n : \Size
		}{\Gamma \sez \sfillsys{\sfillsysclause{\var p : P}{n}} : \Size}{}
	\end{equation}
	\begin{equation}
		\inference{
			\shp \Gamma \sez i : \IX \qquad
			\Gamma \sez m, n : \Size \\
			\Gamma, \var p : \paren{(i \idtp \IX 0) \vee (i \idtp \IX 1)}[\varsigma] \sez m[\wknvar p] = n[\wknvar p] : \Size
		}{\Gamma \sez m = n : \Size}{} \quad
		\inference{
			\Gamma \sez m, n : \Size
		}{\Gamma \sez m \sqcup n : \Size}{}
	\end{equation}
	satisfying the expected equations.
\end{proposition}
For closed types $T$, the set $T \dsub \gamma$ is independent of $\gamma : \PSub W \Gamma$; hence we will denote it as $\DSub W T$. Similarly, we will write $W \Dsez t : T$ for $W \Dsez t : T \dsub \gamma$.
\begin{proof}
	For $\gamma : \DSub W \Gamma$, we set $(\DSub W \Size) = \IN^{\PSub{()}{\shp W}}$, i.e.\ a term $n : \Size \dsub \gamma$ consists of a natural number for every vertex of the cube $\shp W$, which is $W$ with all path dimensions contracted. Put differently still, a term $n : \DSub W \Size$ consists of a natural number for every vertex of the cube $W$, such that numbers for path-adjacent vertices are equal. Writing $n \ssub \psi$ for the vertex corresponding to $\psi : \PSub{()}{\shp W}$, we define accordingly $n \psub \vfi \ssub \psi = n \ssub{\shp \vfi \circ \psi}$. This implies that we have in general $n \ssub \psi = n \psub{\psi'} \ssub{\id_{()}}$ for any $\psi' : \PSub{()}{W}$ such that $\shp \psi' = \psi$. Such a $\psi'$ always exists, e.g. $\psi' = (\psi, 0/\var i^\IP \in W)$. Hence, we will avoid the $\ssub \loch$ notation and say that a term $W \Dsez n : \Size$ is determined by all of its vertices $() \Dsez n \psub \vfi : \Size$ which are in fact functions $\IN^{\PSub{()}{()}}$ but can be treated as naturals since $\PSub{()}{()}$ is a singleton. Every such term $n$ has the property that $() \Dsez n \psub \vfi : \Size$ is independent of how $\vfi$ treats path variables, i.e. if $\shp \vfi = \shp \psi$, then $n \psub \vfi = n \psub \psi$.
	
	To see that $\Size$ is discrete, pick $(W, \ctxpath{\var i}) \Dsez n : \Size$. Then $n \psub{0/\var i, \facewkn{\var i}} = n$ because $\shp(0/\var i, \facewkn{\var i}) = \id$.
	
	We define $W \Dsez \szero : \Size$ by setting $() \Dsez \szero \psub{\vfi} = 0 : \Size$ for all $\vfi : \PSub{()}{W}$.
	
	We define $W \Dsez \ssuc n : \Size$ by setting $() \Dsez (\ssuc n) \psub \vfi = n \psub \vfi + 1 : \Size$.
	
	Assume we have $\shp \Gamma \sez P \prop$ and $\Gamma, \var p : P[\varsigma] \sez n : \Size$. Then we define $\Gamma \sez \sfillsys{\sfillsysclause{\var p : P}{n}}$ as follows: pick $\gamma : \DSub{()}{\Gamma}$. Then we set $\sfillsys{\sfillsysclause{\var p : P}{n}} \dsub \gamma = n \dsub{\gamma, \star / \var p}$ if $P [\varsigma] \dsub \gamma = \accol{\star}$, and $\sfillsys{\sfillsysclause{\var p : P}{n}} \dsub \gamma = 0$ if $P [\varsigma] \dsub \gamma = \eset$. We need to show that this respects paths, i.e. if $\vfi, \psi : \PSub{()}{W}$, $\gamma : \DSub{W}{\Gamma}$ and $\shp \vfi = \shp \psi$, then we should prove that $\sfillsys{\sfillsysclause{\var p : P}{n}} \dsub{\gamma \vfi} = \sfillsys{\sfillsysclause{\var p : P}{n}} \dsub{\gamma \psi}$. First, we show that $\varsigma \gamma \vfi = \varsigma \gamma \psi$. Since $\varsigma$ decomposes as $\kappa\inv \inquotshp$, it suffices to show that $\kappa \varsigma \gamma \vfi = \kappa \varsigma \gamma \psi$. Now we have
	\begin{equation}
		\kappa \quotshp \Gamma \circ \varsigma \Gamma \circ \gamma \circ \vfi
		= \fpshadj \shp\inv(\varsigma \Gamma \circ \gamma \circ \vfi) \circ \varsigma()
		= \fpshadj \shp\inv(\varsigma \Gamma \circ \gamma) \circ \shp \vfi \circ \varsigma(),
	\end{equation}
	and similar for $\psi$, which proves the equality since $\shp \vfi = \shp \psi$. But this implies that $P[\varsigma]\dsub{\gamma \vfi} = P[\varsigma]\dsub{\gamma \psi}$. Since we also have $n \dsub{\gamma \vfi} = n \dsub{\gamma \psi}$, we can conclude that $\sfill$ respects paths.
	
	For the equality expressing codiscreteness, assume the premises and pick $\gamma : \DSub{()}{\Gamma}$. Then we have $i\dsub{\varsigma \gamma} : \DSub{()}{\IX}$, which is either 0 or 1. Hence, $\paren{(i \idtp \IX 0) \vee (i \idtp \IX 1)} [\varsigma] \dsub \gamma = \accol \star$, and we find
	\begin{equation}
		m \dsub \gamma = m[\wknvar p] \dsub{\gamma, \star} = n[\wknvar p] \dsub{\gamma, \star} = n \dsub \gamma.
	\end{equation}

	We define $W \Dsez m \sqcup n : \Size$ by setting $() \Dsez (m \sqcup n) \psub \vfi : \Size$ equal to the maximum of $m \psub \vfi$ and $n \psub \vfi$.
\end{proof}

We can easily lift $\szero$, $\ssuc$ and $\smax{}{}$ to $\sharp \Size$, e.g. we have $\var m : \Size \sez \ssuc \var m : \Size$, whence $\ftrvar \sharp m : \sharp \Size \sez \ftrtm{\sharp}{(\ssuc \var m)} : \sharp \Size$, and then using substitution we can derive
\begin{equation}
	\inference{
		\Gamma \sez n : \sharp \Size
	}{\Gamma \sez \ftrtm{\sharp}{(\ssuc \var m)}[\bullet, n/\ftrvar \sharp m] : \sharp \Size}{}.
\end{equation}
We will denote the latter term as $\ssuc n$. Similarly, we can lift $\sfill$:
\begin{equation}
	\inference{
	\inference{
		\inference{
		\inference{
			\Gamma \sez P \prop
		}{\flat \Gamma \sez \flat P \prop}{}
		}{\shp \flat \Gamma \sez (\flat P)[\varsigma \inv] \prop}{}
		\qquad
		\inference{
			\Gamma, \var p : P \sez n : \sharp \Size
		}{\flat \Gamma, \ftrvar \flat p : \flat P \sez \iota\inv(n[\kappa]) : \Size}{}
	}{\flat \Gamma \sez \sfillsys{\sfillsysclause{\ftrvar \flat p : (\flat P)[\varsigma\inv]}{\iota\inv(n[\kappa])}} : \Size}{}
	}{\Gamma \sez \iota \paren{\sfillsys{\sfillsysclause{\ftrvar \flat p : (\flat P)[\varsigma\inv]}{\iota\inv(n[\kappa])}}}[\kappa]\inv : \sharp \Size}{}
\end{equation}
We will denote this result as $\sfillsyssharp{\sfillsysclause{\var p : P}{n}}$. Note that $\varsigma()$ and $\kappa()$ are the identity. For $\gamma : \DSub{()}{\Gamma}$, we have
\begin{align*}
	\sfillsyssharp{\sfillsysclause{\var p : P}{n}} \dsub{\gamma}
	&= \iota \paren{\sfillsys{\sfillsysclause{\ftrvar \flat p : (\flat P)[\varsigma\inv\Gamma]}{\iota\inv(n[\kappa\Gamma])}}}[\kappa\Gamma]\inv \dsub{\gamma} \\
	&= \iota \paren{\sfillsys{\sfillsysclause{\ftrvar \flat p : (\flat P)[\varsigma\inv\Gamma]}{\iota\inv(n[\kappa\Gamma])}}} \dsub{\fpshadj \shp(\gamma \circ \varsigma ()\inv)} \\
	&= \iota \paren{\sfillsys{\sfillsysclause{\ftrvar \flat p : (\flat P)[\varsigma\inv\Gamma]}{\iota\inv(n[\kappa\Gamma])}}} \dsub{\fpshadj \shp(\gamma)} \\
	&= \fpshadj \flat \paren{\sfillsys{\sfillsysclause{\ftrvar \flat p : (\flat P)[\varsigma\inv\Gamma]}{\iota\inv(n[\kappa\Gamma])}} \dsub{\fpshadj \shp(\gamma) \circ \kappa()}} \\
	&= \fpshadj \flat \paren{\sfillsys{\sfillsysclause{\ftrvar \flat p : (\flat P)[\varsigma\inv\Gamma]}{\iota\inv(n[\kappa\Gamma])}} \dsub{\fpshadj \shp(\gamma)}}.
\end{align*}
Now
\begin{equation}
	(\flat P)[\varsigma\inv \Gamma][\varsigma \Gamma] \dsub{\fpshadj \shp(\gamma)}
	= (\flat P) \dsub{\fpshadj \shp \Gamma} = P \dsub \gamma.
\end{equation}
So we make the expected case distinction: if $P \dsub \gamma = \accol \star$, then
\begin{align*}
	\sfillsyssharp{\sfillsysclause{\var p : P}{n}} \dsub{\gamma}
	&= \fpshadj \flat \paren{\iota\inv(n[\kappa\Gamma]) \dsub{\fpshadj \shp(\gamma)}}
	= n[\kappa\Gamma]\dsub{\fpshadj \shp(\gamma)} = n \dsub \gamma.
\end{align*}
If $P \dsub \gamma = \eset$, then we just get $0$. So while it looks ugly, this is precisely the construction we would expect.
\begin{proposition}
	We have an inequality proposition
	\begin{equation}
		\inference{\Gamma \sez m, n : \sharp \Size}{\Gamma \sez m \leq n \prop}{}
	\end{equation}
	that satisfies reflexivity, transitivity, $0 \leq n$, $\ssuc m \leq \ssuc n$ if $m \leq n$, $\sfillsyssharp{\sfillsysclause{\var p : P}{m}} \leq \sfillsyssharp{\sfillsysclause{\var p : P}{n}}$ if $m \leq n$, and $m \leq m \sqcup n$ and $n \leq m \sqcup n$.
\end{proposition}
\begin{proof}
	Assume we have $W \Dsez \fpshadj \flat(m), \fpshadj \flat(n) : \sharp \Size$, i.e. $\flat W \Dsez m, n : \Size$. Then we set $(\fpshadj \flat(m) \leq \fpshadj \flat(n))$ equal to $\accol \star$ if $() \Dsez m \psub \vfi \leq n \psub \vfi : \Size$ for every $\vfi : \PSub{()}{\flat W}$. Otherwise, we set it equal to $\eset$. This can be shown to satisfy the required properties.
\end{proof}
\begin{proposition}
	We have
	\begin{equation}
		\inference{
			(\mu, \beta : \mu \to \sharp) \in \accol{(\coshp, \iota \vartheta), (\Id, \iota), (\sharp, \id)} \\
			\Gamma, \ftrvar \mu n : \mu \Size \sez A \type \\
			\Gamma \sez f : \Pi(\ftrvar \mu n : \mu \Size).\paren{  \Pi(\ftrvar \mu m : \mu \Size).(\uparrow \beta(\ftrvar \mu m) \leq \beta(\ftrvar \mu n)) \to A[\ftrvar \mu m/\ftrvar \mu n]  } \to A
		}{\Gamma \sez \fix^\mu\,f : \Pi(\ftrvar \mu n : \mu\Size).A}{}
	\end{equation}
	(where we omit weakening and other uninteresting parts of substitutions).
\end{proposition}
\begin{proof}
	We define $\fix^\mu\,f = \lambda \ftrvar \mu n.f~\ftrvar \mu n~(\lambda \ftrvar \mu m.\lambda \var p.\fix^\mu~f~\ftrvar \mu m)$ (where we omit weakening substitutions), which we will prove to be a well-founded definition by induction essentially on the greatest vertex of $\ftrvar \mu n$.
	\begin{description}
		\item[$\mu = \Id$] Pick $\gamma : \DSub W \Gamma$ and $n : \DSub W \Size$. Let $\omega : \PSub{()}{W}$ attain a maximal vertex $n \psub \omega$ of $n$. We show that we can define $(\fix~f)\dsub \gamma \cdot n$ as $f \dsub \gamma \cdot n \cdot (\lambda \var m.\lambda \var p.\fix~f~\var m) \dsub \gamma$, assuming that $(\fix~f) \dsub \gamma \cdot m$ is already defined for all $m : \DSub W \Size$ such that all vertices of $m$ are less than $n \psub \omega$. In other words, we have to show that $(\lambda \var m.\lambda \var p.\fix~f~\var m) \dsub \gamma$ is already defined. But this function is determined completely by defining terms of the form $(\lambda \var m.\lambda \var p.\fix~f~\var m) \dsub \gamma \cdot m \cdot p$, where $W \Dsez m : \Size$ and $W \Dsez p : \ssuc \iota(m) \leq \iota(n)$. Note that $\iota(m) = \fpshadj \flat(m \psub{\kappa})$. The existence of $p$ then implies that $m \psub{\kappa \vfi} + 1 \leq n \psub{\kappa \vfi}$ for all $\vfi : \PSub{()}{\flat W}$. Since any $\psi : \PSub{()}{W}$ factors as $\psi = \psi \circ \kappa () = \kappa W \circ \flat \psi$, we can conlude that all vertices of $m$ are less than the corresponding ones of $n$, hence less than $n \psub \omega$. Then $(\fix~f) \dsub \gamma \cdot m$ is already defined, and one can show that
		\begin{equation}
			(\lambda \var m.\lambda \var p.\fix~f~\var m) \dsub \gamma \cdot m \cdot p = (\fix~f) \dsub \gamma \cdot m.
		\end{equation}
		Hence, the definition is well-founded.
		
		\item[$\mu = \sharp$] Pick $\gamma : \DSub W \Gamma$ and $\fpshadj \flat(n) : \DSub{W}{\sharp \Size}$, i.e. $n : \DSub{\flat W}{\Size}$. Let $\omega : \PSub{()}{\flat W}$ attain a maximal vertex $n \psub \omega$ of $n$. We show that we can define $(\fix^\sharp~f)\dsub \gamma \cdot \fpshadj \flat(n)$ as $f \dsub \gamma \cdot \fpshadj \flat(n) \cdot (\lambda \ftrvar \sharp m.\lambda \var p.\fix^\sharp~f~\ftrvar \sharp m) \dsub \gamma$. So all $(\lambda \ftrvar \sharp m.\lambda \var p.\fix^\sharp~f~\ftrvar \sharp m) \dsub \gamma \cdot m \cdot p = (\fix^\sharp~f) \dsub \gamma \cdot m$ have to be defined. But the existence of $W \Dsez p : \fpshadj \flat(m+1) \leq \fpshadj \flat(n)$ asserts that $(\fix^\sharp~f) \dsub \gamma \cdot m$ is already defined by the induction hypothesis. 
		
		\item[$\mu = \coshp$] Analogous. \qedhere
	\end{description}
\end{proof}

\subsubsection{The type $\Size$}
We can now proceed with the interpration of $\Size$ in ParamDTT. Just like with $\Nat$, the interpretation of $\szero$, $\ssuc$ and $\smax{}{}$ is entirely straightforward. For t-Size-fill, we have
\begin{align*}
	&\interp{\inference{
		\Gamma \judtm{P}{\IF} \qquad
		\Gamma, \ctxface{P} \judtm{n}{\Size}
	}{\Gamma \judtm{\sfillsys{\sfillsysclause P n}}{\Size}}{t-Size-fill}}
	= \\
	&\inference{
		\inference{
		\inference{
		\inference{
			\interp \Gamma \sez \interp P : \PropD
		}{\shp \interp \Gamma \sez \ElD{\interp P} \prop}{}
		}{\sharp \shp \interp \Gamma \sez \sharp \ElD{\interp P} \prop}{}
		}{\shp \interp \Gamma \sez (\sharp \ElD{\interp P})[\iota] \prop}{}
		\qquad
		\interp \Gamma, \var p : (\sharp \ElD{\interp P})[\sharp \varsigma][\iota] \sez \interp n : \Size
	}{\interp \Gamma \sez \sfillsys{\sfillsysclause{\var p : (\sharp \ElD{\interp P})[\iota]}{\interp n}} : \Size}{}
\end{align*}
Then we can also interpet t=-Size-codisc.

\subsubsection{The inequality type}
The inequality type is interpreted as
\begin{equation}
	\interp{
		\inference{
			\Gamma \judtm{m, n}{\El~\Size}
		}{\Gamma \judtm{m \leq n}{\El~\uni 0}}{t-$\leq$}
	} =
	\inference{
	\inference{
	\inference{
	\inference{
		\interp \Gamma \sez \interp m, \interp n : \Size
	}{\shp \interp \Gamma \sez \interp m[\varsigma]\inv, \interp n[\varsigma]\inv : \Size}{}
	}{\sharp \shp \interp \Gamma \sez \ftrtm{\sharp}{(\interp m[\varsigma]\inv)}, \ftrtm{\sharp}{(\interp n[\varsigma]\inv)} : \sharp \Size}{}
	}{\sharp \shp \interp \Gamma \sez \ftrtm{\sharp}{(\interp m[\varsigma]\inv)} \leq \ftrtm{\sharp}{(\interp n[\varsigma]\inv)} \prop}{}
	}{\interp \Gamma \sez \tycodeDD{\ftrtm{\sharp}{(\interp m[\varsigma]\inv)} \leq \ftrtm{\sharp}{(\interp n[\varsigma]\inv)}} : \uniDD_0}{}
\end{equation}
We have $\interp{\El~ m \leq n} = \ftrtm{\sharp}{\interp m} \leq \ftrtm{\sharp}{\interp n}$, and $\interp{\El~m \leq n}[\iota] = \forsub \sharp \interp m \leq \forsub \sharp \interp n$.

As an example of how we interpret simple inequality axioms, we take the following:
\begin{equation}
	\interp{
		\inference{
			\sharp \setminus \Gamma \judtm{n}{\El~\Size}
		}{\Gamma \judtm{\name{zero}_\leq~ n}{\El~0 \leq n}}{t-$\leq$-zero}
	} =
	\inference{
		\interp \Gamma \sez \forsub \sharp \interp n : \sharp \Size
	}{\interp \Gamma \sez \star : \szero \leq \forsub \sharp \interp n}{}.
\end{equation}

The filling rule is a bit more complicated. We need to prove
\begin{equation}
	\inference{
		\sharp \setminus \Gamma \judtm P \IF \qquad
		\sharp \setminus \Gamma, \ctxface P \judtm{m, n}{\Size} \qquad
		\Gamma, \ctxface P \judtm{e}{m \leq n}
	}{\Gamma \judtm{\leqfillsys{\leqfillsysclause P e}}{\sfillsys{\sfillsysclause P m} \leq \sfillsys{\sfillsysclause P n}}}{t-$\leq$-fill}.
\end{equation}
First, we unpack the proposition (see the section on face predicates):
\begin{equation}
	\inference{
		\interp{\sharp \setminus \Gamma \judtm P \IF} = \paren{\interp{\sharp \setminus \Gamma} \sez \interp P : \PropD}
	}{\interp{\Gamma} \sez (\sharp \ElD \interp P)[\sharp \varsigma][\iota] \prop}{}.
\end{equation}
We have $\interp \Gamma, \var p : (\sharp \ElD \interp P)[\sharp \varsigma][\iota] \sez \interp e : \forsub \sharp m \leq \forsub \sharp n$, and we need to prove
\begin{equation}
	\interp \Gamma \sez \ldots : \forsub \sharp \paren{ \sfillsys{\sfillsysclause{\var p : (\sharp \ElD{\interp P})[\iota]}{\interp m}} } \leq \forsub \sharp \paren{ \sfillsys{\sfillsysclause{\var p : (\sharp \ElD{\interp P})[\iota]}{\interp n}} }.
\end{equation}
Now, precisely in those cases where the $\sfill$s evaluate to $\interp m$ and $\interp n$ respectively, we have evidence that $\forsub \sharp \interp m \leq \forsub \sharp \interp n$. This allows us to construct the conclusion.

\subsubsection{The $\sfix$ rule}
The fix rule
\begin{equation*}
	\inference{
		\Gamma, \ctxvar \nu n {\El~\Size} \judty{\El~A} \\
		\Gamma \judtm{f}{\El~\prodvar \nu n \Size.(\prodvar \nu m \Size.(\ssuc\,m \leq n) \to A[\varclr m / \varclr n]) \to A}
	}{\Gamma \judtm{\sfix^\nu\,f}{\El~\prodvar \nu n \Size.A}}{t-fix}
\end{equation*}
is interpreted as
\begin{equation}
	\inference{
		(\nu, \beta : \nu \to \sharp) \in \accol{(\coshp, \iota \vartheta), (\Id, \iota), (\sharp, \id)} \\
		\interp \Gamma, \ftrvar \nu n : \nu \Size \sez \interp{\El~A}[\iota] \dtype \\
		\interp \Gamma \sez \interp f : \Pi(\ftrvar \nu n : \nu \Size).(\Pi(\ftrvar \nu m : \nu \Size).(\beta(\ftrvar \nu m) \leq \beta(\ftrvar \nu n)) \to \interp{\El~A}[\iota][\ftrvar \nu m/\ftrvar \nu n]) \to \interp{\El~A}[\iota]
	}{\interp \Gamma \sez \fix^\nu \interp f : \Pi(\ftrvar \nu n : \nu \Size).\interp{\El~A}[\iota]}{}.
\end{equation}
Since the model supports the definitional version of the equality axiom for $\fix$, the axiom itself can be interpreted as an instance of reflexivity.

\part{A Presheaf Model of Dependent Type Theory with Degrees of Relatedness}\label{part:reldtt}

\chapter{Finite-depth cubical sets} \label{ch:dcubecat}
In our paper \cite{reldtt}, we present a multi-mode type theory, where we assign every type and every judgement a depth $n$ where $n+1 \in \IN$ is the number of relations available in the type, and we assign every dependency a modality which suits the depth of its domain and codomain.

In this chapter, we build the presheaf category $\widehat{\dcubecat n}$ of depth $n$ cubical sets as a model for the depth $n$ fragment of our type theory, and for every modality $\mu : m \to n$ available in the type theory, we build a CwF morphism $F : \widehat{\dcubecat m} \to \widehat{\dcubecat n}$ between the corresponding presheaf categories.

The left division operation in the type system is modelled by another CwF morphism that is left adjoint to the modality by which we divide.

Because we use BCM-style operators in the paper \cite{moulin-param3,moulin}, we need the separated product to model context extension with an interval variable; a concept which only exists for cubical sets that have no operation for taking the diagonal. Because we want left division to preserve the structure of contexts, we need to show that the functors interpreting left divisions, preserve the separated product.

\section{The category of cubes of depth $n$}
In \cref{sec:cube}, we defined the category of cubes $\cubecat$, whose objects were cubes with just one flavour of dimensions. In \cref{sec:bpcube}, we defined bridge/path cubes, which had two flavours: bridge dimensions and path dimensions. Here we define the category $\dcubecat n$ of \textdef{cubes of depth $n$}, where $n+1 \in \IN$. A cube of depth $n$ has $n+1$ flavours of dimensions, called 0-bridges up to $n$-bridges, where 0-bridges will also be called paths.

The definition is self-evident: a cube of depth $n$ is a set of variables $W_0 \subseteq \aleph$ equipped with a function $W_0 \to \accol{0, \ldots, n}$ that assigns a flavour to every dimension. However, we denote it in a more type-theoretic style as
\begin{equation}
	W = (\var i_1 : \Idim{n_1}, \ldots, \var i_k : \Idim{n_k}).
\end{equation}
In particular, since there can only be a function $W \to \accol{}$ if $W = \eset$, the only cube of depth $-1$ is the 0-dimensional cube $()$.

A face map $\vfi : \PSub V W$ assigns to every variable $(\var i : \Idim n) \in W$ a variable $(\var i \psub \vfi : \Idim m) \in V$ such that $m \geq n$. Hence, when $m \geq n$, we can think of $\Idim m$ as a \emph{sub}type of $\Idim n$.

Note that $\pointcat \cong \dcubecat{-1}$, $\cubecat \cong \dcubecat 0$ and $\bpcubecat \cong \dcubecat 1$.

\section{Reshuffling functors}\label{sec:reshuffling-functors}
A \textdef{cubical set of depth $n$} is a presheaf $\Gamma$ over $\dcubecat n$. They are collected in the presheaf category $\widehat{\dcubecat n}$. We can think of $\Gamma$ as a set equipped with $n+1$ proof-relevant relations, as well as, of course, the equality relation. In intuitive discussions, we will write $\rlookup i \Gamma$ to refer to the $i$-bridge relation and $\rlookup \eqty \Gamma$ to refer to the equality relation on $\Gamma$. We will also write $\Gamma = \rellist{\rlookup \eqty \Gamma}{\rlookup 0 \Gamma, \ldots, \rlookup n \Gamma}$. Note that equality implies path-connectedness due to the existence of the face map $(\var i/\novar) : (\var i : \Idim 0) \to ()$, and that $m$-bridge-connectedness implies $n$-bridge-connectedness when $m \leq n$ due to the existence of the face map $(\var i^{\Idim n} / \var i^{\Idim m}) : (\var i : \Idim n) \to (\var i : \Idim m)$. We will denote this intuitively as $(\rlookup \eqty \Gamma) \subseteq \rlookup 0 \Gamma \subseteq \ldots \rlookup n \Gamma$.

In this section, we construct a family of functors $F : \widehat{\dcubecat m} \to \widehat{\dcubecat n}$ called \textdef{reshuffling functors} that essentially reshuffle the relations that a presheaf $\Gamma$ consists of. This family will contain, among others, the functors $\cohpi, \cohdisc, \cohfget, \cohcodisc, \cohpaths, \shp, \flat, \sharp$ and $\coshp$.

\subsection{Reshuffling functors on cubical sets, intuitively}\label{sec:reshuffling-functors-on-cubical-sets-intuitively}
A reshuffling functor $F : \widehat{\dcubecat m} \to \widehat{\dcubecat n}$ is specified up to isomorphism by a monotonically increasing function
\begin{equation}
	F : \accol{\eqty \leq 0 \leq \ldots \leq n} \to \accol{\eqty \leq 0 \leq \ldots \leq m \leq \top} : k \mapsto \rlookup k F.
\end{equation}
and will also be denoted as $F = \reshlist{\rlookup \eqty F}{\rlookup 0 F, \ldots, \rlookup n F}$.
Its action on a presheaf is defined intuitively by $\rlookup{i}{F \Gamma} = \rlookup{(\rlookup i F)}{\Gamma}$, where $\rlookup \top \Gamma$ is understood to be the total relation `true' on $\Gamma$. Composition is then given by $i \cdot (F \circ G) = (i \cdot F) \cdot G$, where $\top \cdot G$ is understood to mean $\top$. If $\rlookup \eqty F \neq (\eqty)$, this means that $F$ will identify all $(\rlookup \eqty F)$-bridge-connected points in $\Gamma$.

For example:
\begin{equation}
	\reshlist{0}{1, 3, \top}\Gamma = \rellist{ \rlookup 0 \Gamma }{ \rlookup 1 \Gamma, \rlookup 3 \Gamma, \rlookup \top \Gamma}.
\end{equation}

\begin{thesis}\label{thm:reshuffle-adjoint}
	A pair of reshuffling functors $L : \widehat{\dcubecat n} \to \widehat{\dcubecat m}$ and $R : \widehat{\dcubecat m} \to \widehat{\dcubecat n}$ is adjoint ($L \dashv R$) if and only if
	\begin{equation}
		\forall i \in \accol{\eqty, 0, \ldots, n, \top}, j \in \accol{\eqty, 0, \ldots, m, \top}. i \leq j \cdot L \Leftrightarrow i \cdot R \leq j.
	\end{equation}
	By consequence:
	\begin{itemize}
		\item A reshuffling functor $R$ has a left adjoint $L$ if and only if $(\eqty \cdot R) = (\eqty)$, where $j \cdot L$ is the greatest index $i \in \accol{\eqty, 0, \ldots, n, \top}$ such that $i \cdot R \leq j$ (even when $j = \top$).
		\item A reshuffling functor $L$ has a right adjoint $R$ if and only if $j \cdot L < \top$ for all $j < \top$, where $i \cdot R$ is the least index $j \in \accol{\eqty, 0, \ldots, m, \top}$ such that $i \leq j \cdot L$ (even when $i = \top$).
	\end{itemize}
\end{thesis}
\begin{example}
	For example, we have the following chain of adjunctions between $\widehat{\dcubecat 3}$ and $\widehat{\dcubecat 2}$:
	\begin{equation}
		\reshlist{0}{0, 2, 3} \dashv
		\reshlist{\eqty}{\eqty, 1, 1, 2} \dashv
		\reshlist{\eqty}{1, 1, 3} \dashv
		\reshlist{\eqty}{0, 0, 2, 2} \dashv
		\reshlist{\eqty}{0, 2, 2} \dashv
		\reshlist{\eqty}{0, 1, 1, \top}.
	\end{equation}
	To see this in action, consider the following table. In every row, we have a functor $L$, which we apply to a general object $S$ of the domain of $L$ to obtain an object $LS$ of its codomain. The $\mapsto$ arrow indicates which is which. In the row below, we apply the right adjoint $R$, which of course has domain and codomain swapped, to a general object $T$. One observes that the presheaf maps $LS \to T$ correspond to the presheaf maps $S \to RT$. Recall that the relations are always ordered from strict to liberal.
	\begin{equation}
		\begin{array}{c | c c c}
			\text{functor} & \in \widehat{\dcubecat 3} & & \in \widehat{\dcubecat 2} \\ \hline
			\reshlist{0}{0, 2, 3} & \rellist{A_\eqty}{A_0, A_1, A_2, A_3} & \mapsto & \rellist{A_0}{A_0, A_2, A_3} \\
			\bot & \downarrow & \cong & \downarrow \\
			\reshlist{\eqty}{\eqty, 1, 1, 2} & \rellist{B_\eqty}{B_\eqty, B_1, B_1, B_2} & \mapsfrom & \rellist{B_\eqty}{B_0, B_1, B_2} \\
			\bot & \downarrow & \cong & \downarrow \\
			\reshlist{\eqty}{1, 1, 3} & \rellist{C_\eqty}{C_0, C_1, C_2, C_3} & \mapsto & \rellist{C_\eqty}{C_1, C_1, C_3} \\
			\bot & \downarrow & \cong & \downarrow \\
			\reshlist{\eqty}{0, 0, 2, 2} & \rellist{D_\eqty}{D_0, D_0, D_2, D_2} & \mapsfrom & \rellist{D_\eqty}{D_0, D_1, D_2} \\
			\bot & \downarrow & \cong & \downarrow \\
			\reshlist{\eqty}{0, 2, 2} & \rellist{E_\eqty}{E_0, E_1, E_2, E_3} & \mapsto & \rellist{E_\eqty}{E_0, E_2, E_2} \\
			\bot & \downarrow & \cong & \downarrow \\
			\reshlist{\eqty}{0, 1, 1, \top} & \rellist{F_\eqty}{F_0, F_1, F_1, \top} & \mapsfrom & \rellist{F_\eqty}{F_0, F_1, F_2}
		\end{array}
	\end{equation}
	For example, in the first adjunction, on the left we see that:
	\begin{itemize}
		\item[$=$] Equality in $A$ implies equality in $B$ (which follows from the condition for 0 since $A_= \subseteq A_0$).
		\item[0] A path in $A$ implies equality in $B$.
		\item[1] A 1-bridge in $A$ implies a 1-bridge in $B$ (which follows from the condition for 2 since $A_1 \subseteq A_2$).
		\item[2] A 2-bridge in $A$ implies a 1-bridge in $B$.
		\item[3] A 3-bridge in $A$ implies a 2-bridge in $B$.
	\end{itemize}
	On the right, we get:
	\begin{itemize}
		\item[$=$] A path in $A$ implies equality in $B$.
		\item[0] A path in $A$ implies a path in $B$ (which follows from the condition for $=$ since $B_= \subseteq B_0$).
		\item[1] A 2-bridge in $A$ implies a 1-bridge in $B$.
		\item[2] A 3-bridge in $A$ implies a 2-bridge in $B$.
	\end{itemize}
	So on both sides, we have the same non-redundant properties.
\end{example}
\begin{example}[$\widehat{\cubecat}$ and $\widehat{\bpcubecat}$] \label{eg:bpcube}
	The category of cubical sets $\widehat \cubecat$ is isomorphic to $\widehat{\dcubecat 0}$, and the category of bridge/path cubical sets $\widehat \bpcubecat$ is isomorphic to $\widehat{\dcubecat 1}$, where the paths are the 0-bridges and the bridges are the 1-bridges. The functors from \cref{part:paramdtt} are all reshuffling functors:
	\begin{equation}
		\begin{array}{l | c c c | l}
			\text{functor} & \in \widehat{\dcubecat 1} & & \in \widehat{\dcubecat 0} \\ \hline
			\cohpi = \reshlist{0}{1} & \rellist{A_\eqty}{A_0, A_1} & \mapsto & \rellist{A_0}{A_1} & \text{quotient by path relation} \\
			\bot & \downarrow & \cong & \downarrow \\
			\cohdisc = \reshlist{\eqty}{\eqty, 0} & \rellist{B_\eqty}{B_\eqty, B_0} & \mapsfrom & \rellist{B_\eqty}{B_0} & \text{add discrete paths} \\
			\bot & \downarrow & \cong & \downarrow \\
			\cohfget = \reshlist{\eqty}{1} & \rellist{C_\eqty}{C_0, C_1} & \mapsto & \rellist{C_\eqty}{C_1} & \text{forget paths} \\
			\bot & \downarrow & \cong & \downarrow \\
			\cohcodisc = \reshlist{\eqty}{0, 0} & \rellist{D_\eqty}{D_0, D_0} & \mapsfrom & \rellist{D_\eqty}{D_0} & \text{add codisc. paths / disc. bridges} \\
			\bot & \downarrow & \cong & \downarrow \\
			\cohpaths = \reshlist{\eqty}{0} & \rellist{E_\eqty}{E_0, E_1} & \mapsto & \rellist{E_\eqty}{E_0} & \text{forget bridges} \\
			\bot & \downarrow & \cong & \downarrow \\
			\cohshirr = \reshlist{\eqty}{0, \top} & \rellist{F_\eqty}{F_0, \top} & \mapsfrom & \rellist{F_\eqty}{F_0} & \text{add codiscrete bridges}
		\end{array}
	\end{equation}
	This yields 5 adjoint endofunctors on $\widehat{\dcubecat 1}$:
	\begin{equation}
		\begin{array}{l | c c c | l}
			\text{functor} & \in \widehat{\dcubecat 1} & & \in \widehat{\dcubecat 1} \\ \hline
			\shp = \cohdisc \cohpi = \reshlist{0}{0, 1} & \rellist{A_\eqty}{A_0, A_1} & \mapsto & \rellist{A_0}{A_0, A_1} & \text{quotient by path relation} \\
			\qquad \bot & \downarrow & \cong & \downarrow \\
			\flat = \cohdisc \cohfget = \reshlist{\eqty}{\eqty, 1} & \rellist{B_\eqty}{B_\eqty, B_1} & \mapsfrom & \rellist{B_\eqty}{B_0, B_1} & \text{discrete paths} \\
			\qquad \bot & \downarrow & \cong & \downarrow \\
			\sharp = \cohcodisc \cohfget = \reshlist{\eqty}{1, 1} & \rellist{C_\eqty}{C_0, C_1} & \mapsto & \rellist{C_\eqty}{C_1, C_1} & \text{codiscrete paths} \\
			\qquad \bot & \downarrow & \cong & \downarrow \\
			\coshp = \cohcodisc \cohpaths = \reshlist{\eqty}{0, 0} & \rellist{D_\eqty}{D_0, D_0} & \mapsfrom & \rellist{D_\eqty}{D_0, D_1} & \text{discrete bridges} \\
			\qquad \bot & \downarrow & \cong & \downarrow \\
			\shirr = \cohshirr \cohpaths = \reshlist{\eqty}{0, \top} & \rellist{E_\eqty}{E_0, E_1} & \mapsto & \rellist{E_\eqty}{E_0, \top} & \text{codiscrete bridges}
		\end{array}
	\end{equation}
\end{example}

\subsection{The 2-poset of reshuffles}
In this section, we construct the simplest setting in which we can formally phrase \cref{thm:reshuffle-adjoint}: a category whose objects are just cube depths and whose morphisms are formal reshuffling functor specifications. Then, we phrase the lemma and prove it.

\begin{definition}
	A \textdef{2-poset}, \textdef{(1, 2)-category} or \textdef{poset-enriched category} $\cat C$ consists of:
	\begin{itemize}
		\item A set of objects $\Obj(\cat C)$,
		\item For each $x, y \in \Obj(\cat C)$, a \emph{poset} of morphisms $\Hom(x, y)$,
		\item For each $x \in \Obj(\cat C)$, an identity morphism $\id_x \in \Hom(x, x)$,
		\item For each $x, y, z \in \Obj(\cat C)$, an \emph{increasing} map $\circ : \Hom(y, z) \times \Hom(x, y) \to \Hom(x, z)$,
	\end{itemize}
	such that $\circ$ is associative and unital with respect to $\id$.
\end{definition}

\begin{definition}
	In a 2-poset, we say that a morphism $L \in \Hom(x, y)$ is \textdef{left adjoint} to $R \in \Hom(y, x)$ ($L \dashv R$) if $\id_x \leq RL$ and $LR \leq \id_y$. In that case, we have $L = LRL$ and $R = RLR$.
\end{definition}

\begin{lemma}
	In a 2-poset, if $L \dashv R$ and $L' \dashv R$, then $L = L'$. The dual result also holds.
\end{lemma}
\begin{proof}
	We have
	\begin{equation}
		L \leq LRL' \leq L' \leq L'RL \leq L. \qedhere
	\end{equation}
\end{proof}

\begin{definition}\label{def:reshufflecat}
	We define the \textdef{2-poset of reshuffles} $\reshufflecat$ as follows:
	\begin{itemize}
		\item $\Obj(\reshufflecat) = \set{n}{n+1 \in \IN}$,
		\item $\Hom(m, n)$ is the set of \textdef{reshuffles} from $m$ to $n$, i.e. increasing functions
		\begin{equation}
			F : \accol{\eqty \leq 0 \leq \ldots \leq n} \to \accol{\eqty \leq 0 \leq \ldots \leq m \leq \top} : k \mapsto \rlookup k F.
		\end{equation}
		denoted as $F = \reshlist{\rlookup \eqty F}{\rlookup 0 F, \ldots, \rlookup n F}$.
		We say that $F \leq G$ whenever $\rlookup i F \leq \rlookup i G$ for all $i \in \accol{=, 0, \ldots, n}$.
		\item $\id_n = \reshlist{=}{0, \ldots, n}$.
		\item $G \circ F$ is defined by $\rlookup i {(G \circ F)} = \rlookup{(\rlookup i G)}{F}$ where we define $\top \cdot F = \top$.
	\end{itemize}
	
	Sometimes, we need to consider only reshuffles that have at least a certain number of left/right adjoints. The 2-posets that contain only those reshuffles (with the same order relation) will be denoted as $\reshufflecatllr$, where the number of white bullets on the left/right is the minimal number of left/right adjoints.
\end{definition}

\begin{lemma}\label{thm:reshuffle-adjoint-reshufflecat}
	A pair of reshuffles $L \in \Hom(n, m)$ and $R \in \Hom(m, n)$ is adjoint ($L \dashv R$) if and only if
	\begin{equation}
		\forall i \in \accol{\eqty, 0, \ldots, n, \top}, j \in \accol{\eqty, 0, \ldots, m, \top}. i \leq j \cdot L \Leftrightarrow i \cdot R \leq j.
	\end{equation}
	By consequence:
	\begin{itemize}
		\item A reshuffle $R$ has a left adjoint $L$ if and only if $(\eqty \cdot R) = (\eqty)$, where $j \cdot L$ is the greatest index $i \in \accol{\eqty, 0, \ldots, n, \top}$ such that $i \cdot R \leq j$ (even when $j = \top$).
		\item A reshuffle $L$ has a right adjoint $R$ if and only if $j \cdot L < \top$ for all $j < \top$, where $i \cdot R$ is the least index $j \in \accol{\eqty, 0, \ldots, m, \top}$ such that $i \leq j \cdot L$ (even when $i = \top$).
	\end{itemize}
\end{lemma}
\begin{proof}
	First, assume that we have reshuffles $L \in \Hom(m, n)$ and $R \in \Hom(n, m)$ satisfying
	\begin{equation}
		\forall i \in \accol{\eqty, 0, \ldots, m, \top}, j \in \accol{\eqty, 0, \ldots, n, \top}. i \leq j \cdot L \Leftrightarrow i \cdot R \leq j.
	\end{equation}
	We show that $L \dashv R$, i.e. that $\id_m \leq RL$ and $LR \leq \id_n$.
	
	To see the former, we need to show that for all $i$, we have $i \leq (i \cdot R) \cdot L$ which is equivalent to $i \cdot R \leq i \cdot R$ and hence true.
	
	To see the latter, we need to show that for all $j$, we have $(j \cdot L) \cdot R \leq j$ which is equivalent to $j \cdot L \leq j \cdot L$ and hence true.
	
	Conversely, assume that $L \dashv R$. Then we have
	\begin{align}
		i \leq j \cdot L &\Rightarrow i \cdot R \leq j \cdot LR \leq j, \\
		i \cdot R \leq j &\Rightarrow i \leq i \cdot RL \leq j \cdot L,
	\end{align}
	so all adjoint pairs satisfy the given equivalence.
	
	The equivalence is unsatisfiable if $(\eqty \cdot R) \neq (\eqty)$, because then $(\eqty \cdot L)$ is undefined.
	
	It is also unsatisfiable if $m \cdot L = \top$ or $(\eqty \cdot L) = \top$ because then we get $\top \cdot R < \top$ which contradicts the definition $\top \cdot R := \top$.
\end{proof}
\begin{corollary}\label{thm:count-adjoints}
	A reshuffle $F \in \Hom(m, n)$ with $m, n \geq 0$ has
	\begin{equation*}
		\begin{array}{c c c | l}
			\text{1 left adjoint} & \Leftrightarrow & (\eqty \cdot F) = (\eqty) & F = \reshlist{\eqty}{\ldots} \\
			\text{2 left adjoints} & \Leftrightarrow & (\eqty \cdot F) = (\eqty) \wedge 0 \cdot F \geq 0 & F = \reshlist{\eqty}{\geq 0, \ldots} \\
			\text{3 left adjoints} & \Leftrightarrow & (\eqty \cdot F) = (\eqty) \wedge 0 \cdot F = 0 & F = \reshlist{\eqty}{0, \ldots} \\
			\text{1 right adjoint} & \Leftrightarrow & n \cdot F < \top & F = \reshlist{\ldots}{\ldots, < \top} \\
			\text{2 right adjoints} & \Leftrightarrow & n \cdot F = m & F = \reshlist{\ldots}{\ldots, m} \\
		\end{array}
	\end{equation*}
	When either domain or codomain is $-1$ \emph{and the other is not}, we give exhaustive lists:
	\begin{equation*}
		\begin{array}{c | c | c}
			& \Hom(-1, n) & \Hom(m, -1) \\
			\hline
			\text{1 left adjoint} & \reshlist{\eqty}{\eqty, \ldots, \eqty, \top, \ldots, \top} & \reshlist{\eqty}{} \\
			\text{2 left adjoints} & \reshlist{\eqty}{\top, \ldots, \top} & \reshlist{\eqty}{} \\
			\text{3 left adjoints} & \reshlist{\eqty}{\top, \ldots, \top} & \text{none} \\
			\text{1 right adjoint} & \reshlist{\eqty}{\eqty, \ldots, \eqty} & \reshlist{< \top}{} \\
			\text{2 right adjoints} & \reshlist{\eqty}{\eqty, \ldots, \eqty} & \reshlist{m}{}
		\end{array}
	\end{equation*}
	The only reshuffles of type $\Hom(-1, -1)$ are $\reshlist{\eqty}{}$ and $\reshlist{\top}{}$; the former is the identity and hence self-adjoint, the latter has neither a left nor a right adjoint.
\end{corollary}
\begin{proof}
	The characterization of having a single adjoint comes immediately from \ref{thm:reshuffle-adjoint-reshufflecat}. In order to have $n+1$ adjoints, a reshuffle needs to have $n$ adjoints and the same needs to hold for its adjoint. Going through the characterization of adjoint reshuffles then inevitably leads to the data above.
\end{proof}

\subsection{Reshuffling functors on cubes}\label{sec:reshuffle-on-cubes}
In this section, we want to define reshuffling functors between categories of cubes. Let $F$ be a reshuffle with right adjoint $G$. We know that a defining substitution $\DSub{(\var i : \Idim n)}{\Gamma}$ expresses an $n$-bridge in $\Gamma$. However, what does a defining substitution $\DSub{F(\var i : \Idim n)}{\Delta}$ express? By using the `right adjoint' $G = \fpsh F$ (see \cref{thm:lifting-adjunctions} and \cref{thm:yoneda-and-left-adjoint}), we know that this corresponds to a defining substitution $\DSub{(\var i : \Idim n)}{G \Delta}$, which we know is an $(\rlookup n G)$-bridge in $\Delta$ (or, if $\rlookup n G = (\eqty)$, equality in $\Delta$). Hence, we should set $F(\var i : \Idim n) := (\var i : \Idim{\rlookup n G})$ (or, if $\rlookup n G = (\eqty)$, we should set $F(\var i : \Idim n) := ()$ so that $F$ identifies source and target of every $n$-bridge). This requires that $F$ has a right adjoint $G$, and that $G$ does not contain $\top$, because $\Idim{\top}$ is not an available dimension flavour.\footnote{Suppose we had cubes such as $(\var i : \Idim \top)$. Then a defining substitution $\DSub{(\var i : \Idim \top)}{\Gamma}$ would consist of a source and a target that satisfy `true'. This means that a defining substitution with domain $(\var i : \Idim \top)$ is completely determined by its source and target. If we want to enforce that condition for all presheaves, we have to move to a \emph{sheaf} model of type theory. Categories of sheaves can be given the structure of a CwF, but the naive implementation of the universe satisfies the sheaf condition in general only up to isomorphism \cite{sheaf-universes}, causing further complications. For this reason, we want to avoid having to use a sheaf model.} In other words, $F$ must have two further right adjoints, i.e. $n \cdot F = m$ by \cref{thm:count-adjoints}.
Since a 2-poset is a special case of a 2-category, with a unique morphism $F \to G$ whenever $F \leq G$, we can consider 2-functors from a 2-poset to a 2-category.
\begin{definition}\label{def:reshuffle-on-cubes}
	Let $\reshufflecatrr$ be the sub-2-poset of $\reshufflecat$ that only contains reshuffles that have 2 right adjoints in $\reshufflecat$ (a property which is closed under identity and substitution). We have a 2-functor $\cubecat : \reshufflecatrr \to \Cat$ defined as follows:
	\begin{itemize}
		\item The object $n$ is mapped to the category $\dcubecat n$.
		\item The reshuffle $F$ with right adjoint $G$ is mapped to a functor, also denoted $F$, which maps a cube $(W, \var i : \Idim{n})$ to
		\begin{itemize}
			\item $(FW, \var i : \Idim{\rlookup n G})$ if $\rlookup n G \neq (\eqty)$,
			\item $FW$ if $\rlookup n G = (\eqty)$.
		\end{itemize}
		and a face map $(\vfi, \var j^{\Idim m}/\var i^{\Idim n})$, where necessarily we have $m \geq n$, to
		\begin{itemize}
			\item $(F\vfi, \var j^{\Idim{\rlookup m G}}/\var i^{\Idim{\rlookup n G}})$ if $\rlookup m G \geq \rlookup n G > (\eqty)$,
			\item $(F\vfi, \var j^{\Idim{\rlookup m G}}/\novar)$ (or more accurately $F \vfi$ in case $\var j$ occurs again in $\vfi$) if $\rlookup m G > \rlookup n G = (\eqty)$,
			\item $F\vfi$ if $\rlookup m G = \rlookup n G = (\eqty)$.
		\end{itemize}
		Functors arising this way are called \textdef{reshuffling functors on cubes}.
		\item Whenever $F \leq F'$ (with $F \dashv G$ and $F' \dashv G'$, implying that $G' \leq G$), we associate to this a natural transformation $\iota : F \to F'$ which sends the cube $(W, \var i : \Idim{n})$ to the face map
		\begin{itemize}
			\item $(\iota W, \var i^{\Idim{\rlookup n G}}/\var i^{\Idim{\rlookup n {G'}}}) : (FW, \var i : \Idim{\rlookup n G}) \to (F'W, \var i : \Idim{\rlookup n {G'}})$ if $\rlookup n G \geq \rlookup n {G'} > (\eqty)$,
			\item $(\iota W, \var i^{\Idim{\rlookup n G}}/\novar) : (FW, \var i : \Idim{\rlookup n G}) \to F'W$ if $\rlookup n G > \rlookup n {G'} = (\eqty)$,
			\item $\iota W : FW \to F'W$ if $\rlookup n G = \rlookup n {G'} = (\eqty)$.
		\end{itemize}
		Natural transformations arising in this way are called \textdef{reshuffle casts of cubes}.
	\end{itemize}
	This is straightforwardly verified to be a well-defined 2-functor.
\end{definition}
Note that 2-functors in general preserve adjointness of morphisms; this is trivial from the characterization of adjointness using unit and co-unit. By consequence, if $F \dashv G$ in $\reshufflecat$, then we automatically know that $F \dashv G$ as reshuffling functors on cubes, with the reshuffle casts as their unit and co-unit.

The following lemma generalizes \cref{thm:uniqueness-of-nattrans-base}:
\begin{lemma}\label{thm:uniqueness-of-nattrans-base-reshuffle}
	Let $F, F' : \dcubecat m \to \dcubecat n$ be two reshuffling functors. Then all natural transformations $\nu : F \to F'$ are equal to a (the) reshuffle cast.
\end{lemma}
\begin{proof}
	Pick a natural transformation $\nu : F \to F'$. By the same reasoning as in the proof of \cref{thm:uniqueness-of-nattrans-base}, $\nu$ is completely determined and substitutes every variable of $FW'$ with the variable of $FW$ that goes by the same name.
	
	We show that $\nu = \iota$. It is sufficient to show that $F \leq F'$, since in that case $\iota$ exists and is therefore equal to $\nu$. Let $F \dashv G$ and $F' \dashv G'$, then we may equivalently show that $G' \leq G$. Since $G$ and $G'$ have a left adjoint, we know that $(\rlookup \eqty {G'}) = (\rlookup \eqty G) = (\eqty)$. Now pick $k \in \accol{0, \ldots, m}$. Then we get $\nu : F(\var i : \Idim k) \to F'(\var i : \Idim k)$. Several cases are possible:
	\begin{itemize}
		\item If $\rlookup k G = (\eqty)$, then $F(\var i : \Idim k) = ()$. Since $\nu$ substitutes variables for themselves, this must mean that $F'(\var i : \Idim k) = ()$ as well, i.e. $\rlookup k {G'} = (\eqty)$.
		
		\item If $\rlookup k {G'} = (\eqty)$, then necessarily it is less than or equal to $\rlookup k G$.
		
		\item If $\rlookup k G \neq (\eqty)$ and $\rlookup k {G'} \neq (\eqty)$, then we have $\nu = (\var i / \var i) : (\var i : \Idim{\rlookup k G}) \to (\var i : \Idim{\rlookup k {G'}})$, which is only possible if $\rlookup k {G'} \leq \rlookup k G$. \qedhere
	\end{itemize}
\end{proof}

\subsection{Reshuffling functors on cubical sets, formally}\label{sec:reshuffling-functors-on-cubical-sets-formally}
In this section, we try to define reshuffling functors between cubical set categories. The modalities available internally will be the reshuffles that have two left adjoints; we call these \textdef{right reshuffling functors}. In order to model left division, we need to consider the left adjoint of any modality (in the paper \cite{reldtt} this is called a contramodality), so we will also consider reshuffling functors that have a left and a right adjoint; we call these \textdef{central reshuffling functors}. However, the reshuffles that have no left adjoint, i.e. those that redefine the equality relation, have no role to play in the general modality structure. The reshuffle $\cohpi_0^n = \reshlist{0}{1, \ldots, n} : n \to n-1$ will also be needed as a functor on cubical sets, in order to reason about the universe and in order to model existential quantification, but we will treat it separately in \cref{sec:reldtt-cohpi}.

From \cref{sec:psh-transform}, we know that a functor $F : \catV \to \catW$ gives rise to three functors $\lpsh F \dashv \fpsh F \dashv \rpsh F$ where $\lpsh F, \rpsh F : \widehat \catV \to \widehat \catW$ and $\fpsh F : \widehat \catW \to \widehat \catV$. Here, the functor $\lpsh F$ (which we did not give a general construction for) extends $F$ in the sense that $\lpsh F \circ \yoneda \cong \yoneda \circ F$ (\cref{thm:yoneda-and-left-adjoint}). This functor is not necessarily a morphism of CwFs, but $\fpsh F$ is a strict morphism of CwFs (\cref{thm:fpsh-strict-cwf-morphism}) and $\rpsh F$ is a (possibly non-strict) morphism of CwFs (\cref{thm:rpsh-cwf-morphism}).

If we have $F \dashv G$, then we know that $\fpsh F \dashv \fpsh G$. By uniqueness of the adjoint, this gives us $\lpsh F \dashv \lpsh G \cong \fpsh F \dashv \fpsh G \cong \rpsh F \dashv \rpsh G$.
In other words, every chain of $n$ adjoint functors between two base categories, gives us a chain of $n+2$ adjoint functors between the presheaf categories, the first $n$ of which extend the original ones.

\begin{definition}\label{def:central-reshuffling-functor}
	We have a 2-functor $\widehat{\cubecat} : \reshufflecatlr \to \Cat$ defined as follows:
	\begin{itemize}
		\item The object $n$ is mapped to the category $\widehat{\dcubecat n}$.
		\item The reshuffle $F$ with left adjoint $L$ is mapped to $\fpsh L$, which we also denote $F$.
		Functors arising this way are called \textdef{central reshuffling functors on cubical sets}.
		\item Whenever $F \leq F'$ (with $L \dashv F$ and $L' \dashv F'$, implying that $L' \leq L$), we associate to this a natural transformation $\iota : F \to F'$ which is the lifting $\fpsh \iota$ of $\iota : L' \to L$.
		Natural transformations arising in this way are called (central) \textdef{reshuffle casts of cubes}.
	\end{itemize}
	This is straightforwardly verified to be a well-defined 2-functor.
\end{definition}
\begin{proposition}\label{thm:uniqueness-of-nattrans-psh-central-reshuffle}
	Let $F, F' : \widehat{\dcubecat m} \to \widehat{\dcubecat n}$ be two central reshuffling functors. Then all natural transformations $\nu : F \to F'$ are equal to a (the) reshuffle cast.
\end{proposition}
\begin{proof}
	This follows immediately from \cref{thm:unlift-nattrans} and \cref{thm:uniqueness-of-nattrans-base-reshuffle}.
\end{proof}
\begin{definition}\label{def:right-reshuffling-functor}
	We have a pseudofunctor (2-functor that respects identity and composition only up to coherent isomorphism) $\widehat{\cubecat} : \reshufflecatll \to \Cat$ defined as follows:
	\begin{itemize}
		\item The object $n$ is mapped to the category $\widehat{\dcubecat n}$.
		\item The reshuffle $F$ with left adjoints $K \dashv L \dashv F$ is mapped to $\rpsh K$, which we also denote $F$.
		Functors arising this way are called \textdef{right reshuffling functors on cubical sets}.
		\item Whenever $F \leq F'$ (with $K \dashv L \dashv F$ and $K' \dashv L' \dashv F'$, implying that $K \leq K'$), we associate to this a natural transformation $\iota : F \to F'$ which is the lifting $\rpsh \iota$ of $\iota : K \to K'$.
		Natural transformations arising in this way are called (right) \textdef{reshuffle casts of cubes}.
	\end{itemize}
	The following lemmas imply that this is a well-defined pseudofunctor.
\end{definition}
\begin{lemma}
	The operation defined in \cref{def:right-reshuffling-functor} preserves identity and composition up to (a priori possibly incoherent) isomorphism.
\end{lemma}
\begin{proof}
	Since both $\rpsh \id$ and $\Id$ are right adjoint to $\fpsh \id = \Id$, they are isomorphic.
	
	Since both $\rpsh{(G \circ F)}$ and $\rpsh G \circ \rpsh F$ are right adjoint to $\fpsh{(G \circ F)} = \fpsh F \circ \fpsh G$, they are isomorphic.
\end{proof}
\begin{proposition}\label{thm:uniqueness-of-nattrans-psh-right-reshuffle}
	Let $F, F' : \widehat{\dcubecat m} \to \widehat{\dcubecat n}$ be two right reshuffling functors. Then all natural transformations $\nu : F \to F'$ are equal to a (the) reshuffle cast.
\end{proposition}
\begin{proof}
	Note that, if $L \dashv R$, then natural transformations $G \to HL$ correspond to natural transformations $GR \to H$. Hence, if $K \dashv L \dashv F$ and $K' \dashv L' \dashv F'$, then we have
	\begin{equation}
		(F \to F') = (\rpsh K \to \rpsh{K'}) \cong (\Id \to \rpsh{K'} \fpsh K) \cong (\fpsh{K'} \to \fpsh K).
	\end{equation}
	Then the result follows from \cref{thm:uniqueness-of-nattrans-psh-central-reshuffle}.
\end{proof}
\begin{proposition}
	The restrictions of the 2-functor defined in \cref{def:central-reshuffling-functor} and the pseudofunctor defined in \cref{def:right-reshuffling-functor} to the 2-poset $\reshufflecatllr$, are isomorphic.
\end{proposition}
We will consider $F = \fpsh L$ a strict equality, and $F \cong \rpsh K$ merely an isomorphism, to avoid confusion when both are defined.
\begin{proof}
	Let $K \dashv L \dashv F \dashv R$. Then $\rpsh K \cong \fpsh L$ as both are right adjoint to $\fpsh K$. This isomorphism is natural, by uniqueness of reshuffle casts.
\end{proof}
We could similarly define left reshuffling functors; however we do not do so as we will only need $\cohpi_0^n : \widehat{\dcubecat n} \to \widehat{\dcubecat{n-1}}$.
\begin{theorem}[No maze theorem]\label{thm:no-maze}
	Under very general conditions on the reshuffling functors $F, G, H, K$, all natural transformations
	\begin{equation}
		\Hom(F \loch, G \loch) \to \Hom(H \loch, K \loch),
	\end{equation}
	are equal. A specification of the precise conditions would be less comprehensible than the proof below, from which the conditions can be inferred.
\end{theorem}
This theorem states that we often need not, as we construct a natural transformation, leave a notational trail of crumbs to express how we got where we are: all roads to the same destination are equal; our model is not a maze.
\begin{proof}
	We assume that $F$ has a right adjoint $R$; if $G$ has a left adjoint $L$, we can proceed in a similar fashion. Then $\Hom(F\loch, G\loch) \cong \Hom(\loch, RG\loch)$, so it suffices to consider natural transformations
	\begin{equation}
		\zeta : \Hom(\loch, G'\loch) \to \Hom(H \loch, K \loch)
	\end{equation}
	where $G' = RG$.
	
	Applying $\zeta$ to $\id : G' \Gamma \to G' \Gamma$, we get $\nu\Gamma := \zeta(\id) : H G' \Gamma \to K \Gamma$. Hence, we have a natural transformation $\nu : HG' \to K$. Moreover, $\nu$ determines $\zeta$; indeed, $\zeta(\sigma) = \zeta(\id \circ \sigma) = \zeta(\id) \circ HG' \sigma = \nu \circ HG'\sigma$.
	
	Thus, it suffices to show that all natural transformations $\nu : H' \to K$ are equal, where $H' = HG'$. We ultimately want to use uniqueness of central or right reshuffling casts. If $H'$ contains some central reshuffling functors on the outside, we can bring these to the right as right reshuffling functors. If $K$ contains some central reshuffling functors on the inside, we can use naturality to bring these to the left as right reshuffling functors.
\end{proof}
\begin{remark}[No maze principle]\label{remark:no-maze}
	We will use the no maze theorem as a legitimation for not bothering to prove equality of transformations that fundamentally do nothing and arose as compositions of reshuffle casts. The no maze theorem says that under very general conditions, such transformations are necessarily equal. But we stretch its interpretation, and conclude that even when other transformations are around, those transformations must be rather contrived, so when dealing with non-contrived transformations, we will assume equality without proof.
\end{remark}

\subsection{Some useful reshuffles}\label{sec:useful-reshuffles}
Before we proceed, we will give names to some useful reshuffles.

The reshuffle $\cohpi_{[0, \ell]}^m$ quotients out the relations at positions $0$ up to $\ell$ (or equivalently just the one at $\ell$) from a presheaf of depth $m$. We require $\ell \geq -1$ and $m \geq \ell$.
\begin{align}
	\cohpi_{[0, -1]}^m &:= \reshlist{\eqty}{0, \ldots, m} = \id : m \to m, & (\ell = -1), \\ \nn
	\cohpi_{[0, \ell]}^m &:= \reshlist{\ell}{\ell + 1, \ldots, m} : m \to m - (\ell + 1), & (\ell \geq 0).
\end{align}
The \textdef{discrete} reshuffle $\cohdisc_{[k+1, \ell]}^m$ inserts the strictest possible relations at positions $k+1$ up to $\ell$, obtaining a presheaf of depth $m$. We require $k+1 \geq 0$, $\ell \geq k$ (if $\ell = k$ then we obtain the identity reshuffle) and $m \geq \ell$.
\begin{align}
	\cohdisc_{[0, \ell]}^m &:= \reshlist{\eqty}{\eqty, \ldots, \eqty, 0, \ldots, m-(\ell + )} : m - (\ell + 1) \to m, & (k+1 = 0), \\ \nn
	\cohdisc_{[k+1, \ell]}^m &:= \reshlist{\eqty}{0, \ldots, k, k, \ldots, k, k+1, \ldots, m - (\ell - k)} : m - (\ell - k) \to m, & (k+1 > 0).
\end{align}
The \textdef{forgetful} reshuffle $\cohfget_{[k+1, \ell]}^m$ forgets the relations $k+1$ up to $\ell$ in a presheaf of depth $m$. It has the same requirements on $k, \ell$ and $m$.
\begin{align}
	\cohfget_{[k+1, \ell]}^m &:= \reshlist{\eqty}{0, \ldots, k, \ell+1, \ldots, m} : m \to m - (\ell - k).
\end{align}
The \textdef{codiscrete} reshuffle $\cohcodisc_{[k+1, \ell]}^m$ inserts the most liberal relations possible at positions $k+1$ up to $\ell$, obtaining a presheaf of depth $m$. It has the same requirements on $k, \ell$ and $m$.
\begin{align}
	\cohcodisc_{[k+1, \ell]}^m &:= \reshlist{\eqty}{0, \ldots, k, k+1, \ldots, k+1, k+1, \ldots, m - (\ell - k)} : m - (\ell - k) \to m, & (\ell < m), \\ \nn
	\cohcodisc_{[k+1, m]}^m &:= \reshlist{\eqty}{0, \ldots, k, \top, \ldots, \top} : k \to m, &(\ell = m).
\end{align}
Note that, if both exist, we have $\cohdisc_{[k+1, \ell]}^m = \cohcodisc_{[k, \ell-1]}^m$. For each of the functors that are indexed with an interval, if both ends of the interval are equal ($k+1 = \ell$), then we write e.g. $\cohdisc_\ell^m = \cohdisc_{[\ell, \ell]}^m : m - 1 \to m$. We have
\begin{align}
	\cohpi_{[0, \ell]}^m \dashv \cohdisc_{[0, \ell]}^m &\dashv \cohfget_{[0, \ell]}^m \dashv \cohcodisc_{[0, \ell]}^m, & (k+1 = 0), \\ \nn
	\ldots \dashv \cohcodisc_{[k, \ell-1]}^m = \cohdisc_{[k+1, \ell]}^m &\dashv \cohfget_{[k+1, \ell]}^m \dashv \cohcodisc_{[k+1, \ell]}^m, & (k+1 > 0).
\end{align}
Furthermore, the following equations hold:
\begin{equation}
	\cohpi_{[0, \ell]}^{m} \circ \cohdisc_{[0, \ell]}^{m} = \id, \qquad
	\cohfget_{[k+1, \ell]}^m \cohdisc_{[k+1, \ell]}^m = \id, \qquad
	\cohfget_{[k+1, \ell]}^m \cohcodisc_{[k+1, \ell]}^m = \id.
\end{equation}
\begin{example}
	In \cref{eg:bpcube} and \cref{part:paramdtt}, we had
	\begin{align}
		\cohpi \cong \cohpi_0^1 &= \reshlist{0}{1} : 1 \to 0 \\ \nn
		\dashv \cohdisc \cong \cohdisc_0^1 &= \reshlist{\eqty}{\eqty, 0} : 0 \to 1 \\ \nn
		\dashv \cohfget \cong \cohfget_0^1 &= \reshlist{\eqty}{1} : 1 \to 0 \\ \nn
		\dashv \cohcodisc \cong \cohcodisc_0^1 = \cohdisc_1^1 &= \reshlist{\eqty}{0, 0} : 0 \to 1 \\ \nn
		\dashv \cohpaths \cong \cohfget_1^1 &= \reshlist{\eqty}{0} : 1 \to 0 \\ \nn
		\dashv \cohshirr \cong \cohcodisc_1^1 &= \reshlist{\eqty}{0, \top} : 0 \to 1.
	\end{align}
	We have $\cohpi \cohdisc \cong \cohfget \cohdisc = \cohfget \cohcodisc = \cohpaths \cohcodisc = \cohpaths \cohshirr = \reshlist{\eqty}{0} = \Id$.
\end{example}

\chapter{Fibrancy}\label{ch:fibrancy}
In order to support the identity extension lemma of parametricity \cite{reynolds}, we need to be able to interpret dependent type theory in a model that uses only \emph{discrete} types. In this chapter, we show that basic type constructions preserve discreteness, and we also prove a few other intermediate results. Some results we prove not for discreteness specifically, but for general \emph{robust} notions of fibrancy, of which discreteness is an instance. The concept of robustness is to our knowledge novel.

\section{Robust notions of fibrancy}\label{sec:fibrancy}
\subsection{Robust notions of fibrancy}
\begin{definition}
	Let $\cat C$ be any category. Its \textdef{arrow category} $\Arr(\cat C)$ is the functor space $\cat C^{\accol{\bullet \to \bullet}}$, i.e. the objects are the morphisms of $\cat C$, and the morphisms are commutative squares.
\end{definition}
\begin{definition}
	Let $\cat C$ be a category. A \textdef{class of morphisms in $\cat C$} is any category $\cat I$ together with a functor $\cat I \to \Arr(\cat C)$ that is an isofibration\footnote{Meaning, essentially, that if a morphism belongs to a given class, then all morphisms isomorphic to it, also belong to that same class.}.
\end{definition}
\begin{definition}
	Let $\cat C$ be a category and $i : \cat I \to \Arr(\cat C)$ a class of morphisms. A \textdef{right lifting operation} $r$ for a morphism $\vfi : x \to y$ provides a diagonal for every diagram of the form
	\begin{equation}
		\xymatrix{
			a \ar[r] \ar[d]_{i(\eta)}
			& x \ar[d]^{\vfi}
			\\
			b \ar[r] \ar@{.>}[ru]_{r(\eta)}
			& y
		}
	\end{equation}
	naturally in $\eta \in \cat I$. The category of morphisms $\vfi$ equipped with right lifting operations, forms a new class of morphisms $R(\cat I)$ (where the morphism-squares are defined in the obvious way).
	
	We dually define \textdef{left lifting operations} and the class $L(\cat I)$.
\end{definition}
\begin{definition}
	A class of morphisms $i : \cat I \to \cat C$ \textdef{has pullbacks} if, for every $\eta \in \cat I$ such that $i(\eta) : a \to b$ and every $\vfi : b' \to b$, there is a morphism $\eta[\vfi] \to \eta$ that is mapped to a pullback square
	\begin{equation}
		\xymatrix{
			a[\vfi] \ar[r] \ar[d]_{i(\eta[\vfi])} \pullbackcorner
			& a \ar[d]^{i(\eta)}
			\\
			b' \ar[r]_\vfi
			& b,
		}
	\end{equation}
	naturally in $b'$.
\end{definition}
\begin{definition}
	A class of morphisms is \textdef{robust} if it is of the form $R(\cat I)$, where $\cat I$ has pullbacks.
\end{definition}

\subsection{A model with only fibrant types}
In this section, we will show that if we restrict the general presheaf CwF $\widehat \catW$ to those types $T$ for which $\pi : \Gamma.T \to \Gamma$ belongs to a fixed robust class of morphisms $R(\cat I)$, called \textdef{fibrations}, then we obtain again a CwF $\pshfib \catW$. Morphisms of the class $\cat I \to \widehat \catW$ will be called \textdef{horn inclusions}.

\subsubsection{The category with families $\pshfib \catW$}
\begin{lemma}
	$\pshfib \catW$ is a well-defined category with families (see \cref{def:cwf}).
\end{lemma}
\begin{proof}
	The only thing we need to prove in order to show this is that $\Ty$ still has a morphism part, i.e. that fibrancy is preserved under substitution.

	So pick a substitution $\sigma : \Theta \to \Gamma$ and a fibrant type $\Gamma \sez T \ftype$. Pick a horn inclusion $\eta : \Lambda \to \Delta$. Then we need to find a diagonal the square on the left:
	\begin{equation}
		\xymatrix{
			\Lambda \ar[r] \ar[d]_\eta
			& \Theta.T[\sigma] \ar[r]^{\sigma \subext} \ar[d]_\pi \pullbackcorner
			& \Gamma.T \ar[d]_\pi
			\\
			\Delta \ar[r]\ar@{.>}@/_/[ru] \ar@{.>}@/_/[rru]
			& \Theta \ar[r]_\sigma
			& \Gamma.
		}
	\end{equation}
	However, we know that there is a diagonal $\Delta \to \Gamma.T$, by fibrancy of $T$. Since $\Theta.T[\sigma]$ is a pullback (as is easy to show), we get the required lifting.
\end{proof}

\subsubsection{Dependent sums}
\begin{lemma}
	The category with families $\pshfib \catW$ supports dependent sums.
\end{lemma}
\begin{proof}
	One can easily show that $\Gamma.\Sigma A B \cong \Gamma.A.B$, where $\Gamma.\Sigma A B \xrightarrow \pi \Gamma$ corresponds to $\Gamma.A.B \xrightarrow \pi \Gamma.A \xrightarrow \pi \Gamma$. Since a composition of fibrations is again a fibration (as the lifting operations can be composed), and isomorphisms are fibrations, we can conclude that $\Sigma A B$ is fibrant if $A$ and $B$ are.
\end{proof}

\subsubsection{Dependent products}
\begin{lemma}
	The category with families $\pshfib \catW$ supports dependent products.
\end{lemma}
In fact, the proof of this lemma proves something stronger:
\begin{lemma}\label{thm:disc-prod-fib}
	Given a context $\Gamma$, an arbitrary type $\Gamma \sez A \type$ and a fibrant type $\Gamma.A \sez B \ftype$, the type $\Pi A B$ is fibrant.
\end{lemma}
\begin{proof}
	Pick a lifting problem
	\begin{equation}
		\xymatrix{
			\Lambda \ar[rr]^{(\sigma \eta, f)} \ar[d]_\eta
			&& \Gamma.\Pi A B \ar[d]^\pi
			\\
			\Delta \ar[rr]_{\sigma} \ar@{.>}[rru]^{(\sigma, g)}
			&& \Gamma.
		}
	\end{equation}
	Since $\Delta.(\Pi A B)[\sigma]$ is a pullback, it is equivalent to prove that we can lift $\id : \Delta \to \Delta$ to $\Delta \to \Delta.(\Pi A B)[\sigma]$, i.e. without loss of generality we may assume that $\sigma = \id$ and need to solve the lifting problem
	\begin{equation}
		\xymatrix{
			\Lambda \ar[rr]^{(\eta, f)} \ar[d]_\eta
			&& \Delta.\Pi A B \ar[d]^\pi
			\\
			\Delta \ar[rr]_{\id} \ar@{.>}[rru]^{(\id, g)}
			&& \Delta.
		}
	\end{equation}
	In other words, we are looking for a term $\Delta \sez g : \Pi A B$ such that $g[\eta] = f$, or equivalently a term $\Delta.A \sez b : B$ such that $b[\eta \subext] = \ap~f$. This boils down to solving the lifting problem
	\begin{equation}
		\xymatrix{
			\Lambda.A[\eta] \ar[rr]^{(\eta \subext, \ap~f)} \ar[d]_{\eta \subext}
			&& \Delta.A.B \ar[d]^\pi
			\\
			\Delta.A \ar[rr]_{\id} \ar@{.>}[rru]^{(\id, b)}
			&& \Delta.A.
		}
	\end{equation}
	But $\eta \subext$ is the pullback of the horn inclusion $\eta$ under $\pi : \Delta.A \to \Delta$. Hence, it is a horn inclusion as well, so the diagonal exists by fibrancy of $B$.
\end{proof}

\subsubsection{Identity types and propositions}
\begin{lemma}
	If all horn inclusions are surjective, then $\pshfib \catW$ supports identity types.
\end{lemma}
We even have a stronger result:
\begin{lemma}
	If all horn inclusions are surjective, then propositions are fibrant. \qed
\end{lemma}

\subsubsection{Glueing}
\begin{lemma}
	If all horn inclusions are surjective, the category with families $\pshfib \catW$ supports glueing.
\end{lemma}
\begin{proof}
	Suppose we have $\Gamma \sez A \ftype$, $\Gamma \sez P \prop$, $\Gamma.P \sez T \ftype$ and $\Gamma.P \sez f : T \to A[\pi]$. It suffices to show that $G = \Gluesys{A}{\Gluesysclauseb{P}{T}{f}}$ is fibrant. It is not hard to see that we have a pullback diagram
	\begin{equation}
		\xymatrix{
			\Gamma.G \ar[r] \ar[d] \pullbackcorner
			& \Gamma.A \ar[d]
			\\
			\Gamma.\Pi P T \ar[r]
			& \Gamma.(P \to A).
		}
	\end{equation}
	Now all three corners that determine the pullback, have a solution to any given lifting problem. Since surjective horn inclusions can be lifted in at most one way, these solutions are necessarily compatible. Then the pullback itself also has a solution.
\end{proof}
Unfortunately there are counterexamples for the same claim about welding.

\subsection{Fibrant replacement}
\begin{definition}
	A \textdef{fibrant replacement} of a type $\Gamma \sez T \type$ is a fibrant type $\Gamma \sez \fibrepl T \ftype$ together with a function $\Gamma \sez \infibrepl : T \to \fibrepl T$ such that every function $\Gamma \sez f : T \to S$ to a fibrant type $S$, factors as $f = g \circ \infibrepl$ in such a way that $g$ respects the lifting operation and is unique in doing so.
\end{definition}
Note that the definition implies that all fibrant replacements are isomorphic, and that every fibrant type is its own fibrant replacement.
\begin{proposition}\label{thm:fibrepl-subst}
	Assume that all horn inclusions are surjective.\footnote{This requirement is stronger than necessary.} If we have $\sigma : \Theta \to \Gamma$ and $\Gamma \sez T \type$ has a fibrant replacement $\Gamma \sez \fibrepl T \ftype$, then $(\fibrepl T)[\sigma]$ is a fibrant replacement of $T[\sigma]$.
\end{proposition}
\begin{proof}
	We define a type $\Gamma \sez A \type$ by setting $A \dsub \gamma = \sigma\inv(\gamma)$. Then $\Gamma.A \cong \Theta$, with $\pi : \Gamma.A \to \Gamma$ corresponding to $\sigma : \Theta \to \Gamma$. In other words, without loss of generality we assume that $\sigma$ is a weakening substitution $\pi : \Gamma.A \to \Gamma$ for a variable of a potentially non-fibrant type $A$.
	
	Now pick a fibrant type $\Gamma, \var x : A \sez S \ftype$ and a function $\Gamma, \var x : A \sez f : T[\wknvar x] \to S$. We show that $f$ factors uniquely over $\infibrepl[\wknvar x] : T[\wknvar x] \to (\fibrepl T)[\wknvar x]$. We do not have to show that the resulting function $g : (\fibrepl T)[\wknvar \pi]$ respects the lifting operation, as there is always at most one lifting operation. Clearly, swapping some arguments, we have $\Gamma \sez \lambda \var t . \lambda \var x . f[\wknvar t \subext]~(\var t [\wknvar x]) : T \to \Pi(\var x : A).S$. Moreover, $\Pi(\var x : A).S$ is fibrant, so we have
	\begin{equation}
		\lambda \var t . \lambda \var x . f[\wknvar t \subext]~(\var t [\wknvar x]) = h \circ \infibrepl
	\end{equation}
	for some unique $\Gamma \sez h : \fibrepl T \to \Pi(\var x : A).S$.
	
	Then we get $\Gamma, \var x : A \sez g := \lambda \var r . h[\wknvar r \wknvar x]~\var r~(\var x[\wknvar r]) : (\fibrepl T)[\wknvar x] \to S$. We now claim that
	\begin{equation}
		\paren{\lambda \var r . h[\wknvar r \wknvar x]~\var r~(\var x[\wknvar r])} \circ \paren{\infibrepl[\wknvar x]} = f.
	\end{equation}
	This is easily checked syntactically. Moreover, $g$ is the only such function, since $h$ is unique and $g$ was obtained from $h$ by an invertible manipulation.
	\begin{equation}
		\xymatrix{
			\Gamma.A.T[\pi] \ar[rr]^{\infibrepl[\pi]} \ar@/^{2.5em}/[rrrr]^f \ar[dd]_{\pi \subext} \ar[rd]^\pi
			&& \Gamma.A.(\fibrepl T)[\pi] \ar@{.>}[rr] \ar[dd]_{\pi \subext} \ar@{->>}[ld]_{\pi}
			&& \Gamma.A.S \ar@{->>}[llld]_\pi
			\\
			& \Gamma.A \ar[dd]_(.7){\pi}
			\\
			\Gamma.T \ar[rr]^(.6){\infibrepl} \ar@/^{2.5em}/[rrrr]^(.8){\lambda \var t . \lambda \var x . f[\wknvar t \subext]~(\var t [\wknvar x])} \ar[rd]^\pi
			&& \Gamma.\fibrepl T \ar@{.>}[rr]^h \ar@{->>}[ld]_{\pi}
			&& \Gamma.\Pi A S \ar@{->>}[llld]_\pi
			\\
			& \Gamma
		}
	\end{equation}
\end{proof}
\begin{proposition}
	If all horn inclusions are surjective, then every type $\Gamma \sez T \type$ has a fibrant replacement $\Gamma \sez \fibrepl T \ftype$, naturally in $\Gamma$ (i.e. the operation $\fibrepl$ commutes with substitution on the nose).
\end{proposition}
\begin{proof}
	As in \cref{sec:tyshp}, we define $\fibrepl T$ by quotienting out the least equivalence relation that makes $T$ fibrant. From \cref{thm:fibrepl-subst}, we know that $(\fibrepl T)[\sigma] \cong \fibrepl(T[\sigma])$. Equality follows because both are quotients of $T[\sigma]$ and the isomorphism preserves representants.
\end{proof}

\subsection{Fibrant co-replacement}
\begin{definition}
	A \textdef{fibrant co-replacement} of a type $\Gamma \sez T \type$ is a fibrant type $\Gamma \sez \fibcorepl T \ftype$ together with a function $\Gamma \sez \outfibcorepl : \fibcorepl T \to T$ such that every function $\Gamma \sez f : S \to T$ from a fibrant type $S$, factors as $f = \outfibcorepl \circ g$ in such a way that $g$ respects the lifting operation and is unique in doing so.
\end{definition}
\begin{corollary}
	Let $\Gamma \sez A, B \type$. Then functions $A \to \fibcorepl B$ are in 1-to-1-correspondence with functions $\fibrepl A \to B$. If both constructions exist in general, then we can say $\fibrepl \dashv \fibcorepl$. \qed
\end{corollary}

\subsection{Fibrancy and functors}
Assume we have a functor $F : \cat C \to \cat D$, where both $\cat C$ and $\cat D$ are equipped with a notion of horn inclusions.
\begin{proposition}\label{thm:preserve-fibrancy}
	If $F$ has a left adjoint $L$ that preserves horn inclusions, then $F$ preserves fibrancy.
\end{proposition}
\begin{proof}
	Pick a fibrant map $\tau : \Gamma' \to \Gamma$. We show that $F \tau$ is fibrant. Pick a lifting problem
	\begin{equation}
		\xymatrix{
			\Lambda \ar[r] \ar[d]_\eta
			& F\Gamma' \ar[d]^{F\tau}
			\\
			\Delta \ar[r]
			& F\Gamma.
		}
	\end{equation}
	This corresponds to a lifting problem
	\begin{equation}
		\xymatrix{
			L\Lambda \ar[r] \ar[d]_{L\eta}
			& \Gamma' \ar[d]^{\tau}
			\\
			L\Delta \ar[r]
			& \Gamma,
		}
	\end{equation}
	which has a solution that carries over to the original problem.
\end{proof}

\begin{proposition}\label{thm:partial-fib-repl}
	For fibrant types $T$ in the appropriate contexts, we have invertible rules:
	\begin{equation}
		\binference{
			\Gamma.\fibrepl S \sez t : T
		}{
			\Gamma.S \sez t[\subext \infibrepl] : T[\subext \infibrepl]
		}{}
		\qquad
		\binference{
			\fibrepl \Gamma \sez t : T
		}{
			\Gamma \sez t[\infibrepl] : T[\infibrepl]
		}{}
		\qquad
		\binference{
			(\fibrepl \Gamma).S \sez t : T
		}{
			\Gamma.S[\infibrepl] \sez t[\infibrepl \subext] : T[\infibrepl \subext]
		}{}
	\end{equation}
	and moreover, $\fibrepl \Gamma.\fibrepl T \cong \fibrepl(\Gamma.T[\infibrepl])$.
\end{proposition}
\begin{proof}
	For the first rule, note that terms $\Gamma.\fibrepl S \sez t : T$ correspond to functions $\Gamma \sez f : \fibrepl S \to \Sigma(\fibrepl S)T$ whose first projection is the identity. Since $\Sigma(\fibrepl S)T$ is fibrant, these correspond to functions $\Gamma \sez f' : S \to \Sigma(\fibrepl S)T$ whose first projection is $\infibrepl$, which in turn correspond to terms $\Gamma.S \sez t' : T[\infibrepl]$.
	
	The second rule is proven analogously or can be seen as a special case of the first rule, as contexts are essentially closed types.
	
	The third rule is proven by going via the fibrant type $\Pi S T$.
	
	We certainly have a map $\Gamma.T[\infibrepl] \to \fibrepl \Gamma.\fibrepl T$, which factors over $\infibrepl : \Gamma.T[\infibrepl] \to \fibrepl(\Gamma.T[\infibrepl])$ as the codomain is fibrant. The converse map is constructed in the same spirit as the rules proven above. They are inverses as both maps act on quotients and preserve representants.
\end{proof}

\section{Discreteness}
\subsection{Definition and characterization}
\begin{definition}
	In the CwF $\widehat{\dcubecat n}$, we say that a defining substitution $\gamma : \DSub{(W, \var i : \Idim k)}{\Gamma}$ or a defining term $(W, \var i : \Idim k) \Dsez t : T \dsub \gamma$ is \textdef{degenerate in $\var i$} if it factors over $(\facewkn{\var i}) : \PSub{(W, \var i : \Idim k)}{W}$.
\end{definition}
Thinking of $\var i$ as a variable, this means that $\gamma$ and $t$ do not refer to $\var i$. Thinking of $\var i$ as a dimension, this means that $\gamma$ and $t$ are flat in dimension $\var i$. Note that $t$ can only be degenerate in $\var i$ if $\gamma$ is.
\begin{corollary}
	For a defining substitution $\gamma : \DSub{(W, \var i : \Idim k)}{\Gamma}$ or a defining term $(W, \var i : \Idim k) \Dsez t : T \dsub \gamma$, the following are equivalent:
	\begin{enumerate}
		\item $\gamma$/$t$ is degenerate in $\var i$,
		\item $\gamma = \gamma \circ (0/\var i, \facewkn{\var i})$; $t = t \psub{0/\var i, \facewkn{\var i}}$,
		\item $\gamma = \gamma \circ (1/\var i, \facewkn{\var i})$; $t = t \psub{1/\var i, \facewkn{\var i}}$. \qed
	\end{enumerate}
\end{corollary}
\begin{definition}
	In the CwF $\widehat{\dcubecat n}$, where $n \geq 0$:
	\begin{itemize}
		\item We call a context \textbf{discrete} if all of its cubes are degenerate in every $\Idim 0$-dimension (also called path dimension).
	
		\item We call a map $\rho : \Gamma' \to \Gamma$ \textbf{discrete} if every defining substitution $\gamma$ of $\Gamma'$ is degenerate in every $\Idim 0$-dimension in which $\rho \circ \gamma$ is degenerate.
	
		\item We call a type $\Gamma \sez T \type$ \textbf{discrete} (denoted $\Gamma \sez T \dtype$) if every defining term $t : T \dsub \gamma$ is degenerate in every $\Idim 0$-dimension in which $\gamma$ is degenerate.
	\end{itemize}
	In the CwF $\widehat{\dcubecat{-1}}$, the presheaf structure does not provide a path relation.\footnote{Remember that the only primitive context is $()$, so that cubical sets of depth $-1$ are essentially sets.} We use the `true' relation instead of the path relation, leading to the following definitions:
	\begin{itemize}
		\item We call a context $\Gamma$ \textbf{discrete} if all defining substitutions $\DSub {()} \Gamma$ are equal.
		\item We call a map $\rho : \Gamma' \to \Gamma$ \textbf{discrete} if it is injective.
		\item We call a type $\Gamma \sez T \type$ \textbf{discrete} (denoted $\Gamma \sez T \dtype$) if, for every $\gamma : \DSub{()}{\Gamma}$, all defining terms $() \Dsez t : T \dsub \gamma$ are equal.
	\end{itemize}
	A more specific but unpractically long name for this definition would be `homogeneous 0-discreteness', as only the homogeneous 0-bridge relation is required to be discrete.
\end{definition}
\begin{proposition}
	A type $\Gamma \sez T \type$ is discrete if and only if $\pi : \Gamma.T \to \Gamma$ is discrete.
\end{proposition}
\begin{proof}
	We first treat the case where $n \geq 0$.
	\begin{itemize}
		\item[$\Rightarrow$] Assume that $T$ is discrete. Pick $(\gamma, t) : \DSub{(W, \var i : \Idim 0)}{\Gamma.T}$ such that $\pi \circ (\gamma, t) = \gamma$ is degenerate in $\var i$. Then $t$ is degenerate in $\var i$ by discreteness of $T$ and so is $(\gamma, t)$.
		\item[$\Leftarrow$] Assume that $\pi$ is discrete. Pick $t : T \dsub \gamma$ where $\gamma$ is degenerate in $\var i$. Then $(\gamma, t)$ is degenerate in $\var i$ since $\pi(\gamma, t) = \gamma$, and hence $t$ is degenerate in $\var i$.
	\end{itemize}
	Now let $n = -1$.
	\begin{itemize}
		\item[$\Rightarrow$] Assume that $T$ is discrete. Pick $(\gamma, t), (\gamma', t') : \DSub{()}{\Gamma.T}$ such that $\pi \circ (\gamma, t) = \pi \circ (\gamma', t')$, i.e. $\gamma = \gamma'$. Then $t = t'$ by discreteness of $T$, and hence $(\gamma, t) = (\gamma', t')$ so that $\pi$ is injective.
		\item[$\Leftarrow$] This is equally obvious. \qedhere
	\end{itemize}
\end{proof}
\begin{proposition}
	A context $\Gamma$ is discrete if and only if $\Gamma \to ()$ is discrete.
\end{proposition}
\begin{proof}
	If $n \geq 0$, it suffices to note that every defining substitution of $()$ is degenerate in every dimension. If $n = -1$, the claim is trivial.
\end{proof}

\begin{proposition}\label{thm:discreteness-yoneda-horns}
	For $n \geq 0$, a map $\rho : \Gamma' \to \Gamma$ is discrete if and only if it has the lifting property with respect to all maps $(\facewkn{\var i}) : \yoneda(W, \var i : \Idim 0) \to \yoneda W$.
	
	For $n = -1$, a map $\rho : \Gamma' \to \Gamma$ is discrete if and only if it has the lifting property with respect to the maps $(\id, \id) : \yoneda W \uplus \yoneda W \to \yoneda W$.\footnote{Recall that $W$ is necessarily $()$, such that $\yoneda W$ is the terminal presheaf; hence we might as well write $\Bool \to ()$.}
\end{proposition}
\begin{proof}
	We first treat the case where $n \geq 0$.
	\begin{itemize}
		\item[$\Rightarrow$] Suppose that $\rho$ is discrete and consider a square
		\begin{equation}
			\xymatrix{
				\yoneda(W, \var i : \Idim 0) \ar[r]^{\gamma'} \ar[d]_{(\facewkn{\var i})}
				& \Gamma' \ar[d]^{\rho}
				\\
				\yoneda W \ar[r]_{\gamma}
				&
				\Gamma.
			}
		\end{equation}
		Then the defining substitution $\rho \circ \gamma' : \DSub{(W, \var i : \Idim 0)}{\Gamma}$ clearly factors over $(\facewkn{\var i})$ so that it is degenerate in $\var i$. By degeneracy of $\rho$, the same holds for $\gamma'$, yielding the required diagonal.
		
		\item[$\Leftarrow$] Suppose that $\rho$ has the lifting property and take $\gamma' : \DSub{(W, \var i : \Idim 0)}{\Gamma'}$ such that $\rho \circ \gamma'$ is degenerate in $\var i$. This gives us a square as above, which has a diagonal, showing that $\gamma'$ is degenerate.
	\end{itemize}
	For the case where $n = -1$, it suffices to note that our presheaves are essentially sets and hence being injective is equivalent to lifting $\Bool \to ()$.
\end{proof}
\begin{definition}\label{def:dcubecat-horn-inclusion}
	Let $\cat I$ be the category whose objects are morphisms $\eta : \Lambda \to \Delta$ where $\Lambda \cong \Delta \times \yoneda(\var i : \Idim 0)$ and $\eta$ corresponds to $\pi_1$ under this isomorphism. The morphisms in $\cat I$ from $\eta' : \Lambda' \to \Delta'$ to $\eta : \Lambda \to \Delta$ are just morphisms $\sigma : \Delta \to \Delta'$, from which we can build a square by composing $\sigma \times \yoneda(\var i : \Idim 0)$ with the isomorphisms. For $n = -1$, we use $\Bool$ instead of $\yoneda(\var i : \Idim 0)$.
\end{definition}
\begin{corollary}\label{thm:dcubecat-lifting-property}
	A map $\rho : \Gamma' \to \Gamma$ is discrete if and only if it has the right lifting property with respect to all horn inclusions.
\end{corollary}
\begin{proof}
	The `if' part is trivial as the maps from \cref{thm:discreteness-yoneda-horns} are horn inclusions. To prove the `only if' part, pick a diagram
	\begin{equation}
		\xymatrix{
			\Theta \times \Upsilon \ar[r]^\sigma \ar[d]_{\pi_1}
			& \Gamma \ar[d]^\rho
			\\
			\Theta \ar[r]_{\sigma'}
			& \Gamma'
		}
	\end{equation}
	where $\pi_1$ is an arbitrary horn inclusion (note that any horn inclusion is, up to isomorphism, of this form) and $\rho$ is discrete. We need to show that $\sigma$ factors over $\pi_1$. When working at depth $-1$, then $\rho$ is injective, so $\sigma$ must factor over $\pi_1$ as $\rho \sigma$ does.
	
	Assume that we are working at depth $n \geq 0$. Then $\Upsilon = \yoneda(\var i : \Idim 0)$. Pick defining substitutions $(\theta, \vfi), (\theta, \vfi') : \DSub{W}{\Theta \times \yoneda(\var i : \Idim 0)}$ with equal first projections. We show that $\sigma(\theta, \vfi) = \sigma(\theta, \vfi')$. Without loss of generality, we may assume that $W = (V, \var j : \Idim 0, \var j ' : \Idim 0)$ and $\vfi = (\var j / \var i)$ and $\vfi' = (\var j'/\var i)$ and $\theta$ is degenerate in $\var j$ and $\var j'$. Indeed, precomposition with a single further primitive substitution gives us an arbitrary defining substitution.
	
	Then $\rho \sigma (\theta, \vfi) = \sigma' \theta$ is also degenerate in $\var j$ and $\var j'$, and so the same holds for $\sigma (\theta, \vfi) = \gamma (\var j / \novar, \var j'/\novar)$. But then we have:
	\begin{equation}
		\sigma (\theta, \vfi') = \sigma (\theta, \vfi) (\var j'/\var j, \var j/\var j') = \gamma (\var j / \novar, \var j'/\novar) (\var j'/\var j, \var j/\var j') = \gamma (\var j / \novar, \var j'/\novar) = \sigma (\theta, \vfi). \qedhere
	\end{equation}
\end{proof}

\subsection{A model with only discrete types}
\begin{proposition}\label{thm:disc-is-robust}
	Discreteness is robust. \qed
\end{proposition}
Hence, the results from \cref{sec:fibrancy} apply, including those that require horn inclusions to be surjective. Moreover, we can support welding:
\begin{lemma}
	The category with families $\ddisc n$ supports welding.
\end{lemma}
\begin{proof}
	Suppose we have $\Gamma \sez A \dtype$, $\Gamma \sez P \prop$, $\Gamma.P \sez T \dtype$ and $\Gamma.P \sez f : A[\pi] \to T$. It suffices to show that $\Omega = \Weldsys{A}{\Weldsysclauseb P T f}$ is discrete.
	
	First consider the case where $n \geq 0$. Pick $(W, \var i : \Idim 0) \Dsez w : \Omega \dsub \gamma$ where $\gamma$ is degenerate along $\var i$.
	\begin{itemize}
		\item If $P \dsub \gamma = \accol \star$, then $w \psub{0/\var i, \facewkn{\var i}}^\Omega = w \psub{0/\var i, \facewkn{\var i}}^{T[\id, \star]} = w$ by discreteness of $T$.
	
		\item If $P \dsub \gamma = \eset$, then $w \psub{0/\var i, \facewkn{\var i}}^\Omega = w \psub{0/\var i, \facewkn{\var i}}^{A} = w$ by discreteness of $A$.
	\end{itemize}
	
	Now consider the case where $n = -1$. Pick $W \Dsez w, w' : \Omega \dsub \gamma$.
	\begin{itemize}
		\item If $P \dsub \gamma = \accol \star$, then $\Omega \dsub \gamma = T \dsub{\gamma, \star}$ and $w = w'$ by discreteness of $T$.
	
		\item If $P \dsub \gamma = \eset$, then $\Omega \dsub \gamma = A \dsub \gamma$ and $w = w'$ by discreteness of $A$. \qedhere
	\end{itemize}
\end{proof}

\subsection{The functor $\cohpi_0$}\label{sec:reldtt-cohpi}
In \cref{sec:reshuffling-functors-on-cubical-sets-formally}, we defined central and right reshuffling functors on cubical sets, but not left ones. In other words, if $F : m \to n$ is a reshuffle with no left adjoint (because it redefines the equality relation), then we did not define a corresponding functor on cubical sets. This can be done in general if $F$ has two right adjoints, but we only need the special case where $F = \cohpi_0^n = \reshlist{0}{1, \ldots, n} : n \to n - 1$. This functor always has two right adjoints: $\cohpi_0^n \dashv \cohdisc_0^n \dashv \cohfget_0^n$.

In this section, we show that $\cohdisc_0^n \cohfget_0^n$ (which we could call $\flat_0^n$) is the discrete co-replacement functor on contexts. Hence, it is right adjoint to the discrete replacement functor $\quotshp$. Then $\cohfget_0^n \quotshp \dashv \cohdisc_0^n \cohfget_0^n \cohdisc_0^n \cong \cohdisc_0^n$, so we can define $\cohpi_0^n := \cohfget_0^n \quotshp : \widehat{\dcubecat{n}} \to \widehat{\dcubecat{n-1}}$ (where $n \geq 0$).

\begin{proposition}
	Let $n \geq 0$. Then the functor $\flat_0^n := \cohdisc_0^n \cohfget_0^n : \widehat{\dcubecat n} \to \widehat{\dcubecat n}$ is the discrete coreplacement functor.
\end{proposition}
\begin{proof}
	Pick a presheaf map $\sigma: \Theta \to \Gamma$, where $\Theta$ is discrete. We have to show that $\sigma$ factors uniquely over the reshuffle cast $\iota : \flat_0^n \Gamma \to \Gamma$.
	
	With a similar approach as in \cref{thm:discrete-contexts-and-cohesion}, we can show that $\iota : \flat_0^n \Theta \to \Theta$ is an isomorphism. Then by naturality of $\iota : \flat_0^n \to \Id$, we have $\sigma = \iota \circ (\flat \sigma \circ \iota\inv)$ which shows existence of the factorization.
	
	To show uniqueness, pick some $\tau, \tau' : \Theta \to \flat_0^n \Gamma$ such that $\iota \circ \tau = \iota \circ \tau' = \sigma : \Theta \to \Gamma$. Applying $\flat_0^n$, we find that $\flat_0^n \tau = \flat_0^n \tau' : \flat_0^n \Theta \to \flat_0^n \Gamma$. Postcomposing with $\iota : \flat_0^n \Gamma \to \Gamma$ and using naturality, we get $\tau \circ \iota = tau' \circ \iota : \flat_0^n \Theta \to \Gamma$. Precomposition with $\iota\inv : \Theta \cong \flat_0^n \Theta$ yields the desired result.
\end{proof}
\begin{definition}
	The functor $\cohpi_0^n : \widehat{\dcubecat{n}} \to \widehat{\dcubecat{n-1}}$ (where $n \geq 0$) is defined as $\cohpi_0^n := \cohfget_0^n \quotshp$, where $\quotshp$ is the discrete replacement functor on contexts.
\end{definition}
Then we know that $\cohpi_0^n = \cohfget_0^n \quotshp \dashv \flat_0^n \cohcodisc_0^n = \cohdisc_0^n \cohfget_0^n \cohcodisc_0^n = \cohdisc_0^n$ as intended.

\subsection{Discreteness and reshuffling functors}
\begin{proposition}
	A right reshuffling functor $H : \widehat{\dcubecat m} \to \widehat{\dcubecat n}$ preserves discreteness if and only if $m = -1$ or [$n \geq 0$ and $\rlookup 0 H = 0$].
\end{proposition}
\begin{remark}\label{remark:three-left-adjoints}
	Since $H$ is assumed to be a right reshuffling functor, we already know that the reshuffle $H$ has two left adjoints. The condition $m = -1 \vee (n \geq 0 \wedge 0 \cdot H = 0)$ is then equivalent to saying that $H$ has in fact three left adjoints (\cref{thm:count-adjoints}). In the paper \cite{reldtt}, right reshuffling functors are called modalities and central reshuffling functors are called contramodalities. So we claim that a modality preserves discreteness if and only if its left adjoint contramodality is again a modality. When a modality $\mu$ preserves discreteness, the $\mu$-modal $\Sigma$-type need not be quotiented and hence there exists a first projection, as is also observed in the paper.
\end{remark}
\begin{proof}
	We first show the `only if' part from the absurd. That is: we assume $m \geq 0$ and [$n = -1$ or $0 \cdot H > 0$].
	\begin{itemize}
		\item If $m \geq 0$ and $n = -1$, then $\Bool$ is discrete but $H \Bool$ is not, as it contains distinct points.
		
		\item If $m \geq 1$ and $n \geq 0$ and $\rlookup 0 H > 0$, then $\yoneda(\var i : \Idim 1)$ is discrete but $H \yoneda(\var i : \Idim 1)$ contains a non-trivial path from $(0/\var i)$ to $(1/\var i)$.
	
		\item If $m = 0$ and $n \geq 0$ and $\rlookup 0 H = \top$, then $\Bool$ is discrete but $H \Bool$ contains a non-trivial path from $\false$ to $\true$.
	\end{itemize}
	
	For the `if' part, by \cref{thm:preserve-fibrancy}, it suffices to show that the left adjoint $L \dashv H$ preserves horn inclusions. We know that $L$ respects products, so it suffices to show:
	\begin{description}
		\item[If $m = -1$ and $n = -1$] that $L \Bool \cong \Bool$, which is true since $L = \reshlist{\eqty}{} = \Id$,
		\item[If $m = -1$ and $n \geq 0$] that $L \yoneda(\var i : \Idim 0) \cong \Bool$, which is true since defining substitutions $\fpshadj K (\vfi) : \DSub{W = ()}{L \yoneda (\var i : \Idim)}$ (where $K \dashv L$) correspond to defining substitutions $\vfi : \DSub{K() = ()}{\yoneda(\var i : \Idim 0)}$ and these are the objects $(0/\var i)$ and $(1/\var i)$,
		\item[If $m \geq 0$ and $n \geq 0$] that $L \yoneda(\var i : \Idim 0) \cong \yoneda(\var i : \Idim 0)$. Tedious manipulations reveal that $L\yoneda(\var i : \Idim 0) \cong \yoneda(\var i : \Idim{0 \cdot H}) = \yoneda(\var i : \Idim 0)$. \qedhere
	\end{description}
\end{proof}
\begin{proposition}[Modal elimination of the discrete replacement]\label{thm:reldtt-modal-tyshp}
	For any right reshuffling functor $H$, we have a natural transformation on types $\iota : H \quotshp \to \quotshp H$ (where $\quotshp$ is the discrete replacement) such that $\iota \circ H \inquotshp = \inquotshp H$.
\end{proposition}
\begin{proof}
	Let $p$ be the greatest index such that $p \cdot H < \top$. Then we can decompose $H$ as $\cohcodisc_{[p+1, n]}^n G$, where $G = \reshlist{\eqty}{0 \cdot H, \ldots, p \cdot H} : m \to p$ has two left adjoints and a right one (\cref{thm:count-adjoints}). Then we can prove the theorem separately for $G$ and $\cohcodisc_{[p+1, n]}^n$, and then compose $\cohcodisc_{[p+1, n]}^n G \quotshp \to \cohcodisc_{[p+1, n]}^n \quotshp G \to \quotshp \cohcodisc_{[p+1, n]}^n G$. Also, we only construct the natural transformation; we invite the courageous reader to keep track of the commutativity property as we go.
	
	So we first construct $G \quotshp \to \quotshp G$. Let $R$ be right adjoint to $G$. Then it is sufficient to construct $\quotshp \to R \quotshp G$.  Having three left adjoints, $R$ preserves discreteness (\cref{remark:three-left-adjoints}), so it is sufficient to construct $\Id \to R \quotshp G$. We have $\iota : \Id \to RG$ and $R \inquotshp G : RG \to R \quotshp G$ (where $\inquotshp$ is $\infibrepl$ for discreteness).
	
	Then we construct $\cohcodisc_{[p+1, n]}^n \quotshp \to \quotshp \cohcodisc_{[p+1, n]}^n$.
	\begin{description}
		\item[$p = -1$] Note that the functor $\cohcodisc_{[0, n]}^n$ maps a set to a codiscrete depth $n$ cubical set. It is straightforward to show that $\cohcodisc_{[0, n]}^n$ preserves surjectivity of presheaf maps. Hence, $\cohcodisc_{[0, n]}^n \inquotshp : \cohcodisc_{[0, n]}^n \to \cohcodisc_{[0, n]}^n \quotshp$ is surjective. Meanwhile, $\quotshp_{-1}$ can be seen as the reshuffling functor for $\reshlist{\top}{}$ and quotients out the trivially true relation. Thus, if we have some $\gamma : \DSub W {\cohcodisc_{[0, n]}^n \quotshp \Gamma}$, then it is the image under $\cohcodisc_{[0, n]}^n \inquotshp$ of some $\gamma' : \DSub{W}{\cohcodisc_{[0, n]}^n \Gamma}$. Then $\inquotshp \circ \gamma'$ proves that $\DSub{W}{\quotshp \cohcodisc_{[0, n]}^n \Gamma}$ is inhabited --- hence a singleton --- and we can map $\gamma$ to $\inquotshp \circ \gamma'$.
		
		\item[$p \geq 0$] Then $\cohcodisc_{[p+1, n]}^n : p \to n$ has four left adjoints: $\cohfget_{[p, n-1]}^n \dashv \cohcodisc_{[p, n-1]}^n = \cohdisc_{[p+1, n]}^n \dashv \cohfget_{[p+1, n]}^n \dashv \cohcodisc_{[p+1, n]}^n$. Hence, $\cohfget_{[p+1, n]}^n$ and $\cohcodisc_{[p+1, n]}^n$ preserve discreteness.
		
		We first show that any cubical set of the form $\quotshp \cohcodisc_{[p+1, n]}^n \Gamma$ is $i$-codiscrete for all $i \geq p$, i.e. that all $i$-bridges in it are fully determined by their source and target. We do this by proving that the (separated/cartesian) exponential presheaf $(\yoneda(\var i : \Idim i) \to \quotshp \cohcodisc_{[p+1, n]}^n \Gamma)$ is isomorphic to $\quotshp \cohcodisc_{[p+1, n]}^n \Gamma \times \quotshp \cohcodisc_{[p+1, n]}^n \Gamma$. We have the following diagram, which we discuss below in detail:
		\begin{equation}
			\xymatrix{
				(\yoneda(\var i : \Idim i) \to \quotshp \cohcodisc_{[p+1, n]}^n \Gamma)
					\ar@{->>}[rr]_{(\loch \circ (0/\var i), \loch \circ (1 / \var i))}
				&& {\quotshp \cohcodisc_{[p+1, n]}^n \Gamma \times \quotshp \cohcodisc_{[p+1, n]}^n \Gamma}
				\\
				\\
				{\quotshp (\yoneda(\var i : \Idim i) \to \cohcodisc_{[p+1, n]}^n \Gamma)}
					\ar@{=}[rr]_{\quotshp (\loch \circ (0/\var i), \loch \circ (1 / \var i))} \ar@{~}@<.8ex>[rr]
					\ar@{->>}[uu]
				&&
				{\quotshp (\cohcodisc_{[p+1, n]}^n \Gamma \times \cohcodisc_{[p+1, n]}^n \Gamma)}
					\ar@{=}[uu] \ar@{~}@<.8ex>[uu]
				\\
				\\
				{(\yoneda(\var i : \Idim i) \to \cohcodisc_{[p+1, n]}^n \Gamma)}
					\ar@{=}[rr]_{(\loch \circ (0/\var i), \loch \circ (1 / \var i))} \ar@{~}@<.8ex>[rr]
					\ar@{->>}[uu]_{\inquotshp}
					\ar@{->>}@/^{8em}/[uuuu]^{\inquotshp \circ \loch}
				&&
				{\cohcodisc_{[p+1, n]}^n \Gamma \times \cohcodisc_{[p+1, n]}^n \Gamma}
					\ar@{->>}[uu]^{\inquotshp}
					\ar@{->>}@/_{8em}/[uuuu]_{\inquotshp \times \inquotshp}
			}
		\end{equation}
		On the left, we always consider presheaves of $i$-bridges, whereas on the right, we consider presheaves of pairs consisting of a source and a target. Below, we consider presheaves of $i$-bridges and pairs in $\cohcodisc_{[p+1, n]}^n \Gamma$; by construction of $\cohcodisc_{[p+1, n]}^n$, these presheaves are isomorphic. In the middle, we consider the discrete replacement of the same presheaves, that are then also isomorphic. On top, we consider presheaves of $i$-bridges and pairs in $\quotshp \cohcodisc_{[p+1, n]}^n \Gamma$; here, isomorphism is what we want to prove. What we already know is that we can take the source and target of an $i$-bridge; we will see later that this operation is surjective, i.e. any source and target are connected by at least one $i$-bridge.
		
		We get from the lower row to the middle by using the natural transformation $\inquotshp : \Id \to \quotshp$ which is always surjective. We can map $i$-bridges in $\cohcodisc_{[p+1, n]}^n \Gamma$ to $i$-bridges in $\quotshp \cohcodisc_{[p+1, n]}^n \Gamma$. This operation is \emph{not} obviously surjective, but it is because a defining substitution
		\begin{equation}
			\DSub{W}{(\yoneda(\var i : \Idim i) \to \Theta)}
		\end{equation}
		is the same thing as a substitution
		\begin{equation}
			(\yoneda W * \yoneda(\var i : \Idim i)) \to \Theta
		\end{equation}
		(where $*$ is either the separated or the cartesian product), which in turn is essentially a defining substitution $\DSub{(W, \var i : \Idim i)}{\Theta}$, and $\inquotshp$ is surjective on defining substitutions. Similarly, we can map pairs in $\cohcodisc_{[p+1, n]}^n \Gamma$ to pairs in $\quotshp \cohcodisc_{[p+1, n]}^n \Gamma$. That operation is surjective because a pair consists of two unrelated defining substitutions and $\inquotshp$ is surjective on defining substitutions.
		
		Because the presheaves in the top row are discrete (since separated/cartesian exponential presheaves are discrete if their codomain is, and since the cartesian product preserves discreteness), the curved arrows factor over the discrete replacement; a bit of diagram chasing shows that the factors are then also surjective. At this point, more diagram chasing also reveals that the upper horizontal arrow is indeed surjective.
		
	We claim that $\quotshp \Theta \times \quotshp \Phi$ is in general a discrete replacement of $\Theta \times \Phi$. This is shown by exploiting the fact that the cartesian product is left adjoint to the cartesian exponential and the fact that the cartesian exponential is discrete if its codomain is. Hence, the upper right vertical arrow is an isomorphism.
		
		Then finally, we can walk in a loop along the upper rectangle in the diagram, implying that all of its edges are isomorphisms.
		
		Then we can conclude that $i$-bridges in $\quotshp \cohcodisc_{[p+1, n]}^n \Gamma$ are fully determined by their source and target for all $i > p$ and hence that $\quotshp \cohcodisc_{[p+1, n]}^n \Gamma$ is, up to isomorphism, in the image of $\cohcodisc_{[p+1, n]}^n$. Because $\cohfget_{[p+1, n]}^n \cohcodisc_{[p+1, n]}^n \cong \Id$, we can then conclude that
		\begin{equation}
			\quotshp \cohcodisc_{[p+1, n]}^n \Gamma \cong \cohcodisc_{[p+1, n]}^n \cohfget_{[p+1, n]}^n \quotshp \cohcodisc_{[p+1, n]}^n \Gamma.
		\end{equation}
		Moreover, in all this reasoning, we did not break naturality in $\Gamma$, so we may conclude
		\begin{equation}
			\quotshp \cohcodisc_{[p+1, n]}^n \cong \cohcodisc_{[p+1, n]}^n \cohfget_{[p+1, n]}^n \quotshp \cohcodisc_{[p+1, n]}^n
		\end{equation}
		and it suffices to construct
		\begin{equation}
			\cohcodisc_{[p+1, n]}^n \quotshp \to \cohcodisc_{[p+1, n]}^n \cohfget_{[p+1, n]}^n \quotshp \cohcodisc_{[p+1, n]}^n.
		\end{equation}
		Diving under $\cohcodisc_{[p+1, n]}^n$ and using that $\cohfget_{[p+1, n]}^n$ preserves discreteness, we need $\Id \to \cohfget_{[p+1, n]}^n \quotshp \cohcodisc_{[p+1, n]}^n$, which is proven by composing $\iota : \Id \to \cohfget_{[p+1, n]}^n \cohcodisc_{[p+1, n]}^n$ with $\cohfget_{[p+1, n]}^n \inquotshp \cohcodisc_{[p+1, n]}^n : \cohfget_{[p+1, n]}^n \cohcodisc_{[p+1, n]}^n \to \cohfget_{[p+1, n]}^n \quotshp \cohcodisc_{[p+1, n]}^n$. \qedhere
	\end{description}
\end{proof}

\commentout{
}

\commentout{
}

\chapter{Semantics of RelDTT}\label{ch:reldtt}

In this chapter, we interpret the typing rules of the type system from our LICS '18 paper \cite{reldtt} --- which we call RelDTT for relational dependent type theory --- in the categories of depth $n$ cubical sets.

\section{Modalities}
We will use $\reshufflecatll$ as the 2-poset of internal modalities (see \cref{def:reshufflecat} for a formal definition, and \cref{sec:reshuffling-functors} (introduction) and \cref{sec:reshuffling-functors-on-cubical-sets-intuitively} for intuition). By \cref{thm:count-adjoints} it contains precisely those reshuffles $\mu : m \to n$ for which $(\rlookup{\eqty}{\mu}) = (\eqty)$ and (if applicable) $\rlookup 0 \mu \geq 0$. Since $(\eqty \cdot \mu)$ is fixed, in the paper \cite{reldtt}, we write
\begin{equation}
	\angles{0 \cdot \mu, \ldots, n \cdot \mu} \quad \text{for} \quad \reshlist{\eqty}{0 \cdot \mu, \ldots, n \cdot \mu}.
\end{equation}

The 2-poset of contramodalities is then given by $\reshufflecatlr$ and consists of precisely those reshuffles $\kappa : n \to m$ for which $(\eqty \cdot \kappa) = \eqty$ and all $i \cdot \kappa < \top$. In the paper \cite{reldtt}, we similarly omit the value of $(\eqty \cdot \kappa)$ and moreover write $\bot$ instead of $(\eqty)$.

Although the composition of a contramodality (central reshuffle) $\kappa : n \to m$ with a modality (right reshuffle) $\rho : p \to n$ is not necessarily in $\reshufflecatll$, the operation $\mu \circ \loch : (p \to m) \to (p \to n)$ on $\reshufflecatll$ has a left adjoint $\mu \setminus \loch : (p \to n) \to (p \to m)$. It is computed as follows (where $\kappa \dashv \mu$ in $\reshufflecat$):
\begin{itemize}
	\item $\rlookup{(\eqty)}{(\mu \setminus \nu)} = (\eqty)$,
	\item $\rlookup i {(\mu \setminus \nu)} = 0$ if $\rlookup{i}{\kappa \nu} = (\eqty)$,
	\item $\rlookup i {(\mu \setminus \nu)} = \rlookup{i}{\kappa \nu}$ if $\rlookup{i}{\kappa \nu} \geq 0$.
\end{itemize}

We will use the abbreviations introduced in \cref{sec:useful-reshuffles}. Some modalities play a remarkable role:
\begin{itemize}
	\item $\cohfget_0^{n+1} = \reshlist{\eqty}{1, 2, \ldots, n+1} : n+1 \to n$ is \textbf{parametricity}.
	\item $\cohcodisc_0^{n+1} = \reshlist{\eqty}{0, 0, 1, \ldots, n} : n \to n+1$ is \textbf{structurality}.
	\item $\cohcodisc_{[1, n]}^n \cohfget_{[1, m]}^m = \reshlist{\eqty}{0, \top, \ldots, \top} : m \to n$ is \textbf{shape-irrelevance}.
	\item $\cohcodisc_{[0, n]}^n \cohfget_{[0, m]}^m = \reshlist{\eqty}{\top, \ldots, \top} : m \to n$ is \textbf{irrelevance}.
	\item $\cohdisc_{[1, n]}^n \cohfget_{[1, m]}^m = \reshlist{\eqty}{0, \ldots, 0} : m \to n$ is \textbf{ad hoc polymorphism}.
\end{itemize}

\section{Judgements and the meta-type of the interpretation function}
The judgements take the following forms:
\begin{itemize}
	\item $\Gamma \sez_n \Ctx$, meaning $\interp \Gamma \in \widehat{\dcubecat n}$,
	\item $\Gamma \sez_n T \type$, meaning $\interp \Gamma \sez_{\widehat{\dcubecat n}} \interp T \dtype$,
	\item $\Gamma \sez_n t : T$, meaning $\interp \Gamma \sez_{\widehat{\dcubecat n}} \interp t : \interp T$,
	\item $\Gamma \sez_1 i : \IX$ ($i$ is an interval variable or constant), meaning $\interp \Gamma \sez_{\widehat{\dcubecat 1}} \interp i : \IX$, where $\IX = \yoneda(\var i : \Idim 1)$,
	\item $\Gamma \sez_{n+1} P : \IF^n$ ($P$ enodes a face proposition), meaning $\interp \Gamma \sez_{\widehat{\dcubecat{n+1}}} \interp P : \cohdisc_0^{n+1} \Prop$,
	\item $\Gamma \sez_n P \prop$ ($P$ is a face proposition), meaning $\interp \Gamma \sez_{\widehat{\dcubecat{n}}} \interp P \prop$.
\end{itemize}
Definitional equality is always interpreted as equality of mathematical objects.
In the model, we will write $\sez_n$ as shorthand for $\sez_{\widehat{\dcubecat n}}$.

We maintain the following invariant, which we will prove each time we introduce a way to derive $\Gamma \sez_n \Ctx$.
\begin{proposition}
	Whenever $\Gamma \sez_n \Ctx$ and $\kappa \dashv \mu$, we have $\interp{\mu \setminus \Gamma} \cong \kappa \interp \Gamma$.
\end{proposition}
In our proofs, we will often rely on the following lemma:
\begin{lemma}\label{thm:left-division-semantics}
	Let $\mu : m \to n$ and $\kappa \dashv \mu$. Then for any $\rho : p \to n$, we have $\kappa \circ \rho \leq \mu \setminus \rho$ and for any type $\Gamma \sez_p T \dtype$, the types $\kappa \rho \Gamma \sez_m \kappa \rho T, ((\mu \setminus \rho) T)[\iota] \type$ are isomorphic.
\end{lemma}
\begin{proof}
	The inequality $\kappa \circ \rho \leq \mu \setminus \rho$ follows immediately from the construction of $\loch \setminus \loch$. The only aspects of $\kappa \rho T$ and $((\mu \setminus \rho)T)[\iota]$ that are potentially different, are the $i$-bridges where $i \cdot \kappa = (\eqty)$. There the $i$-bridges of $\kappa \rho T$ are all degenerate, whereas those of $((\mu \setminus \rho) T)[\iota]$ are the 0-bridges of $T$. However, the substitution with $\iota : \kappa \rho \to \mu \setminus \rho$ ensures that only $i$-bridges remain that live over $i$-bridges in $\kappa \rho \Gamma$, which are all degenerate. In other words, we are considering 0-bridges in $T$ that live over degenerate bridges in $\Gamma$, but those are constant by discreteness of $T$. Hence, they are all constant.
	
	We make this precise. Let $\theta \dashv \kappa \dashv \mu$ and $\chi \dashv \psi \dashv \rho$. Pick a defining substitution $\fpshadj \theta(\rpshadj \chi(\sigma)) : \DSub{W}{\kappa \rho \Gamma}$. It arises from a substitution $\sigma : \psi \yoneda \theta W \to \Gamma$. Terms
	\begin{equation}
		W \Dsez \fpshadj \theta(\rpshadj \chi(t)) : (\kappa \rho T) \dsub{\fpshadj \theta(\rpshadj \chi(\sigma))}
	\end{equation}
	are just terms
	\begin{equation}
		\psi \yoneda \theta W \sez t : T[\sigma].
	\end{equation}
	Meanwhile, we have
	\begin{equation}
		\iota \circ \fpshadj \theta(\rpshadj \chi(\sigma)) = \rpshadj{\theta \setminus \chi}(\sigma \circ \iota) : \DSub{W}{(\mu \setminus \rho) \Gamma}
	\end{equation}
	where we denote the left adjoints to $\mu \setminus \rho$ using the symbols $\theta \setminus \chi \dashv \psi / \kappa \dashv \mu \setminus \rho$. The second no-maze-transformation $\iota$ has type $(\psi / \kappa)\yoneda \to \psi \yoneda \theta$.
	
	Terms
	\begin{equation}
		W \Dsez \rpshadj{\mu \setminus \rho}(t) : ((\mu \setminus \rho) T) \dsub{\iota \circ \fpshadj \theta(\rpshadj \chi(\sigma))}
	\end{equation}
	are then simply terms
	\begin{equation}
		(\psi/\kappa) \yoneda W \sez t : T [\sigma \circ \iota].
	\end{equation}
	So the question is: what is the difference between $(\psi/\kappa) \yoneda W$ and $\psi \yoneda \theta W \cong \psi \theta \yoneda W$?

	Note that $\psi$ and $(\psi / \kappa)$ preserve the cartesian product as they are lifted functors, and $\theta$ preserves it at least when applied under the Yoneda-embedding. So we can pull $W$ apart and consider the action of $\psi / \kappa$ and $\psi \circ \theta$ on individual variables. One can show that $\yoneda(\var i : \Idim i)$ is mapped (up to isomorphism) to $\yoneda(\var i : \Idim{i \cdot (\mu \setminus \rho)})$ and $\yoneda(\var i : \Idim{i \cdot (\kappa \circ \rho)})$, except when the reshuffles yield $\top$, then we end up with $\Bool$ instead.
	
	So the only case where there is a difference is when $i \cdot \kappa = (\eqty)$; in that case we get $\yoneda(\var i : \Idim 0)$ and $\yoneda()$ respectively. But then we are considering paths $t$ over $\sigma \circ \iota = \sigma \circ (\var i / \novar)$ which are degenerate by discreteness of $T$.
\end{proof}

\section{Contexts and variables}
\subsection{Empty context}
\begin{equation}
	\interpinference{n \geq -1}{\ctx_n}{c-em} =
	\inference{n \geq -1}{() \ctx}{}
\end{equation}
Clearly, if $\kappa \dashv \mu$, then $\kappa \interp{()} = \kappa () = () = \interp{()} = \interp{\mu \setminus ()}$.

\subsection{Context extension with a typed variable}
\begin{equation}
	\interpinference{
		\Gamma \sez_n \Ctx \\
		\mu : m \to n \\
		\mu \setminus \Gamma \sez_m T \type
	}{\Gamma, \ctxresh \mu x T \sez_n \Ctx}{c-ext}
	\vmess{
	=
	\inference{
		\interp \Gamma \in \widehat{\dcubecat n} \\
		\kappa \dashv \mu : m \to n \\
		\kappa \interp \Gamma \sez_m \interp T \dtype
	}{\interp \Gamma, \var x : (\mu \interp T) [\iota] \sez_n \Ctx}{}
	}
\end{equation}
If $\kappa' \dashv \mu' : m' \to n$, then
\begin{align*}
	\interp{\mu' \setminus (\Gamma, \ctxresh \mu x T)}
	&= \interp{(\mu' \setminus \Gamma), \ctxresh {\mu' \setminus \mu} x T} \\
	&= \interp{\mu' \setminus \Gamma}, \var x : ((\mu' \setminus \mu) \interp T)[\iota] \\
	&\cong (\kappa' \interp \Gamma), \var x : ((\mu' \setminus \mu) \interp T)[\iota].
\end{align*}
We have $(\mu' \setminus \mu) \kappa \interp \Gamma \sez (\mu' \setminus \mu) \interp T \dtype$. When we substitute with $\iota : \kappa' \circ \mu \to \mu' \setminus \mu$, by \cref{thm:left-division-semantics}, we get
\begin{equation}
	\kappa' \mu \kappa \interp \Gamma \sez ((\mu' \setminus \mu) \interp T) [\iota] \cong \kappa' \mu \interp T \dtype.
\end{equation}
Finally we substitute with $\kappa' \iota : \kappa' \to \kappa' \mu \kappa$ to find that our extended context is indeed isomorphic to
\begin{equation}
	\paren{\kappa' \interp \Gamma, \var x : (\kappa' \mu \interp T)[\kappa' \iota]} \cong \kappa' \paren{\interp \Gamma, \var x : (\mu \interp T)[\iota]}.
\end{equation}

\subsection{Cartesian context extension with an interval variable}
\begin{equation}
	\interpinference{
		\Gamma \sez_n \Ctx \\
		\mu : 1 \to n
	}{\Gamma, \ctxresh \mu i \IX \sez_n \Ctx}{}
	\vmess{
	= \inference{
		\interp \Gamma \ctx \\
		\mu : 1 \to n
	}{
		\interp \Gamma, \var i : \mu \IX \ctx
	}{}
	}
\end{equation}
The left division invariant is proven as for ordinary context extension.

\commentout{
}

\subsection{Context extension with a proof synthesis variable}
\begin{equation}
	\interpinference{
		\Gamma \sez_n \Ctx \\
		\mu : 0 \to n
	}{
		\Gamma, \ctxresh{\mu}{i}{\IJud}
	}{c-ext-j}
	\vmess{
	= \inference{
		\interp \Gamma \ctx \\
		\mu : 0 \to n
	}{
		\interp \Gamma, \var i : \mu \Bool \ctx
	}{}
	}
\end{equation}
The left division invariant is proven as before.

\subsection{Variable rules}
The variable rule
\begin{equation}
	\inference{
		\Gamma \sez_n \Ctx \\
		(\ctxresh \mu x T) \in \Gamma \\
		\mu \leq \idmod : n \to n
	}{\Gamma \sez_n x : T}{t-var}
\end{equation}
is interpreted using a combination of weakening and the following rule:
\begin{equation}
	\interpinference{
		\Gamma, \ctxresh{\mu}{x}{T} \sez_n \Ctx \\
		\mu \leq \idmod : n \to n
	}{
		\Gamma, \ctxresh{\mu}{x}{T} \sez_n x : T
	}{}
	\vmess{
		= \inference{
			\interp \Gamma, \var x : (\mu \interp T)[\iota] \ctx \\
			\iota : \mu \to \id : n \to n
		}{
			\interp \Gamma, \var x : (\mu \interp T)[\iota] \sez \iota(\var x) : \interp T[\iota][\wknvar x]
		}{}
	}
\end{equation}
The rule
\begin{equation}
	\inference{
		\Gamma \sez_1 \Ctx \\
		(\ctxresh \mu i \IX) \in \Gamma \\
		\mu \leq \idmod : 1 \to 1
	}{\Gamma \sez_1 i : \IX}{}
\end{equation}
is interpreted analogously.

\section{Universes}

\subsection{Shape-irrelevant universe polymorphism}
Since we assume that the metatheory has Grothendieck universes, the category $\widehat{\dcubecat n}$ --- as any presheaf category --- has a hierarchy of Hofmann-Streicher universes \cite{psh-universes}
\begin{equation}
	\uniPsh_0 \subseteq \uniPsh_1 \subseteq \ldots
\end{equation}
which are all closed types. However, this hierarchy does not classify discrete types. We will instead start from the hierarchy
\begin{equation}
	\uniNDD_0 \subseteq \uniNDD_1 \subseteq \ldots
\end{equation}
where $\uniNDD_\ell$ classifies discrete types of level $\ell$, i.e. the defining terms $W \Dsez T : \uniNDD_\ell$ are discrete types $\yoneda W \sez \dEl~T \dtype$.

These universes are not discrete, however; they are non-discrete universes of discrete types (whence the abbreviation NDD). Indeed, an $i$-bridge from $A$ to $B$ in $\uniNDD_\ell$ specifies what it means to give an $i$-bridge from $a : A$ to $b : B$. Instead, what we want is that all 0-bridges in the universe are constant, while the meaning of $i$-bridges from $a : A$ to $b : B$ is specified by an $(i+1)$-bridge from $A$ to $B$. We can achieve this by applying $\cohdisc_0^{n+1} = \reshlist{\eqty}{\eqty, 0, \ldots, n} : n \to n+1$ to $\uniPsh_\ell$, yielding a hierarchy
\begin{equation}
	\uniDD_0 := \cohdisc_0^{n+1} \uniNDD_0 \subseteq \uniDD_1 := \cohdisc_0^{n+1} \uniNDD_1 \subseteq \ldots
\end{equation}
of universes in $\widehat{\dcubecat{n+1}}$ classifying types from $\widehat{\dcubecat n}$.

However, this only gives us metatheoretic universe polymorphism, and we want internal shape-irrelevant universe polymorphism. For this reason, we need to construct
\begin{equation}
	\inference{
		\Gamma \sez_{n+1} \ell : \cohcodisc_{[0, n+1]}^{n+1} \cohfget_0^0 \IN
	}{
		\Gamma \sez_{n+1} \uniDD_{\ell} \dtype.
	}{}
\end{equation}
Recall that $\cohcodisc_{[0, n+1]}^{n+1} \cohfget_0^0 = \reshlist{\eqty}{\top, \ldots, \top} : 0 \to n+1$ is irrelevance; the polymorphic universe as a term should depend shape-irrelevantly on the level, but since we start by constructing it as a type, we need an irrelevant dependence.

So pick a defining substitution $\gamma : \DSub{W}{\Gamma}$. If $W$ is a $d$-dimensional cube, then $\ell \dsub {\gamma}$ specifies $2^d$ natural numbers and nothing else. Let $\min(\ell \dsub \gamma)$ be the least one among them. Then we take as defining terms
\begin{equation}
	W \Dsez_{n+1} T : \uniDD_{\ell} \dsub{\gamma}
\end{equation}
precisely the defining terms
\begin{equation}
	W \Dsez_{n+1} T : \uniDD_{\min(\ell \dsub \gamma)} \dsub{()}.
\end{equation}
Restriction can be imported from $\uniDD_{\min(\ell \dsub \gamma)}$. Indeed, if $\vfi : \PSub V W$, then $\min(\ell \dsub{\gamma \vfi}) \geq \min(\ell \dsub \gamma)$ so that $\uniDD_{\min(\ell \dsub \gamma)} \subseteq \uniDD_{\min(\ell \dsub {\gamma\vfi})}$ and hence $V \Dsez_{n+1} T \psub \vfi : \uniDD_{\min(\ell \dsub \gamma)} \subseteq \uniDD_{\min(\ell \dsub {\gamma\vfi})}$.

If $\Gamma \sez \ell : \cohcodisc_{[0, n]}^n \cohfget_0^0 \IN$, then we also introduce a notation $\Gamma \sez T \dtype_\ell$ which means that $T$ is a discrete type such that, for any $\gamma : \DSub W \Gamma$, the set $T \dsub \gamma$ has level $\min(\ell \dsub \gamma)$. For example, we have $\Gamma \sez \uniDD_\ell \dtype_{\suc~\ell}$.

We interpret:
\begin{equation*}
	\interpinference{
		n \geq -1 \\
		\Gamma \sez_{n+2} \Ctx \\
		\shirrmod \setminus \Gamma \sez_0 \ell : \IN
	}{
		\Gamma \sez_{n+2} \unidepth \ell n : \El~\unidepth{\suc~\ell}{n+1}
	}{t-Uni}
	=
	\inference{
	\inference{
	\inference{
	\inference{
	\inference{
	\inference{
	\inference{
		\interp{\shirrmod \setminus \Gamma} \sez_0 \ell : \IN
	}{\interp \Gamma \sez_{n+2} \iota(\ell) : \cohcodisc_{[1, n+2]}^{n+2} \IN}{}
	}{\quotshp \interp \Gamma \sez_{n+2} \iota(\ell) : \cohcodisc_{[1, n+2]}^{n+2} \IN}{}
	}{\cohfget_0^{n+2} \quotshp \interp \Gamma \sez_{n+1} \iota(\ell) : \cohfget_0^{n+2} \cohcodisc_{[1, n+2]}^{n+2} \IN}{}
	}{\cohpi_0^{n+2} \interp \Gamma \sez_{n+1} \iota(\ell) : \cohcodisc_{[0, n+1]}^{n+1} \cohfget_0^{0} \IN}{$\cong$}
	}{\cohpi_0^{n+2} \interp \Gamma \sez_{n+1} \uniDD_{\iota(\ell)} \dtype_{\iota(\suc~\ell)}}{}
	}{\cohpi_0^{n+2} \interp \Gamma \sez_{n+1} \tycode{\uniDD_{\iota(\ell)}} : \uniNDD_{\iota(\suc~\ell)}}{}
	}{\interp \Gamma \sez_{n+2} \iota(\tycode{\uniDD_{\iota(\ell)}}) : \uniDD_{\iota(\suc~\ell)}}{}
\end{equation*}
where $\iota$ is used to denote any no-maze-transformation.

\subsection{Decoding types}
\begin{equation}
	\interpinference{
		\parmod \setminus \Gamma \sez_{n+1} T : \unidepth \ell n
	}{
		\Gamma \sez_n \El~T \type
	}{ty}
	=
	\inference{
	\inference{
	\inference{
		\interp{\parmod \setminus \Gamma} \sez_{n+1} T : \uniDD_\ell
	}{\interp \Gamma \sez_n \iota(T) : \cohfget_0^{n+1} \uniDD_\ell}{}
	}{\interp \Gamma \sez_n \iota(T) : \uniNDD_\ell}{}
	}{\interp \Gamma \sez_n \El~\iota(T) \dtype}{}
\end{equation}

\subsection{Universe cumulativity}
\begin{equation}
	\interpinference{
		\Gamma \sez_{n+1} T : \unidepth \ell n
	}{
		\Gamma \sez_{n+1} T : \unidepth{\suc~\ell}{n}
	}{t-cumul}
	=
	\inference{
		\interp \Gamma \sez_{n+1} \interp T : \uniDD_\ell \qquad
		\inference{
		}{\interp \Gamma \sez_{n+1} \uniDD_\ell \subseteq \uniDD_{\suc~\ell} \dtype}{}
	}{\interp \Gamma \sez_{n+1} \interp T : \uniDD_{\suc~\ell}}{}
\end{equation}

\subsection{Conversion}
The rule
\begin{equation}
	\inference{
		\Gamma \sez_n t : T \\
		\Gamma \sez_n T \jeq T' \type
	}{\Gamma \sez_n t : T'}{t-conv}
\end{equation}
is trivial as $\interp T = \interp{T'}$.

\subsection{Excluded middle}
We need to interpret
\begin{equation}
	\interpinference{
		\hocmod \setminus \Gamma \sez_{n+1} T : \El~\unidepth \ell n
	}{\Gamma \sez_n \name{lem}~T : \El(T \uplus T \to \Empty)}{}
	=
	\inference{
	\inference{
	\inference{
		\interp{\hocmod \setminus \Gamma} \sez_{n+1} \interp T : \uniDD_\ell
	}{\interp \Gamma \sez_{n+1} \iota(\interp T) : \cohdisc_{[1, n+1]}^{n+1} \cohfget_{[1, n+1]}^{n+1} \uniDD_\ell}{}
	}{\interp \Gamma \sez_{n+1} \iota(\interp T) : \cohdisc_{[0, n+1]}^{n+1} \cohfget_{[0, n]}^{n} \uniNDD_\ell}{}
	}{\interp \Gamma \sez_{n+1} \interp{\name{lem}~T} : \El~\iota(\interp T) \uplus \El~\iota(\interp T) \to \Empty}{}
\end{equation}
So pick some $\gamma : \DSub{W}{\interp \Gamma}$. Then we get $W \Dsez_{n+1} \iota(\interp T) \dsub \gamma : \paren{\cohdisc_{[0, n+1]}^{n+1} \cohfget_{[0, n]}^{n} \uniNDD_\ell} \dsub \gamma$. Now $\iota(\interp T) \dsub \gamma$ necessarily arises as some $() \Dsez T' : \uniNDD_{\omega} \dsub{()}$, i.e. it is a closed discrete type. If it is empty, then we can construct $\El~\iota(\interp T) \to \Empty$. If it is not, then it contains a point and hence also degenerate cubes of any shape on that point. Then we can use the axiom of choice to define $\interp{\name{lem}~T} \dsub \gamma$ as the degenerate cube on the chosen point.

\section{Interval}
The rule
\begin{equation}
	\inference{
		\Gamma \sez_1 \Ctx
	}{\Gamma \sez_1 0, 1 : \IX}{}
\end{equation}
is interpreted in the obvious way.

The rules for $\IJud$ are also straightforward to interpret.

\section{$\Pi$-types}

\subsection{Type formation}
For any reshuffle $\mu : m \to n$, we define $\bar \mu : m+1 \to n+1$ by $\bar \mu = \cohfget_0^{n+1} \setminus (\mu \circ \cohfget_0^{m+1})$. This has the property that $\cohfget_0^{n+1} \bar \mu = \mu \cohfget_0^{m+1}$ and that $\bar \mu \cohdisc_0^{m+1} = \cohdisc_0^{n+1} \bar \mu$.
\begin{equation}
	\interpinference{
		\kappa \dashv \mu : m \to n \\
		\bar \mu \setminus \Gamma \sez_{m+1} A : \El~\unidepth \ell m \\
		\Gamma, \ctxresh{(\parmod \setminus \mu)}{x}{\El~A} \sez_{n+1} B : \El~\unidepth \ell n
	}{\Gamma \sez_{n+1} (\ctxresh \mu x A) \to B : \El~\unidepth \ell n}{t-Pi} =
\end{equation}
\begin{equation*}
	\inference{
	\inference{
	\inference{
		\inference{
		\inference{
		\inference{
		\inference{
		\inference{
			\interp{\bar \mu \setminus \Gamma} \sez_{m+1} \interp A : \uniDD_\ell
		}{\interp \Gamma \sez_{n+1} \iota(\interp A) : \bar \mu \uniDD_\ell}{}
		}{\interp \Gamma \sez_{n+1} \iota(\interp A) : \cohdisc_0^{n+1} \mu \uniNDD_\ell}{}
		}{\kappa \cohpi_0^{n+1} \interp \Gamma \sez_n \iota(\interp A) : \uniNDD_\ell}{}
		}{\kappa \cohpi_0^{n+1} \interp \Gamma \sez_n \El~\iota(\interp A) \dtype_\ell}{}
		}{\cohpi_0^{n+1} \interp \Gamma \sez_n (\mu \El~\iota(\interp A))[\iota] \type_\ell}{}
		~
		\inference{
		\inference{
		\inference{
		\inference{
			\interp \Gamma, \var x : \paren{(\cohfget_0^{n+1} \setminus \mu) \El~\iota(\interp A)}[\iota] \sez_{n+1} \interp B : \uniDD_\ell
		}{\cohpi_0^{n+1}  \paren{ \interp \Gamma, \var x : \paren{(\cohfget_0^{n+1} \setminus \mu) \El~\iota(\interp A)}[\iota] } \sez_n \iota(\interp B) : \uniNDD_\ell}{}
		}{\cohpi_0^{n+1} \interp \Gamma, \var x : \cohfget_0^{n+1} \quotshp \paren{(\cohfget_0^{n+1} \setminus \mu) \El~\iota(\interp A)}[\iota] \sez_n \iota(\interp B) : \uniNDD_\ell}{}
		}{\cohpi_0^{n+1} \interp \Gamma, \var x : (\mu \El~\iota(\interp A))[\iota] \sez_n \iota(\interp B) : \uniNDD_\ell}{}
		}{\cohpi_0^{n+1} \interp \Gamma, \var x : (\mu \El~\iota(\interp A))[\iota] \sez_n \El~\iota(\interp B) \dtype_\ell}{}
	}{\cohpi_0^{n+1} \interp \Gamma \sez_n \Pi(\var x : (\mu \El~\iota(\interp A))[\iota]).\El~\iota(\interp B) \dtype_\ell}{}
	}{\cohpi_0^{n+1} \interp \Gamma \sez_n \tycode{\Pi(\var x : (\mu \El~\iota(\interp A))[\iota]).\El~\iota(\interp B)} : \uniNDD_\ell}{}
	}{\interp \Gamma \sez_n \iota \paren{\tycode{\Pi(\var x : (\mu \El~\iota(\interp A))[\iota]).\El~\iota(\interp B)}} : \uniDD_\ell}{}
\end{equation*}
We have used among others the following properties:
\begin{itemize}
	\item $\cohpi_0^{n+1} := \cohfget_0^{n+1} \quotshp$,
	\item $\quotshp(\Gamma.T[\inquotshp]) \cong \quotshp \Gamma.\quotshp T$,
	\item $\mu \leq \cohfget_0^{n+1}(\cohfget_0^{n+1} \setminus \mu)$,
	\item $\uniDD_\ell = \cohdisc_0^{p+1} \uniNDD_\ell$.
\end{itemize}
The function type over the interval is interpreted similarly:
\begin{equation}
	\inference{
		\mu : 1 \to n \\
		\Gamma, \ctxresh{(\parmod \setminus \mu)}{i}{\IX} \sez_{n+1} B : \unidepth \ell n
	}{\Gamma \sez_{n+1} (\ctxresh \mu i \IX) \to B : \unidepth \ell n}{}
\end{equation}

\subsection{Abstraction and application}
Nothing is remarkable about the interpretation of the following rules:
\begin{equation}
	\inference{
		\Gamma, \ctxresh \mu x A \sez_n b : B
	}{\Gamma \sez \lambda(\ctxresh \mu x A).b : (\ctxresh \mu x A) \to B}{t-lam}
\end{equation}
\begin{equation}
	\inference{
		\Gamma, \ctxresh \mu i \IX \sez_n b : B
	}{\Gamma \sez \lambda(\ctxresh \mu i \IX).b : (\ctxresh \mu i \IX) \to B}{}
\end{equation}
\begin{equation*}
	\inference{
		\mu : m \to n \\
		\Gamma \sez_n f : (\ctxresh \mu x A) \to B \\
		\mu \setminus \Gamma \sez_m a : A
	}{\Gamma \sez_n \apresh \mu f a : B[a/x]}{t-app} \quad
	\inference{
		\mu : 1 \to n \\
		\Gamma \sez_n f : (\ctxresh \mu i \IX) \to B \\
		\mu \setminus \Gamma \sez_1 i : \IX
	}{\Gamma \sez_n \apresh \mu f i : B[i/x]}{} \quad
\end{equation*}

\section{$\Sigma$-types}
\subsection{Type formation}
The $\Sigma$-type formation rule
\begin{equation}
	\inference{
		\kappa \dashv \mu : m \to n \\
		\bar \mu \setminus \Gamma \sez_{m+1} A : \El~\unidepth \ell m \\
		\Gamma, \ctxresh{(\parmod \setminus \mu)}{x}{\El~A} \sez_{n+1} B : \El~\unidepth \ell n
	}{\Gamma \sez_{n+1} (\ctxresh \mu x A) \times B : \El~\unidepth \ell n}{t-Sigma}
\end{equation}
is interpreted in almost the same way as the $\Pi$-type formation rule. A crucial difference is that the $\Sigma$-type is not automatically discrete if the codomain is. Hence, we apply $\quotshp$ before moving back from the type judgement into the universe. So we have
\begin{equation}
	\interp{(\ctxresh \mu x A) \times B} = \iota \paren{\tycode{\quotshp \Sigma(\var x : (\mu \El~\iota(\interp A))[\iota]).\El~\iota(\interp B)}}.
\end{equation}
We similarly interpret:
\begin{equation}
	\inference{
		\mu : 1 \to n \\
		\Gamma, \ctxresh{(\parmod \setminus \mu)}{i}{\IX} \sez_{n+1} B : \unidepth \ell n
	}{\Gamma \sez_{n+1} (\ctxresh \mu i \IX) \times B : \unidepth \ell n}{}
\end{equation}
We do not support pair types dual to $\multimap$, because if they can exist at all, they would require special treatment in the context, similar to the Moulin-style interval itself.

\subsection{Pair construction}
\begin{equation}
	\interpinference{
		\mu : m \to n \\
		\mu \setminus \Gamma \sez_m a : \El~A \\
		\Gamma \sez_n b : \El~B[a/x]
	}{\Gamma \sez \pairresh \mu a b : \El~(\ctxresh \mu x A) \times B}{t-pair}
	=
\end{equation}
\begin{equation*}
	\inference{
	\inference{
		\inference{
			\interp{\mu \setminus \Gamma} \sez_m \interp a : \El~\iota(\interp A)
		}{\interp \Gamma \sez_n \iota(\interp a) : (\mu \El~\iota(\interp A))[\iota]}{}
		\quad
		\interp \Gamma \sez_n \interp b : \El~\iota(\interp B) [\iota(\interp a) / \var x]
	}{\interp \Gamma \sez_n (\iota(\interp a), \interp b) : \Sigma(\var x : (\mu \El~\iota(\interp A))[\iota]).\El~\iota(\interp B)}{}
	}{\interp \Gamma \sez_n \inquotshp (\iota(\interp a), \interp b) : \quotshp \Sigma(\var x : (\mu \El~\iota(\interp A))[\iota]).\El~\iota(\interp B)}{}
\end{equation*}
We similarly interpret:
\begin{equation}
	\inference{
		\mu : 1 \to n \\
		\mu \setminus \Gamma \sez_1 k : \IX \\
		\Gamma \sez_n b : B[k/i]
	}{\Gamma \sez \pairresh \mu k b : (\ctxresh \mu i \IX) \times B}{}
\end{equation}

\subsection{Pair elimination}
In order to interpret
\begin{equation}
	\inference{
		\mu : m \to n \qquad \nu : n \to p \\
		\Gamma, \ctxresh{\nu}{z}{\El~(\ctxresh \mu x A) \times B} \sez_p C \type \\
		\Gamma, \ctxresh{\nu \mu}{x}{\El~A}, \ctxresh{\nu}{y}{\El~B} \sez_p c : C[\pairresh \mu x y / z] \\
		\nu \setminus \Gamma \sez_n t : \El~(\ctxresh \mu x A) \times B
	}{\Gamma \sez_p \ind_\times^\nu(z.C, x.y.c, t) : C[t/z]}{t-indpair}
\end{equation}
we rely on substitution and instead interpret
\begin{equation}
	\interpinference{
		\mu : m \to n \qquad \nu : n \to p \\
		\Gamma, \ctxresh{\nu}{z}{\El~(\ctxresh \mu x A) \times B} \sez_p C \type \\
		\Gamma, \ctxresh{\nu \mu}{x}{\El~A}, \ctxresh{\nu}{y}{\El~B} \sez_p c : C[\pairresh \mu x y / z]
	}{\Gamma, \ctxresh{\nu}{z}{\El~(\ctxresh \mu x A) \times B} \sez_p \ind_\times^\nu(z.C, x.y.c, z) : C}{t-indpair} =
\end{equation}
\begin{equation*}
	\inference{
	\inference{
	\inference{
	\inference{
		\interp \Gamma, \var x : (\nu \mu \El~\iota(\interp A))[\iota], \var y : (\nu \El~\iota(\interp B))[\iota] \sez_p \interp c : C[(\var x, \var y)/\var z]
	}{\interp \Gamma, \var z : \paren{\Sigma(\var x : (\nu \mu \El~\iota(\interp A))[\iota]).(\nu \El~\iota(\interp B))[\iota]} \sez_p \interp c [\fst~\var z/\var x, \snd~\var z/\var y] : C}{4}
	}{\interp \Gamma, \var z : \paren{\nu \Sigma(\var x : (\mu \El~\iota(\interp A))[\iota]).\El~\iota(\interp B)}[\iota] \sez_p \interp c [\fst~\var z/\var x, \snd~\var z/\var y][\iota] : C}{3}
	}{\interp \Gamma, \var z : \paren{\quotshp \nu \Sigma(\var x : (\mu \El~\iota(\interp A))[\iota]).\El~\iota(\interp B)}[\iota] \sez_p \interp c [\fst~\var z/\var x, \snd~\var z/\var y][\iota][\inquotshp]\inv : C}{2}
	}{\interp \Gamma, \var z : \paren{\nu \quotshp \Sigma(\var x : (\mu \El~\iota(\interp A))[\iota]).\El~\iota(\interp B)}[\iota] \sez_p \interp c [\fst~\var z/\var x, \snd~\var z/\var y][\iota][\nu \inquotshp]\inv : C}{1}
\end{equation*}
This is best read from bottom to top. In (1), we use \cref{thm:reldtt-modal-tyshp} which gives us a transformation $\iota : \nu \quotshp \to \quotshp \nu$. In (2), we use discreteness of $C$ to get rid of $\quotshp$. In (3), we use \cref{thm:right-adjoint-preserves-types} which states that $\nu$ preserves $\Sigma$-types up to isomorphism. In (4), we simply split up the $\Sigma$-type.

We similarly interpret:
\begin{equation}
	\inference{
		\mu : 1 \to n \qquad \nu : n \to p \\
		\Gamma, \ctxresh{\nu}{z}{\El~(\ctxresh \mu i \IX) \times B} \sez_p C \type \\
		\Gamma, \ctxresh{\nu \mu}{i}{\IX}, \ctxresh{\nu}{y}{\El~B} \sez_p c : C[\pairresh \mu x y / z] \\
		\nu \setminus \Gamma \sez_n t : \El~(\ctxresh \mu i \IX) \times B
	}{\Gamma \sez_p \ind_\times^\nu(z.C, i.y.c, t) : C[t/z]}{}
\end{equation}

\section{Identity types}

\subsection{Type formation}
\begin{equation}
	\interpinference{
		\mu : m \to n \\
		\bar \mu \setminus \Gamma \sez_{m+1} A : \El~\unidepth \ell m \\
		(\parmod \setminus \mu) \setminus \Gamma \sez_m a, b : \El~A
	}{\Gamma \sez_{n+1} \idtpresh \mu A a b : \El~\unidepth \ell {n}}{t-Id} =
\end{equation}
\begin{equation*}
	\inference{
	\inference{
	\inference{
	\inference{
	\inference{
	\inference{
	\inference{
		\interp{(\parmod \setminus \mu) \setminus \Gamma} \sez_m \interp a, \interp b : \El~\iota(\interp A)
	}{\interp \Gamma \sez_{n+1} \iota(\interp a), \iota(\interp b) : ((\cohfget_0^{n+1} \setminus \mu) \El~\iota(\interp A)) [\iota]}{}
	}{\quotshp \interp \Gamma \sez_{n+1} \iota(\interp a), \iota(\interp b) : ((\cohfget_0^{n+1} \setminus \mu) \El~\iota(\interp A)) [\iota]}{}
	}{\cohfget_0^{n+1} \quotshp \interp \Gamma \sez_{n} \iota(\interp a), \iota(\interp b) : \cohfget_0^{n+1} ((\cohfget_0^{n+1} \setminus \mu) \El~\iota(\interp A)) [\iota]}{}
	}{\cohpi_0^{n+1} \interp \Gamma \sez_{n} \iota(\interp a), \iota(\interp b) : (\mu \El~\iota(\interp A)) [\iota]}{}
	}{\cohpi_0^{n+1} \interp \Gamma \sez_{n} \iota(\interp a) \idtp{(\mu \El~\iota(\interp A)) [\iota]} \iota(\interp b) \dtype_0}{}
	}{\cohpi_0^{n+1} \interp \Gamma \sez_{n} \tycode{  \iota(\interp a) \idtp{(\mu \El~\iota(\interp A)) [\iota]} \iota(\interp b)  }: \uniNDD_\ell}{}
	}{\interp \Gamma \sez_{n+1} \iota \paren{  \tycode{  \iota(\interp a) \idtp{(\mu \El~\iota(\interp A)) [\iota]} \iota(\interp b)  }  }: \uniDD_\ell}{}
\end{equation*}
where we have used that $\parmod \circ (\parmod \setminus \mu) = \mu$.

\subsection{Reflexivity}
\begin{equation}
	\interpinference{
		\mu \setminus \Gamma \sez_{m} a : \El~A
	}{\Gamma \sez_{n} \refl^\mu~a : \El~\idtpresh \mu A a a}{}
	=
	\inference{
	\inference{
		\interp{\mu \setminus \Gamma} \sez_m \interp a : \El~\iota(\interp A)
	}{\interp \Gamma \sez_n \iota(\interp a) : (\mu \El~\iota(\interp A))[\iota]}{}
	}{\interp \Gamma \sez_n \refl~\iota(\interp a) : \iota(\interp a) \idtp{(\mu \El~\iota(\interp A))[\iota]} \iota(\interp a)}{}
\end{equation}

\subsection{Reflection rule and UIP}
\begin{equation}
	\interpinference{
		\mu : m \to n \\
		\Gamma \sez_n e : \El~\idtpresh{\mu}{A}{a}{b}
	}{\mu \setminus \Gamma \sez_m a \jeq b : \El~A}{t-eq-rflct} =
	\inference{
	\inference{
		\interp \Gamma \sez_n \interp e : \iota(\interp a) \idtp{(\mu \El~\iota(\interp A))[\iota]} \iota(\interp b)
	}{\interp \Gamma \sez_n \iota(\interp a) = \iota(\interp b) : (\mu \El~\iota(\interp A))[\iota]}{}
	}{\interp{\mu \setminus \Gamma} \sez_m \iota(\interp a) = \iota(\interp b) : \El~\iota(\interp A)}{}
\end{equation}
We also have
\begin{equation}
	\inference{
		\Gamma \sez_n e, e' : \idtpresh \mu A a b
	}{\Gamma \sez_n e \jeq e' : \idtpresh \mu A a b}{t-eq-uip}
\end{equation}
From these two properties, function extensionality

\begin{equation}
	\inference{
		\Gamma \sez_n e : (\ctxresh \nu x A) \to \idtpresh \mu B {\apresh {\mu \setminus \nu} f x}{\apresh {\mu \setminus \nu} g x}
	}{\Gamma \sez_n \funext^\mu e : \idtpresh{\mu}{((\ctxresh{(\mu \setminus \nu)} x A) \to B)} f g}{}
\end{equation}
\begin{equation}
	\inference{
		\Gamma \sez_n e : (\ctxresh \nu i \IX) \to \idtpresh \mu B {\apresh {\mu \setminus \nu} f i}{\apresh {\mu \setminus \nu} g i}
	}{\Gamma \sez_n \funext^\mu e : \idtpresh{\mu}{((\ctxresh{(\mu \setminus \nu)} i \IX) \to B)} f g}{}
\end{equation}
and the J-rule
\begin{equation}
	\inference{
		\mu : m \to n \qquad \nu : n \to p \\
		(\nu \mu) \setminus \Gamma \sez_m a, b : A \\
		\Gamma, \ctxresh{\nu \mu}{y}{A}, \ctxresh{\nu}{w}{\idtpresh \mu A a y} \sez_{p} C \type \\
		\nu \setminus \Gamma \sez_{n} e : \idtpresh \mu A a b \\
		\Gamma \sez_p c : C[a/y, \refl^\mu~a/w]
	}{\Gamma \sez_p \J_{\tyresh \mu A}^\nu(a, b, y.w.C, e, c) : C[b/y, e/w]}{}
\end{equation}
can be derived.
The $\J$-rule proves in particular $(\tyresh{\nu}{\idtpresh \mu A a b}) \to \idtpresh{\nu \mu} A a b$.

\section{Inductive types}
\subsection{Empty and unit types}
These are trivial.

\subsection{Booleans}
The type formation rule and the constructors are trivial. After proof synthesis and substitution, the elimination rule takes the form
\begin{equation}
	\inference{
		\nu : 0 \to p \\
		\Gamma, \ctxresh{\nu}{x}{\El~\Bool} \sez_p C \type \\
		\ctxresh{\nu}{j}{\IJud}, \Gamma \sez_p c : C[\pparen{\true, \false}~j/x]
	}{\Gamma, \ctxresh{\nu}{x}{\El~\Bool} \sez_p \ind_\Bool^\nu(x.C, j.c, x) : C}{t-indBool}
\end{equation}
The interpretation of this rule is trivial, since $\interp{\El~\Bool} = \interp{\IJud} = \Bool$.

\subsection{Natural numbers}
The type formation rule and the constructors are trivial. The elimination rule is straightforward if $\nu = \reshlist{\eqty}{0, \ldots, 0}$. If $\nu = \cohcodisc_{[q+1, p]}^p \circ \reshlist{\eqty}{0, \ldots, 0}$, then we can first create an ad hoc function to $\cohfget_{[q+1, p]}^p C$ and then show that the image is $(q+1)$-codiscrete in $C$ using the synthesized proofs.

\commentout{

}

\section{Path degeneracy}
We need to interpret
\begin{equation}
	\inference{
		\underline 0 := \reshlist{\eqty}{1, \ldots, 1} : 1 \to n \\
		\Gamma \sez_n A \type \\
		\Gamma \sez_n p : (\ctxresh \mu i \IX) \to A
	}{\Gamma \sez_n \degax~p : \idtpresh{\idmod}{((\ctxresh {\underline 0} i \IX) \to A)}{p}{\lambda(\ctxresh {\underline 0} i \IX).\apresh {\underline 0} p 0}}{}
\end{equation}

In fact, the model validates the definitional conclusion $\Gamma \sez_n p = \lambda(\ctxresh {\underline 0} i \IX).\apresh {\underline 0} p 0 : (\ctxresh {\underline 0} i \IX) \to A$, as is immediately clear from discreteness.

\section{Glueing and welding}
We simply sketch the rules without giving semantics.

\subsection{Face predicates}
In the model, we have a subobject $\IX \to \reshlist{\eqty}{1, 1} \IX$, which defines a proposition $\var i : \reshlist{\eqty}{1, 1} \IX \sez_1 (\var i \in \IX) \prop$. We also have a subobject $\cohfget_0^1 \IX \to \cohfget_0^1 \IX \times \cohfget_0^1 \IX$, which defines a proposition $\var i, \var j : \cohfget_0^1 \IX \sez_0 \idpr{\var i}{\var j} \prop$. We have the following rules:
\begin{equation}
	\interp{
		\inference{
			\Gamma \sez_{n+1} \Ctx
		}{
			\Gamma \sez_{n+1} \bot, \top : \IF^n
		}{}
	}
	=
	\inference{
	\inference{
	\inference{
		\interp \Gamma \sez_{n+1} \Ctx
	}{\cohpi_0^{n+1} \interp \Gamma \sez_n \bot, \top \prop}{}
	}{\cohpi_0^{n+1} \interp \Gamma \sez_n \tycode{\bot}, \tycode{\top} : \Prop}{}
	}{\interp \Gamma \sez_n \iota(\tycode{\bot}), \iota(\tycode{\bot}) : \cohdisc_0^{n+1} \Prop}{}
\end{equation}
\begin{equation}
	\interp{
		\inference{
			\Gamma \sez_{n+1} P, Q : \IF^n
		}{
			\Gamma \sez_{n+1} P \vee Q, P \wedge Q : \IF^n
		}{}
	}
	=
	\inference{
	\inference{
	\inference{
	\inference{
	\inference{
		\interp \Gamma \sez_{n+1} P, Q : \cohdisc_0^{n+1} \Prop
	}{\cohpi_0^{n+1} \interp \Gamma \sez_n \iota(P), \iota(Q) : \Prop}{}
	}{\cohpi_0^{n+1} \interp \Gamma \sez_n \El~\iota(P), \El~\iota(Q) \prop}{}
	}{\cohpi_0^{n+1} \interp \Gamma \sez_n \El~\iota(P) \odot \El~\iota(Q) \prop}{$\odot \in \accol{\vee, \wedge}$}
	}{\cohpi_0^{n+1} \interp \Gamma \sez_n \tycode{\El~\iota(P) \odot \El~\iota(Q)} : \Prop}{}
	}{\interp \Gamma \sez_n \iota(\tycode{\El~\iota(P) \odot \El~\iota(Q)}) : \cohdisc_0^{n+1}\Prop}{}
\end{equation}
\begin{equation}
	\interp{
		\inference{
			\strmod \setminus \Gamma \sez_1 s, t : \IX
		}{
			\Gamma \sez_2 \idpr s t : \IF^1
		}{}
	}
	=
	\inference{
	\inference{
	\inference{
	\inference{
	\inference{
	\inference{
		\interp{ \strmod \setminus \Gamma} \sez_1 s, t : \IX
	}{\interp \Gamma \sez_2 \iota(s), \iota(t) : \cohcodisc_0^2 \IX}{}
	}{\quotshp \interp \Gamma \sez_2 \iota(s), \iota(t) : \cohcodisc_0^2 \IX}{$\cohcodisc_0^2 \IX$ is disc.}
	}{\cohfget_0^2 \quotshp \interp \Gamma \sez_1 \iota(s), \iota(t) : \IX}{}
	}{\cohfget_0^2 \quotshp \interp \Gamma \sez_1 {\iota(s)} \idtp{\IX} {\iota(t)} \prop}{}
	}{\cohpi_0^2 \interp \Gamma = \cohfget_0^2 \quotshp \interp \Gamma \sez_1 \tycode{{\iota(s)} \idtp{\IX} {\iota(t)}} : \Prop}{}
	}{\interp \Gamma \sez_1 \iota(\tycode{\idpr{\iota(s)}{\iota(t)}}) : \cohdisc_0^2 \Prop}{}
\end{equation}
\begin{equation}
	\interp{
		\inference{
			\reshlist{\eqty}{0, 1, 1} \setminus \Gamma \sez_1 t : \IX
		}{
			\Gamma \sez_2 (t \in \IX) : \IF^1
		}{}
	}
	=
	\inference{
	\inference{
	\inference{
	\inference{
	\inference{
	\inference{
	\inference{
		\interp{\reshlist{\eqty}{0, 1, 1} \setminus \Gamma} \sez_1 t : \IX
	}{\interp \Gamma \sez_2 \iota(t) : \reshlist{\eqty}{0, 1, 1} \IX}{}
	}{\interp \Gamma \sez_2 \iota(t) : \reshlist{\eqty}{0, 1, 1}}{$0 \cdot \IX \cong (\eqty \cdot \IX)$}
	}{\interp \Gamma \sez_2 \iota(t) : \cohdisc_0^2 \reshlist{\eqty}{1, 1} \IX = \reshlist{\eqty}{\eqty, 1, 1} \IX}{}
	}{\cohpi_0^2 \interp \Gamma \sez_1 \iota(t) : \reshlist{\eqty}{1, 1} \IX}{}
	}{\cohpi_0^2 \interp \Gamma \sez_1 (\iota(t) \in \IX) \prop}{}
	}{\cohpi_0^2 \interp \Gamma \sez_1 \tycode{\iota(t) \in \IX} : \Prop}{}
	}{\interp \Gamma \sez_2 \iota(\tycode{\iota(t) \in \IX}) : \cohdisc_0^2 \Prop}{}
\end{equation}
\begin{equation}
	\interpinference{
		\mu : m \to n \\
		\bar \mu \setminus \Gamma \sez_{m+1} P : \IF^m
	}{
		\Gamma \sez_{n+1} \name{Box}^\mu~P : \IF^n
	}{}
	=
	\inference{
	\inference{
	\inference{
	\inference{
	\inference{
	\inference{
	\inference{
		\kappa \dashv \mu : m \to n \\
		\interp{\bar \mu \setminus \Gamma} \sez_{m+1} P : \cohdisc_0^{m+1} \Prop
	}{\interp{\Gamma} \sez_{n+1} \iota(P) : \bar \mu \cohdisc_0^{m+1} \Prop = \cohdisc_0^{n+1} \mu \Prop}{}
	}{\kappa \cohpi_0^{n+1} \interp \Gamma \sez_{m} \iota(P) : \Prop}{}
	}{\kappa \cohpi_0^{n+1} \interp \Gamma \sez_{n} \El~\iota(P) \prop}{}
	}{\mu \kappa \cohpi_0^{n+1} \interp \Gamma \sez_{n} \mu \El~\iota(P) \prop}{}
	}{\cohpi_0^{n+1} \interp \Gamma \sez_{n} \mu \El~\iota(P) \prop}{}
	}{\cohpi_0^{n+1} \interp \Gamma \sez_{n} \tycode{\mu \El~\iota(P)} : \Prop}{}
	}{\interp \Gamma \sez_{n+1} \tycode{\mu \El~\iota(P)} : \cohdisc_0^{n+1} \Prop}{}
\end{equation}
\begin{equation}
	\interp{
		\inference{
			\parmod \setminus \Gamma \sez_{n+1} P : \IF^n
		}{
			\Gamma \sez_n \El~P \prop
		}{}
	}
	=
	\inference{
	\inference{
	\inference{
		\interp{\parmod \setminus \Gamma} \sez_{n+1} P : \cohdisc_0^{n+1} \Prop
	}{\interp \Gamma \sez_n \iota(P) : \cohfget_0^{n+1} \cohdisc_0^{n+1} \Prop}{}
	}{\interp \Gamma \sez_n \iota(P) : \Prop}{}
	}{\interp \Gamma \sez_n \El~\iota(P) \prop}{}
\end{equation}
\begin{equation}
	\interpinference{
		\Gamma \sez_n \Ctx \\
		\mu : m \to n \\
		\mu \setminus \Gamma \sez_m P \prop
	}{
		\Gamma, \tyresh \mu P \sez_n \Ctx
	}{}
	=
	\inference{
		\interp \Gamma \sez_n \Ctx \\
		\kappa \dashv \mu : m \to n \\
		\kappa \interp \Gamma \sez_m P \prop
	}{\interp \Gamma, \_ : (\mu P)[\iota] \sez_n \Ctx}{}
\end{equation}
Similar to the ordinary context extension rule, one can show that this rule preserves the invariant given by \cref{thm:left-division-semantics}. We can moreover show that the following obvious substitutions are isomorphisms:
\begin{align*}
	\interp{\Gamma, \ctxresh{\inf\accol{\mu, \nu}}{\var k}{\IX}, \Delta[\var k/\var i, \var k/\var j]} &\cong \interp{\Gamma, \ctxresh{\mu}{\var i}{\IX}, \ctxresh{\nu}{\var j}{\IX}, \tyresh{\rho}{\idpr{\var i}{\var j}}, \Delta} \\
	\interp{\Gamma, \Delta[\var k/\var i, 0/\var j]} &\cong \interp{\Gamma, \ctxresh{\mu}{\var i}{\IX}, \tyresh{\rho}{\idpr{\var i}{0}}, \Delta} \\
	\interp{\Gamma, \Delta[\var k/\var i, 1/\var j]} &\cong \interp{\Gamma, \ctxresh{\mu}{\var i}{\IX}, \tyresh{\rho}{\idpr{\var i}{1}}, \Delta} \\
	\interp{\Gamma, \Delta} & \cong \interp{\Gamma, \top, \Delta} \\
	\eset &\cong \interp{\Gamma, \bot, \Delta} \\
	\interp{\Gamma, P, Q, \Delta} &\cong \interp{\Gamma, P \wedge Q, \Delta} \\
	\interp{\Gamma, P, \Delta} \uplus_{\interp{\Gamma, P \wedge Q, \Delta}} \interp{\Gamma, Q, \Delta} &\cong \interp{\Gamma, P \vee Q, \Delta} \\
	\interp{\Gamma, \ctxresh{\inf \accol{\mu, \nu}}{\var i}{\IX}, \Delta} &\cong \interp{\Gamma, \ctxresh{\mu}{\var i}{\IX}, \tyresh{\nu}{(\var i \in \IX)}, \Delta} \\
	\interp{\Gamma, \tyresh{\nu \circ \mu}{P}, \Delta} &\cong \interp{\Gamma, \tyresh{\nu}{\name{Box}^\mu~p}, \Delta}
\end{align*}
This proves soundness of the unification rules.

\subsection{Systems}
When terms are defined on different subobjects of the same context, and they are compatible when the subobjects overlap, then we can paste them together and define a term on the union of all given subobjects.

\subsection{Welding}
We have a rule
\begin{equation}
	\interpinference{
		\Gamma \sez_{n+1} P : \IF^n \\
		\Gamma \sez_{n+1} A : \unidepth \ell n \\
		\Gamma, \tyresh{\strmod}{\El~P} \sez_{n+1} T : \unidepth \ell n \\
		(\strmod \setminus \Gamma), \El~P \sez_n f : \El~(A \to T)
	}{ \Gamma \sez_{n+1} \Weldsys{A}{\Weldsysclauseb{P}{T}{f}} : \unidepth \ell n }{}
\end{equation}
where $\Weldsys{A}{\Weldsysclauseb{\top}{T}{f}} \jeq T$. In order to interpret this, we first refurbish the premises:
\begin{equation}
	\inference{
	\inference{
		\interp \Gamma \sez_{n+1} P : \IF^n
	}{\cohpi_0^{n+1} \interp \Gamma \sez_n \iota(P) : \Prop}{}
	}{\cohpi_0^{n+1} \interp \Gamma \sez_n \El~\iota(P) \prop}{}
	\qquad
	\inference{
	\inference{
		\interp \Gamma \sez_{n+1} A : \uniDD_\ell
	}{\cohpi_0^{n+1} \interp \Gamma \sez_n \iota(A) : \uniNDD_\ell}{}
	}{\cohpi_0^{n+1} \interp \Gamma \sez_n \El~\iota(A) \dtype_\ell}{}
\end{equation}
\begin{equation}
	\inference{
	\inference{
	\inference{
		\interp \Gamma, \_ : \cohcodisc_0^{n+1} \El~\iota(P)[\iota] \sez_{n+1} T : \uniDD_\ell
	}{\cohpi_0^{n+1} (\interp \Gamma, \_ : \cohcodisc_0^{n+1} \El~\iota(P)[\iota]) \sez_n \iota(T) : \uniNDD_\ell }{}
	}{\cohpi_0^{n+1} \interp \Gamma, \_ : \El~\iota(P) \sez_n \iota(T) : \uniNDD_\ell}{\cref{thm:partial-fib-repl}}
	}{\cohpi_0^{n+1} \interp \Gamma, \_ : \El~\iota(P) \sez_n \El~\iota(T) \dtype_\ell}{}
\end{equation}
\begin{equation}
	\inference{
	\inference{
	\inference{
		\interp{\strmod \setminus \Gamma, \El~P} \sez_n f : \El~\iota(A)[\iota] \to \El~\iota(T)[\iota]
	}{\cohfget_0^{n+1} \interp \Gamma, \El~\iota(P)[\iota] \sez_n \iota(f) : \El~\iota(A)[\iota] \to \El~\iota(T)[\iota]}{}
	}{\quotshp\cohfget_0^{n+1} \interp \Gamma, \El~\iota(P)[\iota] \sez_n \iota(f) : \El~\iota(A)[\iota] \to \El~\iota(T)[\iota]}{\cref{thm:partial-fib-repl}}
	}{\cohpi_0^{n+1} \interp \Gamma, \El~\iota(P) \sez_n \iota(f) : \El~\iota(A) \to \El~\iota(T)}{\cref{thm:reldtt-modal-tyshp}}
\end{equation}
This allows us to build, in context $\cohpi_0^{n+1} \interp \Gamma$, the $\Weld$-type, which can then be turned into an element of $\uniNDD_\ell$ in context $\interp \Gamma$.

The following rule is interpreted using the $\weld$ constructor in the model.
\begin{equation}
	\interpinference{
		\Gamma \sez_n a : A
	}{\Gamma \sez_n \weldsys{\sysclauseb{P}{f}} ~ a : \Weldsys{A}{\Weldsysclauseb{P}{T}{f}}}{}
\end{equation}
We also need to model the induction principle:
\begin{equation}
	\interpinference{
		\nu : n \to p, \quad p \cdot \nu < \top \\
		\Gamma, \ctxresh \nu y {\Weldsys{A}{\Weldsysclauseb{P}{T}{f}}} \sez_p C \type \\
		\Gamma, \ctxresh \nu x A \sez c : C[\weldsys{\sysclauseb{P}{f}} ~ x/y] \\
		\Gamma, \tyresh \nu P, \ctxresh \nu y T \sez_p d : C \\
		\Gamma, \tyresh \nu P, \ctxresh \nu x A \sez_p c \jeq d[\apresh \idmod f x / y] : C[\apresh \idmod f x / y]
	}{\Gamma \sez_p \ind_\Weld^\nu(y.C, \sys{\sysclauseb{P}{y.d}}, x.c, b) : C[b/y]}{}
\end{equation}
But the requirement that $p \cdot \nu < \top$ implies that $\nu$ has a right adjoint and is therefore a central reshuffling functor, i.e. a lifted functor, that preserves $\Weld$ on the nose (\cref{thm:lifted-preserves-types}).

\subsection{Glueing}
The interpretation of $\Glue$ is analogous.
\begin{equation}
	\interpinference{
		\Gamma \sez_{n+1} P : \IF^n \\
		\Gamma \sez_{n+1} A : \unidepth \ell n \\
		\Gamma, \tyresh \strmod P \sez_{n+1} T : \unidepth \ell n \\
		(\strmod \setminus \Gamma), P \sez_n f : T \to A
	}{ \Gamma \sez_{n+1} \Gluesys{A}{\Weldsysclauseb{P}{T}{f}} : \unidepth \ell n }{}
\end{equation}
where $\Gluesys{A}{\Weldsysclauseb{\top}{T}{f}} \jeq T$.
\begin{equation}
	\interpinference{
		\Gamma \sez_n b : \Gluesys{A}{\Weldsysclauseb{P}{T}{f}}
	}{\Gamma \sez_n \ungluesys{\sysclauseb{P}{f}}~b : A}{}
\end{equation}
\begin{equation}
	\interpinference{
		\Gamma, P \sez_n t : T \\
		\Gamma \sez_n a : A \\
		\Gamma, P \sez_n \apresh \idmod f t \jeq a
	}{\Gamma \sez_n \gluesys{a}{\sysclauseb P t}}{}
\end{equation}

\commentout{
}

\bibliographystyle{alphaurl}
\bibliography{techrep-refs.bib}

\end{document}